\documentclass[11pt,a4paper,reqno]{amsart}

\usepackage[utf8]{inputenc}
\usepackage[T1]{fontenc}
\usepackage{mathptmx}  
\usepackage{helvet}    
\usepackage{courier}   
\usepackage[foot]{amsaddr}

\usepackage[textwidth = 2cm]{todonotes}
\usepackage{amsmath}
\usepackage{amsfonts}
\usepackage{amssymb}
\usepackage{amsthm}
\usepackage{verbatim}
\usepackage{dsfont}
\usepackage{fullpage}
\usepackage{microtype}
\usepackage[libertine]{newtxmath}
\usepackage{newtxtext}
\allowdisplaybreaks

\usepackage[tt=false]{libertine} 
\usepackage{caption}
\usepackage{bm}
\usepackage{bbm}
\usepackage{hyperref, color}
\hypersetup{colorlinks=true,citecolor=blue, linkcolor=blue, urlcolor=blue}
\usepackage[linesnumbered,boxed,ruled,vlined]{algorithm2e}
\usepackage{algorithmicx}
\usepackage[numbers]{natbib}
\usepackage{xcolor}
\usepackage{enumerate} 
\usepackage[shortlabels]{enumitem}
\usepackage{tabularx}
\usepackage{array}
\usepackage{cleveref}
\usepackage{pifont}
\usepackage{xifthen}
\usepackage{xstring}

\usepackage{stackengine}
\usepackage{scalerel}
\usepackage{graphicx}

\DeclareMathAlphabet{\mathcal}{OMS}{cmsy}{m}{n}

\newcolumntype{L}[1]{>{\raggedright\arraybackslash}p{#1}}
\newcolumntype{C}[1]{>{\centering\arraybackslash}m{#1}}
\newcolumntype{R}[1]{>{\raggedleft\arraybackslash}p{#1}}

\usepackage{makecell}
\usepackage{footnote}
\makesavenoteenv{tabular}

\newtheorem{theorem}{Theorem}[section]
\newtheorem{observation}[theorem]{Observation}

\newtheorem*{claim*}{Claim}
\newtheorem{condition}[theorem]{Condition}

\newtheorem{lemma}[theorem]{Lemma}
\newtheorem{proposition}[theorem]{Proposition}
\newtheorem{corollary}[theorem]{Corollary}
\theoremstyle{definition}

\newtheorem{definition}[theorem]{Definition}
\newtheorem{remark}[theorem]{Remark}
\newtheorem*{remark*}{Remark}

\def\^#1{\mathbb{#1}} 
\def\*#1{\mathbf{#1}} 
\def\+#1{\mathcal{#1}} 
\def\-#1{\mathrm{#1}} 
\def\=#1{\boldsymbol{#1}} 
\def\!#1{\mathtt{#1}}
\def\@#1{\mathscr{#1}}

\newcommand{\abs}[1]{\left | #1 \right |}
\newcommand{\set}[1]{\left\{#1\right\}}

\newcommand{\id}[1]{\ensuremath{\mathbbm{1}}\left[#1\right]}

\newcommand{\floor}[1]{\left\lfloor #1 \right\rfloor}
\newcommand{\ceil}[1]{\left\lceil #1 \right\rceil}
\renewcommand{\Pr}[2][]{ \ifthenelse{\isempty{#1}}
  {\mathop{\mathbf{Pr}}\left[#2\right]} {\mathop{\mathbf{Pr}}_{#1} \left[#2\right]} }
\newcommand{\E}[2][]{ \ifthenelse{\isempty{#1}}
  {\mathbf{E}\left[#2\right]}
  {\mathbf{E}_{#1}\left[#2\right]} }

\newcommand{\cstarfro}[1]{\@C^{#1}_{\star \text{-} \mathsf{frozen}}}
\newcommand{\sstarfro}[1]{\@P^{#1}_{\star \text{-} \mathsf{frozen}}}
\newcommand{\cstarbad}[1]{\@C^{#1}_{\star \text{-} \mathsf{bad}}}
\newcommand{\sstarbad}[1]{\@P^{#1}_{\star \text{-} \mathsf{bad}}}
\newcommand{\vefro}[1]{V^{#1}_{\@P-\mathsf{frozen}}}

\newcommand{\cstara}[1]{\@C^{#1}_{\times,1}}
\newcommand{\cstarb}[1]{\@C^{#1}_{\times,2}}
\newcommand{\sstar}[1]{\@S^{#1}_{\star}}

\newcommand{\csinf}[1]{\@C^{#1}_{\star \text{-} \mathsf{inf}}}
\newcommand{\cvinf}[1]{\@C^{#1}_{v \text{-} \mathsf{inf}}}
\newcommand{\ssinf}[1]{\@S^{#1}_{\star \text{-} \mathsf{inf}}}
\newcommand{\cscon}[1]{\@C^{#1}_{\star \text{-} \mathsf{con}}}
\newcommand{\cspformula}{PDC\xspace}
\newcommand{\csppformula}{PDC\xspace}

\newcommand{\perman}{\mathsf{per}}

\newcommand{\poly}{\textnormal{\textsf{poly}}}

\newcommand{\pname}{permutation set\xspace}
\newcommand{\pnames}{permutation sets\xspace}

\newcommand{\ol}{\overline}

\newcommand{\LLL}{Lov{\'a}sz Local Lemma\xspace}

\newcommand{\vbl}{\mathsf{vbl}}
\newcommand{\e}{\textnormal{e}}

\newcommand{\True}{\!{True}}
\newcommand{\False}{\!{False}}
\newcommand{\tmix}{\tau_{\mathrm{mix}}}
\newcommand{\dtv}{d_{\mathrm{TV}}}
\newcommand{\bigmid}{\; \big \lvert \;}

\newcommand{\MCMC}{\textnormal{\textsf{MCMC}}}

\newcommand{\truncatedsampling}{\textnormal{\textsf{RejectionSampling}}}
\newcommand{\samplepermutation}{\textnormal{\textsf{Sample}}}
\newcommand{\countpermutation}{\textnormal{\textsf{Count}}}
\newcommand{\idealsamplepermutation}{\textnormal{\textsf{IdealSample}}}

\newcommand{\qgl}[1]{{\color{red}{#1}}}

\newcommand{\sampleblocksize}{\Delta \log (n/\varepsilon) }

\newcommand{\lipschitzconstant}{k\conditiondomainsize\Delta^2\log (n/\varepsilon)}
\newcommand{\conditiondomainsize}{L}

\newcommand{\coupleM}{\@M^{\circ}}
\newcommand{\cfro}[1]{\@C_{\mathrm{frozen}}{#1}}
\newcommand{\cfronull}{\@C_{\mathrm{frozen}}}
\newcommand{\cdif}[1]{\@C_{\mathrm{diff}}{#1}}
\newcommand{\cdifnull}{\@C_{\mathrm{diff}}}
\newcommand{\cdiscre}[1]{\@C_{\mathrm{discrep}}{#1}}
\newcommand{\cdiscrenull}{\@C_{\mathrm{discrep}}}
\newcommand{\pfro}[1]{\@P_{\mathrm{frozen}}{#1}}
\newcommand{\pfronull}{\@P_{\mathrm{frozen}}}
\newcommand{\mfro}[1]{\@M_{\mathrm{frozen}}{#1}}
\newcommand{\mfronull}{\@M_{\mathrm{frozen}}}
\newcommand{\mafro}[1]{\@M_{1,\mathrm{frozen}}{#1}}
\newcommand{\mafronull}{\@M_{1,\mathrm{frozen}}}
\newcommand{\mbfro}[1]{\@M_{2,\mathrm{frozen}}{#1}}
\newcommand{\mbfronull}{\@M_{2,\mathrm{frozen}}}

\newcommand{\vset}{V_{\mathrm{set}}}
\newcommand{\pcurrent}{\+P^{\dagger}}

\newcommand{\Succ}{\mathsf{Succ}}

\newcommand{\Dis}{\!{Dis}}
\newcommand{\pdisc}{\+P_{\bot}}
\newcommand{\pzeta}{\+P_{\theta}}
\newcommand{\pdisci}{\+P^{\ast}_{\bot}}
\newcommand{\pdiff}{\+P_{u}}
\newcommand{\pdiffinit}{\+P_{0}}

\newcommand{\cdisc}{\+C_{\bot}}
\newcommand{\csmall}{\+C_{s}}
\newcommand{\psmall}{\mathcal{Z}}
\newcommand{\cremain}{\+C_{u}}

\newcommand{\pfinal}{{\+{L}}}
\newcommand{\pinit}{{\+{I}}}

\newcommand{\pdecomex}{\+P'}

\newcommand{\qtmpone}{Z_1}
\newcommand{\qtmptwo}{Z_2}
\newcommand{\sizeuperbound}{\gamma^{\theta}}
\newcommand{\sizeuperboundsquare}{\gamma^{2\theta}}

\newcommand{\sizeuperboundpfinal}{\alpha}
\newcommand{\sizeuperboundlp}{\gamma}
\newcommand{\weightfunction}{\widehat{f}}
\newcommand{\weightconstant}{\zeta}

\newcommand{\inputtrivial}{\Phi_1, \Phi_2,\pdecomex,P}

\newcommand{\inputonestep}{\Phi_1, \Phi_2,\pdecomex,P,P_i,t}

\newcommand{\easycase}{{easy-to-handled factorized formulas}\xspace}
\newcommand{\hardcase}{{challenging factorized formulas}\xspace}

\newcommand{\initialcouple}
{\textnormal{\textsf{InitialCouple}}}

\newcommand{\onestep}{{\mathsf{CP}}}
\newcommand{\cpperm}{\mathsf{CP}_{\mathsf{perm}}}
\newcommand{\couple}{\mathsf{Couple}}
\newcommand{\supp}{\!{supp}}

\newcommand{\distcouple}{\mathsf{dist}_{\@M}}

\newcommand{\edgepermutation}{\+E_{P}}
\newcommand{\edgeconstraint}{\+E_C}
\newcommand{\propagate}{propagation trajectory\xspace}

\newcommand{\witnesswrtp}[1]{\+W[#1]}
\newcommand{\pref}{\mathsf{suf}\xspace}

\title{Sampling permutations satisfying constraints \\within the lopsided local lemma regime}

\author{Kun He}
\address[Kun He]{Renmin University of China. \textnormal{E-mail: \url{hekun.threebody@foxmail.com}}}

\author{Guoliang Qiu}
\address[Guoliang Qiu]{Shanghai Jiao Tong University, China. \textnormal{E-mail: \url{guoliang.qiu@sjtu.edu.cn}}}

\author{Xiaoming Sun}
\address[Xiaoming Sun]{ICT, Chinese Academy of Sciences, China. \textnormal{E-mail: \url{sunxiaoming@ict.ac.cn}}}

\begin{document}

\begin{abstract}
Sampling a random permutation with restricted positions, or equivalently approximating the permanent of a 0-1 matrix, is a fundamental problem in computer science, with several notable results achieved over the years.
However, existing algorithms typically exhibit high computational complexity. 
Achieving the optimal running time remains elusive, even for nontrivial subsets of the problem.
Furthermore, existing algorithms primarily focus on a \emph{single} permutation, leaving many combinatorial problems involving \emph{multiple} constrained permutations unaddressed.

For a single permutation, we achieve the optimal running time $O(n^2)$ for approximating the permanent of a very dense $n \times n$ 0-1 matrix, where each row and column contains at most $\sqrt{(n-2)/20}$ zeros. This result serves as a fundamental building block in our sampling algorithm for multiple permutations.

We further introduce a general model called \emph{permutations with disjunctive constraints} (PDC) for handling multiple constrained permutations. We propose a novel Markov chain-based algorithm for sampling nearly uniform solutions of PDC within a lopsided Lov{\'a}sz Local Lemma (LLL) regime. For uniform PDC formulas, where all constraints are of the same width and all permutations are of the same size, our algorithm runs in nearly linear time with respect to the number of variables.

Previous approaches for sampling LLL relied on the variable model, where the underlying probability space is a product space. In contrast, the sampling problem of PDC encounters a fundamental challenge: the random variables within each permutation in the joint probability space are \emph{not} mutually independent, leading to long-range correlations that previous factorization techniques cannot cut down.
To tackle this challenge, we introduce a novel sampling framework called \emph{correlated factorization} and a new concept in the path coupling analysis, termed the \emph{inactive vertex}. 
As a result, a sampling LLL beyond the variable model is established. 
\end{abstract}

\maketitle
\setcounter{tocdepth}{1}


\tableofcontents

\setcounter{page}{0} \thispagestyle{empty} \vfill
\pagebreak

\section{Introduction}\label{sec-intro}
\medskip
\noindent
Given a positive integer $n$, 
sampling all the permutations of $[n]$ uniformly at random is one of the most fundamental problems in probability theory.
This problem, usually known as card shuffling for mathematicians,
has a long history~\cite{renyi1984diary}, and has led to numerous profound results~\cite{aldous1983random,aldous1986shuffling,diaconis1988group,bayer1992trailing}.
Compared to the original card shuffling problem, which involves no specific constraints on the permutations, the sampling of permutations that satisfy various constraints has also attracted significant attention in the fields of probability, statistics, and computer science.
Numerous important problems can be viewed as specialized instances of the problem of sampling permutations with constraints. The following are some well-studied examples:
\noindent
\begin{itemize}
\item \textbf{sampling permutations with restricted positions}.
An instance of \emph{permutations with restricted positions} (PRP) includes a positive integer $n$ and a collection of sets $R_1,R_2,\ldots,R_n\subseteq [n]$. The sampling of PRP is the problem of sampling permutations of $[n]$ uniformly at random from $\prod_{i\in [n]}R_i$.

\item \textbf{sampling perfect matchings of a bipartite graph}. Given a bipartite graph with $n$ vertices in each of the two partitions, a perfect matching is a subset of edges such that each vertex is adjacent to exactly one edge. 

\item \textbf{approximating the permanent of a 0-1 matrix.} 
Given a 0-1 square matrix $A = (a_{i,j})_{n\times n}$ where $a_{i,j}\in \set{0,1}$ for any $i,j\in [n]$, the permanent of $A$ is defined as 
\[\!{per}(A) \triangleq \sum_{\sigma }\prod_{i\in [n]}a_{i,\sigma(i)},\]
where the sum is over all permutations $\sigma$ of $[n]$.
Computing the permanent of a matrix is one of the first problems shown to be $\#$P-complete~\cite{valiant1979complexity}, even if the matrix is a 0-1 matrix where the row and column sums are at least $n/2$. 
By the self-reducibility property of the problem~\cite{jerrum1986random}, any polynomial time algorithm for generating a random permutation $\sigma$ where $\prod_{i\in [n]}a_{i,\sigma(i)}\neq 0$ can be used to efficiently approximate $\!{per}(A)$.
\end{itemize}
The equivalence of these three problems is well-known in the literature~\cite{jerrum1986random}.

Several effective algorithms have been introduced for both sampling PRP and approximating the permanent of a 0-1 matrix. 
Broder~\cite{broder1986hard} initially devised a Markov chain method based on switching for PRP sampling and permanent approximation. He also proved that, even for \(\gamma\)-dense 0-1 matrices where the row and column sums are at least \(\gamma n\) with a fixed \(\gamma \in (0,1)\), calculating the permanent exactly remains $\#$P-hard.
Later, Jerrum and Sinclair demonstrated that Broder's Markov chain operates in \(O(n^8 \log n)\) time for \(1/2\)-dense matrices~\cite{jerrum1989approximating}. 
Jerrum, Sinclair, and Vigoda introduced another Markov chain where the parameters were fine-tuned to approach the target distribution \cite{jerrum2004polynomial}. This innovation led to an algorithm capable of approximating the permanent in \(O(n^{10}(\log n)^3)\) time, establishing the first fully polynomial randomized approximation scheme (FPRAS) for the permanent of matrices with nonnegative entries. Bezáková et al. later enhanced this time complexity to \(O(n^7(\log n)^4)\) \cite{bezakova2008accelerating}.
For \(\gamma\)-dense cases where \(\gamma > 1/2\), Huber and Law introduced a more refined sampling algorithm for PRPs with running time \(O\left(n^4 \log n + n^{2+(1-\gamma)/(2\gamma - 1)}\right)\) \cite{huber2008fast}.

Beyond single permutations with restricted positions, many combinatorial problems can be modeled as multiple permutations with constraints. Examples include Latin squares~\cite{keedwell2015latin}, independent transversals and perfect matchings in multipartite hypergraphs~\cite{erdHos1994independent, aharoni2009perfect}, the strong chromatic number~\cite{HS17}, and hypergraph packing~\cite{LS07}.
Furthermore, the sampling problem of multiple permutations with constraints arises naturally in resource-allocation settings (see \Cref{sec-application-intro}).
There are few results on sampling multiple permutations with constraints, and all are limited to specific problems. McKay and Wormald developed an algorithm based on ``switchings'' to generate random Latin rectangles~\cite{mckay1991uniform}. Cooper, Dyer, and Handley introduced the first rigorous polynomial mixing-time bound, $O(n^{10}\log n)$, for natural Markov chains that sample discordant permutations and Latin rectangles~\cite{cooper2011networks}.

Multiple permutations satisfying given constraints can be captured within the framework of constraint satisfaction problems (CSP).
Formally, given any positive integer $m$,
an instance of CSP on $m$ permutations
is denoted by $\Phi = \left(\+P,\+Q, \+C \right)$,
where $(\+{P}, \+Q)$ specifies the 
permutations and $\+C$ specifies the constraints. 
In this setup, \(\+P = (P_1, \ldots, P_m)\) is a list of variable sets, with $P_1, \ldots, P_m$
being disjoint sets. Similarly, \(\+Q = (Q_1, \ldots, Q_m)\) is a list of sets of numbers, where each \(Q_i\) has the same size as its corresponding \(P_i\).
Intuitively, \(P_i\) represents the variables in one permutation, and \(Q_i\) is the set of values those variables can take. An assignment of values to all the variables in $V\triangleq P_1\cup \cdots \cup P_m$ is called \emph{valid} if, for each \(i\), the values assigned to the variables in \(P_i\) form a permutation of the values in \(Q_i\).
The constraints $\+C$ define the conditions that any valid assignment must satisfy.
Specifically, we focus on disjunctive constraints as in SAT, where each constraint $C\in \+C$ has the form  
$(v_1\neq c_1) \lor \cdots \lor (v_{\ell}\neq c_{\ell})$, with $\ell>0$; here, each $v_i$ is a variable and $c_i$ is a value it can take.
A \emph{satisfying} assignment for $\Phi$ is a valid one that meets all the constraints in $\+C$.
Two constraints $C = \cdots \lor (v\neq c)\lor \cdots$ and $C'=\cdots \lor (v'\neq c')\lor \cdots$ are called \emph{related} (denoted as $C\sim C'$)
if either $v=v'$, or $v$ and $v'$ belong to the same set $P_i$ and $c=c'$.
In the context of the lopsided \emph{Lov$\acute{a}$sz Local Lemma} (LLL), this relation reflects a negative dependence between the events of $C$ and $C'$ being violated~\cite{HS17}.
In this paper, we focus on formulas for \emph{permutations with disjunctive constraints} (PDC).

To illustrate the PDC framework, consider the problem of uniformly sampling perfect matchings in a 3-partite hypergraph. Let $H=(V_1,V_2,V_3,E)$ with $V_i=\{v^i_1,\dots,v^i_q\}$ for $i=1,2,3$, where each hyperedge $e\in E\subseteq V_1\times V_2\times V_3$ contains one vertex from each part. A perfect matching is a subset $S\subseteq E$ that covers every vertex exactly once. Encode this sampling task as a PDC instance $\Phi=(\mathcal P,\mathcal Q,\mathcal C)$ with $\mathcal P=(P_1,P_2)$, where $P_1=V_1$ and $P_2=V_2$, and $\mathcal Q=(Q_1,Q_2)$ with $Q_1=Q_2=[q]$. For every $i,j,k\in[q]$ such that $(v^1_i,v^2_j,v^3_k)\notin E$, include the constraint $(v^1_i\neq k)\lor(v^2_j\neq k)$. Thus, if $\sigma$ is a satisfying assignment of $\Phi$ in which $(v^1_i = k)$ and $(v^2_j = k)$ for some $k\in[q]$, we must have $(v^1_i,v^2_j,v^3_k)\in E$. Because each set $P_\ell$ where $\ell\in \{1,2\}$ must be assigned a permutation of $Q_\ell$, any satisfying assignment $\sigma$ yields exactly $q$ such triples, one for each $k$, and these triples form a perfect matching of $H$; the converse also holds.

Given a PDC formula $\Phi$, the following are its key parameters:
\begin{itemize}
    \item \emph{width} $k=k_{\Phi}\triangleq \max_{C\in \+C}\abs{\vbl(C)}$ where $\vbl(C)$ denotes the set of variables used by $C$;
 
    \item \emph{variable degree} $d = d_{\Phi} \triangleq \max_{v\in V}\abs{\{C\in \+{C}\mid v\in \vbl(C)\}}$;
    \item \emph{constraint degree} $\Delta = \Delta_{\Phi} \triangleq \max_{C\in \+{C}}\abs{\{C'\in \+{C}\mid C'\sim C\}}$;\footnote{The constraint degree $\Delta$ should be distinguished from the \emph{lopsidependency degree} $D$, which is the maximum degree of the lopsidependency graph: $D\triangleq\max_{C\in \+{C}}\abs{\{C'\in \+{C}\setminus\{C\}\mid C\sim C'\}}$. Note that $\Delta=D+1$.}
     \item \emph{minimal \pname size} $q = q_{\Phi} \triangleq \min_{P\in \+P}\abs{P}$;
    \item \emph{violation probability} $p=p_{\Phi}=\max_{C\in \+{C}}\^P[\neg C]$, where $\^P \triangleq \^P_\Phi$ denotes the uniform distribution over all valid assignments, and $\neg C$ denotes the event that $C$ is violated.
\end{itemize}
Currently, no effective algorithm is available for sampling PDC.

\subsection{Sampling LLL}\label{sec-previous-works}
Uniformly sampling a random solution that satisfies a set of constraints has garnered significant attention in recent years. The Lovász Local Lemma (LLL) is a powerful tool for sampling constraint-satisfying solutions, demonstrating the feasibility of avoiding all ``bad" events under certain weak dependency conditions~\cite{EL75}. It has been shown that CSP sampling is tractable only within an LLL-like regime \( p \Delta^2 \gtrsim 1 \), where \(\gtrsim\) disregards lower-order terms and constant factors~\cite{BGGGS19, galanis2023inapproximability}. Efficient algorithms for near-uniform sampling of solutions in this regime are collectively known in the literature as the \emph{sampling LLL}.
Recent progress in this field is notable~\cite{GJL19, Moi19, guo2019counting, galanis2019counting, harris2020new, FGYZ20, feng2021sampling, Vishesh21sampling, Vishesh21towards, HSW21, galanis2021inapproximability, feng2022improved, he2022sampling, he2023deterministic, WY24}. Prior approaches to the sampling LLL are based on product probability spaces defined by sets of mutually independent random variables, commonly known as the variable model~\cite{MT10, Beyond}. However, the sampling problem for PDC extends \emph{beyond} this variable model. In PDC, the underlying probability space is a \emph{joint} space in which the random variables associated with each permutation are not mutually independent.
This inherent dependency among variables makes PDC sampling particularly challenging.

To understand the technical barrier in the sampling of PDC, let’s examine the sampling of PRP, a specific case of PDC involving a single permutation. Existing fast algorithms for the sampling LLL leverage formula factorization. 
Given a CSP instance \( \Phi \), if a subset of \( \Phi \)’s variables is correctly assigned, a proportion of the constraints will be satisfied. By removing all satisfied constraints, the remaining CSP can be factorized into disjoint sub-formulas of logarithmic size, making sampling on the reduced formula straightforward.
However, this approach does not apply to PRP. 
A PRP cannot be factorized into smaller instances even if most constraints are satisfied, as its variables must collectively form a single permutation. 
Similarly, for PDC formulas involving multiple permutations, the techniques used in the sampling LLL also fail due to the interdependence of the large permutations.

To characterize the joint probability space of PDC, a more suitable tool is the lopsided LLL~\cite{erdHos1991lopsided}, which can yield improved bounds on PDC formulas compared to the standard LLL. The lopsided LLL is a variant of the LLL that applies to \emph{lopsidependent} ``bad" events exhibiting negative correlation. Specifically, given a collection \(\mathcal{B}\) of bad events in a probability space, for any event \(A\) defined on that space and any subset \(S \subseteq \mathcal{B}\), we say that \(A\) is non-lopsidependent with respect to \(S\) if
\begin{align*}
    \Pr{A\; \big \vert \; \bigwedge_{B\in S} \overline{B}} \le \Pr{A}.
\end{align*}
In other words, the events in \(S\) are either independent of or positively correlated with \(A\). Furthermore, the lopsidependency among events can be represented by the \emph{lopsidependency graph}, an undirected graph \(G = (\mathcal{B}, E)\) that satisfies the following property: for any \(B \in \mathcal{B}\) and any subset \(S \subseteq \mathcal{B} \setminus \bigl(\Gamma(B) \cup \{B\}\bigr)\), where \(\Gamma(B)\) denotes the neighbors of \(B\) in \(G\), the event \(B\) is non-lopsidependent with respect to \(S\). 
The lopsided LLL then provides a condition ensuring that the probability of none of the events in \(\mathcal{B}\) occurring is positive (see \Cref{thm-lopsiLLL}).
As a fundamental tool in probabilistic combinatorics,
the lopsided LLL underpins a wide range of applications~\cite{erdHos1991lopsided, HarrisS14a}. However, research within this framework has so far been limited to existential and constructive results, leaving a significant gap: sampling algorithms tailored to the lopsided LLL condition remain elusive. The design of a PDC sampling algorithm can offer critical insights into developing sampling algorithms in the lopsided LLL context. Such a breakthrough would not only address this long-standing deficiency but also pave the way for a wealth of new techniques, greatly advancing the sampling applications based on the lopsided LLL.

In summary, two related problems in the sampling of constrained permutations need to be addressed:
\begin{itemize}
\item \emph{Optimal algorithm for the sampling of permutations with constraints.} 
Current sampling algorithms primarily rely on the Markov chain method. To establish rapid mixing of these Markov chains, the canonical path approach is often used, but this tends to yield loose bounds on running time. The current best algorithm for approximating the permanent has a running time of \( O(n^7 (\log n)^4) \)~\cite{bezakova2008accelerating}. Even for \( \gamma \)-dense PRPs, the novel rejection sampling algorithm requires \( O\left( n^4 \log n + n^{2 + (1-\gamma)/(2\gamma - 1)} \right) \)~\cite{huber2008fast}. 
The running times of these algorithms remain far from optimal. Developing optimal-time sampling algorithms for even a nontrivial subset of the problem would be highly significant, as it could offer valuable insights into the strategies that an optimal algorithm should employ.

\item \emph{Efficient algorithm for the sampling of multiple permutations.}
Existing sampling algorithms for multiple permutations with constraints are limited to specific problems, and no efficient sampling algorithm currently exists for the general PDC model, which captures a wide range of important combinatorial problems. 
Developing an efficient sampling algorithm for PDC would address the sampling challenges of these combinatorial objects by providing a robust method for sampling   permutations with complex constraints. 
Furthermore, such an algorithm would represent a crucial step in extending the sampling LLL from the standard LLL condition to the more general lopsided LLL, thereby broadening the scope of sampling techniques to accommodate intricate dependency structures. Consequently, an efficient sampling algorithm for PDC remains highly desirable.
\end{itemize}

Indeed, prior to our work, no nontrivial class of PRPs was known for which the sampling algorithm achieved optimal running time.
Additionally, it was not known whether the sampling problem for PDC is polynomial-time tractable under a lopsided LLL condition.

\subsection{Main results}
In this paper, we resolve the aforementioned open questions. 
For a single permutation, we propose a fast sampling algorithm that runs in optimal \(O(n^2)\) time for a specific class of PRPs. Moreover, this algorithm constitutes a key technique in our approach to sampling PDC formulas. For multiple permutations, we introduce a fast algorithm for sampling satisfying assignments of PDC formulas in the lopsided LLL regime. 
This is the first sampling LLL beyond the variable model.
Furthermore, our method achieves nearly linear runtime for uniform PDC formulas, where all constraints have the same width and all permutations are of equal size.

\subsection*{Results on a single permutation}
Let $\gamma\in (0,1]$.
Given any \(n \times n\) 0-1 matrix, we define it as \emph{very dense} if each row sum and each column sum are greater than \(n - \sqrt{(n-2)/20}\). Similarly, for any integer \(n \geq 2\) and any sets \(R_1, \dots, R_n \subseteq [n]\), we call the instance \emph{very dense} if \(\abs{R_i} \geq n - \sqrt{(n-2)/20}\) for each \(i \in [n]\), and \(\abs{\{k \in [n] \mid j \in R_k\}} \geq n - \sqrt{(n-2)/20}\) for each \(j \in [n]\).

Our first result for very dense 0-1 matrices provides a tight approximation of the permanent.
For any $x\geq 0$ and $y>0$, define 
\begin{equation}\label{eq-define-gxy}
\begin{aligned}
g(x,y)\triangleq 2\pi y \cdot \exp\left(\frac{1}{3y}\right) \cdot  \left(\frac{y}{\e}  \right)^{y}\cdot  \left(1-\frac{x}{y^2} \right)^{y}.
\end{aligned}
\end{equation}
Then we have the following theorem.
\begin{theorem}\label{thm-esitimator-result}
For any 0-1 matrix $A$ of size $n\times n$, let $\rho$ denote the total number of $0$s in $A$. If each row sum and each column sum are larger than $n - \sqrt{(n-2)/20}$, then we have 
\begin{align*}
\frac{g(\rho,n)}{\sqrt{2\pi n} \;\e^2} \leq \perman(A) \leq  19\cdot\frac{g(\rho,n)}{\sqrt{2\pi n} \;\e^2}.
\end{align*}
\end{theorem}
This theorem is restated as \Cref{thm-estimator-num-prp} in the context of PRP. 
It serves as a basic ingredient for our sampling and approximate counting algorithms for very dense PRP, just as in Theorems \ref{thm-sampler-PRP-LLL} and \ref{thm-estimator-PRP-LLL}.
The bound we establish in this theorem is remarkably sharp, tight up to a constant factor.
The conditions on the row sum and the column sum guarantee that the degree of the lopsidependent graph is bounded such that the lopsided LLL is applicable (see \Cref{sec-lll}).
Using the lopsided LLL, we achieve a good approximation of the permanent, accurate to within a constant factor. 
This conclusion appears somewhat surprising
because the lower bound provided by LLL is typically considered to be loose.

Based on \Cref{thm-esitimator-result},
we have the following exact sampler for PRP in the lopsided LLL regime.

\begin{theorem}[Sampler for very dense PRP in the lopsided LLL regime]\label{thm-sampler-PRP-LLL}
There exists an exact sampler such that given as input any integer {$n \geq 2$} and any $R_1,\cdots,R_n\subseteq [n]$ where {$\abs{R_i} \geq n - \sqrt{(n-2)/20}$} for each $i\in [n]$
and $\abs{\{k\in [n]\mid j\in R_k\}} \geq n - \sqrt{(n-2)/20}$ for each $j\in [n]$, it outputs a random permutation uniformly from $\prod_{i\in [n]}R_i$, in expected running time $O(n^{2})$. In addition, there is also an approximate sampler such that given as input any $\varepsilon\in (0,1)$ with an instance satisfying the same condition, it outputs random permutations with total variance distance of $\varepsilon$ from the uniform distribution, in running time $O(n^{2}\cdot \log (1/\varepsilon))$.
\end{theorem}

Note that for the input described in the preceding theorem, any sampling algorithm for a very dense PRP necessarily takes at least $\Omega(n^2)$ time, since simply reading the input already costs $\Omega(n^2)$.
Therefore, our algorithm attains an optimal running time.
Previously, the best sampling algorithm for very dense PRP had a running time of $\tilde{O}\left(n^4\right)$~\cite{huber2008fast}. 
\Cref{thm-sampler-PRP-LLL} provides a significant speedup for these instances.
We also obtain a fast approximate counting algorithm for very dense PRP, which is stated in \Cref{thm-estimator-PRP-LLL}.
Theorems \ref{thm-sampler-PRP-LLL} and \ref{thm-estimator-PRP-LLL} are basic ingredients for our sampler of PDC.

\subsection*{Results on multiple permutations}
Recall the parameters $p,\Delta,k,q$ of $\Phi$. 
A PDC formula $\Phi$ is $(k,q)$-uniform if all constraints have the same length $k$ and all permutations have the same size $q$. We have the following theorem for $(k,q)$-uniform PDC formulas.

\begin{theorem} \label{theo:unif}
The following holds for some positive constants $q_{\min}$ and $c$.
There is an algorithm such that given as input any $\varepsilon\in (0,1)$ and any $(k,q)$-uniform PDC formula $\Phi$ with $n$ variables where $\Phi$ satisfies $q\geq q_{\min}$ and 
    \[k\ge 24, \quad q^k\geq c k^{24}\Delta^{32},\]
it outputs random valid assignments of $\Phi$ with total variance distance of $\varepsilon$ from the uniform distribution of all satisfying assignments of $\Phi$, in time $\widetilde{O}({k\Delta^2 n^{1.001}/\varepsilon})$.
\end{theorem}

The formal statement of our theorem is provided in \Cref{theo:unif-formal}. The sampling problem for hypergraph coloring can be reduced to that for \((k,q)\)-uniform formulas\footnote{The main idea of the reduction is as follows. Consider the problem of $q$-coloring a $k$-uniform hypergraph.
For each variable $v$ of the hypergraph, introduce $q-1$ dummy variables, denoted $v_1,\dots,v_{q-1}$. Let $P_v = \{v, v_1, \dots, v_{q-1}\}$ represent a set of variables in the PDC, with the domain $Q_v = \{0, 1, \dots, q-1\}$. By further decomposing the non-monochromatic constraints into disjunctive constraints, the reduction is completed.
}. Combined with known hardness results for hypergraph coloring~\cite{BGGGS19, galanis2023inapproximability}, this reduction implies that the sampling problem for \((k,q)\)-uniform formulas is intractable when \(\Delta^2 \gtrsim q^{k}\). Therefore, to render the problem tractable, it is natural to assume an LLL-like regime in which \(q^{k} \gtrsim \Delta^c\) for some constant \(c\). Our theorem establishes such a regime for \((k,q)\)-uniform formulas. Furthermore, our algorithm runs in nearly linear time in \(n\), which is significantly faster than previous algorithms for multiple permutations~\cite{cooper2011networks}.

Additionally, we have the following algorithm for general PDC formulas.
\begin{theorem}\label{main-theorem}
    The following holds for some positive constants $q_{\min}$ and $c$.
    There exists an algorithm such that given as input any $\varepsilon\in (0,1)$ and any PDC formula $\Phi$ with $n$ variables where $\Phi$ satisfies $q\geq q_{\min}$ and
    \begin{align}
    {cpk^{518}\Delta^{786}}\leq 1, \label{eq-main-theorem}
    \end{align}
    it outputs random valid assignments of $\Phi$ with total variance distance of $\varepsilon$ from the uniform distribution of all satisfying assignments of $\Phi$, in time $\widetilde{O}({k\Delta n^{7}/\varepsilon^2})$.
\end{theorem}

The formal statement of the theorem appears in \Cref{thm-main-formal}. 
While the constants in \eqref{eq-main-theorem} can be further optimized, obstacles remain in achieving bounds comparable to those in the sampling LLL for the variable model~\cite{Vishesh21sampling,HSW21,he2022sampling,he2023deterministic,WY24}. In particular, the extra condition \(q \geq q_{\min}\) can be eliminated by employing the mark/unmark paradigm~\cite{Moi19}, which is capable of handling small variables and permutations; however, this modification may complicate the overall algorithm design.
Moreover, the time complexity in \Cref{main-theorem} is higher than that in \Cref{theo:unif} due to challenges posed by extreme instances in general PDC sampling, where one permutation is significantly larger than the others. Nevertheless, the time complexity of our sampling algorithm for general PDC is comparable to that of the PRP sampling algorithm~\cite{jerrum2004polynomial}. Given that the sampling challenge for PDC is inherently more complex than that for PRP, the demonstrated time complexity is considered acceptable.

Our marginal sampler for 
drawing from marginal distributions (\Cref{Alg:samplepermutation})
is a core component of the sampling LLL.
This sampler diverges from all previous marginal samplers designed for product probability spaces.
Notably, it enables the rapid sampling of nearly uniform satisfying solutions in the joint probability space of PDC, transcending the variable model within the local lemma framework, an advancement not previously achieved.
Our marginal sampler is established based on the fast sampling and approximation algorithms for very dense PRP (Theorems \ref{thm-sampler-PRP-LLL} and \ref{thm-estimator-PRP-LLL}). 
By combining these tools, we develop a novel sampling framework called \emph{correlated factorization} (see \Cref{sec-breaking-large-components}) for the joint probability space.
Consequently, the marginal sampler is established.

\subsection{Applications}\label{sec-application-intro}
Our sampling algorithms in Theorems \ref{theo:unif} and \ref{main-theorem} are especially well suited to resource-allocation problems of the following form.  There are $m$ resource classes, each holding $q$ distinct units. These units must be allocated among $\ell$ tasks so that each task receives exactly $k$ units, at most one unit from each class, while any surplus units may remain unassigned.
The $k$ resources assigned to a single task are subject to mild compatibility constraints.  In this setting the $q$ resources within any fixed class naturally form a permutation of length $q$, while the compatibility requirements translate into length-$k$ disjunctive constraints on the $k$ variables associated with each task. Several illustrative applications are presented below.

\textbf{Teacher assignment problem.}
Consider a grade with $\ell$ classes and $m$ distinct subjects.  Each class must offer $k$ subjects to its students, and—because the classes have different specializations—the required sets of $k$ subjects may differ from one class to another. 
For each subject $i \in [m]$, exactly $q$ teachers are qualified, and at least $q - L$ of them, designated as the subset $S_i$, are senior enough to serve as homeroom teachers.
No teacher is qualified to teach more than one subject.
We wish to assign teachers to classes under the following conditions:
\begin{itemize}
\item Each teacher is assigned to at most one class.
\item Every class receives exactly one teacher for each of its $k$ required subjects.
\item Every class includes at least one senior instructor to serve as the homeroom teacher.
\end{itemize}
To ensure fairness across classes and subjects, an assignment is selected uniformly at random from all allocations that satisfy these constraints.

\textbf{Reviewer assignment problem.}
Assume $\ell$ proposals are to be evaluated. Reviewers are recruited from $m$ subject areas, with exactly $q$ reviewers assigned to each area. For every subject $i \in [m]$, at least $q - L$ of these reviewers, denoted $S_i$, are senior enough to serve as committee chairs.
We assign reviewers to proposals under the following constraints:
\begin{itemize}
\item Each reviewer serves on exactly one proposal.
\item Each proposal is assigned $k$ reviewers, one from each of the $k$ subject areas most relevant to it.
\item Every proposal’s committee includes at least one senior reviewer, who acts as chair.
\end{itemize}
To preserve confidentiality and fairness, the final assignment is selected uniformly at random from the set of all allocations that satisfy these constraints.

In \Cref{sec-application}, we prove that if $L$ is $\tilde{O}(q^{1/32})$ where high-order terms are suppressed, then these two sampling problems can be solved in almost linear time.
Prior to this work, no non-trivial algorithm was known for either sampling problem in the regime where both $\ell$ and $m$ are large.

\vspace{0.3cm}

Our sampling algorithm also has intriguing applications in combinatorics, such as factors of independent transversals in multipartite hypergraphs.
The problem of establishing sufficient
conditions for the existence of independent transversals and their factors in multipartite hypergraphs has been widely explored in the literature~\cite{erdHos1994independent,yuster2021factors}, though many questions remain open.
For the sampling of factors of independent transversals,
no nontrivial algorithm has previously been proposed.
We demonstrate, via reduction, that the sampling problem becomes intractable beyond the LLL-like regime. Additionally, we present the first efficient sampling algorithm for this problem within the LLL-like regime. 
These results are formally stated in \Cref{sec-application}.
\section{Technical overview}\label{sec-tech-tool}
In this paper, we provide an optimal algorithm for the sampling of very dense PRP and a novel algorithm
for the sampling of PDC in the lopsided LLL regime.
In this section, we will provide a detailed exposition of the algorithmic techniques.

\subsection{Techniques for handling a single permutation}\label{sec-tech-single-perm}
In this subsection, we summarize our techniques for handling very dense PRPs, which also serve as essential building blocks for sampling PDC.

Our approximation and sampling algorithms for very dense PRPs follow the approach outlined in \cite{huber2006exact}. Given a 0–1 matrix \(A\), Huber proposed an upper bound \(U(A)\) for the permanent \(\perman(A)\) (see \Cref{lemma-upper-bound-fa-perm}) and subsequently designed a novel sampling algorithm for PRP based on this bound. The algorithm runs in approximately \(O(n^2) \cdot U(A)/\perman(A)\) time. Consequently, to demonstrate the algorithm's efficiency, it suffices to derive a lower bound \(L(A)\) for \(\perman(A)\) and show that the ratio \(U(A)/L(A)\) is small. Previously, the lower bound \(L(A)\) was estimated using Van der Waerden’s inequality, which tends to be loose when the number of zeros in the rows and columns of \(A\) varies significantly.

Instead of relying on Van der Waerden’s inequality, we derive a tighter lower bound for very dense matrices using the lopsided LLL. Let \(A\) be a 0–1 matrix of size \(n \times n\), and let \(\sigma = (\sigma(1), \dots, \sigma(n))\) be a uniformly random permutation of \([n]\). Define \(\mathcal{E}_{i,j}\) as the bad event \(\sigma(i)=j\) for every \(i,j \in [n]\) such that \(A_{i,j}=0\). Consequently, we have
$\perman(A) = n! \cdot \Pr{\overline{\bigcup_{i,j} \mathcal{E}_{i,j}}}$.
Moreover, the lopsidependency structure among these bad events is captured by the lopsidependency graph. When \(A\) is a very dense 0–1 matrix, the degree of this graph is well-bounded, allowing the application of the lopsided LLL to obtain a lower bound for 
$\Pr{\overline{\bigcup_{i,j} \mathcal{E}_{i,j}}}$. 
Remarkably, the lower bound derived using the lopsided LLL for \(\perman(A)\) differs from \(U(A)\) by only a constant factor, thereby serving as a tight lower bound for \(\perman(A)\) up to that constant factor (see \Cref{thm-esitimator-result}). 
Thus, the PRP algorithm in \cite{huber2006exact} effectively runs in approximately \(O(n^2)\) time (see \Cref{thm-sampler-PRP-LLL}).

\subsection{Techniques for handling multiple permutations}\label{sec-tech-samper-pcc}
Our sampling algorithm for PDC follows the Markov chain framework for sampling a solution of CSP in the LLL regime.
This framework uses state compression to project the state space onto a smaller, well-connected space, where the barrier of disconnectivity could be circumvented because the images of the projection might collide.
The Markov chain simulates a single-site Glauber dynamics in the compressed state space to draw an approximate
sample from the uniform distribution $\mu$ of all solutions.
In each step of the Glauber dynamics, 
the assignment of a randomly selected variable \( v \) is updated according to the marginal distribution of \( v \), conditioned on the projected assignments \( X \) of other variables.
The standard rejection sampling is employed to update the assignment of $v$.
To make the update efficient, 
each constraint should be satisfied by $X$ with a fairly high probability.
This allows the remaining CSP formula to be factorized into disjoint sub-formulas of logarithmic sizes after removing all the satisfied constraints.
To demonstrate the rapid mixing of Glauber dynamics,
path coupling~\cite{FGYZ20,feng2021sampling} and 
information percolation~\cite{Vishesh21sampling,HSW21}
are employed to show a long connected path between two uncoupled variables is impossible. 

Our sampling algorithm for PDC also leverages the state compression technique.
In contrast to the state compression designed for the product probability space~\cite{feng2021sampling},
we design a novel structure-preserved state compression for the joint probability space of PDC.
Then our algorithm simulates a
\emph{permutation-wise Glauber dynamics} in the compressed state space to draw an approximate
sample from the uniform distribution $\mu$ of all solutions, 
rather than simulating a
\emph{Glauber dynamics} as in previous methods. 
The permutation-wise Glauber dynamics is employed here to simultaneously update the assignments of all variables within a permutation set, considering the dependency between these variables.
In each transition of the permutation-wise Glauber dynamics, a permutation set $P$ is picked uniformly at random and then the assignments of all the variables in $P$ are updated according to the marginal distribution on $P$, conditioned on the projected assignments $X$ of other variables in the compressed state space.
However, even after removing satisfied constraints under $X$ and factorizing the remaining PDC formula,
the sub-formula containing $P$ remains large if $P$ itself is large. 
To make the update in each step efficient,
we introduce a sampling framework called \emph{correlated factorization}, where the sub-formula containing $P$ is further broken down into two correlated parts, allowing one part to be sampled almost independently of the other.
To demonstrate the permutation-wise Glauber dynamics is fast mixing, 
we employ the path coupling technique.
The presence of long-range correlations from large permutations introduces new challenges.
A key concept in our analysis is the \emph{inactive} vertex, where the discrepancy at the vertex cannot further spread out through any constraint.
We establish the contraction property of the discrepancy by showing that most vertices in a large assigned permutation set are inactive.

In the following, we present our techniques individually. 
These techniques are designed to address long-range correlations between the variables. 
They may inspire the development of other sampling algorithms based on the lopsided LLL that extend beyond the variable model. 
For each PDC formula $\Phi$, let $\mu_{\Phi}$ denote the uniform distribution over its solutions.

\subsubsection{\emph{\textbf{Structure-preserved state compression}}}
Similar to the Glauber dynamics for sampling CSP solutions in the LLL regime, the state space of the permutation-wise Glauber dynamics for PDC may also be disconnected. To overcome this challenge, we propose a \emph{structure-preserved state compression} for PDC.

Consider the PDC formula \(\Phi = (\+P, \+Q, \+C)\) defined over \(m\) permutations, where \(\+P = (P_1, P_2, \dots, P_m)\) and \(\+Q = (Q_1, Q_2, \dots, Q_m)\). A structure-preserved state compression of \(\Phi\) involves decomposing each permutation set \(P_i\) into several smaller sets, denoted \(P_{i,1}, \dots, P_{i,\ell_i}\). Instead of operating on the original solution space, our permutation-wise Glauber dynamics runs on the compressed state space. In each compressed state, every subset \(P_{i,j} \subseteq P_i\) is paired with a corresponding domain \(Q_{i,j} \subseteq Q_i\) such that \(|Q_{i,j}| = |P_{i,j}|\) and the domains \(Q_{i,1}, \dots, Q_{i,\ell_i}\) partition \(Q_i\). This state comprises all valid assignments in which the partial assignment for the variables in \(P_{i,j}\) is a permutation of the values in \(Q_{i,j}\). In effect, each compressed state represents all solutions of the PDC formula \(\Phi' = (\+P', \+Q', \+C)\), where
$\+P' = (P_{1,1}, \dots, P_{1,\ell_1}, \dots, P_{m,1}, \dots, P_{m,\ell_m})$
and
$
\+Q' = (Q_{1,1}, \dots, Q_{1,\ell_1}, \dots, Q_{m,1}, \dots, Q_{m,\ell_m})$.
This approach allows multiple solutions of \(\Phi\) to be represented within the same state, effectively overcoming issues of disconnectivity (see \Cref{lemma-unique-distribution}).

Recall that two constraints  
$C = \cdots \lor (v \neq c) \lor \cdots$ and $C' = \cdots \lor (v' \neq c') \lor \cdots$
are defined to be related, denoted \(C \sim C'\), if either \(v = v'\), or \(v\) and \(v'\) belong to the same set \(P_i\) and \(c = c'\). Consequently, any two unrelated constraints neither share a variable nor forbid the same value for variables within the same permutation set. Importantly, this relationship is maintained even when each \(P_i\) is partitioned into individual subsets \(P_{i,1}, \dots, P_{i,\ell_i}\); that is, the non-association \(C \not\sim C'\) persists within the compressed states. Let \(\Phi'\) denote the corresponding PDC formula for a compressed state. Then, the lopsidependency graph for \(\Phi'\) is a subgraph of that for \(\Phi\), thereby preserving the original lopsidependency structure within each compressed state (see \Cref{lemma:lop-graph}). This property is fundamental to our proofs.

\subsubsection{\emph{\textbf{Correlated factorization}}}\label{sec-breaking-large-components}
Given a PDC formula $\Phi$ defined on $m$ permutations and a state compression of $\Phi$, where each permutation set $P_i$ is decomposed into multiple individual permutation sets denoted as $P_{i,1},\cdots, P_{i,\ell_i}$,
we assume that the formula induced by each state in the compressed state space still satisfies a strong lopsided LLL condition.
The subsequent challenge is how to implement the transition of the permutation-wise Glauber dynamics.
Given the current state of the permutation-wise Glauber dynamics,
\emph{w.l.o.g.}, let $P_m$ be the picked permutation set at the current transition.
At this state, each permutation set in $P_{i,1},\cdots, P_{i,\ell_i}$ has been assigned a subset of $Q_i$ as its domain for each $i\in [m]$.
Our task is to update the domains of $P_{m,1},\cdots, P_{m,\ell_m}$ according to the current domains of the permutation sets $P_{1,1},\cdots, P_{1,\ell_1},\cdots,P_{m-1,1},\cdots, P_{m-1,\ell_{m-1}}$. 
Given the domains of these permutation sets,
many constraints of the PDC formula $\Phi$ have been satisfied.
After removing these satisfied constraints,
$\Phi$ is “factorized” into subformulas,
where the correlated variables are assigned to the same subformula.
Since all variables in \( P_m \) are correlated, they are assigned to the same subformula, denoted as \( \Phi' = (\+P' \cup \{P_m\}, \+Q' \cup \{Q_m\}, \+C') \), where \( \+P' \) represents the permutation sets correlated with \( P_m \) and \( \+Q' \) represents their domains.
We remark that all permutation sets in $\+P'$ are small as a result of the decomposition.
A crucial step in implementing the transition is to draw a sample from $\mu_{\Phi'}$, the uniform distribution over the solutions of $\Phi'$. 
If $P_m$ is small, then $\Phi'$ is also small, allowing for efficient rejection sampling from $\mu_{\Phi'}$. However, this task becomes significantly more challenging when $P_m$ is large.

Given a subformula $\Phi'=(\+P'\cup \{P_m\},\+Q'\cup\{Q_m\},\+C')$ where $P_m$ is large, the constraint set $\+C'$ can be partitioned into four disjoint subsets $\+C_0,\+C_1,\+C_1^{\ast},\+C_2$:
\begin{itemize}
\item $C\in \+C_0$ if $C$ does not depend on any variable in $P_m$;
\item $C\in \+C_1$ if $C$ depends on exactly one variable in $P_m$ and no other variables outside of $P_m$;
\item $C\in \+C_1^{\ast}$ if $C$ depends on exactly one variable in 
$P_m$ along with other variables outside of $P_m$;
\item $C\in \+C_2$ if $C$ depends on at least two variables in $P_m$.
\end{itemize}
Let \( \Phi^{\dagger} = (\+P' \cup \{P_m\}, \+Q' \cup \{Q_m\}, \+C_0 \cup \+C_1 \cup \+C_1^{\ast}) \). By the lopsided LLL, most solutions of \( \Phi^{\dagger} \) are also solutions of \( \Phi' \) since the constraints in \( \+C_2 \) are rarely unsatisfied when \( P_m \) is large. Consequently, one can sample from \( \mu_{\Phi'} \) using standard rejection sampling: first draw a sample \( \sigma \) from \( \mu_{\Phi^{\dagger}} \) and then accept \( \sigma \) if it also satisfies \( \Phi' \). The main challenge then is to sample from \( \mu_{\Phi^{\dagger}} \).

Let \( \Phi_0 = (\+P', \+Q', \+C_0) \) and \( \Phi_1 = (\{P_m\}, \{Q_m\}, \+C_1) \). If \( \+C_1^{\ast} = \emptyset \), then \( \Phi^{\dagger} \) can be further factorized into \( \Phi_0 \) and \( \Phi_1 \). Additionally, \( \Phi_0 \) can be broken down into smaller subformulas because all permutation sets in \( \+P' \) are smaller, and many constraints in \( \Phi \) are satisfied.
Moreover, the sampling of \( \Phi_1 \) aligns with the PRP sampling problem. Therefore, we can first draw a random solution \( x \) for \( \Phi_0 \) using rejection sampling and then draw a random solution \( y \) for \( \Phi_1 \) using \Cref{thm-sampler-PRP-LLL}. The concatenation of \( x \) and \( y \) then yields a random solution for \( \Phi^{\dagger} \).

Unfortunately, \( \+C_1^{\ast} \) is typically not empty. Despite this, we still factorize \( \Phi_0 \) from \( \Phi^{\dagger} \), even though the variables in \( \Phi_0 \) are correlated with \( P_m \). We show that the distribution \( \mu_{\Phi_0} \) closely approximates the marginal distribution of \( \mu_{\Phi^{\dagger}} \) on the variables in \( \Phi_0 \). Thus, the following efficient sampler for \( \mu_{\Phi^{\dagger}} \) can be constructed:

\begin{itemize}
\item Sample a random solution \( x \) for \( \Phi_0 \) using rejection sampling.
\item Given \( x \) as the assignment for the variables in \( \Phi_0 \), \( \Phi^{\dagger} \) becomes a PRP formula \( \Phi^{\dagger}_x \) on \( P_m \). Sample a random solution \( y \) for \( \Phi^{\dagger}_x \) with the sampling algorithm for PRP (see \Cref{thm-sampler-PRP-LLL}).
\item Accept the concatenation \( z \) of \( x \) and \( y \) as a solution of \( \Phi^{\dagger} \) with a probability proportional to the solution size of \( \Phi^{\dagger}_x \), which can be estimated with the approximate counting algorithm for PRP (see \Cref{thm-estimator-PRP-LLL}).
\end{itemize}
This final step adjusts \( z \) to follow the distribution \( \mu_{\Phi^{\dagger}} \). We call this framework \emph{correlated factorization}.
The main idea of our framework is as follows: although the variables in \( \Phi_0 \) and \( P_m \) are correlated, the solution size of the PRP formula \( \Phi^{\dagger}_x \) for random \( x \) has a good concentration, allowing the variables in \( \Phi_0 \) to be sampled nearly independently of \( P_m \).

\subsubsection{\emph{\textbf{Inactive vertices in path coupling}}}
We employ the path coupling rather than the canonical path to obtain a tight mixing time of our permutation-wise Glauber dynamics. 
Within a single step of the coupling procedure (\Cref{Alg:coupling}), discrepancies may spread among variables within the same permutation set and across variables linked by the same constraint. 
Our objective is to establish that the discrepancy exhibits a contraction property, meaning it diminishes over time. 
However, the long-range correlation due to large permutations introduces new challenges to our analysis. Specifically, for a given permutation set $P$, if any variable $v$ in $P$ is uncoupled currently, all the variables in $P$ can potentially become uncoupled in the subsequent step of the algorithm. This issue is especially troublesome when $P$ is large, as it leads to a substantial increase in the number of uncoupled vertices, impeding the reduction of discrepancies.

A key concept in our proof is the notion of an \emph{inactive} vertex.
Given two partial assignments $X_1,X_2$ for a formula $\Phi$,
we call a vertex $v$ inactive if all constraints associated with $v$ are satisfied under both assignments $X_1$ and $X_2$.
Otherwise, $v$ is called active.
For an inactive vertex $v$, it is possible that $X_1(v) \neq X_2(v)$.
However, this discrepancy at $v$ cannot further spread out through any constraint.
As a result, the discrepancy can be bounded by considering only the active discrepant vertices.
Our key insight is the following: while discrepancies within a large permutation set $P$ can spread rapidly across $P$, most vertices in $P$ are inactive. 
This is primarily because, given the large size of $P$, the constraints associated with $P$ are likely to be satisfied with a high probability under a random assignment of $P$.
Consequently, the growth in active discrepancies is constrained. Nevertheless, achieving the contraction property for the discrepancy demands significant analytical effort.

\subsubsection{\emph{\textbf{Summary}}}
Although our sampling algorithm for PDC follows the general Markov chain framework employed in prior sampling LLL, it introduces a fundamentally novel marginal sampler (\Cref{Alg:samplepermutation}) that departs significantly from existing approaches. 
The marginal sampler is a central component in the Markov chain approach, which updates the assignment of a single site at each transition.
In previous LLL-based sampling algorithms, the marginal sampler is typically implemented via simple rejection sampling. However, this method breaks down in the PDC setting due to the presence of long-range correlations among variables within large permutations.
To overcome this critical barrier, we develop a new marginal sampler based on a fine-grained decomposition of the subformula associated with a large permutation into several correlated components. 
The centerpiece of the construction is a structural fact about one of these components, a very dense PRP, whose number of feasible solutions is sharply concentrated around its expectation.
This insight allows us to rigorously bound the correlations between components and ensure correct sampling.
Moreover, to guarantee the correct marginal distribution over the subformula, we propose new sampling and approximation algorithms specifically tailored to very dense PRPs. These methods not only extend the algorithmic toolbox for handling dense permutation structures but also contribute a novel perspective to the design of marginal samplers under complex dependencies. Our contribution thus advances the state-of-the-art in sampling for PRPs and offers a principled approach to marginal sampling for long-range correlated structures beyond the variable model in the LLL framework.
\section{Preliminaries}\label{sec-Preliminaries}

\subsection{Notations}

Throughout this paper, let $\^R$ be the set of real numbers, $\^R_+$ be the set of non-negative real numbers, and $\^N_+$ be the set of non-negative integers. 
For any positive integer $n$, we use $[n]$ to denote the set $\{1,2,\ldots,n\}$, and $n!$ to denote the factorial of $n$.
For simplicity,  
we use $x_{[n]}$ to represent $x_1,x_2,\cdots,x_n$.
Similarly, we will use $x_{[n]\setminus\{i\}}$ to represent 
$x_1,\cdots,x_{i-1},x_{i+1},\cdots,x_{n}$ and 
$x_{i,[n]}$ to represent $x_{i,1},x_{i,2},\cdots,x_{i,n}$.
Furthermore, to simplify notation, we sometimes use $x_{[n]}$ to denote the sequence $(x_{[n]}) = (x_{1},x_{2},\cdots,x_{n})$.

Given a \csppformula formula on $m$ permutations $\Phi=\left(\+P,\+Q,\+C\right)$ where $\+P = (P_{[m]}) = (P_1,P_2,\cdots,P_m)$ and
$\+Q = (Q_{[m]}) = (Q_1,Q_2,\cdots,Q_m)$,
let $\Omega_{\Phi}$ be the set of all satisfying assignments of $\Phi$, $\Omega^{\ast}_{\Phi}$ be the set of all valid assignments of $\Phi$,
and $V_{\Phi}\triangleq  P_1\cup \cdots \cup P_m$ be the union of all sets of variables in $\+P$. The subscript $\Phi$ could be omitted if it is clear from the context.
For convenience, we refer to $P_i$ as a \pname and $Q_i$ as its domain, and the correspondence between the \pname and its domain can also be indicated by $Q_i=\+Q(P_i)$ for each $i\in [m]$. 
Moreover, we will also treat non-repetitive sequences as sets, and abuse the set notations if it is clear from the context.
{For any sequence of \pnames $S=(P'_1,\cdots,P'_\ell)\subseteq \+P$, let $\+Q(S)$ denote the sequence of domains
$\left(\+Q(P'_1),\cdots,\+Q(P'_\ell)\right)$. }
For any $T\subseteq  V$, let $\+C(T)$ be the set of constraints $C\in \+C$ such that $\vbl(C)\cap T \neq \emptyset$, where $\vbl(C)$ is the set of variables used by $C$.
For each constraint $C\in \+C$, we say $u\in V$ is on $C$ if $u\in \vbl(C)$.
Additionally, we say $P\in \+P$ is on $C$ if $P\cap \vbl(C)\neq \emptyset$.

For convenience, we will use $a <_{q} b$ and $a >_{q} b$ to denote $a < b^{1+o_{{q}}(1)}$ and $a > b^{1-o_{{q}}(1)}$ respectively, where $q$ is the minimal permutation set size for the discussed PDC formula $\Phi$.
We denote $a=_{q} b$ if both $a <_{q} b$ and $a >_{q} b$ hold.

Given any assignment $\sigma\in \Omega^{\ast}$, let $\sigma(S)=\bigcup_{v\in S} \; \sigma(v) $ denote the value of $S$ in $\sigma$ for any subset of variables $S\subseteq V$. In particular, we use the notation $\sigma(v)$ instead of $\sigma(\set{v})$ to denote the assignment of $v$ for any $v\in V$. Moreover, we use $\sigma_S$ to denote the partial assignment on any subset of variables $S\subseteq V$ induced from $\sigma$. In other words, for any variable $v$ in $S$, we have $\sigma_S(v)=\sigma(v)$. Sometimes, we may use $v=c$ or $v\neq c$ to denote the event that $\sigma(v)=c$ or $\sigma(v)\neq c$ for convenience. 

Consider the uniform probability space over the valid assignments in $\Omega^*$. Recall that $\neg C$ is the event that the constraint $C$ is violated. 
We also abuse $C$ to denote the event that $C$ is satisfied when there is no ambiguity.
We say a PDC formula $\Phi$ is \emph{satisfiable} if there exist some satisfying assignments for $\Phi$ and use $\mu=\mu_\Phi$ to denote the uniform distribution over all satisfying assignments of $\Phi$. For any $S\subseteq V$, we use $\mu_S$ to denote the marginal distribution among the partial assignment $\sigma_S$ induced from $\mu$. Again, we use $\mu_v$ instead of $\mu_{\set{v}}$ for any $v\in V$ in convention.

Given a PDC formula $\Phi = (\+{P}, \+Q, \+{C})$, the \emph{simplification} of $\Phi$ is the formula $(\+{P}, \+Q, \+{C}')$ where $\+C'\subseteq \+C$ is the set of unsatisfied constraints given that the variables in $P$ take values from $\+Q(P)$ for each $P\in \+P$. The PDC formula $\Phi$ can be naturally represented as a (multi-)hypergraph $H_{\Phi}$ according to its simplification $(\+{P}, \+Q, \+{C}')$. Specifically, each variable $v\in V$ corresponds to a vertex in $H_{\Phi}$, each permutation $P\in \+{P}$ corresponds to a hyperedge $P$, and each constraint $C\in\+{C}'$ corresponds to a hyperedge $\vbl(C)$ in $H_{\Phi}$.
We slightly abuse the notation and write $H_{\Phi}=(V, \+P\cup \+{C}')$.
Let $H_{i}=(V_i, \+P_i\cup \+{C}_i)$, where $1\le i\le K$, denote all $K\ge 1$ connected components within $H_{\Phi}$.
Each connected component corresponds to a formula denoted as $\Phi_i=(\+P_i, \+Q_i, \+{C}_i)$. 
We refer to $\Phi_1, \Phi_2, \cdots, \Phi_K$ as the \emph{factorization} of $\Phi$. It is evident that {$\Phi=\Phi_1\land\Phi_2\land \cdots \land \Phi_K$}, and the uniform distribution $\mu_{\Phi}$ over all satisfying assignments of $\Phi$ is the product of the distributions $\mu_{\Phi_i}$ among the formulas $\Phi_i$ for $1\leq i\leq K$.

\subsection{Lopsided \LLL}\label{sec-lll}
Recall the definitions of lopsidependent events and the lopsidependency graph in \Cref{sec-previous-works}.
The lopsided LLL establishes a condition ensuring that the probability of none of the bad events occurring is greater than 0. This condition can be stated as follows:

\begin{theorem}[\cite{AS16}]\label{thm-lopsiLLL}
Given a \emph{lopsidependency} graph $G=(\+{B},E)$ w.r.t. the events $\+{B}$, 
if there is a function $x:\+{B}\rightarrow (0,1)$ such that for any $B\in\+{B}$,
\begin{align}\label{eq-condition-asym-lop}
\Pr{B}\le x(B)\prod_{B'\in\Gamma(B)}(1-x(B'))
\end{align}
then
$$
\Pr{\bigwedge_{B\in\+{B}}\overline B}\ge \prod_{B\in \+{B}}(1-x(B))>0.
$$
In particular, if $\-e\Pr{B}(D+1)\leq 1$ for each $B\in\+{B}$ where $D$ is the maximum degree of $G$, then 
$$
\Pr{\bigwedge_{B\in\+{B}}\overline B}\ge \prod_{B\in \+{B}}\left(1-\mathrm{e}\Pr{B}\right)>0.
$$
\end{theorem}

{By \Cref{thm-lopsiLLL}, we can obtain upper bounds of the conditional probabilities of events, the proof of which is deferred to~\Cref{sec:lll}.}

\begin{corollary} \label{Coro:LS}
Assume the condition \eqref{eq-condition-asym-lop} in \Cref{thm-lopsiLLL}.
For {any event $A$} and any subset $T\subseteq \+B$,
if $A$ is non-lopsidependent with the events in $\+B \setminus T$, 
we have
$$
\Pr{A \; \big \lvert  \; \bigwedge_{B\in \+B} \overline{B}} \le \Pr{A} \cdot \prod_{B\in T} (1-x(B))^{-1}.
$$
In particular, if $\mathrm{e}\Pr{B}(D+1)\leq 1$ for each $B\in\+{B}$ where $D$ is the maximum degree of $G$, then 
$$
\Pr{A \bigmid \bigwedge_{B\in \+B} \overline{B}} \le \Pr{A} \cdot \prod_{B\in T} (1-\-e\Pr{B})^{-1}.$$
\end{corollary}

\smallskip
Given any \cspformula formula $\Phi=(\+P,\+Q,\+C)$, we apply the lopsided LLL by considering the uniform probability space over its valid assignments $\Omega^{\ast}$ and the collection of ``bad'' events $\+B=\+B_\Phi\triangleq \set{\neg C\mid C\in \+C} $, where $\neg C$ is the event that $C$ is violated. The lopsidependency between the events in $\+B$ can be characterized by the undirected graph $G_\Phi^{\!{lop}}=(\+B,E_{\Phi}^{\!{lop}})$ where $\set{\neg C,\neg {C'}}\in E_{\Phi}^{\!{lop}}$ if and only if $C\sim C'$.
{Following the proof in~\cite{LS07}, one can verify that $G^{\!{lop}}_{\Phi}$ is a lopsidependency graph of the events in $\+{B}$.} 

Recall that $\^P$ is the uniform distribution over all valid assignments, and $\neg C$ denotes the event that $C$ is violated.
According to~\Cref{thm-lopsiLLL} and~\Cref{Coro:LS}, we immediately have the following lemma.
\begin{lemma}\label{prop-PDC-local-uniformity}
    Given any PDC formula $\Phi=(\+P,\+Q,\+C)$ satisfying $\-ep\Delta\leq 1$, it holds that 
    \begin{align*}
        \^P\left[ \bigwedge_{C\in \+C}  C\right]\geq \prod_{C\in \+C}\left(1-\mathrm{e}  \^P[\neg C]\right)>0.
    \end{align*}
    In addition, for any subset of valid assignments $A\subseteq \Omega^*$ and subset of constraints $\+C'\subseteq \+C$ such that $A$ is non-lopsidependent with the violation of the constraints in $\+C\setminus \+C'$, we have
    \begin{align*}
        \^P\left[A \bigmid \bigwedge_{C\in \+C}  C\right] \le \^P[A] \cdot \prod_{C\in \+C'} (1-\e  \^P[\neg C])^{-1}.
    \end{align*}
\end{lemma}

\subsection{Markov chain and path coupling}\label{subsec:pathcoupling}
The Markov Chain Monte Carlo (MCMC) method is widely used for approximate sampling. A Markov chain $\left(X_t\right)_{t\geq 0}$ defined on a discrete state space $\Omega$ can be specified by the transition matrix $M\in \mathbb{R}^{\Omega\times \Omega}$ where $M(x,y)=\Pr{X_{t+1}=y\mid X_t=x}$ for all $t\in \mathbb{N}$ and $x,y\in \Omega$. Thus, the Markov chain can be denoted by the transition matrix $M$. A distribution $\mu$ over the state space $\Omega$ is called the stationary distribution of the Markov chain $M$ if $\mu=\mu M$.
It is well known that $\mu$ is a stationary distribution of $M$ if $M$ is \emph{reversible} with respect to $\mu$, or equivalently, satisfies the following \emph{detailed balance condition}:
$$\forall x,y\in \Omega\quad \mu(x)M(x,y) = \mu(y)M(y,x).$$
The following result on the convergence of Markov chain is well known.
\begin{theorem}[\cite{LPW+17}]\label{Theo:Chain}
Consider the Markov chain $\{X_t\}_{t\ge 0}$ defined on $\Omega$ with the transition matrix $M\in \^R^{\Omega\times \Omega}$. If $M$ satisfies the following conditions:
\begin{itemize}
    \item \emph{irreducible}: $\forall x,y\in \Omega$, $\exists t>0$ such that $M^t(x,y)>0$;
    \item \emph{aperiodic}: $\forall x\in \Omega$, $\mathsf{gcd}(\{t>0 \mid M^t(x,x)>0\})=1$;
\end{itemize}
then there is a unique stationary distribution $\mu$. Furthermore, $\forall x\in\Omega, \lim_{t\rightarrow +\infty} \Pr{X_t=x} =\mu(x)$.
\end{theorem}

Given a Markov chain $M$ with stationary distribution $\mu$ , its convergence rate can then be captured by the \emph{mixing time} defined as:
$$
\tmix(\delta)\triangleq \max_{X_0\in\Omega} \min_t \{t\mid \dtv(X_t,\mu) \le \delta\},
$$ where $\dtv(\mu,\nu)\triangleq\frac{1}{2}\sum_{x\in\Omega}\abs{\mu(x)-\nu(x)}$ is the  total variance distance defined for any probability distributions $\mu$ and $\nu$ over $\Omega$.

\emph{Coupling} is a useful method for bounding the discrepancy of the probability distributions. Given probability distributions $\mu$ and $\nu$ defined on the same discrete space $\Omega$, a coupling of $\mu$ and $\nu$ is a joint distribution $(X,Y)$ over $\Omega\times \Omega$ such that the marginal distribution of $X$ (resp. $Y$) is $\mu$ (resp. $\nu$). 
\begin{lemma}\label{prop:coupling}
Let $\@C$ be a coupling between $\mu$ and $\nu$. It holds that $\dtv(\mu,\nu)\le 
\Pr[(X,Y)\sim \@C]{X\ne Y}$, and there exists an optimal coupling that achieves equality.
\end{lemma}

Path coupling is a useful technique for bounding the mixing time of the Markov chains. Consider a Markov chain $\left(X_t\right)_{t\geq 0}$ defined on a discrete state space $\Omega$ with stationary distribution $\mu$. ($\Omega$ doesn't need to be the support of $\mu$.) To construct a coupling between the Markov chain transition conditioned on every pair of configurations in $\Omega$, it allows us to consider pairs on a pre-metric of $\Omega$ instead, leading to a significantly simpler analysis. The notion of pre-metric can be stated as follows.

\begin{definition}[pre-metric]\label{def-pre-metric-brief}
A pre-metric on $\Omega$ is a weighted connected undirected graph with $\Omega$ as its vertices such that for any
edge $(x, y)$, the weight of the edge is equal to the weighted shortest-path distance between $x$ and $y$. In other words, the length of every other path from
$x$ to $y$ is at least the weight of the edge $(x, y)$.
\end{definition}

We can extend the pre-metric to a metric, denoted as $d$, on $\Omega$, where $d(x,y)$ is the shortest path distances in the pre-metric from $x$ to $y$ for any $x,y\in \Omega$. In the following, it is shown that if we can construct the couplings between pairs on the pre-metric such that the distance defined by $d$ exhibits contraction under the one-step transition of the Markov chain, then there exists a coupling between every pair of configurations in $\Omega$ that exhibits the contraction property.

\begin{lemma}[\cite{BD97}]\label{Lemma:pathcoupling}
Consider a Markov chain $\left(X_t\right)_{t\geq 0}$ defined on $\Omega$ with stationary distribution $\mu$. Let $G=(\Omega,E)$ be a pre-metric on $\Omega$ where each edge in $E$ has a weight no less than 1,
$d(\cdot,\cdot)$ be the weighted shortest path distance on $G$,
and $d_{\max}$ be $\max_{x,y\in \Omega} d(x,y)$. 
For any $\varepsilon>0$,
if there exists a coupling
$(X_t,Y_t)_{t\ge 0}$ of a Markov chain defined for each $(x,y)\in E$ such that
\begin{align}\label{eqn:coupling-contraction}
    \E{d(X_{t+1},Y_{t+1})\mid X_t=x,Y_t=y}\le (1-\varepsilon) d(x,y)
\end{align}
for each $(x,y)\in E$,
then there exists a coupling $(X_t,Y_t)_{t\ge 0}$ such that the contraction~\eqref{eqn:coupling-contraction} holds for each $(x,y)\in \Omega\times \Omega$. Moreover, we have $\tmix(\delta)\le \left\lceil\frac{1}{\varepsilon} \ln \frac{d_{\max}}{\delta}\right\rceil$ for each $\delta>0$.
\end{lemma}

\subsection{Inequalities}
McDiarmid's inequality is a powerful concentration result that provides a sharp bound on the probability that a function of random variables deviates from its expected value, assuming the function satisfies the bounded differences condition described below.

\begin{definition}[functions with bounded differences]
Let
$\mathbf{c} = (c_1, \ldots, c_n) \in \mathbb{R}_+^n$
be a given vector. A function \(f: \Omega \to \mathbb{R}\) is said to have bounded differences with respect to \(\mathbf{c}\) if for any $\=x =(x_1,\ldots,x_n),\=x'=(x'_1,\ldots,x'_n)\in \=\Omega$, it holds that
\[
\left| f(\mathbf{x}) - f(\mathbf{x}') \right| \le \sum_{i=1}^n c_i \, \mathbf{1}_{\{x_i \neq x'_i\}},
\]
where each \(c_i\) is referred to as the \(i\)-th difference coefficient of \(f\).
\end{definition}

\begin{theorem}[McDiarmid's inequality~\cite{mcdiarmid1989method}]\label{thm-McDiarmid-inequlaity}
Suppose \(f: \Omega_1\times\cdots\times \Omega_n \to \mathbb{R}\) satisfies the bounded differences property with respect to some \(\mathbf{c} \in \mathbb{R}_+^n\), and let \(\mathbf{X} = (X_1, \ldots, X_n)\) be a vector of independent random variables, with each \(X_i\) taking values in \(\Omega_i\). Then for any $t>0$, the tail probability satisfies
\begin{equation}\label{eq-McDiarmid}
\Pr{ \abs{f(\=X) - \E   {f(\=X)}} \ge t }
\le 2\exp \left(  - \frac{2t^2}{ \|\=c\|^2_2}  \right). 
\end{equation}
\end{theorem}

Initially conjectured by Minc~\cite{minc1963upper} and subsequently proven by Bregman~\cite{bregman1973some}, Bregman's Theorem establishes a sharp upper bound on the permanent.
According to the relationship between the permanent and the permutation, it also provides an upper bound on the number of permutations.
\begin{theorem}[Bregman’s Theorem~\cite{minc1963upper,bregman1973some}]\label{thm-Bregman}
Let $n$ be a positive integer and $R_1,R_2,\cdots,R_n$ be subsets of $[n]$.
Let $r_i = \abs{R_i}$ for each $i\in [n]$.
Then the number of permutations of $[n]$ in $\prod_{i\in [n]}R_i$ is upper bounded by 
$\prod_{i\in [n]}(r_i!)^{1/r_i}$.
\end{theorem}

The following inequalities are also used throughout our subsequent discussion.

\begin{proposition}[Bernoulli's inequality]\label{equality-ab}
    Given any real numbers $a\geq -1, b\ge 1$, we have
    $(1+a)^{b}\geq 1+ab$.
\end{proposition}

 \begin{proposition}[Stirling's formula]\label{thm-Stirling}
For any positive integer $n$, we have
\begin{align*}
 \sqrt{2\pi n}\left(\frac{n}{\-e}\right)^n \exp\left(\frac{1}{12n+1}\right)\leq n! \leq  \sqrt{2\pi n}\left(\frac{n}{\-e}\right)^n \exp\left(\frac{1}{12n}\right).
\end{align*}
\end{proposition}

\section{Results on a single permutation}\label{sec:prp}
This section presents our results on PRP. 
\subsection{Very dense PRP in the LLL regime}
This subsection presents our results on very dense PRP.

\subsubsection{\emph{\textbf{Lopsided LLL is almost tight for very dense PRP}}}
For any positive integer $a$, define $f(a)$ recursively as follows.
\begin{equation}
f(1) \triangleq \-e; \quad\forall a\geq 1,  \quad  f(a+1) = f(a)+1+\frac{1}{2f(a)} + \frac{0.6}{2f(a)^2}.
\end{equation}
The following two lemmas have been proved in \cite{huber2006exact}.

\begin{lemma}\label{lemma-upper-bound-fa}
For any $a\geq 2$, $f(a) \leq a + 0.5\ln a + 1.65.$
\end{lemma}

\begin{lemma}\label{lemma-upper-bound-fa-perm}
Given any integer {$n \geq 2$} and any $R_1,\cdots,R_n\subseteq [n]$,
let $\Omega$ denote the set of permutations of $[n]$ in $\prod_{i\in [n]}R_i$. 
Then we have
\begin{align*}
\abs{\Omega} \leq \prod_{i\in [n]}\frac{f(\abs{R_i})}{{\-e}}.
\end{align*}
\end{lemma}

 {Recall the function $g(x,y)$ defined in \eqref{eq-define-gxy}}.
The following theorem provides an approximation for the number of solutions of very dense PRP, which is tight up to a constant factor.
\begin{theorem}[lopsided LLL is almost tight for very dense PRP]\label{thm-estimator-num-prp}
Given any integer {$n \geq 2$} and any $R_1,\cdots,R_n\subseteq [n]$,
assume that {$\abs{R_i} \geq n - \sqrt{(n-2)/20}$} for each $i\in [n]$
and 
$\abs{\{k\in [n]\mid j\in R_k\}} \geq n - \sqrt{(n-2)/20}$ for each $j\in [n]$.
Let $\rho$ denote $n^2 - \sum_{i\in [n]}\abs{R_i}$ and $\Omega$ denote the set of permutations of $[n]$ in $\prod_{i\in [n]}R_i$. 
Then we have
\begin{align*}
\frac{g(\rho,n)}{\sqrt{2\pi n} \;\-e^2} \leq \abs{\Omega} \leq \prod_{i\in [n]}\frac{f(\abs{R_i})}{{\-e}}\leq  19\cdot\frac{g(\rho,n)}{\sqrt{2\pi n} \;\-e^2}.
\end{align*}
\end{theorem}
The above theorem serves as a basic ingredient for our sampling and approximate counting algorithms for very dense PRP, just as in Theorems \ref{thm-sampler-PRP-LLL} and \ref{thm-estimator-PRP-LLL}.
In ~\Cref{thm-estimator-num-prp},
the conditions on $\abs{R_i}$ for each $i\in[n]$ and $\abs{\{k\in [n]\mid j\in R_k\}}$ for each $j\in[n]$ guarantee that the degree of the lopsidependent graph is bounded such that the lopsided LLL is applicable\footnote{To achieve a nearly tight bound, we have stronger requirement on the degree of the lopsided-dependent graph than $\-e p\Delta < 1$.}.
The lower bound is proved with the lopsided LLL,
which is tight up to a constant factor for the number of solutions of very dense PRP.
This conclusion appears somewhat surprising
because the lower bound provided by LLL is typically considered to be loose.
In many cases, there can be an exponential gap between the lower bound established by LLL and the number of solutions.
However, this is not the case for very dense PRP.

\begin{proof}[Proof of \Cref{thm-estimator-num-prp}]
By \Cref{lemma-upper-bound-fa-perm}, to prove this theorem, it is sufficient to prove 
\begin{align}
\abs{\Omega}&\geq \frac{g(\rho,n)}{\sqrt{2\pi n} \;\-e^2}, \label{eq-grhon-lb}\\
\prod_{i\in [n]}\frac{f(\abs{R_i})}{{\-e}}&\leq  19\cdot\frac{g(\rho,n)}{\sqrt{2\pi n} \;\-e^2}.\label{eq-upper-bpund-fri}
\end{align}
\underline{Proof of \eqref{eq-grhon-lb}.}
Let $\sigma = (\sigma(1),\cdots,\sigma(n))$ be a uniformly random permutation of $[n]$.
Let $A_{i,j}$ denote the bad event $\sigma(i) = j$ for each $i\in [n],j\notin R_{i}$.
The lopsidependency between these events can be characterized by the lopsidependency graph, and \Cref{thm-lopsiLLL} can be applied.
Let $\Delta$ be the degree of the lopsidependency graph.
Then we have 
\begin{align}\label{eq-upper-bound-delta-n}
\Delta \leq \max_{i\in [n]}(n - \abs{R_i}) + \max_{j\in [n]}\abs{\{k\in [n]\mid j\in R_k\}}\leq \sqrt{(n-2)/5}.
\end{align}
Therefore, if $n\leq 6$, we have $\Delta < 1$.
Thus, $\Delta = 0$. We have 
\begin{align}\label{eq-sixdelta-n}
6\Delta + 1\leq n.
\end{align}
If $n> 6$, we also have \eqref{eq-sixdelta-n}.
In addition, we have 
\begin{align}
\rho = n^2 - \sum_{i\in [n]}\abs{R_i} \leq n^2 - n\min_{i\in [n]}\abs{R_i}= n\max_{i\in [n]}(n - \abs{R_i})\leq n\Delta.
\end{align}
Combined with \eqref{eq-upper-bound-delta-n},
we have 
\begin{align}\label{eq-upperbound-threedeltaovernmuniusonesqure}
\frac{3\Delta}{(n-1)^2} \leq \frac{3\Delta}{3(n+1)\Delta^2} = \frac{1}{(n+1)\Delta} \leq \frac{1}{\rho+1}.
\end{align}

Next, we show that $x(A_{i,j})$ can be set as $(1/n)\cdot \left(1-{1}/{n}\right)^{-2\Delta}$ for each bad event $A_{i,j}$ in ~\Cref{thm-lopsiLLL}. This boils down to showing that
\begin{align}\label{eqn-LLL-condition-single}
    1/n \leq  1/n\cdot \left(1-1/n\right)^{-2\Delta}\cdot \left(1- (1/n)\cdot \left(1-1/n\right)^{-2\Delta} \right)^{\Delta},
\end{align} 
since the violation probability of each constraint can be bounded by ${1}/{n}$, and the degree of the lopsidependency graph is bounded by $\Delta$. 
Note that 
\begin{align*}
\left(1- 1/n\right)^{-2\Delta}\geq \exp(2\Delta/n)\geq 1+2\Delta/n.
\end{align*}
In addition, by \Cref{equality-ab} we have
\[\left(1- (1/n)\cdot \left(1-1/n\right)^{-2\Delta} \right)^{\Delta} \geq 1 - \frac{\Delta}{n\cdot (1-1/n)^{2\Delta}} \geq 1 - \frac{\Delta}{n - 2\Delta}.\]
Therefore, \eqref{eqn-LLL-condition-single} holds if
\begin{align*}
    \left(1+\frac{2\Delta}{n}\right)\cdot \left(1-\frac{\Delta}{n-2\Delta}\right)\geq 1,
\end{align*} which follows from \eqref{eq-sixdelta-n} through simple calculation. 
Let $$g'(\rho,n) = \sqrt{2\pi n}\cdot \left(\frac{n}{\-e}  \right)^{n}\cdot \exp\left(\frac{1}{12n+1}\right)\cdot \left(1-\frac{1}{n}\right) ^{\rho} \cdot \frac{1}{\-e}.$$
By~\Cref{thm-lopsiLLL}, we have
\begin{align*}
     &\quad \abs{\Omega}\geq  n!\cdot \prod_{i\in [n],j\not\in R_i}(1-x(A_{i,j}))= n!\cdot(1-x(A_{i,j}))^{\rho}  \tag{by the definition of $\rho$}\\
    &\geq \sqrt{2\pi n}\cdot \left(\frac{n}{\-e}  \right)^{n}\cdot \exp\left(\frac{1}{12n+1}\right)\cdot \left(1-\frac{1}{n}\cdot \left(1-\frac{1}{n}\right)^{-2\Delta}\right)^{\rho}   \tag{by~\Cref{thm-Stirling}}\\
    &=\sqrt{2\pi n}\cdot \left(\frac{n}{\-e}  \right)^{n}\cdot \exp\left(\frac{1}{12n+1}\right)\cdot \left(1-\frac{1}{n}\cdot\left(1+\frac{1}{n-1}\right)^{2\Delta}\right) ^{\rho}   \\
    &\geq \sqrt{2\pi n}\cdot \left(\frac{n}{\-e}  \right)^{n}\cdot \exp\left(\frac{1}{12n+1}\right)\cdot \left(1-\frac{1}{n}\left(1+\frac{3\Delta}{n-1}\right)\right) ^{\rho} \tag{$\spadesuit$}\\
    &= \sqrt{2\pi n}\cdot \left(\frac{n}{\-e}  \right)^{n}\cdot \exp\left(\frac{1}{12n+1}\right)\cdot \left(1-\frac{1}{n}\right) ^{\rho}\cdot \left(1-\frac{3\Delta}{(n-1)^2}\right) ^{\rho}   \\
    &\geq \sqrt{2\pi n}\cdot \left(\frac{n}{\-e}  \right)^{n}\cdot \exp\left(\frac{1}{12n+1}\right)\cdot \left(1-\frac{1}{n}\right) ^{\rho}\cdot \frac{1}{\-e} \tag{by $(1 - 1/(\rho+1))^{\rho}\geq 1/\-e$ and \eqref{eq-upperbound-threedeltaovernmuniusonesqure}}
    \\
    & = g'(\rho,n),
\end{align*}
where $(\spadesuit)$ holds since 
  \begin{align*}
    \left(1+\frac{1}{n-1}\right)^{2\Delta}&\leq 1+\frac{2\Delta}{n-1}+\sum_{i\geq 2} \frac{2\Delta\cdot(2\Delta -1) \cdots(2\Delta -i+1)}{i!}\left(\frac{1}{n-1}\right)^i \tag{by binomial series }\\
    &\leq 1+\frac{2\Delta}{n-1}+\sum_{i\geq 2} \left(\frac{2\Delta}{n-1}\right)^i\\
    &\leq 1+\frac{2\Delta}{n-1} +\left(\frac{2\Delta}{n-1}\right)^2\cdot \left(\frac{1}{1-{2\Delta}/({n-1})} \right)\\
    &=1+\frac{2\Delta}{n-1} + \frac{4\Delta^2}{(n-1)(n-1-{2\Delta})}\\
    &\leq  1+\frac{3\Delta}{n-1} \tag{by \eqref{eq-sixdelta-n}}.
\end{align*}
Thus, we have
     \begin{align*}
        \frac{\abs{\Omega}}{g(\rho,n)} &\geq \frac{g'(\rho,n)}{g(\rho,n)}\\&=\frac{\sqrt{2\pi n}\cdot \exp\left({{1}/{(12n+1)}-1}\right)\cdot \left({n}/{e}  \right)^{n} \cdot \left(1-{1}/{n}\right)^\rho }{ 2\pi n \cdot \exp({{1}/{(3n})}) \cdot  \left({n}/{e}  \right)^{n}\cdot  \left(1-{\rho}/{n^2} \right)^{n}}\\
        &=\frac{\sqrt{2\pi n}\cdot \exp\left({{1}/{(12n+1)}}\right) }{2\pi n \cdot \exp({{1}/{(3n})}+1)    }\cdot \left( \frac{ \left(1-{1}/{n}\right)^{\rho/n}}{1-{\rho}/{n^2}}\right) ^{n}\\
        &\geq \frac{\sqrt{2\pi n}\cdot \exp\left({{1}/{(12n+1)}}\right) }{2\pi n \cdot \exp({{1}/{(3n})}+1)} \tag{by ~\Cref{equality-ab}}\\
        &\geq \frac{1}{\sqrt{2\pi n} \; \-e^{2}}.   
    \end{align*}
Thus, \eqref{eq-grhon-lb} is proved.
\vspace{0.5cm}

\noindent \underline{Proof of \eqref{eq-upper-bpund-fri}. }Note that 
\begin{align*}
\prod_{i\in [n]}f(\abs{R_i}) &\leq  \prod_{i\in [n]} (\abs{R_i}+0.5\ln \abs{R_i} + 1.65)\tag{by \Cref{lemma-upper-bound-fa}}\\
&= \prod_{i\in [n]} \abs{R_i}\left( 1 + \frac{0.5\ln \abs{R_i} + 1.65}{\abs{R_i}}\right)\\
& \leq \prod_{i\in [n]} \abs{R_i}\exp\left( \frac{0.5\ln \abs{R_i} + 1.65}{\abs{R_i}}\right)\tag{by $1+x\leq \exp(x)$ \text{ for each } x}\\
& \leq \prod_{i\in [n]} \abs{R_i} \left(5.3\sqrt{\abs{R_i}}\right)^{1/\abs{R_i}}\tag{by $\exp(1.65)\leq 5.3$}.
\end{align*}
Thus, we have
\begin{align*}
\prod_{i\in [n]} \left(\frac{\abs{R_i}}{\-e}\cdot  \left(5.3\sqrt{\abs{R_i}}\right)^{{1}/\abs{R_i}}\right)
&\leq \left(\prod_{i\in [n]} \frac{\abs{R_i}}{\-e}\right)\cdot  \left(5.3\sqrt{n}\right)^{\frac{n}{ n - \sqrt{(n-2)/20}}}\tag{ by {$ n - \sqrt{(n-2)/20} \leq \abs{R_i} \leq n$}}\\
&= \left(5.3\sqrt{n}\right)^{\frac{n}{ n - \sqrt{(n-2)/20}}} \cdot \left(\frac{n}{e} \right)^{n}\cdot \prod_{i\in [n]}{\left(1-\frac{n - \abs{R_i}}{n} \right) }\\
& \leq \left(5.3\sqrt{n}\right)^{\frac{n}{ n - \sqrt{(n-2)/20}}} \cdot \left(\frac{n}{e} \right)^{n}\cdot  \left(1-\frac{\rho}{n^2}\right)^{n} \tag{by AM-GM inequality}\\
& = \frac{\left(5.3\sqrt{n}\right)^{\frac{n}{ n - \sqrt{(n-2)/20}}}\cdot g(\rho,n)}{2\pi n\cdot \exp(1/3n)} \tag{by \eqref{eq-define-gxy}}\\
& \leq\frac{\-e^2\left(5.3\sqrt{n}\right)^{\frac{n}{ n - \sqrt{(n-2)/20}}}}{\sqrt{2\pi n}\cdot \exp(1/3n)} \cdot  \frac{g\left(x,n\right)}{\sqrt{2\pi n} \;\-e^2}\\
& \leq 19 \cdot \frac{g\left(x,n\right)}{\sqrt{2\pi n} \;\-e^2}\\
\end{align*}
By the above two inequalities, \eqref{eq-upper-bpund-fri} is immediate. 
The theorem is proved.
\end{proof}

\subsubsection{\emph{\textbf{Exact Sampling for very dense PRP in the LLL regime}}}
In this subsection, we present our approximation and sampling algorithm for very dense PRP in the LLL regime.
Our algorithms are identical to those in~\cite{huber2006exact}, but our analysis has been significantly refined.
Rather than lower bounding the number of permutations with Van der Waerden’s inequality as in ~\cite{huber2006exact},
we lower bound the number of permutations with the lopsided LLL as shown in \Cref{thm-estimator-num-prp}.

Similar to \Cref{thm-sampler-PRP-LLL},
we also have a fast approximate counting algorithm for very dense PRP.

\begin{theorem}[approximate counting for dense PRP in the LLL regime]\label{thm-estimator-PRP-LLL}
There exists an algorithm such that given as input any real numbers $\varepsilon,\delta\in (0,1)$, any integer {$n \geq 2$} and any $R_1,\cdots,R_n\subseteq [n]$ satisfying the condition in~\Cref{thm-estimator-num-prp}, it outputs a $1\pm\varepsilon$ approximation of $\abs{\Omega}$ with probability at least $1-\delta$, in running time $O(n^{2}\cdot\varepsilon^{-2}\cdot\log (1/\delta))$, where $\Omega$ is the set of permutations of $[n]$ in $\prod_{i\in [n]}R_i$.
\end{theorem}

Theorems \ref{thm-sampler-PRP-LLL} and \ref{thm-estimator-PRP-LLL} are basic ingredients of our sampler for PDC, which are immediate by \Cref{thm-estimator-num-prp} and the following theorem adapted from \cite{huber2006exact}.

\begin{theorem}
Let $\varepsilon\in (0,1)$, $\delta\in (0,1)$, \(n \ge 2\) be an integer, and \(R_1, \dots, R_n \subseteq [n]\). Define \(\Omega\subseteq \prod_{i=1}^n R_i\) as the set of all permutations \(\pi\) of \([n]\) where \(\pi(i) \in R_i\) for each \(i\). There exists an exact sampler that, given any \(n, R_1, \dots, R_n\) as input, outputs a uniformly random permutation from \(\Omega\), in expected running time 
\[O\left(n^{2}\cdot\frac{\prod_{i\in [n]}f(\abs{R_i})}{\-e^n\abs{\Omega}}\right).\]
In addition, there also exists an approximate sampler such that given as input any $\varepsilon, n, R_1, \dots, R_n$, it outputs a random permutation with total variance distance of $\varepsilon$ from the uniform distribution of \(\Omega\), in running time 
\[O\left( n^2 \cdot \frac{\prod_{i\in [n]}f(\abs{R_i})}{\-e^n\abs{\Omega}} \cdot \log \frac{1}{\varepsilon} \right).\]
There also exists an algorithm such that given as input any $\varepsilon,\delta,n,R_1,\cdots,R_n$, it outputs a $1\pm\varepsilon$ approximation of $\abs{\Omega}$ with probability at least $1-\delta$, in running time 
\[O\left(\frac{n^2}{\varepsilon^{2}} \cdot\frac{\prod_{i\in [n]}f(\abs{R_i})}{\-e^n\abs{\Omega}} \cdot\log \frac{1}{\delta}\right).\]
\end{theorem}

\smallskip

\section{Structure-preserved state compression}\label{sec:compression}

In this section, we introduce the state compression for PDC formulas. 
As discussed in \Cref{sec-intro}, 
the Markov chain approach for sampling LLL is facilitated by the utilization of the state compression technique.
This technique enables the projection of the state space onto a more compact subspace, overcoming the challenge of the disconnectivity barrier by ensuring the projected images potentially intersect and are well connected. 
Previous state compressions have been limited to product probability spaces.
In this work, we present a novel structure-preserved state compression technique for PDC formulas, {ensuring that the probability spaces of the projected images are still induced by permutations.}
Our state compression is realized by the decomposition of permutations. 

\subsection{Decomposition of permutations}\label{sec-decomposition-permutation}
Recall that \(x_{[n]} = (x_1, x_2, \dots, x_n)\). For simplicity, we sometimes use the notations \(x_1, x_2, \dots, x_n\) and \((x_1, x_2, \dots, x_n)\), as well as \(x_{[n]}\) and \((x_{[n]})\), interchangeably.
Given a PDC formula on $m$ permutations $\Phi=(\+P,\+Q,\+C)$ with $\+P = P_{[m]} = (P_1,\cdots,P_m)$ and $\+Q = Q_{[m]} = (Q_1,\cdots,Q_m)$,
the decomposition of $\+P$ can be defined as the following sequence of \pnames
$$\+P'=\left(P_{1,[\ell_1]},\cdots,P_{m,[\ell_m]}\right) = \left(P_{1,1},\cdots,P_{1,\ell_1},\cdots,P_{m,1},\cdots,P_{m,\ell_m}\right)$$
where $\ell_i>0$ and $P_{i,[\ell_i]} = (P_{i,1},\cdots,P_{i,\ell_i})$ is a partition of $P_i$ for each $i\in [m]$.
For any \(i\in [m]\) and \(t\in [\ell_i]\), let
$P = \bigcup_{j\ge t} P_{i,j}$.
We denote by \(\+P'[P]\) the partition \((P_{i,t}, P_{i,t+1}, \dots, P_{i,\ell_i})\) of \(P\) in \(\+P'\).
Furthermore, with a slight abuse of notation, we use $\+P'\setminus \+P'[P] \cup \set{P}$ to denote the sequence of \pnames
\[
    \left(P_{1,[\ell_1]},\cdots, P_{i-1,[\ell_{i-1}]},P_{i,1}, \cdots, P_{i,t-1},P,P_{i+1,[\ell_{i+1}]},\cdots,P_{m,\ell_m}\right).
\]
In particular, for each $i\in [m]$, we have $\+P'[P_i]= P_{i,[\ell_i]}$, and $\+P'\setminus \+P'[P_i] \cup \set{P_i}$ is the decomposition of $\+P$ obtained by replacing the partition of $P_i$ in $\+P'$ with $P_i$. 
For brevity, let $\+P'\circ P_i \triangleq\+P'\setminus \+P'[P_i] \cup \set{P_i}$.
The notations about $\+Q$
can be defined in a similar way.

Furthermore, given $\+P'=\left(P_{1,[\ell_1]},\cdots,P_{m,[\ell_m]}\right)$, and 
$\+Q'=\left(Q_{1,[\ell_1]},\cdots,Q_{m,[\ell_m]}\right),$
the pair $(\+P',\+Q')$ forms a decomposition of $(\+P,\+Q)$ if $\abs{P_{i,j}} = \abs{Q_{i,j}}$ for each $i\in [m]$ and $ j\in [\ell_i]$.

Consider any  PDC formula $\Phi = (\+P,\+Q,\+C)$ with a decomposition $(\+P',\+Q')$ of $(\+P,\+Q)$. The formula induced by $\Phi$ and 
$(\+P',\+Q')$, denoted as $\Phi[\+P',\+Q']$,
is the formula $\Phi'=(\+P',\+Q',\+C')$,
where $\+C'\subseteq \+C$ is the set of unsatisfied constraints given that the variables in $P'$ take values from $\+Q'(P')$ for each $P'\in \+P'$. 
Furthermore, for any valid assignment $\sigma$ of $\Phi$,
we define $\+Q[\Phi,\+P',\sigma]$ as the decomposition $\+Q'$ of $\+Q$ 
such that $(\+P',\+Q')$ is a decomposition of $(\+P,\+Q)$, and $\+Q'(P) = \sigma(P) = \bigcup_{v\in P} \; \sigma(v)$ for each $P\in \+{P}'$.

A crucial property of the formula induced by the decomposition is that the non-lopsidependent constraints in the original formula remain non-lopsidependent. This property follows directly from the definition of the lopsidependency graph for PDC formulas in \Cref{sec-lll}; a proof is provided in \Cref{appendix-compression} for reference.

\begin{lemma} \label{lemma:lop-graph}
Consider any PDC formula $\Phi=(\+P,\+Q,\+C)$ with a decomposition $\+P'$ of $\+P$. For any $\sigma \in \Omega^{\ast}$ and any constraints $C_1,C_2\in \+C$, if $\set{\neg C_1,\neg C_2}\not\in E_\Phi^{\!{lop}}$, then $\set{\neg C_1,\neg C_2}\not \in  E_{\Phi'}^{\!{lop}}$ where $\Phi'=\Phi[\+P',\+Q[\Phi,\+P',\sigma]]$. 
\end{lemma}

{
For any $\zeta\in (0,1]$, we introduce the notion of $\zeta$-decomposition.
\begin{definition}\label{defintion-zeta-decomposition}
Given any \csppformula formula $\Phi=(\+P,\+Q,\+C)$ with a real number $\zeta\in (0,1]$, 
a decomposition $\+P'$ is called a $\zeta$-decomposition of $\+P$ if
\begin{align}
&\forall P_1\in \+P,P_2\in \+P'[P_1], \quad 
\abs{\abs{P_1}^{\zeta} -\abs{P_2}} = O(1).\label{eq-condition-zetadecom}
\end{align}
\end{definition}

\begin{remark}\label{remark-theta-decomposition}
For any constant $\zeta\in (0,1/2]$ and any $\Phi =(\+P,\+Q,\+C)$ where $q_{\Phi}$ is sufficiently large, a $\zeta$-decomposition of $\+P$ always exists.
\end{remark}

The following lemma establishes an upper bound on the violation probability of the decomposed formulas, the proof of which can be found in \Cref{appendix-compression} for reference.

\begin{lemma}\label{lem-violation-decomposed}
Consider a \csppformula formula \(\Phi=(\+P,\+Q,\+C)\) with a \(\zeta\)-decomposition \(\+P'=(P'_1,\dots,P'_\ell)\) of \(\+P\). Assume \(\+Q'=(Q'_1,\dots,Q'_\ell)\) satisfies that, for each \(i\in [\ell]\), \(\abs{Q'_i} =\abs{P'_i}\) and \(Q'_i\subseteq \+Q(P)\), where \(P\) is the unique permutation set in \(\+P\) containing \(P'_i\). Let \(\Phi'=(\+P',\+Q',\+C)\). Then, we have
$p_{\Phi'}<_{q} p^{\zeta}_\Phi$.
\end{lemma}

\subsection{State compression by decomposition}\label{sec-state-comp-decomp}
For any formula $\Phi = (\+P,\+Q,\+C)$ with a decomposition $\+P'$ of $\+P$, let $\Omega[{\Phi},\+P']$ denote the set 
$\left\{\+Q'\mid (\+P',\+Q') \text{ is a decomposition of } (\+P,\+Q)\right\}$.
The distribution $\mu=\mu_\Phi$ naturally induces a distribution $\nu=\nu_{\Phi,\+{P}'}$ over $\Omega[\Phi,\+P']$ where
$$\forall \+Q'\in \Omega[\Phi,\+P'],\quad \nu(\+Q') = \Pr[\sigma\sim\mu]{\+Q[\Phi,\+P',\sigma] =\+Q' }.$$
The following lemma is immediate by the aforementioned definitions.

\begin{lemma}\label{lemma-property-nu}
Given any formula $\Phi = (\+P,\+Q,\+C)$ with a decomposition $\+P'$ of $\+P$,
we have
\begin{equation*}
\forall \sigma\in \Omega^{\ast}, \quad \mu(\sigma)=\nu(\+Q')\cdot \mu_{\Phi'}(\sigma),
\end{equation*}where $\+Q'=\+Q[\Phi,\+P',\sigma]$ and $\Phi'=\Phi[\+P',\+Q']$. 
\end{lemma}

\Cref{lemma-property-nu} characterizes an equivalent method to sample the random satisfying assignments from $\mu$.
Specifically, to sample a satisfying assignment of $\Phi$ from $\mu$ uniformly at random, one can
\begin{enumerate}
    \item Sample a domain $\+Q'$ from the compressed state space ${\Omega}[\Phi,\+P']$ according to $\nu$;
    \item Sample a satisfying assignment of the induced formula $\Phi[\+P',\+Q']$ uniformly at random.
\end{enumerate}
Consequently, it boils down to a sampler for the decomposition of $\+Q$ from $\Omega[\Phi,\+P']$ for sampling the satisfying assignments of $\Phi$.
 
For any {$P\in \+{P}$} and $\+Q'\in \Omega[\Phi,\+P']$,
let $\nu_{P}^{\+Q'}$ denote the marginal distribution of the domains on $\+P'[P]$ induced by $\nu$,
conditioned on the domains corresponding to $ \+P'\setminus \+P'[P]$ being fixed as $\+Q'(\+P'\setminus \+P'[P])$. We can establish the following lemma immediately from the definitions of $\nu_{P}^{\+Q'}$. 
\begin{lemma}\label{lemma-property-nupq}
Given any PDC formula $\Phi = (\+P,\+Q,\+C)$, any decomposition $(\+P',\+Q')$ of $(\+P,\+Q)$, and any $P\in \+P$,
let $\Phi'$ be the formula $(\+P'\circ P,\+Q'\circ \+Q(P),\+C)$. 
Suppose $\+P'[P] = \left(P'_1,P'_2,\cdots,P'_{\ell}\right)$ for some
integer $\ell>0$. It holds that 
\begin{align*}
\forall Q^{\ast}_1,Q^{\ast}_2,\cdots,Q^{\ast}_{\ell},\quad 
\nu_{P}^{\+Q'}\left( Q^{\ast}_1,Q^{\ast}_2,\cdots,Q^{\ast}_{\ell}\right) = 
\Pr[\sigma\sim \mu_{\Phi'}]{\sigma(P'_i) = Q^{\ast}_i \text{ for all } i\in [\ell]}.
\end{align*}
\end{lemma}

\Cref{lemma-property-nupq} characterizes an equivalent method to sample the domains from  $\nu_{P}^{\+Q'}$.
To sample $Q^{\ast}_1\cdots,Q^{\ast}_{\ell}$
from $\nu_{P}^{\+Q'}$, 
one can first draw a sample $\sigma$ from $\mu_{\Phi'}$,
and then let $Q^{\ast}_i = \sigma(P'_i)$ for each $P'_i\in \+P'[P]$.
Thus, we introduce the idealized permutation-wise Glauber dynamics for sampling $Q'$ from $\nu$ as follows:

\medskip
{\centering
\par\addvspace{.5\baselineskip}
\framebox{
  \noindent
  \begin{tabularx}{15cm}{@{\hspace{\parindent}} l X c}
    \multicolumn{2}{@{\hspace{\parindent}}l}{\underline{Idealized permutation-wise Glauber dynamics for $\Phi=(\+P,\+Q,\+C)$ with a decomposition $\+P'$ of $\+P$:}} \\
  1. &  initialize $\+Q'\gets \+Q[\Phi, \+P',\sigma]$ where $\sigma \leftarrow \mbox{a valid assignment of $\Phi$ uniformly at random}$;\\
  2. & repeat the following for sufficiently many iterations:\\
      & \quad a) pick a \pname $P\in\+P$ uniformly at random;\\
      & \quad b) update the domains of $\+P'[P]$ by redrawing its value independently from $\nu_{P}^{\+Q'}$:\\
      & \quad \quad i) draw $\sigma \sim \mu_{\Phi'}$ where $\Phi'=(\+P'\circ P,\+Q'\circ \+Q(P),\+C)$;\\
       & \quad \quad ii) let $\+Q'(P') \gets \sigma(P')$ for each $P'\in \+P'[P]$;\\
  \end{tabularx}
 }
\par\addvspace{.5\baselineskip}}

\medskip
\noindent The following lemma ensures the correctness of the idealized permutation-wise Glauber dynamics, the proof of which is postponed to~\Cref{appendix-compression}.

\begin{lemma}\label{lemma-unique-distribution}
Consider any PDC formula \(\Phi = (\+P, \+Q, \+C)\) together with a decomposition \(\+P'\) of \(\+P\). Suppose that for every \(\+Q'\in \Omega[\Phi,\+P']\) and the corresponding formula \(\Phi' = \Phi[\+P',\+Q']\), we have
$\-e p_{\Phi'} \Delta_{\Phi'} \leq 1$. 
Then, the idealized permutation-wise Glauber dynamics is irreducible, aperiodic, and reversible with respect to \(\nu\). Consequently, it converges to the stationary distribution \(\nu\).
\end{lemma}

\section{Sampling algorithm with state compression}\label{sec:samplealgorithm}
In this section, we introduce the MCMC-based algorithm for sampling the satisfying assignments of the PDC formula. 

\subsection{Sampling algorithm}\label{sec:main-sampling-algorithm}
Given a \cspformula formula $\Phi=(\+P,\+Q,\+C)$, the main idea of our sampling algorithm is to construct a good decomposition $\+P'$ of $\+P$, simulate the idealized permutation-wise Glauber dynamics for the decomposition $\+Q'$ of $\+Q$ as shown in \Cref{sec-state-comp-decomp}, and then sample the satisfying assignment of $\Phi[\+P',\+Q']$ uniformly at random.

\subsubsection{\emph{\textbf{Decomposition construction}}}

We consider the formulas with the decomposition of permutations satisfying the following conditions. 

\begin{condition}\label{condition-state-compression}
Let \(\Delta, k, n \ge 1\) and \(L \ge 2\) be integers, and let \(p, \eta \in (0,1)\) be real numbers. Consider any PDC formula \(\Phi = (\+P,\+Q,\+C)\) satisfying
$\Delta_\Phi \le \Delta, k_\Phi \le k$ and  $|V| = n$,
and let \(\+P'\) be a decomposition of \(\+P\). Then, for each \(P \in \+P\) and every \(P' \in \+P'[P]\), we have
\[
|P'| \le L \quad \text{and} \quad \Bigl||P'| - \min\{|P|^\eta,\, L\}\Bigr| = O(1).
\]
Moreover, for any \(\+Q' \in \Omega[\Phi,\+P']\) and \(\Phi' = \Phi[\+P',\+Q']\), we have
$p_{\Phi'} \le p$.
\end{condition}

With the decomposition satisfying \Cref{condition-state-compression}, one can verify that each \pname $P\in \+P$ where $\abs{P}< L^{1/\eta}$ is decomposed into \pnames $P'\in \+P'$ satisfying $\abs{P'} =_{q} \abs{P}^{\eta}$.
In this way, for any $\+Q'\in \Omega[\Phi,\+P']$, the violation probability of
each constraint in $\Phi[\+P',\+Q']$ is approximately upper bounded by $p^\eta$.
In addition, by \Cref{prop-PDC-local-uniformity}, one can further show that for any random $\+Q'$ appearing in the idealized permutation-wise Glauber dynamics in \Cref{sec-state-comp-decomp}, each constraint in $\+C$ is satisfied with a probability around $1 - p^{1-\eta}$, given that the variables in $P'$ take values from $\+Q'(P')$ for each $P'\in \+P'$. Then most constraints in $\+C$ are satisfied in $\Phi[\+P',\+Q']$.
Furthermore, as shown in subsequent sections, we ensure that each \pname $P'\in \+{P}'$ has a size not exceeding $L$ such that the algorithm can be realized efficiently.

It is worth noting that constructing the decomposition $\+P'$ of $\+P$, which satisfies \Cref{condition-state-compression}, can be easily accomplished by partitioning each $P\in \+P$ into subsets of appropriate sizes. This process is of time complexity $O(\abs{V})$.

\subsubsection{\emph{\textbf{MCMC-based sampling algorithm}}}
Consider any instance satisfying \Cref{condition-state-compression} with an error parameter $\varepsilon\in (0,1)$. The main algorithm (\Cref{Alg:MCMC}) 
then implements a block Markov chain on space 
$\Omega[\Phi,\+P']$ to sample the decomposition $\+Q'$ of $\+Q$ from the distribution $\nu$.
It simulates the permutation-wise Glauber dynamics for $T = \left\lceil 2n\log\left({3n}/{\varepsilon}\right) \right\rceil$ steps to draw a random $\+Q'$
distributed approximately as $\nu$.
At each step, a \pname $P\in \+P$
is picked uniformly at random, and 
$\+Q'(\+P[P])$ is resampled approximately
from the marginal distribution $\nu_{P}^{\+Q'}$. 
Finally, the algorithm draws a random satisfying assignment of the formula $(\+P',\+Q',\+C)$ almost uniformly at random, where $\+Q'$ is distributed approximately from $\nu$. 
According to \Cref{lemma-property-nu},
this process generates a random assignment that approximately follows the distribution of $\mu$.

One of the most technically challenging steps in the sampling algorithm is to draw a sample from the marginal distribution $\nu_{P}^{\+Q'}$.
Recall the factorization of a formula defined in \Cref{sec-Preliminaries}. 
Consider the factorization of the formula $\Phi^{\ast} = (\+P'\circ P,\+Q'\circ \+Q(P),\+C)$, where 
the domains on $\+P'\setminus \+P'[P]$ are fixed as $\+Q'(\+P'\setminus \+P'[P])$ and the domains $\+Q'(\+P'[P])$ are replaced by $\+Q(P)$.
It is evident that
the marginal distribution induced by $\mu_{\Phi^\ast}$ on the variables in $P$ is determined exclusively by the factorized formula $\Phi'$ that contains $P$ in \Cref{Alg:MCMC}. 
In other words, $\mu_{\Phi^\ast}$ and $\mu_{\Phi'}$ have the same marginal distribution on the variables in $P$.
Combining with \Cref{lemma-property-nupq}, one can draw a sample from $\nu_{P}^{\+Q'}$ by sampling $\sigma\sim \mu_{\Phi'}$ and updating $\+Q'(P')$ to $\sigma(P')$ for each $P'\in \+P'[P]$. 
To draw $\sigma\sim \mu_{\Phi'}$, 
we call the subroutine $\samplepermutation(\cdot)$ which is the core of our algorithm.




In the last step of the main algorithm, to draw a random solution for the formula $(\+P',\+Q',\+C)$,
we first factorize $(\+P',\+Q',\+C)$ to formulas $\Phi_1,\cdots, \Phi_K$, and then draw a random solution $\sigma_i$ for each factorized formula $\Phi_i$ by the rejection sampling. 
{The solution of $(\+P',\+Q',\+C)$ is the concatenation of all $\sigma_i$.
We shall ensure that the factorized formulas $\Phi_1,\dots,\Phi_{K}$ have small sizes, which can be achieved by that each \pname $P'\in \+P'$ has a size not exceeding $L$ by \Cref{condition-state-compression}, and each constraint $C\in \+C$ is satisfied with high probability given that the variables in $P'$ take values from $\+Q'(P')$ for each $P'\in \+P'$.}
Thus, the rejection sampling can be efficient for these factorized formulas.

\begin{algorithm}
    \caption{$\MCMC(\Phi,\+P',\varepsilon)$} 
    \label{Alg:MCMC}
    \KwIn{a \cspformula formula $\Phi=(\+P,\+Q,\+C)$, a decomposition $\+P'$ of $\+P$ satisfying \Cref{condition-state-compression}, and a error parameter $0< \varepsilon <1$.}
    \KwOut{a random valid assignment of $\Phi$.}
    $\sigma \leftarrow \mbox{a valid assignment of $\Phi$ uniformly at random}$\;
    $\+Q' \leftarrow \+Q[\Phi,\+P',\sigma]$\;
    $T\gets \left\lceil 2n\log\left({3n}/{\varepsilon}\right) \right\rceil$\;\label{line-main-loops}
    \For{$t= 1$ to $T$}{
        pick a \pname $P$ from $\+P$ uniformly at random\;
        factorize $(\+P'\circ P,\+Q'\circ \+Q(P),\+C)$ to disjoint formulas and 
        let $\Phi'$ be the 
        formula containing $P$\;\label{line-main-factorization}
        $\sigma\gets \samplepermutation(\Phi',P,{\varepsilon}/{\left(3(T+1)\right)})$\;\label{line-main-plargerthanl}        
        \textbf{For} {each $P'\in \+P'[P]$ \textbf{do}} $\+Q'(P') \leftarrow \sigma(P')$\;
    }
   Let $\Phi_1,\dots,\Phi_{K}$ be the factorization of $(\+P',\+Q',\+C)$\;
    \For{$i=1$ to {$K$}}
    {
        $\sigma_i\gets \truncatedsampling(\Phi_i,{\lceil (n/\varepsilon) \log  (3n(T+1)/\varepsilon)\rceil)}$\;\label{line-main-rej-2}
    }
    \Return the concatenation of $\sigma_1,\sigma_2,\cdots,\sigma_K$\;
\end{algorithm}

Our rejection sampling algorithm for drawing from $\mu_{\Phi_i}$ is given in Algorithm 2.
The correctness of this algorithm is folklore. We state without proof.

\begin{lemma}\label{lemma-rejection-sample}
    Consider any PDC formula $\Phi=(\+P,\+Q,\+C)$ with a truncated threshold $T\geq 0$. Conditioned on $\truncatedsampling(\Phi,T)$ terminating at ~\Cref{line:success}, it returns a satisfying assignment $\sigma$ of $\Phi$ uniformly at random. Furthermore, $\dtv(\sigma,\mu_\Phi)$ is upper bounded by the probability that $\truncatedsampling(\Phi,T)$ terminates at ~\Cref{line:fail}.
\end{lemma}


\begin{algorithm} 
    \caption{$\truncatedsampling(\Phi,T)$} 
    \label{alg:truncatedsampling}
    \KwIn{a \cspformula formula $\Phi=(\+P,\+Q,\+C)$ and a truncated threshold $T$.}
    \KwOut{a random valid assignment of $\Phi$.}
    \For{$t=1$ to {$T$}}
    {   
         generate a valid assignment $\sigma$ of $\Phi$ uniformly and independently at random\;
        \textbf{if}{$\ \sigma$ is a satisfying assignment of $\Phi$} \textbf{then} \Return{$\sigma$}\;\label{line:success}
    }
    \Return a valid assignment $\sigma$ of $\Phi$ uniformly at random \label{line:fail}\;
\end{algorithm}

\subsection{The \samplepermutation ~subroutine}\label{subsection-sample}
In this section, we describe the $\samplepermutation(\cdot)$ ~ subroutine, which is the core component of our sampling algorithm. 
Given any instance satisfying~\Cref{condition-state-compression}, we will show that the following invariant is satisfied by the input 
of $\samplepermutation(\cdot)$ subroutine in~\Cref{Alg:MCMC} under certain conditions.
The correctness of the \samplepermutation ~ subroutine is guaranteed by this invariant.
\begin{condition}[invariant for $\samplepermutation$  subroutine]\label{condition-sample}
Let integers $\Delta,k,n,L\geq 1$ and real number $p\in (0,1)$ satisfying {$8\-ep\Delta^2 \leq 1$}.
Consider any PDC formula $\Phi = (\+P,\+Q,\+C)$ with a \pname $P\in\+P$ satisfying the following conditions:
\begin{itemize}
\item $\Delta_{\Phi} \leq \Delta$, $k_\Phi\leq k$, $p_\Phi\leq p$, $\abs{V_{\Phi}} \leq n$, and $\Phi$ cannot be further factorized to smaller formulas;
\item $\abs{P'}\leq L$ for each $P'\in \+P\setminus \set{P}$. 
\end{itemize}
\end{condition}

Our $\samplepermutation(\cdot)$ subroutine significantly differs from previous marginal samplers of sampling LLL.
Significant effort is devoted to tackling the challenges arising from {large} permutations, which give rise to long-term correlations among the random variables.
The conventional approach of factorizing the formula fails when dealing with {large} permutations, as they, along with their neighbor elements in the factorized formula,
need to be sampled collectively. 

{Consequently, in addition to exploiting factorization, we also analyze the correlation between the large permutation and its neighbors using concentration inequalities.
The favorable concentration suggests a weak correlation between the {large permutation} and its neighbors. 
Hence, the neighboring permutations and the large permutation in one factorized formula can be sampled sequentially, where the large permutation is sampled with the sampling algorithm for {PRP}.
The weak correlation only affects the time complexity of the subroutine by a polynomial factor.}

According to \Cref{condition-sample},
the input formula $\Phi=(\+P,\+Q,\+C)$ of $\samplepermutation(\cdot)$ ~ cannot be further factorized.
To draw a random satisfying assignment of $\Phi$,
instead of factorizing $\Phi$ directly,
we begin by omitting the constraints in $\+C(P)$ and factorizing $(\+P,\+Q,\+C\setminus \+C(P))$ into formulas $\Phi_1,\Phi_2,\cdots,\Phi_K$.
By definition,
there must be a factorized formula $(\{P\},\{\+Q(P)\},\emptyset)$ which we designate as $\Phi_K$.
We shall ensure that, with high probability, the factorized formulas $\Phi_1,\Phi_2,\cdots,\Phi_{K-1}$ have small sizes throughout the execution of~\Cref{Alg:MCMC}, which can be formalized as the following condition.
\begin{condition}\label{condition-sample-small}
Consider any instance satisfying~\Cref{condition-sample} with real number $\varepsilon\in (0,1)$.
Let $\Phi_1,\dots,\Phi_{K}$ be the factorization of $(\+P,\+Q,\+C\setminus \+C(P))$ where $\Phi_{K}=(\{P\},\{\+Q(P)\},\emptyset)$, and $\Phi_i=(\+P_i,\+Q_i,\+C_i)$ for each $i\in [K-1]$. It holds that
\begin{equation*}
    \max \{\abs{\+C_1},\cdots,\abs{\+C_{K-1}}\}\leq {\sampleblocksize}.
\end{equation*}
\end{condition}

\begin{remark}\label{remark-factorized}
    In the following pages, we will consistently refer to the PDC formula $\Phi=(\+P,\+Q,\+C)$ that cannot be further factorized to smaller formulas with a \pname $P\in \+P$. For convenience, we preserve the notations $\Phi_1, \Phi_2, \ldots, \Phi_K$ for the factorized formulas of $(\+P, \+Q, \+C\setminus \+C(P))$, where $\Phi_K = (\set{P}, \set{\+Q(P)}, \emptyset)$, and $\Phi_i=(\+P_i,\+Q_i,\+C_i)$ for each $i\in [K-1]$. Moreover, we use $\mu_{[i]}$ to denote $\mu_{\Phi_1}\times \cdots\times \mu_{\Phi_i}$ for any $i\in [K]$.
\end{remark}

In the remaining part of this section, we present the high-level ideas and crucial lemmas for sampling the satisfying assignments of any instances satisfying Condition~\ref{condition-sample-small}. 

We begin with the natural {accept-reject} sampler for sampling $\sigma\sim \mu_{\Phi}$, which is to repeat the following steps until success:
\begin{itemize}
\item draw $\sigma_1,\cdots,\sigma_{K}$ from $\mu_{\Phi_1},\cdots,\mu_{\Phi_{K}}$, independently;
\item accept the concatenation of $\sigma_1,\dots,\sigma_{K}$ as $\sigma$ if the concatenation is a solution of $\Phi$.
\end{itemize}
The correctness of the sampler is folklore. The challenge lies in demonstrating its efficiency.

\subsubsection{\emph{\textbf{Easy-to-handle factorized formulas}}}
We first claim that {if $2\-e\Delta\leq \log(\-e^4n/\varepsilon)$ or $p_{\max}\abs{P}\Delta\leq 1/4$ where $p_{\max}=\max_{C\in \+C(P)}\^P_\Phi(\neg C)$,} the accept-reject sampler has an efficient implementation. The efficiency relies on the following result by the lopsided LLL, which is proved in~\Cref{subsec:correctness}.

\begin{algorithm}
\caption{$\samplepermutation(\Phi,P,\varepsilon)$} \label{Alg:samplepermutation}
    \KwIn{a \cspformula formula $\Phi=(\+P,\+Q,\+C)$, a \pname $P\in \+P$ satisfying~\Cref{condition-sample}, and a parameter $0< \varepsilon<1$.}
    \KwOut{a random valid assignment of $\Phi$}
    Let $\Phi_1,\dots,\Phi_{K}$ be the factorization of $(\+P,\+Q,\+C\setminus \+C(P))$ where $\Phi_{K}=(\{P\},\{\+Q(P)\},\emptyset)$\;\label{line-factorize-sample}
    $p_{\max} \leftarrow \max_{C\in \+C(P)}\^P_\Phi(\neg C)$\;
    \eIf{$2\-e\Delta\leq \log(\-e^4n/\varepsilon)$ or $p_{\max}\abs{P}\Delta\leq 1/4$\label{line-if-sample}}
    {
        {$T\gets \lceil \-e^4n/\varepsilon \cdot  \log(2/\varepsilon) \rceil$}\;\label{line-rejectionsampler-time}
        \For{$j=1$ to {$T$}}
        {
            \For{$i=1$ to {$K-1$}}
            {
                $\sigma_i\gets \truncatedsampling(\Phi_i,\lceil (n/\varepsilon)\cdot \log  (2nT/\varepsilon)\rceil)$\;\label{line-rejection-in-sample-1}
            }
            $\sigma_{K} \leftarrow $ a random permutation of $\+Q(P)$ on $P$\;
            \textbf{if}{\ the concatenation of $\sigma_1,\cdots,\sigma_K$ is a solution of $\Phi$} \textbf{then} \textbf{return} the concatenation\;\label{line-break-sample}
        }
    }
    {
        \For{$i=1$ to {$K-1$}\label{line-sample-core-begin}}
        {   
            $\sigma_i\gets \truncatedsampling(\Phi_i,\lceil (n/\varepsilon)\cdot \log  (12n/\varepsilon)\rceil)$\;\label{line-rejection-in-sample-2}
        }
        $x\gets \max\set{\rho\left(\sigma{[K-1]}\right)-\abs{P}/6,0}$, $\+N\leftarrow g\left(x,\abs{P}\right)$\;
        {$T\gets \lceil 2\sqrt{2\pi\abs{P}}\cdot \textnormal{e}^3 /(1-\exp(-6\-e\ln 2\cdot  \Delta) \cdot \log (6/\varepsilon)\rceil$}\;
        \For{$j=1$ to {$T$}}
    {   
        \For{$t=1$ to {$K-1$}}
        {
            $\sigma_{t}\gets \truncatedsampling(\Phi_{t},\lceil(n/\varepsilon)\cdot \log  (6nT/\varepsilon)\rceil)$\;
        }
        
        \If{$x\leq \rho\left(\sigma{[K-1]}\right)\leq x+\abs{P}/2$}
        {
            $\Phi'\leftarrow \left(\{P\},\{\+Q(P)\},\+C\left[\sigma_{[K-1]}\right]\right), 
            \Phi'' \leftarrow \left(\{P\},\{\+Q(P)\}, \+C^1\left[\sigma_{[K-1]}\right]\right)$\;
            draw $\tau$ from $\mu_{\Phi''}$ with total variance distance {$\frac{\varepsilon}{6T}$} by calling the sampler {in~\Cref{thm-sampler-PRP-LLL}}\;\label{line-coupletau1}
            \If{$\tau$ is also a solution of $\Phi'$}
            {
                $r\sim \!{Unif}[0,1]$\;       
                estimate $\abs{\Omega_{\Phi''}}$ within multiplicative error {$1 \pm \frac{\varepsilon}{12T}$ with probability $1- \frac{\varepsilon}{6T}$} by calling the counting algorithm {in~\Cref{thm-estimator-PRP-LLL}}, and let $\widehat{N}$ be the estimation\;
                \If{$r\leq \widehat{N}/\+N$}
                {
                    \Return the concatenation of $\sigma_1,\dots,\sigma_{K-1},\tau$\;\label{line-sample-core-end}
                }
            }
        }
    }
       
    }   
    
    \Return an arbitrary valid assignment of $\Phi$ uniformly at random\;\label{line-failure-of-sample}
\end{algorithm}

\begin{lemma}\label{lemma-simplecase-repeattime}
Let integers $\Delta,n\geq 1$ and real numbers $p,\varepsilon\in (0,1)$ satisfying {$8\-ep\Delta^2 \leq 1$}.
Consider any PDC formula $\Phi = (\+P,\+Q,\+C)$ where $\Delta_{\Phi} \leq \Delta$, $p_\Phi\leq p$, $\abs{V_{\Phi}} \leq n$ that cannot be further factorized to smaller formulas with a \pname $P\in\+P$. Let $\sigma_1,\cdots,\sigma_{K}$ be independently drawn from $\mu_{\Phi_1},\cdots,\mu_{\Phi_{K}}$, $\sigma$ be the concatenation of $\sigma_1,\dots,\sigma_{K}$,
and $p_{\max} = \max_{C\in \+C(P)}\^P_\Phi[\neg C]$.
If $2\-e\Delta\leq \log(\-e^4n/\varepsilon)$ or $p_{\max}\abs{P}\Delta\leq 1/4$,
then 
$$\Pr[\sigma_{[K]}\sim \mu_{[K]}]{\sigma\in \Omega_{\Phi}} \geq \frac{\varepsilon}{\-e^4 n}.$$ 
\end{lemma}

By \Cref{condition-sample-small}, 
the formulas $\Phi_1,\Phi_2,\cdots,\Phi_{K-1}$ are small.
Then one can draw $\sigma_i\sim \mu_{\Phi_i}$ for each $i\in [K-1]$ efficiently by the rejection sampling. 
To draw $\sigma_K\sim \mu_{\Phi_K}$ for $\Phi_K = (\{P\},\{\+Q(P)\},\emptyset)$, it is sufficient to generate a random permutation of $\+Q(P)$ on $P$.
Combined with \Cref{lemma-simplecase-repeattime}, one can draw $\sigma\sim \mu$ with probability $1 -\varepsilon$ by repeating the aforementioned subroutine
for $\poly(n,1/\varepsilon)$ rounds.
This is realized in Lines \ref{line-if-sample}-\ref{line-break-sample} of \Cref{Alg:samplepermutation}.


\subsubsection{\emph{\textbf{Challenging factorized formulas}}}

The major challenge arises in the case where $2\-e\Delta\geq \log(\-e^4n/\varepsilon)$ and $p_{\max}\abs{P}\Delta\geq 1/4$. 
In this circumstance, 
the probability that the concatenation of $\sigma_1,\cdots,\sigma_{K}$ in the aforementioned procedure is a satisfying assignment of $\Phi$ can be very small and then the sampler is no longer efficient.
Thus, the subroutine in Lines \ref{line-if-sample}-\ref{line-break-sample} fails and new ideas are needed. 
One approach to address this challenge is to design a more efficient sampler that is built upon the assignments $\sigma_1,\cdots,\sigma_{K-1}$ sampled independently from $\mu_{\Phi_1},\cdots,\mu_{\Phi_{K-1}}$.


Specifically, given $\sigma_{[K-1]}$ where $\sigma_i$ is sampled from $\mu_{\Phi_i}$ for each $i\in [K-1]$,
all constraints $C\in \+C\setminus \+C(P)$ are satisfied.
This is because
$C$ must be present in one of the formulas $\Phi_1,\Phi_2,\cdots,\Phi_{K-1}$ and is then satisfied by one of $\sigma_1,\cdots,\sigma_{K-1}$.
In addition,
each unsatisfied constraint $C\in\+C(P)$ can be simplified to a constraints $C'$ defined on $\vbl(C)\cap P$, 
since all the variables in $\vbl(C)\setminus P$ are fixed.
Let 
\begin{align*}
\+C\left[\sigma_{[K-1]}\right] \triangleq 
\left\{\text{simplification of $C$ given $\sigma_{[K-1]}$} \mid C \in \+C \text{ and $C$ is unsatisfied under $\sigma_{[K-1]}$}\right\}
\end{align*}
denote the simplified constraints $C\in \+C$ which are unsatisfied under the condition that $V_{\Phi_i}$ is assigned as $\sigma_i$ for each $i\in [K-1]$.
We remark that $\vbl(C)\subseteq P$ for each $C\in \+C\left[\sigma_{[K-1]}\right]$ by the definition of $\+C\left[\sigma_{[K-1]}\right]$.

We then define the simplified formula with the sampled satisfying assignments on $\Phi_{[K-1]}$ as follows,
\begin{align}\label{eq-define-phi-sigma-kminusone}
\Phi\left[\sigma_{[K-1]}\right] = \left(\{P\},\{\+Q(P)\},\+C\left[\sigma_{[K-1]}\right]\right).
\end{align} A natural idea for sampling $\sigma\sim \mu$ is to sample the satisfying assignment $\tau$ of $\Phi\left[\sigma_{[K-1]}\right]$ uniformly at random, and output the concatenation of $\sigma_1,\cdots,\sigma_{K-1},\tau$. However, upon closer examination, one can see that the method exhibits inherent deviation due to the non-uniform nature of the marginal distribution on $P$ induced by $\mu$. To adjust the probability, we could further apply the technique of filtration. Therefore, one can draw $\sigma\sim \mu$ 
by repeating the following subroutine until success:
\begin{itemize}
\item draw $\sigma_1,\cdots,\sigma_{K-1}$ from $\mu_{\Phi_1},\cdots,\mu_{\Phi_{K-1}}$, independently;
\item draw $\tau$ from $\mu_{\Phi\left[\sigma_{[K-1]}\right]}$ where $\Phi\left[\sigma_{[K-1]}\right]$ is defined in \eqref{eq-define-phi-sigma-kminusone};
\item accept the concatenation of $\sigma_1,\dots,\sigma_{K-1},\tau$ as $\sigma$ with probability $\abs{\Omega_{\Phi\left[\sigma_{[K-1]}\right]}}/\+N$ where $\+N$ is an upper bound of $\abs{\Omega_{\Phi\left[\sigma_{[K-1]}\right]}}$ for all possible $\sigma_{[K-1]}$.
\end{itemize}
In this sampler, the concatenation of $\sigma_1,\dots,\sigma_{K-1},\tau$ is accepted as $\sigma$ with probability proportional to $\abs{\Omega_{\Phi\left[\sigma_{[K-1]}\right]}}$. It can be verified that $\sigma$ follows the distribution of $\mu$.

However, there are still some technical issues with the aforementioned sampler. First, the assignment $\tau$ should be sampled efficiently from $\mu_{\Phi\left[\sigma_{[K-1]}\right]}$. Additionally, the upper bound $\+N$ should not be too large to avoid frequent rejections of the concatenation of $\sigma_1,\dots,\sigma_{K-1},\tau$. We collectively tackle these challenges by reducing them to the counting and sampling task for {PRP}.

Let
\begin{align}\label{eq-define-rho-phione}
\+C^1\left[\sigma_{[K-1]}\right] \triangleq \left\{C\in \+C\left[\sigma_{[K-1]}\right]  \mid \abs{\vbl(C)} = 1 \right\},\quad \rho\left(\sigma_{[K-1]}\right) =
\abs{\+C^1\left[\sigma_{[K-1]}\right]},
\end{align}
and 
\begin{align}
\Phi^1\left[\sigma_{[K-1]}\right] \triangleq \left(\{P\},\{\+Q(P)\},\+C^1\left[\sigma_{[K-1]}\right]\right).
\end{align}
Obviously, 
$\+C^1\left[\sigma_{[K-1]}\right]\subseteq \+C\left[\sigma_{[K-1]}\right]$. Since each constraint $C\in \+C^1\left[\sigma_{[K-1]}\right]$ satisfies $\vbl(C) = 1$, one can verify that $\Phi^1\left[\sigma_{[K-1]}\right]$ is exactly a random PRP.
Therefore, instead of sampling the satisfying assignment from $\mu_{\Phi\left[\sigma_{[K-1]}\right]}$, we first draw a sample $\tau$ from $\mu_{\Phi^1\left[\sigma_{[K-1]}\right]}$ by filtration, and then accept it if $\tau \in \Omega_{\Phi\left[\sigma_{[K-1]}\right]}$. In addition, this method also resolves the challenge of the upper bound estimation in the filtration by calling the counting algorithm for {PRP}.
To summarize, we now sample $\sigma\sim \mu_{\Phi}$ by repeating the following subroutine until success:
\begin{itemize}
\item  draw $\sigma_1,\cdots,\sigma_{K-1}$ from $\mu_{\Phi_1},\cdots,\mu_{\Phi_{K-1}}$, independently;
\item draw a sample from $\mu_{\Phi^1\left[\sigma_{[K-1]}\right]}$ and accept the sample as $\tau$ with probability $\abs{\Omega_{\Phi^1\left[\sigma_{[K-1]}\right]}}/\+N$, where $\+N$ is an upper bound of $\abs{\Omega_{\Phi^1\left[\sigma_{[K-1]}\right]}}$ for all possible $\sigma_{[K-1]}$;
\item if $\tau$ is also a solution of $\Phi\left[\sigma_{[K-1]}\right]$,
accept the concatenation of $\sigma_1,\dots,\sigma_{K-1},\tau$ as $\sigma$ .
\end{itemize}
One can verify that $\sigma$ sampled from the above sampler follows the distribution of $\mu$.

To make this sampler efficient,
the following conditions are also needed:
\begin{enumerate}
\item $\tau \sim \mu_{\Phi^1\left[\sigma_{[K-1]}\right]}$ can be drawn efficiently;
\item the probability that 
$\tau$ is a solution of $\Phi\left[\sigma_{[K-1]}\right]$ is
lower bounded by $1/\poly(n,1/\varepsilon)$;
\item for all possible $\Phi^1\left[\sigma_{[K-1]}\right]$ and some positive constant $c$
\begin{align}\label{eq-upperbound}
\left(\frac{n}{\varepsilon}\right)^{-c} \leq \frac{\abs{\Omega_{\Phi^1\left[\sigma_{[K-1]}\right]}}}{\+N} \leq 1.
\end{align}

\end{enumerate}
For the first item, $\tau \sim \mu_{\Phi^1\left[\sigma_{[K-1]}\right]}$ can be drawn by  
calling the sampling algorithm for PRP, recalling that $\Phi^1\left[\sigma_{[K-1]}\right]$ is exactly an instance of PRP.
For the second item, 
we have that the majority of solutions of $\Phi^1\left[\sigma_{[K-1]}\right]$ also satisfy the constraints of $\Phi\left[\sigma_{[K-1]}\right]$ by the following lemma.
The lemma is immediate by the lopsided LLL, the proof of which is deferred to~\Cref{subsec:correctness}.

\begin{lemma}\label{lemma-solution-phiprime-phiprimeprime}
Let $\Delta\geq 1$. Consider any formula $\Phi = (\{P\},\+Q,\+C)$ where $\Delta_\Phi\leq \Delta$ satisfying $\abs{P}\geq 2\-e\Delta$.
Let $\+C'$ denote $\{C\in \+C\mid \abs{\vbl(C)} = 1\}$ and $\Phi' = (\{P\},\+Q,\+C')$. It holds that
\begin{align*}
\Pr[\sigma\sim \mu_{\Phi'}]{\sigma\in \Omega_{\Phi}} \geq 1 - \frac{2\Delta}{\abs{P}-1}.
\end{align*}
\end{lemma}


Now, let's consider the last item, which requires that $\abs{\Omega_{\Phi^1\left[\sigma_{[K-1]}\right]}}/\+N$ is lower bounded by $1/\poly(n,1/\varepsilon)$.
The main challenge to this item is that the following two conditions cannot be satisfied at the same time:
\begin{itemize}
\item $\+N$ is an upper bound of $\abs{\Phi^1\left[\sigma_{[K-1]}\right]}$ for all possible $\Phi^1\left[\sigma_{[K-1]}\right]$;
\item $\abs{\Omega_{\Phi^1\left[\sigma_{[K-1]}\right]}}/\+N$ is lower bounded by $1/\poly(n,1/\varepsilon)$ for all possible $\sigma_{[K-1]}$.
\end{itemize}
Because for some certain special $\sigma_{[K-1]}$,
the set $\+C^1\left[\sigma_{[K-1]}\right]$ can be 
$\emptyset$, leading to $\abs{\Omega_{\Phi^1\left[\sigma_{[K-1]}\right]}} = \abs{P} ! $.
However, for some other $\sigma_{[K-1]}$, $\abs{\Omega_{\Phi^1\left[\sigma_{[K-1]}\right]}}$ can be as small as $\abs{P} !/\exp(-\Delta)$.

Therefore, rather than a general upper bound of $\abs{\Omega_{\Phi^1\left[\sigma_{[K-1]}\right]}}$ for all possible $\sigma_{[K-1]}$, 
we aim to obtain an estimator $\+N$ where $\+N \geq \abs{\Omega_{\Phi^1\left[\sigma_{[K-1]}\right]}}$ for most $\sigma_{[K-1]}$.
In this way, we can deal with the majority of $\sigma_{[K-1]}$ and then sample a large portion of solutions of the original formula $\Phi$ uniformly. Since the unsampled solutions originate from the negligible portion of $\sigma_{[K-1]}$, their absence only introduces a slight deviation in the distribution.
Formally, we try to obtain an estimator $\+N$
such that
\begin{align}\label{eq-condition-on-uppperbound}
\Pr[\sigma_{[K-1]}\sim \mu_{[K-1]}]{\left(\frac{n}{\varepsilon}\right)^{-c} \leq \frac{\abs{\Omega_{\Phi^1\left[\sigma_{[K-1]}\right]}}}{\+N} \leq 1}\geq  1 - 1/\poly(n,1/\varepsilon).
\end{align}

Our estimator is defined as follows.
Recall the definition of {$g(x,y)$ in \eqref{eq-define-gxy} and the definition of $\rho\left(\sigma_{[K-1]}\right)$ in \eqref{eq-define-rho-phione}.}
For any $\sigma_{[K-1]}$ where $\rho\left(\sigma_{[K-1]}\right)$ satisfies \eqref{eq-condition-on-upperbound},
the following lemma shows that $g(x,\abs{P})$ is an estimator satisfying  \eqref{eq-upperbound}.

\begin{lemma}[Estimator for the number of solutions, {restatement of \Cref{thm-estimator-num-prp}}]\label{lemma-estimation-number-solution}
Let integer $\Delta\geq 1$. Consider any PDC formula $\Phi = (\+P,\+Q,\+C)$ where $\Delta_{\Phi} \leq \Delta$ that cannot be further factorized to smaller formulas with a \pname $P\in\+P$. Assume {$\abs{P}\geq 5\Delta^2+2$}. For any satisfying assignment $\sigma_{[K-1]}$ of $\Phi_{[K-1]}$, we have 
\begin{align}\label{eqn-bounds-estimation}
     \frac{g(\rho\left(\sigma_{[K-1]}\right),\abs{P})}{\sqrt{2\pi\abs{P}} \; \-e^{2}}\leq \abs{\Omega_{\Phi^1\left[\sigma_{[K-1]}\right]}}\leq  g(\rho\left(\sigma_{[K-1]}\right),\abs{P}).
\end{align}
In addition, for any {nonnegative} integer $x$ where 
\begin{align}\label{eq-condition-on-upperbound}
\rho\left(\sigma_{[K-1]}\right) - {\abs{P}}/{2}\leq x\leq \rho\left(\sigma_{[K-1]}\right) \leq \abs{P}\Delta,
\end{align}
we have
\begin{align*}
\frac{g(x,\abs{P})}{\sqrt{2\pi \abs{P}} \;\-e^3} \leq {\abs{\Omega_{\Phi^1\left[\sigma_{[K-1]}\right]}}} \leq g(x,\abs{P}).
\end{align*}
\end{lemma}
Lemma \ref{lemma-estimation-number-solution}
is merely the restatement of~\Cref{thm-estimator-num-prp} and its correctness can be verified by modifying the statements in the proof of~\Cref{thm-estimator-num-prp}.
Combining with \Cref{lemma-estimation-number-solution}, to obtain an upper bound $\+N$ satisfying \eqref{eq-condition-on-uppperbound},
one only needs to obtain a random $x>0$ with probability at least $1 - 1/\poly(n,1/\varepsilon)$ such that 
\begin{align}\label{eq-condition-upperbound-on-x}
\Pr[\sigma_{[K-1]}\sim \mu_{[K-1]}]{x\leq \rho\left(\sigma_{[K-1]}\right) \leq x + {\abs{P}}/{2}}\geq  1 - 1/\poly(n,1/\varepsilon).
\end{align}
and set $\+N$ as $g(x,\abs{P})$.
Fortunately, 
$\rho\left(\sigma_{[K-1]}\right)$ where 
$\sigma_{[K-1]}\sim \mu_{[K-1]}$ exhibits good concentration.

\begin{lemma}[Concentration property of $\rho$]\label{lemma-concentration-rho}
Consider any instance satisfying Condition~\ref{condition-sample-small}, and assume {$\abs{P}\geq 1728\-e^3\ln2 \cdot k^3L^2\Delta^8$}, $2\-e\Delta\geq \log(\-e^4n/\varepsilon)$. It holds that
\begin{align}\label{eq-concentration}
\Pr[\sigma_{[K-1]}\sim \mu_{[K-1]}]{\abs{\rho\left(\sigma_{[K-1]}\right) - \E{\rho\left(\sigma_{[K-1]}\right)}}\geq \abs{P}/12} \leq {\exp(-6\-e\ln 2\cdot  \Delta)}
\end{align}
\end{lemma}

\Cref{lemma-concentration-rho} is proved by McDiarmid's inequality, and its proof is deferred to~\Cref{subsec:correctness}.
The intuition can be stated as follows.
Given any instance $(\Phi,P,\varepsilon)$ satisfying \Cref{condition-sample-small} where $\abs{P}\geq 1728\-e^3\ln2 \cdot k^3L^2\Delta^8$
and the factorized formulas  
$\Phi_1,\dots,\Phi_{K-1}$ are small,
each solution $\sigma_i$ of $\Phi_i$ has minor impact on $\rho\left(\sigma_{[K-1]}\right) =
\abs{\+C^1\left[\sigma_{[K-1]}\right]}$. 
While $\sigma_1,\dots,\sigma_{K-1}$ are drawn from $\mu_{\Phi_1},\dots,\mu_{\Phi_{K-1}}$ independently, $\rho\left(\sigma_{[K-1]}\right)$ concentrations around its mean with high probability.

By \Cref{lemma-concentration-rho}, one can draw $\sigma_{[K-1]}$ from  $\mu_{[K-1]}$ randomly and set $x = \max\left\{\rho\left(\sigma_{[K-1]}\right)-\abs{P}/6, 0\right\}$
and $\+N=g(x,\abs{P})$.
It is easy to verify that \eqref{eq-condition-upperbound-on-x} holds.
With $x$ and $\+N$, one can draw $\sigma$ distributed approximately as $\mu$ by repeating the following subroutine until success: 
\begin{itemize}
\item  draw $\sigma_1,\cdots,\sigma_{K-1}$ from $\mu_{\Phi_1},\cdots,\mu_{\Phi_{K-1}}$, independently.
\item if \eqref{eq-condition-on-upperbound} is satisfied, draw a sample from $\mu_{\Phi^1\left[\sigma_{[K-1]}\right]}$ and accept the sample as $\tau$ with probability $\abs{\Omega_{\Phi^1\left[\sigma_{[K-1]}\right]}}/\+N$; otherwise, skip the last step.
\item if $\tau$ is also a solution of $\Phi\left[\sigma_{[K-1]}\right]$,
accept the concatenation of $\sigma_1,\dots,\sigma_{K-1},\tau$ as $\sigma$.
\end{itemize}
By \Cref{lemma-estimation-number-solution} and \eqref{eq-condition-upperbound-on-x},
one can draw $\sigma\sim \mu_{\Phi}$ with probability $1 -\varepsilon$ by repeating the above subroutine
for $\poly(n,1/\varepsilon)$ rounds.
This sampler is realized in Lines \ref{line-sample-core-begin}-\ref{line-sample-core-end} of \Cref{Alg:samplepermutation}.

\section{Analysis of the sample ~ subroutine}\label{sec-analysis-sample}

In this section, we first analyze the $\samplepermutation(\cdot)$~subroutine for the instance satisfying Conditions~\ref{condition-sample-small}.

\subsection{The correctness of $\samplepermutation$ subroutine}\label{subsubsection-correctness}

The following is the main theorem in this section.
\begin{theorem}\label{thm-sample-correctness}
    Consider any instance satisfying Condition~\ref{condition-sample-small}. Let $p_{\max}=\max_{C\in \+C(P)}\^P_\Phi(\neg C)$.
    If either $p_{\max}\abs{P}\Delta\leq 1/4$ or $\abs{P}\geq 1728\e^3\ln2 \cdot k^3L^2\Delta^8$, the $\samplepermutation(\Phi,P,\varepsilon)$ returns a random valid assignment $\sigma$ of $\Phi$ satisfying $\dtv(\sigma,\mu_\Phi)\leq \varepsilon$ in time $\widetilde{O}(n^4/\varepsilon^2)$.
\end{theorem}
As mentioned in~\Cref{subsection-sample}, the algorithmic techniques differ for the \easycase and the \hardcase. These differences are captured by Lemmas~\ref{lem-sample-correctness-caseI} and \ref{lem-sample-correctness-case2}, respectively, which imply~\Cref{thm-sample-correctness} immediately. Before that, we first analyze the performance of the rejection sampling algorithm on the factorized formula, which serves as a basic subroutine throughout the sampling algorithm.
\begin{lemma}\label{lem-rejectionsample-small-factorized}
    Consider any instance satisfying Condition~\ref{condition-sample-small}. For each $i\in [K-1]$, we have
    \begin{align*}
        \^P_{\Phi_i}\left[\bigwedge_{C\in \+C_i} C \right]\geq \frac{\varepsilon}{n}.
    \end{align*}
\end{lemma}
\begin{proof}
    By \Cref{condition-sample-small}, we have $\abs{\+C_i}\leq \sampleblocksize$ for each $i\in [K-1]$.
    Thus, we have
    \begin{align*}
        \^P_{\Phi_i}\left[\bigwedge_{C\in \+C_i} C \right]&\geq \prod_{C\in \+C_i }
        \left(1- \-e\^P_{\Phi_i}[\neg C] \right)\geq \left(1-\e p\right)^{\Delta\log(n/\varepsilon)}\geq 2^{-\log(n/\varepsilon)}=\frac{\varepsilon}{n},
    \end{align*} where the first inequality holds by~\Cref{prop-PDC-local-uniformity}, and the last inequality holds by~\Cref{equality-ab} and $2\-ep\Delta\leq 1$.
\end{proof}

\subsubsection{\emph{\textbf{Analysis for the \easycase.}}}
Suppose one of the following conditions holds: $2\-e\Delta\leq \log(\e^4n/\varepsilon)$ or $p_{\max}\abs{P}\Delta\leq 1/4$ where $p_{\max}=\max_{C\in \+C(P)}\^P_\Phi(\neg C)$. As stated in ~\Cref{subsection-sample}, the {accept-reject sampler exhibits a high acceptance rate, as shown in ~\Cref{lemma-simplecase-repeattime}. Combining with ~\Cref{lem-rejectionsample-small-factorized}, we obtain the following lemma.

\begin{lemma}\label{lem-sample-correctness-caseI}
    Consider any instance satisfying \Cref{condition-sample-small}, and suppose any of the following conditions hold: $2\-e\Delta\leq \log(\e^4n/\varepsilon)$ or $p_{\max}\abs{P}\Delta\leq 1/4$ where $p_{\max}=\max_{C\in \+C(P)}\^P_\Phi(\neg C)$.
    The $\samplepermutation(\Phi,P,\varepsilon)$ returns a random valid assignment $\sigma$ of $\Phi$ satisfying $\dtv(\sigma,\mu)\leq \varepsilon$.
    {Furthermore, the time complexity of $\samplepermutation(\Phi,P,\varepsilon)$ is $\widetilde{O}\left( n^2/\varepsilon^2\right)$}. 
\end{lemma}
\begin{proof}
    We compare $\samplepermutation(\Phi,P,\varepsilon)$ with the idealized rejection sampler, which repeats the following process for $T\gets \lceil \e^4n/\varepsilon \cdot  \log(2/\varepsilon) \rceil$ times until success:
    \begin{itemize}
        \item draw $\sigma'_1,\cdots,\sigma'_{K}$ from $\mu_{\Phi_1},\cdots,\mu_{\Phi_{K}}$, independently;
        \item accept the concatenation of $\sigma'_1,\dots,\sigma'_{K}$ as $\sigma'$ if the concatenation is a solution of $\Phi$.
    \end{itemize}
    {According to Lemmas~\ref{lemma-rejection-sample} and~\ref{lemma-simplecase-repeattime}, we have $\dtv(\sigma',\mu_\Phi)\leq {\varepsilon}/{2}$.}    

    To prove that $\dtv(\sigma,\mu_\Phi)\leq \varepsilon$, {it suffices to show that $\dtv(\sigma,\sigma')\leq {\varepsilon}/{2}$ by establishing a coupling according to the ~\Cref{prop:coupling}.} 
    Throughout the execution of $\samplepermutation(\Phi,P,\varepsilon)$, let $\+B_j$ denote the event that there are some $i\in[K-1]$ such that $\sigma_i$ is returned at~\Cref{line:fail} in $\truncatedsampling(\Phi_i,\lceil (n/\varepsilon)\cdot \log  (2nT/\varepsilon)\rceil)$ at the iteration $j$ for each $j\in [T]$, and $\+B\triangleq \bigcup_{j\in [T]}\+B_j$. 
    The random assignment $\sigma$ and $\sigma'$ can be coupled perfectly if $\+B$ does not occur. Consequently, the coupling error can be bounded by the probability of the event $\+B$. According to Lemma~\ref{lem-rejectionsample-small-factorized},~\ref{lemma-rejection-sample}, and the union bound, for each iteration $j\in [T]$,
    \begin{align*}
        \Pr{\+B_j}\leq \sum_{i\in [K-1]} \left(1-\frac{\varepsilon}{n}\right)^{\lceil (n/\varepsilon)\cdot \log  (2nT/\varepsilon) \rceil}\leq \sum_{i\in [K-1]} \frac{\varepsilon}{2nT}\leq \frac{\varepsilon}{2T}.
    \end{align*}
    Therefore, the probability of $\+B$ can be bounded by $\varepsilon/2$.
    
    {One can verify that the running time of the algorithm is $\widetilde{O}\left( n^2/\varepsilon^2\right)$.}
\end{proof}


\subsubsection{\emph{\textbf{Analysis for the \hardcase.}}}
In this part, we analyze $\samplepermutation(\Phi,P,\varepsilon)$ subroutine given that {$2\-e\Delta\geq \log(\-e^4n/\varepsilon)$, and $p_{\max}\abs{P}\Delta> 1/4$} where $p_{\max} = \max_{C\in \+C(P)}\^P_\Phi[\neg C]$.}
The main result is the following lemma.

\begin{lemma}\label{lem-sample-correctness-case2}
    Consider instances satisfying Condition~\ref{condition-sample-small}, and assume {$\abs{P}\geq 1728\-e^3\ln2 \cdot k^3L^2\Delta^8, 2\-e\Delta\geq \log(\-e^4n/\varepsilon), p_{\max}\abs{P}\Delta> 1/4$} where $p_{\max} = \max_{C\in \+C(P)}\^P_\Phi[\neg C]$. The $\samplepermutation(\Phi,P,\varepsilon)$ returns a random valid assignment $\sigma$ of $\Phi$ satisfying $\dtv(\sigma,\mu)\leq \varepsilon$.
    {The {time complexity} of $\samplepermutation(\Phi,P,\varepsilon)$ is $\widetilde{O}\left(n^{4}/{\varepsilon^2} \right)$.}
\end{lemma}

{
In the subsequent discussion, we frequently refer to the instances satisfying the following condition.}
\begin{condition}\label{cond:estimation}
    Let integer $\Delta\geq 1$.
    Consider any PDC formula $\Phi = (\+P,\+Q,\+C)$ where $\Delta_{\Phi} \leq \Delta$ that cannot be further factorized to smaller formulas with a \pname $P\in\+P$. Let $\Phi_1,\dots,\Phi_{K}$ be the factorization of $(\+P,\+Q,\+C\setminus \+C(P))$ where $\Phi_{K}=(\{P\},\{\+Q(P)\},\emptyset)$, and $\Phi_i=(\+P_i,\+Q_i,\+C_i)$ for each $i\in [K-1]$.
\end{condition}

Recall the definitions of $\Phi\left[\sigma_{[K-1]}\right]$, $\Phi^1\left[\sigma_{[K-1]}\right]$, $\rho\left(\sigma_{[K-1]}\right)$ and $g(x,y)$ in \Cref{subsection-sample}. 
Consider any instance satisfying~\Cref{cond:estimation}. We use $\Omega_{[K-1]}$ to denote the set of satisfying assignments $\sigma_{[K-1]}$ for $\Phi_{[K-1]}$. 
Given a subset of satisfying assignments $\sigma_{[K-1]}$ of $\Phi_{[K-1]}$, denoted as $\widehat{\Omega}$, we introduce the following idealized rejection sampler~\Cref{Alg:idealsamplepermutation}, which is parameterized by $\widehat{\Omega}$, a positive integer $\+N$, and real numbers $\alpha,\beta \in [0,1]$. This sampler is designed for the convenience of our analysis.
We will prove \Cref{lem-sample-correctness-case2} by comparing \Cref{Alg:idealsamplepermutation} with
\Cref{Alg:samplepermutation}. 

\begin{algorithm}
\caption{$\idealsamplepermutation(\Phi,P,\varepsilon)$} \label{Alg:idealsamplepermutation}
    \KwIn{a \cspformula formula $\Phi=(\+P,\+Q,\+C)$, a \pname $P\in \+P$ satisfying~\Cref{cond:estimation}, and a parameter $0< \varepsilon<1$.}
    \KwOut{a random valid assignment of $\Phi$}
    Let $\Phi_1,\dots,\Phi_{K}$ be the factorization of $(\+P,\+Q,\+C\setminus \+C(P))$ where $\Phi_{K}=(\{P\},\{\+Q(P)\},\emptyset)$\;
    $T\gets \lceil{2/(\alpha\beta)\cdot \log (6/\varepsilon)\rceil}$\;
    \For{$j=1$ to {$T$}}
    {   
        draw $\sigma_1,\cdots,\sigma_{K-1}$ from $\mu_{\Phi_1},\cdots,\mu_{\Phi_{K-1}}$ independently and uniformly at random;
        
        \If{$\sigma_{[K-1]}\in \widehat{\Omega}$}
        {
            $\Phi'\leftarrow \left(\{P\},\{\+Q(P)\},\+C\left[\sigma_{[K-1]}\right]\right), 
            \Phi'' \leftarrow \left(\{P\},\{\+Q(P)\}, \+C^1\left[\sigma_{[K-1]}\right]\right)$\;
            draw a satisfying assignment $\tau$ of the formula $\Phi''$\;\label{line-coupletau2}
            \If{$\tau \in \Omega_{\Phi'}$}
            {
                $r\sim \!{Unif}[0,1]$\;     
                \If{$r\leq \abs{\Omega_{\Phi''}}/\+N$}
                {
                    \Return the concatenation of $\sigma_1,\dots,\sigma_{K-1},\tau$ as $\sigma$\;\label{line-ideal-sample-core-end}
                }
            }
        }
    }    
    \Return an arbitrary valid assignment $\sigma$ of $\Phi$\;\label{line-ideal-failure-of-sample}
\end{algorithm}

The following lemma bounds the error of the idealized sampler in ~\Cref{Alg:idealsamplepermutation}.
\begin{lemma}\label{lem:analysis-idealsample}
    Consider any instance satisfying~\Cref{cond:estimation}. Assume {$\-e p_\Phi\Delta\leq 1$ and $\abs{P}\geq 5\Delta^2+2$}. Given any {non-empty subset $\widehat{\Omega}\subseteq\Omega_{[K-1]}$}, a positive integer $\+N$, and $\alpha,\beta\in [0,1]$ such that for any $\sigma_{[K-1]}\in \widehat{\Omega}$
    \begin{align}\label{eqn-ideal-input-condition}
         \frac{\abs{\widehat{\Omega}}}{\abs{\Omega_{[K-1]}}}\geq \alpha, \ \ \ 
        \beta\leq \frac{\abs{\Omega_{\Phi^1\left[\sigma_{[K-1]}\right]}}}{\+N}\leq 1,
    \end{align}
    the $\idealsamplepermutation(\Phi,P,\epsilon)$ returns a random valid assignment of $\Phi$ satisfying 
    $$\dtv(\sigma,\mu)\leq \left(1-\alpha\right)\cdot 2 \sqrt{2\pi \abs{P}} \cdot \exp(3+2\Delta) + \varepsilon/6.$$
\end{lemma}
\begin{proof}
    Let $\widehat{\Omega}_0$ be the subset of concatenation of satisfying assignments $\sigma_1,\sigma_2,\cdots,\sigma_{K}$ of $\Phi_{K}$ such that $\sigma[K-1]\in \widehat{\Omega}$. In~\Cref{Alg:idealsamplepermutation}, we can see that $\sigma\in \widehat{\Omega}_0$ when it is returned at ~\Cref{line-ideal-sample-core-end}.    
    The probability that the algorithm outputs $\sigma$ is proportional to 
    \begin{align*}
       \frac{\abs{\widehat{\Omega}}}{\abs{\Omega_{{[K-1]}}}}\cdot {\abs{\widehat{\Omega}}^{-1}}\cdot {\abs{\Omega_{\Phi^1\left[\sigma_{[K-1]}\right]}}}^{-1}\cdot \abs{\Omega_{\Phi\left[\sigma_{[K-1]}\right]}}\cdot {\abs{\Omega_{\Phi\left[\sigma_{[K-1]}\right]}}^{-1}}\cdot \frac{\abs{\Omega_{\Phi^1\left[\sigma_{[K-1]}\right]}}}{\+N}=\frac{1}{\abs{\Omega_{{[K-1]}}}\cdot \+N},
    \end{align*} which implies that $\sigma$ follows the uniform distribution over all satisfying assignments in $\widehat{\Omega}_0$, conditioned on $\sigma$ being returned at ~\Cref{line-ideal-sample-core-end}.
    
    Let $E_1$ be the collection of satisfying assignments in $\Omega \setminus \widehat{\Omega}_0$, and $E_2$ be the event that $\sigma$ is returned at ~\Cref{line-ideal-failure-of-sample} in ~\Cref{Alg:idealsamplepermutation}. By the coupling argument with~\Cref{prop:coupling}, one can verify that 
    \begin{align}\label{eqn-dtv-ideal-decomposition}
        \dtv(\sigma,\mu)\leq \mu(E_1)+\Pr{E_2},
    \end{align} {where $\mathbf{Pr}$ is the randomness in ~\Cref{Alg:idealsamplepermutation}.} Therefore, it boils down to establishing the upper bounds for $\mu(E_1)$ and $\Pr{E_2}$.

    
    
    We define $A$ and $B$ as the maximum and minimum sizes of $\Omega_{\Phi\left[\sigma_{[K-1]}\right]}$, respectively, as follows:
    \begin{align*}
        A\triangleq \max_{\sigma_{[K-1]}\in \Omega_{[K-1]}} \abs{\Omega_{\Phi\left[\sigma_{[K-1]}\right]}}, \quad B\triangleq \min_{\sigma_{[K-1]}\in \Omega_{[K-1]}} \abs{\Omega_{\Phi\left[\sigma_{[K-1]}\right]}}.
    \end{align*}
    According to ~\Cref{lemma-estimation-number-solution}, we have
    \begin{align*}
        A\leq \max_{\sigma_{[K-1]}\in \Omega_{[K-1]}} \abs{\Omega_{\Phi^1\left[\sigma_{[K-1]}\right]}}\leq g(0,\abs{P}).
    \end{align*} 
    Meanwhile, by Lemma~\ref{lemma-solution-phiprime-phiprimeprime} and~\ref{lemma-estimation-number-solution}, we have
    \begin{align*}
        B\geq  \min_{\sigma_{[K-1]}\in \Omega_{[K-1]}} \abs{\Omega_{\Phi^1\left[\sigma_{[K-1]}\right]}}\cdot\frac{\abs{P}-1-2\Delta}{\abs{P}-1}\geq  g(\abs{P}{\Delta},\abs{P})\cdot \frac{\abs{P}-1-2\Delta}{\abs{P}-1}\cdot \frac{1}{\sqrt{2\pi \abs{P}} \;\-e^2}.
    \end{align*}
    Therefore, it holds that
    \begin{align*}
        \mu_{\Phi}(E_1) &=  \left({\sum_{\sigma_{[K-1]}\in \Omega_{{[K-1]}}\setminus \widehat{\Omega}}\abs{\Omega_{\Phi\left[\sigma_{[K-1]}\right]}}}\right)\bigg/\left({\sum_{\sigma_{[K-1]}\in \Omega_{{[K-1]}}} \abs{\Omega_{\Phi\left[\sigma_{[K-1]}\right]}}}\right)\\
        &\leq  \left({\sum_{\sigma_{[K-1]}\in \Omega_{{[K-1]}}\setminus \widehat{\Omega}}A}\right)\bigg/\left({\sum_{\sigma_{[K-1]}\in \Omega_{{[K-1]}}} B}\right)\\
        &\leq \left(1-\frac{\abs{\widehat{\Omega}}}{\abs{\Omega_{{[K-1]}}}}\right)\cdot \frac{\abs{P}-1}{\abs{P}-1-2\Delta}\cdot \sqrt{2\pi \abs{P}} \;\-e^3\cdot \left(1-\frac{\Delta}{\abs{P}}\right)^{-\abs{P}}\\
        &\leq \left(1-\alpha\right)\cdot 2\cdot \sqrt{2\pi \abs{P}} \cdot \exp(3+2\Delta) \tag{by $\frac{\abs{P}-1}{\abs{P}-1-2\Delta}\leq 2$ and~\Cref{equality-ab}}.
    \end{align*}
    
    We then consider the probability of the event $E_2$.
    Note that in each iteration of the loop $j$, the algorithm outputs $\sigma$ at ~\Cref{line-ideal-sample-core-end} with probability
    \begin{align*}
         \frac{\abs{\widehat{\Omega}}}{\abs{\Omega_{\Phi_{[K-1]}}}}\cdot \frac{\abs{\Omega_{\Phi\left[\sigma_{[K-1]}\right]}}}{\abs{\Omega_{\Phi^1\left[\sigma_{[K-1]}\right]}}}\cdot \frac{\abs{\Omega_{\Phi^1\left[\sigma_{[K-1]}\right]}}}{\+N} \geq \frac{\alpha\beta}{2},
    \end{align*} which holds {by ~\Cref{lemma-solution-phiprime-phiprimeprime} and $\frac{\abs{P}-1-2\Delta}{\abs{P}-1}\geq \frac{1}{2}$}.
    Consequently, we have
    \begin{align*}
        \Pr{E_2}\leq \left(1-\frac{\alpha\beta}{2}\right)^{\lceil{2/(\alpha\beta)\cdot \log (6/\varepsilon)\rceil}}\leq \varepsilon/6.
    \end{align*}
    
Plugging the upper bound of $\mu(E_1)$ and $\Pr{E_2}$ in \eqref{eqn-dtv-ideal-decomposition}, the proof is immediate.
\end{proof}

Now we are ready to finish the proof of~\Cref{lem-sample-correctness-case2}.

\begin{proof}[Proof of ~\Cref{lem-sample-correctness-case2}]
    We analyze $\samplepermutation(\Phi, P,\varepsilon)$ by coupling \Cref{Alg:samplepermutation} and \Cref{Alg:idealsamplepermutation}.
    

    Let $B_1$ be the event that there exists some $i\in[K-1]$ such that $\sigma_i$ is returned at ~\Cref{line:fail} in the subroutine {$\truncatedsampling(\Phi_i,\lceil (n/\varepsilon)\cdot \log  (12n/\varepsilon)\rceil)$}, some $t\in [K-1]$ such that $\sigma_{t}$ is returned at ~\Cref{line:fail} in the subroutine {$\truncatedsampling(\Phi_{t},\lceil(n/\varepsilon)\cdot \log  (6nT/\varepsilon)\rceil)$}, some $j\in [T]$ such that
    $$\widehat{N}\notin\left[(1-\frac{\varepsilon}{12T})\cdot \abs{\Omega_{\Phi''}},(1+\frac{\varepsilon}{12T})\cdot \abs{\Omega_{\Phi''}}\right]$$
    in $j$-th iteration. According to Lemma~\ref{lem-rejectionsample-small-factorized} and union bound, one can verify that $\Pr{B_1}\leq  5\varepsilon/12$. In addition, if the event $B_1$ does not occur, we have $\sigma_{[K-1]}\sim \mu_{[K-1]}$ by~\Cref{lemma-rejection-sample}, and the estimation of $\abs{\Omega_{\Phi''}}$ with accuracy guarantee.
    
    We also define
    $B_2$ as the event that 
    \begin{align*}
         x\notin \left[\E[\sigma_{[K-1]}\sim \mu_{[K-1]}]{\rho\left(\sigma_{[K-1]}\right)}-\abs{P}/4,\E[\sigma_{[K-1]}\sim \mu_{[K-1]}]{\rho\left(\sigma_{[K-1]}\right)}-\abs{P}/12\right],
    \end{align*}
    By~\Cref{lemma-concentration-rho}, we have {$\Pr{B_2\mid \ol{B_1}}\leq \exp(-6\-e\ln 2\cdot  \Delta)$}. Given that the event $B_2$ does not occur, for any $\sigma_{[K-1]}$ such that 
    \begin{align*}
        \abs{\rho\left(\sigma_{[K-1]}\right)-\E[\sigma_{[K-1]}\sim \mu_{[K-1]}]{\rho\left(\sigma_{[K-1]}\right)}}\leq \abs{P}/12,
    \end{align*}
    it holds that 
    \begin{align*}
        \rho\left(\sigma_{[K-1]}\right)-\abs{P}/2\leq x\leq \rho\left(\sigma_{[K-1]}\right),
    \end{align*} and $\abs{\Omega_{\Phi^1\left[\sigma_{[K-1]}\right]}}/\+N\geq \sqrt{2\pi \abs{P}}\-e^3$ by~\Cref{lemma-estimation-number-solution}.
    Again, by~\Cref{lemma-concentration-rho}, we have 
    \begin{align*}
        \Pr[\sigma_{[K-1]}\sim \mu_{[K-1]}]{x\leq \rho\left(\sigma{[K-1]}\right)\leq x+\abs{P}/2}\geq {1-\exp(-6\-e\ln 2\cdot  \Delta)}.
    \end{align*}

    Then we establish a coupling between $\samplepermutation(\Phi,P,\varepsilon)$ and $\idealsamplepermutation(\Phi,P,\varepsilon)$ given that the events $B_1$ and $B_2$ do not occur. Let $\widehat{\Omega}= \set{\sigma_{[K-1]}\in \Omega_{\Phi_{[K-1]}}\mid x\leq \rho\left(\sigma{[K-1]}\right)\leq x+\abs{P}/2}$, $\+N=g(x,\abs{P})$, {$\alpha=1-\exp(-6\-e\ln 2\cdot  \Delta)$}, and $\beta={1}/{(\sqrt{2\pi \abs{P}} \;\-e^3)}$. One can verify that these satisfiying the conditions in~\Cref{lem:analysis-idealsample} by the aforementioned properties given that $B_2$ does not occur.    
    Let $\tau_1$ be the assignment at~\Cref{line-coupletau1} in $\samplepermutation(\Phi,P,\varepsilon)$, and $\tau_2$ be the assignment at~\Cref{line-coupletau2} in $\idealsamplepermutation(\Phi,P,\varepsilon)$. The coupling is specified as follows:
    \begin{itemize}
        \item share the randomness of assignment $\sigma_{[K-1]}$ and $r$;
        \item couple $\tau_1,\tau_2$ by the optimal coupling if $r\leq \left(1-\frac{\varepsilon}{12T}\right)\cdot \frac{\abs{\Omega_{\Phi''}}}{\+N}$ or $r\geq \left(1+\frac{\varepsilon}{12T}\right)\cdot \frac{\abs{\Omega_{\Phi''}}}{\+N}$.
    \end{itemize}
    One can verify that the coupling errors come from the event that $\tau_1\neq\tau_2$ and 
    $$r\in \left[(1-\frac{\varepsilon}{12T})\cdot \abs{\Omega_{\Phi''}}\big/{\+N},(1+\frac{\varepsilon}{12T})\cdot \abs{\Omega_{\Phi''}}\big/{\+N}\right].$$
    By~\Cref{prop:coupling}, given that the event $B_1$ and $B_2$ do not occur, we have $$\dtv(\samplepermutation(\Phi,P,\varepsilon),\idealsamplepermutation(\Phi,P,\varepsilon))\leq \frac{\varepsilon}{3}.$$
    
    Combining all these with~\Cref{lem:analysis-idealsample}, we have 
    \begin{align*}
        \dtv(\samplepermutation(\Phi,P,\varepsilon),\mu_\Phi)&\leq \frac{11\varepsilon}{12} +{\exp(-6\-e\ln 2\cdot  \Delta)\cdot 2 \sqrt{2\pi \abs{P}} \cdot \exp(3+2\Delta)+\exp(-6\-e\ln 2\cdot  \Delta)}\\
        &\leq \frac{11\varepsilon}{12} + 2\exp(-6\-e\ln 2\cdot  \Delta)\cdot 2 \sqrt{2\pi n} \cdot \exp(3+2\Delta)\\
        &\leq \varepsilon. \tag{by $2\-e\Delta\geq \log(\-e^4n/\varepsilon)$}
    \end{align*}

    {According to~Theorems~\ref{thm-sampler-PRP-LLL} and~\ref{thm-estimator-PRP-LLL}, one can verify that the running time of~\Cref{Alg:samplepermutation} is $\widetilde{O}(n^4/\varepsilon^2)$.}
\end{proof}

\section{Rapid mixing of the idealized permutation-wise Glauber dynamics}\label{sec:rapid-mixing}

Given any \csppformula formula $\Phi=(\+P,\+Q,\+C)$ with width $k=k_\Phi$, constraint degree $\Delta=\Delta_\Phi$, and a decomposition $\+P'$ of $\+P$, we show that the idealized permutation-wise Glauber dynamics 
in \Cref{sec-state-comp-decomp}
is rapid mixing in an LLL-like regime in this section. 
\begin{lemma}\label{lemma-rapidmixing}
    The following holds for any sufficiently large $q_{\min}$ and some constant $c>0$.
   Given any PDC formula $\Phi=(\+P,\+Q,\+C)$ where $q\geq q_{\min}$ with decomposition $\+P'$ satisfying \Cref{condition-state-compression} and $\varepsilon\in (0,1)$, if $c pk^{128}\Delta^{192} \leq 1$, then the idealized permutation-wise Glauber dynamics $M_{\mathrm{Glauber}}$ defined in \Cref{sec-state-comp-decomp} satisfies 
    \begin{align}\label{eq-mixing-time}
        \tmix(M_{\mathrm{Glauber}},\varepsilon)\leq \left\lceil 2n\log\frac{n}{\varepsilon} \right\rceil.
    \end{align}
    
    Moreover, assume $\Phi$ is $(k,q)$-uniform and the decomposition $\+P'$ satisfies $\abs{\abs{P} - r}= O(1)$ for each $P \in \+P'$ and some integer $r$. If $k\geq 24$ and $c k^{12}\Delta^{16} \leq r^{k}$, then \eqref{eq-mixing-time} holds. 
\end{lemma}

In the subsequent discussion, the proofs of some lemmas are deferred to~\Cref{sec:missingproofs}.

\subsection{Mixing time}\label{subsec:rapidmixing} 
As we mentioned in~\Cref{lemma-unique-distribution}, the permutation-wise Glauber dynamics has the unique stationary distribution $\nu$. In this section, we apply the path coupling argument to establish the rapid mixing property of the idealized permutation-wise Glauber dynamics, subject to the condition stated in \Cref{lemma-rapidmixing}.
Given any \csppformula formula $\Phi=(\+P,\+Q,\+C)$ with a decomposition $\+P'=\left(P'_1,\cdots,P'_t\right)$ of $\+P$, 
define
\[
\+V\triangleq\left\{ \+Q'=\left(Q'_1,Q'_2,\cdots,Q'_t\right)\mid \forall i\in [t], \abs{Q'_i} =\abs{P'_i} \text{ and }  Q'_i\subseteq \+Q(P) \text{ for the unique $P\in \+P$ where } P'_i\subseteq P\right\}.
\]
By definition, the set $\+V$ contains all decompositions of $\+Q$ from $\Omega[\Phi,\+P']$. In subsequent discussion, we instead analyze the extended permutation-wise Glauber dynamics $\widetilde{M}_{\mathrm{Glauber}}$ on $\+V$. Given $\+Q'\in \+V$, we pick $P\in \+P$ uniformly at random, sample $\sigma\sim \mu_{\Phi'}$ where $\Phi'=(\+P'\circ P,\+Q'\circ \+Q(P),\+C)$, and update its domains according to $\sigma$. One can verify that the transition of $\widetilde{M}_{\mathrm{Glauber}}$ on $\Omega[\Phi,\+P']$ is consistent with ${M}_{\mathrm{Glauber}}$ and the stationary distribution of $\widetilde{M}_{\mathrm{Glauber}}$ is $\nu$.

As mentioned in~\Cref{subsec:pathcoupling}, to establish the rapid mixing property of the idealized permutation-wise Glauber dynamics, it suffices to define a pre-metric on $\+V$ and construct the couplings between transitions of $\widetilde{M}_{\mathrm{Glauber}}$ conditioned on pairs in the pre-metric that exhibit contraction. 

\subsection*{Pre-metric} 
We define the pre-metric $\+G=(\+V,\+E)$ where
\[
\+E \triangleq \left\{\set{\+Q'_1,\+Q'_2}\mid \+Q'_1,\+Q'_2\in \+V \text{ and } \sum_{P'\in \+P'}\abs{\+Q'_1(P')\setminus\+Q'_2(P')}=1\right\}.
\]

By \Cref{def-pre-metric-brief}, 
one can verify that the weighted graph $\+G = (\+V,\+E)$,
where each $\set{\+Q'_1,\+Q'_2}\in \+E$ is associated with weight $1$, is a pre-metric on $\+V$. Furthermore, the weighted shortest path distance metric induced by $\+G$ can be explicitly defined as the discrepancy between configurations in $\+V$, which is determined by evaluating the number of distinct elements within the domains of the \pnames in $\+P'$. This can be defined as follows:
\begin{align}\label{eq-def-dis}
\forall \+Q'_1,\+Q'_2\in \+V, S\subseteq \+P',\quad  \Dis(\+Q'_1,\+Q'_2,S) \triangleq \sum_{P'\in S}\abs{\+Q'_1(P')\setminus\+Q'_2(P')}.
\end{align}

With a slight abuse of notation, let $\Dis(\+Q'_1,\+Q'_2,P)\triangleq  \Dis(\+Q'_1,\+Q'_2,\+P'[P])$ for any \pname $P\in \+P$, and $\Dis(\+Q'_1,\+Q'_2)\triangleq\Dis(\+Q'_1,\+Q'_2,\+P')$. 


\begin{lemma}\label{lem:distancemetric}
For any $\+Q'_1,\+Q'_2\in \+V$, $\Dis(\+Q'_1,\+Q'_2)$ is the weighted shortest path distance on $\+G$.
\end{lemma}

\subsection*{One-step coupling} 
Let $X'_1,X'_2$ be two domains of $\+P'$ satisfying $\Dis(X'_1,X'_2) = 1$.
We shall construct a coupling of the one-step transitions of the {extended permutation-wise Glauber dynamics} $\left(X'_1,X'_2\right)\rightarrow \left(Y'_{1},Y'_{2}\right)$ satisfying
\begin{align}\label{eqn-contraction}
    \E{\Dis(Y'_{1},Y'_{2})\mid \left(X'_1,X'_2\right)}\leq 1-\frac{1}{2\abs{\+P}}.
\end{align}
Note that for any $\+Q'_1,\+Q'_2\in \+V$, $\Dis(\+Q'_1,\+Q'_2) \leq \sum_{P'\in\+P'}\abs{\+Q'_1(P')} = \abs{V}$.
Combined with ~\eqref{eqn-contraction},~\Cref{lemma-rapidmixing} is proved by Lemmas \ref{Lemma:pathcoupling} and \ref{lemma-unique-distribution}.

For any $\set{X_1',X_2'}\in \+E$, let $\pinit$ denote the unique \pname in $\+P'$ satisfying $\abs{X'_1(\pinit)\setminus X'_2(\pinit)}=1$. 
We then specify the one-step coupling $(X'_1,X'_2)\rightarrow (Y'_1,Y'_2)$ as follows:
\begin{enumerate}
    \item Pick $P\in\+P$ uniformly at random;
    \item Let {$\Phi_1 =(\+P'\circ P,X'_1\circ {\+Q(P)},\+C), 
    \Phi_2 =(\+P'\circ P,X'_2\circ {\+Q(P)},\+C)$};
    \item If $\pinit \subseteq P$, 
    sample $\sigma\sim\mu_{\Phi_{1}}$, and
    set $Y'_1(P') = \sigma(P')$ and $Y'_2(P') = \sigma(P')$ for each $P'\in \+P'$; Otherwise, sample $(\sigma_1,\sigma_2)$ jointly from a coupling of distributions $\mu_{\Phi_{1}}$ and $\mu_{\Phi_{2}}$, and set $Y'_1(P') = \sigma_1(P'),Y'_2(P') = \sigma_2(P')$ for each $P'\in \+P'$.
\end{enumerate}
It is easy to verify that the transitions $X'_1\rightarrow Y'_1$ and
$X'_2\rightarrow Y'_2$ are both faithful copies of the $\widetilde{M}_{\mathrm{Glauber}}$.

\subsection*{Analysis of the path coupling}
The core technical aspect of path coupling lies in demonstrating the contraction of distance through one-step coupling of each pair of domains from the pre-metric. 
More precisely, we shall show that \eqref{eqn-contraction} holds for any $\set{X'_1,X'_2}\in \+E$.
Given any pair of formulas $\set{X'_1,X'_2}\in \+E$, we define the discrepancy $\+D_P$ of each permutation $P\in \+P$ with respect to the one-step coupling $\left(X'_1,X'_2\right)\rightarrow \left(Y'_{1},Y'_{2}\right)$ as follows: 
\begin{align}\label{eq-def-dp}
    \+D_P= \E{\Dis\left(Y'_1,Y'_2,P \right) \mid \mbox{$P$ is picked}}.
\end{align}
By definition, {we have $\+D_{P} = 0$ if $\pinit\subseteq P$.}
Then we have 
\begin{equation}
\begin{aligned}
 \E{\Dis(Y'_1,Y'_2)\mid \left(X'_1,X'_2\right)} &=
1+ \frac{1}{\abs{\+P}}\sum_{P\in \+P}\+D_P - \frac{1}{\abs{\+P}}\\
&=1 - \frac{1}{\abs{\+P}}\left(1 - \sum_{P\in \+P}\+D_P\right).
\end{aligned}
\end{equation}
To prove the inequality in \eqref{eqn-contraction}, it is sufficient to prove the following lemma.
\begin{lemma}\label{lem-DP-bounds}
For any PDC formula $\Phi = (\+P,\+Q,\+C)$ with a decomposition $\+P'$ of $\+P$ satisfying the condition in~\Cref{lemma-rapidmixing}, there exists a collection of one-step couplings $\left(X'_1,X'_2\right)\rightarrow \left(Y'_{1},Y'_{2}\right)$ for any $\set{X'_1,X'_2}\in \+E$ such that
\begin{equation*}
\begin{aligned}
\sum_{P\in \+P}\+D_P \leq \frac{1}{2},
\end{aligned}
\end{equation*}
where $\+D_P$ is defined in \eqref{eq-def-dp}.
\end{lemma}

The remainder of this section is devoted to proving \Cref{lem-DP-bounds}.

\subsection{Coupling construction}\label{subsec:couplingconstruction}
In this section, we clarify the one-step coupling construction stated in~\Cref{lem-DP-bounds}. Before that, we first introduce more notations. 

Given any \csppformula formula $\Phi=(\+P,\+Q,\+C)$, $\+C'\subseteq \+C$, and $S\subseteq V$,
we use $\Lambda(\+C',S)$ to denote the set of forbidden values for the variables in $S$ imposed by the constraints in $\+C'$.
Formally, 
\begin{align}\label{Lambda-def}
    \Lambda(\+C',S) =\{c\mid \exists C\in \+C', v\in S\cap \vbl(C)\mbox{ such that the literal } v\neq c \text{ appears in } C\}.
\end{align}
For simplicity, we use the notation $\Lambda(C,\cdot)$ and $\Lambda(\cdot,v)$ to denote $\Lambda(\set{C},\cdot)$ and $\Lambda(\cdot,\set{v})$ for any $C\in \+C$, $v\in V$, respectively. 

Given any two formulas $\Phi_1 = (\+P,\+Q_1,\+C_1)$, $\Phi_2 = (\+P,\+Q_2,\+C_2)$ and $P\in \+P$, 
let $Q^1$ and $Q^2$ denote the domains of $P$ in $\Phi_1$ and $\Phi_2$, respectively.
We say $P$ is \emph{active} in $\Phi_1$ and $\Phi_2$ if \( Q^1 \cap \Lambda(\+C_1,P)\neq \emptyset \lor  Q^2 \cap \Lambda(\+C_2,P)\neq \emptyset\).
Otherwise, we say $P$ is \emph{inactive} in $\Phi_1$ and $\Phi_2$.
We will omit $\Phi_1$ and $\Phi_2$ if they are clear from the context.

Throughout this section, we fix positive integers $k,\Delta$ and consider any formula $\Phi = (\+P,\+Q,\+C)$ where $k_{\Phi}\leq k$, $\Delta_{\Phi}\leq \Delta$, and each $P\in \+P$ is sufficiently large. Moreover, for any decomposition $\pdecomex$ of $\+P$ and any $P\in \+P$, we will always assume that 
the partition {$\pdecomex[P]=\left(P_1,\cdots,P_{\ell}\right)$ satisfies 
$\abs{P_i} - \abs{P_j} = O(1)$. 
for each $P\in \+P$, $i,j\in [\ell]$.}
We also fix positive integers $\gamma,\alpha$, which are the parameters in our one-step coupling construction.
We will first introduce the coupling subroutines for the permutations in \Cref{subsubsec:couplingforperm}, which will serve as the key component of our one-step coupling algorithm in \Cref{subsubsec:one-step coupling}.

\subsubsection{\textbf{\emph{Couplings for permutations}}}\label{subsubsec:couplingforperm}

In this section, we introduce coupling subroutines for instances that satisfy the following condition, which will hold throughout the coupling algorithm in our later discussion.

\begin{condition}\label{condition-couple-first}
 Let $p,p',\theta\in(0,1)$ satisfy $p'\leq p^{\theta}, 8\-e\Delta p^{\theta}\leq 1$.
 Let $t$ be a positive integer.
Given any two formulas $\Phi_1 = (\+P,\+Q_1,\+C_1)$, $\Phi_2 = (\+P,\+Q_2,\+C_2)$ and any $P\in \+P$ where $p_{\Phi_1}\leq p', p_{\Phi_2}\leq p'$, let $\pdecomex$ be a decomposition of $\+P$
such that for each decomposition $(\pdecomex,\+Q'_1)$ of $(\+P,\+Q_1)$
and each decomposition $(\pdecomex,\+Q'_2)$ of $(\+P,\+Q_2)$,
we have $p_{\Phi'_1}<_{q} p^{\theta}$, $p_{\Phi'_2}<_{q} p^{\theta}$ where 
$\Phi'_1 = (\pdecomex,\+Q'_1,\+C_1)$, $\Phi'_2 = (\pdecomex,\+Q'_2,\+C_2)$.
Assume that $\pdecomex[P] = (P_1,\cdots,P_\ell)$ satisfies $\abs{\abs{P_i} - t^{\theta}} = O(1)$ for each $i\in [\ell]$.
Let $T_1\subseteq \+Q_1(P),T_2\subseteq \+Q_2(P)$ be two sets where $\abs{T_1} = \abs{T_2}\leq 1$,  $\+Q_1(P)\setminus T_1  = \+Q_2(P)\setminus T_2$.
\end{condition}

Given the instances satisfying~\Cref{condition-couple-first},
define the distributions $\nu^1,\nu^2,\nu^1_{i},\nu^2_{i}$ where $i\in [\ell]$ as follows
\begin{align}
\forall Q_1,\cdots,Q_{\ell},\quad \nu^1\left(Q_{[\ell]}\right) &\triangleq \Pr[\sigma\sim \mu_{\Phi_1}]{\sigma(P_i) = Q_i \text{ for all } i\in [\ell]},\\
\forall Q_1,\cdots,Q_{\ell},\quad \nu^2\left(Q_{[\ell]}\right) &\triangleq \Pr[\sigma\sim \mu_{\Phi_2}]{\sigma(P_i) = Q_i \text{ for all } i\in [\ell]},
\end{align}
{\begin{align}
    \forall Q,  \quad\quad \!\!\! \nu^1_i(Q) \triangleq \Pr[\sigma\sim \mu_{\Phi_1}]{\sigma(P_i) = Q \mid T_1\subseteq \sigma(P_i)},\quad 
\nu^2_i(Q) \triangleq \Pr[\sigma\sim \mu_{\Phi_2}]{\sigma(P_i) = Q\mid T_2\subseteq \sigma(P_i)}.
\end{align}}

The following lemma is useful in our later analysis, which provides an upper bound for the event that there exists some forbidden values in a permutation set.

\begin{lemma}\label{lemma-trivial-coupling}
Given any instance satisfying \Cref{condition-couple-first} with {$T_1=T_2=\emptyset$}, we have
\[
\forall i\in [\ell], t\in [2],\quad   \Pr[Q_{[\ell]}\sim\nu^t]{Q_i \cap \Lambda(\+C_t,P_i)\neq \emptyset}  \leq \frac{2k\Delta\abs{P_i}^2}{\abs{P}}.
\]
\end{lemma}

In subsequent pages, we introduce the coupling of the distributions $\nu^1$ and $\nu^2$. Our strategy tries to couple the domain of the permutation set for some $i\in [\ell]$ according to the distribution $\nu_i^1$ and $\nu_i^2$, which is specified in the $\onestep\left(\inputonestep,T_1,T_2\right)$~(\Cref{alg-one-step-coupling}).

\begin{algorithm} 
    \caption{$\onestep\left(\inputonestep,T_1,T_2\right)$} \label{alg-one-step-coupling}
   
    $\Succ\leftarrow \False$\;
    \eIf{${t> \sizeuperboundlp}$}{
        $S \leftarrow \+Q_1(P)\setminus \left(\Lambda(\+C_1,P_i)\cup\Lambda(\+C_2,P_i)\cup T_1\right)$\;
        {
            $\alpha \leftarrow  \binom{\abs{S}}{\abs{P_i}-\abs{T_1}}\cdot \left(\abs{P_i} - \abs{T_1}\right)!\cdot \prod_{j=\abs{T_1}}^{\abs{P_i}-1}(\abs{P}-j - 2\Delta)/(\abs{P}-j )^2$\;
            $\alpha\gets \alpha \cdot \id{T_1\cap \Lambda(\+C_1,P_i)=T_2\cap \Lambda(\+C_2,P_i)=\emptyset}$\;
        }
        
        }{
        $S \leftarrow \+Q_1(P)\setminus T_1$, $q \leftarrow \abs{P}-\abs{P_i} +1$, $p' \leftarrow  p^{\theta}\abs{P_i}/q$\;
        $\alpha \leftarrow  \left( (1-4q\-e p'\Delta)/(1 + 4ep'\Delta ) \right)^{\abs{P_i}-\abs{T_1}}$\;
    }
    draw $r\sim \!{Unif}[0,1]$\;
    \eIf{$r\leq  \alpha$}
    {
        select a subset $T\subseteq S$ uniformly at random with $\abs{T} = \abs{P_i} - \abs{T_1}$\;
        $Q^1_i \leftarrow T\cup T_1$, $Q^2_i \leftarrow T\cup T_2$\; \label{line-set-both-T}
    {\textbf{if} 
    $ {t \leq \sizeuperboundlp} \mbox{ or }{r\geq 1-4\Delta{\abs{P_i}}/\ell}$
    \textbf{ then } $\Succ\leftarrow \True$}\; \label{line:adjust-prob}        
    }{
        draw $Q^1_i$ randomly so that its overall distribution aligns with $\nu^1_i$\;
        draw $Q^2_i$ randomly so that its overall distribution aligns with $\nu^2_i$\;

    }
    
    \Return $(Q^1_i,Q^2_i,\Succ)$\;\label{line-return-one-step}
\end{algorithm}


In \Cref{alg-one-step-coupling}, we apply different domain-coupling procedures according to the size of each permutation set. Specifically:
\begin{itemize}
\item When \(t>\gamma\) (i.e., \(P_i\) is “large”): We couple only the unforbidden values on \(P_i\) from \(\+Q_1(P)\) and \(\+Q_2(P)\), using 
$\alpha$ as a lower bound on the probability of successful coupling. In this setting, if \(T_1 = T_2 = \emptyset\) and \(\Succ\) becomes \(\True\), then \(Q^1_i \setminus T_1 = Q^2_i \setminus T_2\) and \(P_i\) is inactive; consequently, no discrepancy can propagate beyond \(P_i\) (whether through constraints on \(P_i\) or the permutation set \(P\)).

\item When \(t\le \gamma\) (i.e., \(P_i\) is “small”): We couple the entire domain. In that case, if \(T_1 = T_2 = \emptyset\) and \(\Succ\) is \(\True\), the discrepancy can propagate only through constraints on \(P_i\).
\end{itemize}
The following lemma formalizes these observations.

\begin{lemma}\label{lem:meaningofsucc}
In \Cref{alg-one-step-coupling}, 
if $\Succ = \True$, then $Q^1_i\setminus T_1 = Q^2_i \setminus T_2$ and $T_1\subseteq Q^1_i, T_2 \subseteq Q^2_i$.
In addition, if $T_1 = T_2 = \emptyset$, $\Succ = \True$ and ${t> \sizeuperboundlp}$, then all the constraints in $\+C_1(P_i)\cup \+C_2(P_i)$
are satisfied under the condition that the variables in $P_i$ take values from $Q^1_i$ and $Q^2_i$ .
\end{lemma}

\begin{proof}
If $\Succ_i = \True$, then $Q_i^1$ and $Q_i^2$ must be assigned in \Cref{line-set-both-T}. Thus, $Q^1_i\setminus T_1 = Q^2_i \setminus T_2 = T$.
Furthermore, if $T_1 = T_2 = \emptyset$, $\Succ = \True$ and ${t> \sizeuperboundlp}$, we have $S \cap \Lambda(\+C_1,P_i) = \emptyset$ and $Q_i^1 = T$. Combined with $T \subseteq S$, we have $Q_i^1 \cap \Lambda(\+C_1,P_i) = \emptyset$. For constraint $C\in \+C_1(P_i)$, there is a vertex $v\in P_i \cap \vbl(C)$. Because $C$ is a disjunctive constraint, there is a value $c$ such that $C$ is satisfied if $v\ne c$. Therefore, $c\in \Lambda(\+C_1,P_i)$. Combined with $Q_i^1 \cap \Lambda(\+C_1,P_i) = \emptyset$, we have $c\notin Q_i^1$. Thus, $C$ is satisfied under the condition that the variables in $P_i$ take values from $Q^1_i$. The proof also works for constraint $C\in \+C_2(P_i)$.
\end{proof}

If the event \(\Succ = \False\) occurs, even though the discrepancy might be curtailed, we treat it as having propagated and continue revealing additional domains of the permutation set.

The following lemma provides an upper bound for the probability of the event $\Succ = \False$.
\begin{lemma}\label{lemma-cpstep-failprob}
Given any instance satisfying \Cref{condition-couple-first}  and $i\in [\ell]$, at \Cref{line-return-one-step} in \Cref{alg-one-step-coupling}, $(Q^1_i,Q^2_i)$ is a coupling of $\nu_i^1,\nu_i^2$.
In addition, 
we have 
\begin{equation}
\begin{aligned}
\Pr{\Succ = \False} <_{q}
\begin{cases}
{4\Delta\abs{P_i}/\ell}& \text{if $t > \sizeuperboundlp$ and $T_1\cap \Lambda(\+C_1,P_i)=T_2\cap \Lambda(\+C_2,P_i)=\emptyset$,} \\
5\-e \Delta {{p^{\theta}}} \abs{P_i}^2 & \text{otherwise.}
\end{cases} 
\end{aligned}
\end{equation}
{Furthermore, if $t > \sizeuperboundlp$ and $\ell\geq 4\Delta\abs{P_i}$, we have $\Pr{\Succ = \False}=4\Delta\abs{P_i}/\ell$.}
\end{lemma}

Given any \csppformula formula $\Phi=(\+P,\+Q,\+C)$ with a decomposition $\+P'$ of $\+P$, any $P\in \+P$ and any decomposition $\left(P'_{[\ell]},Q'_{[\ell]}\right)$ of $(P,\+Q(P))$, with a little abuse of notations, let 
\[\Phi\left[P'_{[\ell]},Q'_{[\ell]}\right] \triangleq \left(\+P\setminus \{P\}\cup \left\{P'_{[\ell]}\right\},\+Q\setminus \{\+Q(P)\}\cup \left\{Q'_{[\ell]}\right\},\+C\right)\] be the formula induced by $\Phi$ and $P'_{[\ell]},Q'_{[\ell]}$.
Similarly, for any $i\in [\ell]$, with a little abuse of notations, let 
\[\Phi[P'_i,Q'_i] \triangleq (\+P\setminus \{P\}\cup \{P'_i,P\setminus P'_i\},\+Q\setminus \{\+Q(P)\}\cup \{Q'_i,\+Q(P)\setminus Q'_i\},\+C)\] be the formula induced by $\Phi$ and $P'_i,Q'_i$.

We are now ready to introduce the coupling for domains of permutation sets.

\subsection*{{Couplings for permutations with equal domains}}
Given the instance satisfying~\Cref{condition-couple-first} with $T_1 = T_2 = \emptyset$ and any $i\in [\ell]$,
a coupling of $\nu^1,\nu^2$ is specified in $\cpperm\left(\inputonestep\right)$ (\Cref{alg:coupling}).

In~\Cref{alg:coupling}, we first call the coupling subroutine $\onestep(\cdot)$ on $P_i$. According to~\Cref{lem:meaningofsucc}, if $\Succ = \True$, the discrepancy cannot propagate through the permutation on $P$. In this case, we return $P_i$, which corresponds to the variables with an assigned domain. Otherwise, if $\Succ = \False$, the discrepancy propagates through the permutation on $P$. Thus, we attempt to assign domains to the entire permutation $P$. Moreover, we return $P$ and $S\cup \{P_i\}$, where $P$ contains the variables with assigned domains and $S\cup \{P_i\}$ records the active permutation sets.

\begin{algorithm}
    \caption{$\cpperm(\inputonestep)$} \label{alg:coupling}
    initialize $Q^1_{[\ell]}$ and $Q^2_{[\ell]}$ arbitrarily\;
    %
    $(Q^1_i,Q^2_i,\Succ) \leftarrow \onestep(\inputonestep,\emptyset,\emptyset)$\; \label{line:call-cp-step}
    \eIf{$\Succ = \False$}{
        sample $Q^1_{[\ell]\setminus\{i\}}$ such that  $Q^1_{[\ell]}\sim \nu^1$; \\
        sample $Q^2_{[\ell]\setminus \{i\}}$ such that  $Q^2_{[\ell]}\sim \nu^2$;\\
        $S \leftarrow \{P_j\mid j\in [\ell], P_j \text{ is active}\}$\;
         \Return $\left(\Phi_1\left[\+P'[P],Q^1_{[\ell]}\right], \Phi_2\left[\+P'[P],Q^2_{[\ell]}\right],{P,{S\cup \{P_i\}}}\right)$\;\label{line-return-false}
    }{
        \Return $\left(\Phi_1\left[P_i,Q^1_i\right],\Phi_2\left[P_i,Q^2_i\right],{P_i,\emptyset}\right)$\;\label{line-return-true}
    }
   
\end{algorithm}

The analysis of our coupling algorithm is presented in the following lemma, which can be easily verify by~\Cref{lemma-cpstep-failprob}

\begin{lemma}\label{lemma-cpperm-failprob}
Given any instance satisfying \Cref{condition-couple-first} with $T_1 = T_2 = \emptyset$ and any $i\in [\ell]$, in \Cref{alg:coupling}, either $\left(Q^1_{[\ell]},Q^2_{[\ell]}\right)$ is a coupling of $\nu^1,\nu^2$ at \Cref{line-return-false} or $(Q^1_i,Q^2_i)$ is a coupling of $\nu^1_i,\nu^2_i$ at \Cref{line-return-true}. In addition, if ${t> \sizeuperboundlp}$ and {$\ell \geq {4\Delta\abs{P_i}}$, then $\Pr{\Succ = \False} = 4\Delta\abs{P_i}/\ell$.}
\end{lemma}

\subsection*{Couplings for permutations with different domains}

Given any instance satisfying \Cref{condition-couple-first} with $p' = p, t=\abs{P}, T_1 = \{c_1\},T_2 = \{c_2\}$,
we use $f_1,f_2$ to denote the distributions where
\begin{equation}\label{eq-definition-f}
\forall j\in [\ell],\quad f_1(j) = \Pr[\sigma\sim \mu_{\Phi_1}]{c_1\in \sigma(P_j)},
f_2(j) = \Pr[\sigma\sim \mu_{\Phi_2}]{c_2\in \sigma(P_j)}.
\end{equation}
For any $j\in [\ell]$, let $V_j$ be the set of variables $v\in P_j$ such that there exist some constraints $C\in \+C$ which contains the literal $v= c_1$ or $v=c_2$.
Define 
\begin{equation}\label{eq-definition-f-lower-bound}
\begin{aligned}
f_{\mathsf{min}}(j)&\triangleq \max\left\{(1+4\-e\Delta p)\abs{P_j}/\abs{P} - 4\-e\Delta p \abs{P_j},\left(\abs{P_j}-\abs{V_j}\right)\cdot \left(\abs{P}-2\Delta\right)/\abs{P}^2\right\},\\
f_{\mathsf{max}}(j)&\triangleq (1+4\-ep\Delta)\abs{P_j}/\abs{P}.
\end{aligned}
\end{equation}
The following lemma establishes bounds for $f_1$ and $f_2$.
The following lemma establishes bounds for $f_1$ and $f_2$.
\begin{lemma}\label{lemma-lowerbound-f}
For $f_1,f_2$ defined in \eqref{eq-definition-f}, we have for any $j\in [\ell]$, 
\begin{equation}\label{condition-coupling-f-j}
\begin{aligned}
f_{\mathsf{min}}(j) \leq \min\{f_1(j),f_2(j)\}\leq \max\{f_1(j),f_2(j)\}\leq f_{\mathsf{max}}(j).
\end{aligned}
\end{equation}
\end{lemma}
\begin{proof}
The lemma is immediate from~\Cref{lem:marginalub},~\Cref{lem:marginallb} and the Union bound.
\end{proof}

\begin{algorithm}
    \caption{$\initialcouple(\inputtrivial)$} \label{alg-initial-coupling}
    initialize $Q^1_{[\ell]}$ and $Q^2_{[\ell]}$ arbitrarily and let $S \leftarrow \emptyset$\;
    let $(i,j)$ be sampled from a coupling of $(f_1,f_2)$ where $\Pr{i=j = r} \geq f_{\mathsf{min}}(r)$ for each $r\in [\ell]$\;\label{line-coupling-initial-coupling}
    
    \eIf{$i = j$}{
        $(Q_i^1,Q_i^2,\Succ) \leftarrow \onestep(\Phi_1, \Phi_2,\pdecomex,P,P_i,\abs{P},\{c_1\},\{c_2\})$\;\label{line-cpinitial}
        \eIf{$\Succ = \False$}{
            sample $Q^1_{[\ell]\setminus \{i\}}$ such that  $Q^1_{[\ell]}\sim \nu^1$; \\
            sample $Q^2_{[\ell]\setminus \{i\}}$ such that  $Q^2_{[\ell]}\sim \nu^2$\;
        }{
        {\textbf{if} $P_i$ is active \textbf{then} $S \leftarrow \{P_i\}$}\;
        \Return $\left(\Phi_1\left[P_i,Q^1_i\right],\Phi_2\left[P_i,Q^2_i\right],{P_i,S}\right)$\;
        }
    }
    {
        sample $Q^1_{[\ell]}$ such that $Q^1_{[\ell]}$ follows the condition distribution of $\nu^1$ conditional on $c_1\in Q_i$; \\
        sample $Q^2_{[\ell]}$ such that $Q^2_{[\ell]}$ follows the condition distribution of $\nu^2$ conditional on $c_2\in Q_j$\;
    }
    $S \leftarrow \{P_t\mid t\in [\ell], P_t \text{ is active}\}$\;\label{line-init-set-s}
         \Return $\left(\Phi_1\left[\+P'[P],Q^1_{[\ell]}\right], \Phi_2\left[\+P'[P],Q^2_{[\ell]}\right],{P,S}\right)$\;
\end{algorithm}

The coupling of $\nu^1,\nu^2$ is specified in $\initialcouple\left(\inputtrivial\right)$ (\Cref{alg-initial-coupling}). At \Cref{line-coupling-initial-coupling} in the~\Cref{alg-initial-coupling}, the existence of such a coupling of  $(f_1,f_2)$ is by \Cref{lemma-lowerbound-f}. 

In~\Cref{alg-initial-coupling}, we first try to assign $c_1,c_2$ to the same permutation set among $\+P'[P]$. If $c_1,c_2$ have not been assigned to the same permutation set, we assign the entire permutation and return $S$, which corresponds to the collection of permutation sets that could potentially propagate the discrepancy. Otherwise, if $c_1,c_2$ have been assigned to the same permutation set, we implement a similar coupling strategy as $\cpperm(\cdot)$ subroutine, conditioned on $c_1,c_2$ being assigned in $P_i$ respectively.

The following lemma bounds the expected number of permutation sets in $P$ that could propagate the discrepancy and its proof is deferred to~\Cref{sec:miss-proofs}.
\begin{condition}\label{condition-initial}
The following holds for $\theta\in (0,1/2],p,\sizeuperboundlp$ and $P$:
\begin{enumerate}[(a)]
\item if $\abs{P}\leq \sizeuperboundlp$, then $5\-ep^{\theta}\abs{P}^{3\theta}\leq 1$, 
$\-e\left(p\abs{P}^2 + p^{\theta}\abs{P}^{1+2\theta}\right) \min\left\{2\Delta\abs{P}^{2\theta - 1},1\right\}\leq 1$;\label{eq-lem-init-3}
\item if $\abs{P}> \sizeuperboundlp$, then 
$\abs{P}^{2\theta}\min\left\{2\Delta\abs{P}^{2\theta - 1},1\right\}\leq 1$.\label{eq-lem-init-4}
\end{enumerate}
\end{condition}

\begin{lemma}\label{lemma-suc-probability-alg-initial-coupling}
Given any instance satisfying \Cref{condition-couple-first} with $p_{\Phi_1}\leq p, p_{\Phi_2}\leq p, \abs{P} = t, T_1 = \{c_1\}, T_2 = \{c_2\}$ and \Cref{condition-initial}, 
we have 
\begin{align}
\E{\sum_{P'\in S}\abs{P'}}<_{q} (36k+4)\Delta.
\end{align}
\end{lemma}

\subsubsection{\textbf{\emph{One-step coupling construction}}}\label{subsubsec:one-step coupling}
We complete the one-step coupling construction in this section.

Let \(\Phi=(\+P,\+Q,\+C)\) be any PDC formula with a decomposition \(\+P'\) of \(\+P\) satisfying \Cref{condition-state-compression}. Let \(\{X'_1,X'_2\}\in \+E\) and let \(\pinit\) denote the unique \pname in \(\+P'\) for which \(\left|X'_1(\pinit)\setminus X'_2(\pinit)\right|=1\). Consider a constant \(\theta\) such that either \(\theta=\eta\) or \(2\theta\le\eta\); a \(\theta/\eta\)-decomposition \(\+D\) of \(\+P'\); a \(\theta\)-decomposition \(\+D'\) of \(\+P'\); and a parameter \(\alpha\le L\). For sufficiently large \(q_{\min}\) with \(q_{\min}\le q_{\Phi}\), the existence of the decompositions \(\+D\) and \(\+D'\) is guaranteed. Our one-step coupling is specified by \(\couple(\Phi,\+P',X'_1,X'_2,\pinit,\theta,\alpha)\) (see \Cref{Alg:coupling}).


The core idea behind our coupling strategy is to eliminate the discrepancy by assigning domains to the sub-permutation sets in each \(P\in \+P'\) and ensuring that the boundary constraints are satisfied, thereby halting further discrepancy propagation. We denote by \(\vset\) the set of variables whose domains have been assigned in the corresponding permutation sets in \(\pzeta\) during the coupling process. Our goal is to prevent the discrepancy from spreading from \(\vset\) to \(V\setminus \vset\).

In our framework, the discrepancy can propagate in two distinct ways:
\begin{itemize}
    \item \textbf{Type-I discrepancy:} For some \(P\in \+P'\), there exists a sub-permutation set \(P'\in \pzeta[P]\) within \(\vset\) whose assigned domain differs between the coupled configurations. In this case, the discrepancy propagates through \(P\).
    \item \textbf{Type-II discrepancy:} For some constraint \(C\in \+C\) that involves variables from both \(\vset\) and \(V\setminus \vset\), if \(C\) is not satisfied, the discrepancy can propagate through \(C\).
\end{itemize}


We sequentially reveal the domains of the permutation sets in \(\pcurrent\) and couple them by invoking the subroutines \(\initialcouple(\cdot)\) and \(\cpperm(\cdot)\). For Type-I discrepancies, our coupling strategy, described previously via \(\initialcouple(\cdot)\) and \(\cpperm(\cdot)\), reveals the entire permutation \(P\); hence, we only need to track the Type-II discrepancy. Specifically, we employ \(\cremain\) to collect the constraints that could potentially contribute to a Type-II discrepancy, and we use \(\pdiff\) to gather the permutation sets that might propagate the discrepancy through their associated constraints.


For technical reasons, our coupling strategy differs depending on whether the permutation set selected for update, denoted by \(\+L\), is considered “large” or “small” (as determined by the parameter \(\alpha\)). If \(|\+L|\ge \alpha\), we remove the constraints associated with \(\+L\) from \(\cremain\) to avoid assigning domains in \(\+L\) when alternative choices are available. We then apply a more sophisticated coupling on \(\+L\) to bound its discrepancy until no constraint in \(\+C\setminus \+C(\+L)\) intersects both \(\vset\) and \(V\setminus \vset\). Conversely, if \(|\+L| < \alpha\), we treat the constraints on \(\+L\) the same as all other constraints during the coupling procedure.


Our coupling process continues when there are constraints in \(\cremain\) that intersect both \(\vset\) and \(V\setminus\vset\). In that case, we select one such constraint along with an unassigned permutation set associated with it and update its domains using the \(\cpperm(\cdot)\) subroutine (see Lines \ref{line-find-con-coupling} and \ref{line-find-p-coupling}). After updating the domains, we update both \(\cremain\) and \(\pdiff\). One subtle detail is that some assigned permutation sets are not initially added to \(\pdiff\) because they were coupled via \(\cpperm(\cdot)\); however, these sets can still propagate discrepancy through unsatisfied constraints. Therefore, we add them to \(\pdiff\) at \Cref{line-set-pdiff}.

In \Cref{Alg:coupling}, we also maintain several sets to facilitate discrepancy analysis. Specifically, we let \(\pdisci\) denote the collection of sub-permutation sets in \(\pzeta\) that failed to couple in the subroutine \(\onestep(\cdot)\). Additionally, we let \(\pdisc\) represent the set of all “bad" sub-permutation sets in \(\pzeta\), namely those that are either active or have failed to couple in \(\onestep(\cdot)\).
We further define the sets \(\csmall\) and \(\cdisc\) as follows:
\begin{equation*}
\csmall \triangleq \left\{C\in \+C\mid P\in \psmall \text{ for each } P\in \pzeta \text{ on } C\right\},
\end{equation*}
\begin{equation*}
\cdisc \triangleq \{C\in \csmall \mid \vbl(C)\subseteq \vset, \text{$C$ is unsatisfied in $\Phi_1$ or $\Phi_2$}, \text{no $P\in \pdisc$ is on $C$}\}.
\end{equation*}
In other words, \(\csmall\) is the set of constraints in \(\+C\) whose associated sub-permutation sets in \(\pzeta\) are all “small", while \(\cdisc\) consists of those constraints in \(\csmall\) that remain unsatisfied despite having no associated “bad" sub-permutation sets.

\begin{algorithm} \caption{$\couple(\Phi,\+P',X'_1,X'_2,\pinit,\theta,\alpha)$} 
    \label{Alg:coupling}
    \KwIn{a PDC formula $\Phi=(\+P,\+Q,\+C)$, a decomposition $\+P'$ of $\+P$ satisfying \Cref{condition-state-compression}, a pair of domains $\set{X'_1,X'_2} \in \+E$, $\pinit\in \+P'$ where $\abs{X'_1(\pinit)\setminus X'_2(\pinit)}=1$, a constant $\theta$ where $\theta = \eta$ or $2\theta \leq \eta$, a $\theta/\eta$-decomposition $\+D$ of $\+P'$ and a $\theta$-decomposition $\+D'$ of $\+P'$, and $\alpha \leq L$.}
    \KwOut{a pair of domains $(Y'_1,Y'_2)$.}
    
    pick a \pname $\pfinal$ from $\+P$ uniformly at random\;
    $\pcurrent \leftarrow \+P'\setminus \+P'[\pfinal]\cup \{\pfinal\}$\;
    $\Phi_1 \leftarrow (\pcurrent,X'_1\setminus X'_1(\+P'[\pfinal])\cup \set{\+Q(\pfinal)},\+C)$, 
    $\Phi_2 \leftarrow (\pcurrent,X'_2\setminus X'_2(\+P'[\pfinal])\cup \set{\+Q(\pfinal)},\+C)$\;\label{line-set-phi1-phi2}
    $\cremain \leftarrow \{C\in \+C\mid \text{$C$ is unsatisfied in $\Phi_1$ or $\Phi_2$}\}$\;
     \eIf{$\pinit\subseteq \pfinal$}{
         $\sigma_1 \leftarrow $ a random solution of $\Phi_1$, $\sigma_2 \leftarrow \sigma_1$\;\label{line-trival-cp-in-coupling}
    }
    {
        \eIf{$\abs{\pfinal}\geq \sizeuperboundpfinal $}{
            $\cremain\leftarrow \cremain \setminus \+C(\pfinal)$, 
            $\pzeta \leftarrow \+D'$, $A \leftarrow \{P\in \+P'\mid  \abs{P}\leq \sizeuperboundlp\}$, 
            $\psmall \leftarrow \cup_{P\in A} \+D'[P]$ \;
        }
        {
            $\pzeta \leftarrow \+D'\setminus \+D'[\pfinal]\cup \+D[\pfinal]$, $A \leftarrow \{P\in \pcurrent\mid  \abs{P}\leq \sizeuperboundlp\}$, 
            $\psmall \leftarrow \cup_{P\in A} \pzeta[P]$\;
            \label{line-abs-pfinal-lessthan-threshold}
            
        }      
    $\left(\Phi_1,\Phi_2,\vset',\pdiffinit\right)\gets \initialcouple\left({\Phi_1,\Phi_2,\pzeta,\pinit}\right)$\;\label{line-initial-coupling}
    $\vset \leftarrow \vset'$, $\pdiff\leftarrow \pdiffinit,\pdisc\leftarrow \pdiffinit,\pdisci\leftarrow \pdiffinit$\;\label{line-set-vset-pdiff-pdisc}
    \While{$\True$\label{line-while-loop}}
    {
    
    \For{each $C\in \cremain(\vset')$}
    {
        \textbf{if } {$C$ is satisfied in both $\Phi_1$ and  $\Phi_2$ or $\vbl(C)\subseteq \vset$}
        \textbf{ then }
            $\cremain  \leftarrow \cremain\setminus \{C\}$\;

        \If{$\vbl(C)\subseteq \vset$ 
         and $C$ is unsatisfied in $\Phi_1$ or $\Phi_2$\label{line-if-c-unsatisfied}}
        {
            $\pdiff \leftarrow \pdiff\cup \{P\in\psmall\mid C\in \+C(P)\}$\;\label{line-set-pdiff}   
        }
    }
    \eIf{$\exists C\in \cremain \cap \+C(P^{\ast})$  \text{ for some } $P^{\ast}\in \pdiff$\label{line-condition-coupling}}
    {
        let $C$ be the first such constraint and $v$ be the first variable in $\vbl(C)\setminus \vset $\;\label{line-find-con-coupling}
        let $P'\in \pzeta[P]$ for some $P\in \pcurrent$ be the unique \pname in $\pzeta$ where $v\in P'$\;\label{line-find-p-coupling}      $\left(\Phi_1,\Phi_2,\vset',\pdisc'\right)\leftarrow   \cpperm(\Phi_1, \Phi_2,\pzeta,P\setminus \vset,P',\abs{P})$\;\label{line-coupling-perm} 
        $\vset \leftarrow \vset\cup\vset^{'}$, $\pdisc\leftarrow \pdisc\cup \pdisc'$, 
        $\pdiff \leftarrow \pdiff\cup \pdisc'$,
        $\pdisci \leftarrow \pdisci\cup (\pdisc'\cap \{P'\})$\;\label{line-coupling-set-perm}
    }
    {
        $(\sigma_1,\sigma_2)\leftarrow$ the coupling of $\mu_{\Phi_1},\mu_{\Phi_2}$  where $\E{\sum_{P\in \+P'[\pfinal]}\abs{\sigma_1(P)\setminus\sigma_2(P)}}$ is minimal\;\label{line-coupling}
        \textbf{break};\label{line-break}
    }
    }
    }
     $Y'_1(P') \leftarrow \sigma(P')$ and $Y'_2(P') \leftarrow \sigma(P')$ for each $P'\in \+P'$\;
    \Return $(Y'_1,Y'_2)$\;\label{line-set-y}
\end{algorithm}

Let $\pdiff^i$ denote the updated $\pdiff$ at \Cref{line-condition-coupling} during the $i$-th iteration of the while loop from Line \ref{line-while-loop} to  Line \ref{line-break} in \Cref{Alg:coupling}.
Similarly, let $\pdisc^i$, $\Phi_1^i, \Phi_2^i$ and  $\vset^i$ denote the updated $\pdisc$, $\Phi_1, \Phi_2$ and $\vset$ at \Cref{line-condition-coupling} during the $i$-th iteration of the while loop, respectively. 
Additionally, define
\begin{equation*}
\cdisc^i \triangleq \{C\in \csmall \mid \vbl(C)\subseteq \vset^i, \text{$C$ is unsatisfied in $\Phi_1^i$ or $\Phi_2^i$}, \text{\text{no $P\in \pdisc^{i}$ is on $C$}}\}.
\end{equation*}

We have the following simple observations for \Cref{Alg:coupling}.
\begin{observation}\label{observation-alg-coupling}
The following conditions are satisfied by \Cref{Alg:coupling}:
\begin{itemize}
\item $\pdisc \subseteq \pzeta$, $\cdisc\subseteq \csmall$, 
and $C\subseteq \vset$ for each $C\in \cdisc$;
\item for any $P\in \pdisc$ and any $C\in \+C(P)$, we have $C\not \in \cdisc$;
\item for the $C$ and $P'$ in Lines \ref{line-find-con-coupling}-\ref{line-find-p-coupling} of \Cref{Alg:coupling}, we have $C\in \+C(P')$.
\end{itemize}
\end{observation}

One can verify that the correctness of the coupling algorithm as follows. 
\begin{lemma}
    In~\Cref{Alg:coupling}, $\left(X'_1,X'_2\right)\rightarrow \left(Y'_{1},Y'_{2}\right)$ is a one-step coupling of the {extended permutation-wise Glauber dynamics} $\widetilde{M}_{\mathrm{Glauber}}$.
\end{lemma}
\begin{proof}
    One can verify that for any input instance of the subroutine $\initialcouple(\cdot)$ and $\cpperm(\cdot)$,~\Cref{condition-couple-first} holds. Therefore,
    the lemma follows immediately from the~\Cref{lemma-cpperm-failprob} and~\Cref{lemma-suc-probability-alg-initial-coupling}.
\end{proof}

\subsection{Discrepancy Analysis}\label{subsec:discrepancy-analysis}
In this section, we present a probabilistic analysis that yields quantitative bounds on the discrepancy using witness combinatorial structures, thereby completing the proof of ~\Cref{lem-DP-bounds}.

The witness in \Cref{Alg:coupling} can be characterized by the combinatorial structure over the permutations $\pzeta$ and the constraints $\csmall$. 
We emphasize that the random variable $\pzeta$ in \Cref{Alg:coupling} depends on $\pfinal$.
The following definition captures the propagation between these units. 
The discrepancy can propagate through the entire permutation, and this can be captured as follows:
\[
\edgepermutation \triangleq \left\{(P_1,P_2) \mid \exists P\in \+P'\text{ s.t. } P_1,P_2\in \pzeta[P]\right\}.
\]
Meanwhile, the discrepancy can also propagate through the constraints as follows:
\[
\edgeconstraint \triangleq \left\{(C_1,C_2)\mid (C_1,C_2\in \+C) \land ( \exists P\in \psmall\text{ s.t. } C_1,C_2\in \+C(P))\right\}.
\]
For any subsets $S_1,S_2\subseteq \pzeta$,
let $\edgepermutation(S_1,S_2)$ denote the set $\{(u,v)\mid u\in S_1,v\in S_2, (u,v) \in \edgepermutation\}$. 
The discrepancy propagation can then be captured by the \propagate.

\begin{definition}[\propagate]\label{def-prop-traj}
    For each $u,v\in \pzeta\cup \+C$ such that $u\neq v$ and each $t\geq 0$, a sequence $(u = v_0, v_1, v_2, \ldots, v_{t+1}=v)$ is called a \propagate from $u$ to $v$ with length $t$ if the following holds:
    \begin{itemize}
     \item  $(v_i,v_{i+1})\in \edgeconstraint$ for each $i\in [t-1]$;
     \item  either $u \in \pzeta$ and $v_1\in \+C(u)$, or $(u,v_{1})\in \edgeconstraint$;
     \item either $v\in \pzeta$ and $v_{t}\in \+C(v)$,
     or $(v_t,v)\in \edgeconstraint$.
     \end{itemize}
     For each $u,v\in \pzeta\cup \+C$ where $u\neq v$, let $d(u,v)$ denote the minimum length of the \propagate from $u$ to $v$. 
     For any $S_1,S_2\subseteq \pzeta\cup \+C$ and $t\geq 0$, let $\+S(S_1,S_2,t)$ denote $\{(u,v)\mid u\in S_1,v\in S_2,d(u,v)=t\}$. For simplicity, we further use $\+S^{\ast}(S_1,S_2, t)$ to denote $\bigcup_{0\leq i\leq t}\+S(S_1,S_2,i)$.
\end{definition}


With these definitions and notations, we define the witness as follows.

\begin{definition}[witness sequence]\label{def-witness}
A witness sequence, or \emph{w.s.} for short, is a sequence $(v_0,v_1,\cdots,v_t)$ where $v_0\in \pzeta[\pinit]$ and the following conditions hold for each $i\in [t]$ and $0 \leq j< i$:
\begin{itemize}
    \item {$(( v_{i-1},v_{i})\in \edgepermutation\land i>1)$ or $(v_{i-1}, v_{i}) \in \+S^{\ast}(\pzeta\cup\csmall,\pzeta\cup\csmall, 2)$};
    \item $(v_i\neq v_j)\land ((v_i,v_j)\not \in \edgeconstraint)\land(v_i\not\in \+C(v_j) \text{ if } v_j\in \pzeta)\land ((v_{i},v_{j})\not\in \edgepermutation$ \text{ if }{$j\in [i-2]$}).
\end{itemize} 
Given any \emph{w.s.} $s = (v_0,\cdots,v_t)$, define
\begin{align*}
   V_{P}(s) \triangleq \{v_{i}\in \pzeta\mid 1<i\leq t, ( v_{i-1},v_{i})\in \edgepermutation\}, \quad 
V_C(s) \triangleq \{v_i\in \pzeta \setminus V_{P}(s) \mid i\in [t]\},
\end{align*}
where collect the permutations in $v_1,\cdots,v_t$.
We use $V_{P},V_{C}$ to denote $V_{P}(s),V_{C}(s)$ if $s$ is clear from the context. 
We say $s=(v_0,\cdots,v_t)$ \emph{occurs} if 
$v_0\in \pdiffinit$, $V_C\subseteq \pdisci$, 
and $v_i\in \cdisc \cup \pdisc \setminus \pdiffinit$ for each $i\in [t]$.
\end{definition}

Let $\+R \triangleq\{P\in \+P\mid \pinit\not \subseteq P, \abs{P} \leq \sizeuperboundpfinal \}$ and $\overline{\+R} = \{P\in \+P\mid \pinit\not \subseteq P, \abs{P} > \sizeuperboundpfinal \}$. Define the subsets of witness sequences as follows:
\begin{itemize}
\item Consider the case where $\pzeta=\+D'$. Let $\+W$ denote the set of all witness sequences. For any $ P\in \+D'\setminus \psmall$,  
let $\witnesswrtp{P}$ denote $\{(v_0,\cdots,v_t)\in \+W\mid v_t= P\}$. For any $P\in \+D'\cap \psmall$, let $\witnesswrtp{P}$ denote $\{(v_0,\cdots,v_t)\in \+W\mid v_t=P \text{ or } {(v_t,P)\in \+S^{\ast}(\+D'\cup \+C,\psmall,1)}\}$;
\item For any {$P\in \+R$, $P'\in \+D[P]$,
let $\+W[P,P']$ be the witness sequences $(v_0,\cdots,v_t = P')$ under $\pfinal = P$. }
\end{itemize}

We shall reduce the discrepancy analysis to the probabilistic analysis of witness sequences via the following witness percolation lemma and its proof is deferred to~\Cref{subsub:witness-percolation}.
\begin{lemma}[Witness percolation]\label{lemma-witness-sequence-pinpdiff}
For each $P\in \pdiff\setminus \pdiffinit$, there exists a \emph{w.s.} $(v_0, v_1, \cdots,v_t)$ occurring in~\Cref{Alg:coupling} where {$v_t=P \text{ or } (v_t,P)\in \+S^{\ast}(\pzeta\cup \csmall,\+Z,1)$}. In addition, if $P\in \pdisc$ or $P\not\in \+Z$, then $v_t = P$.
\end{lemma}

By \Cref{lemma-witness-sequence-pinpdiff}, we have the following two corollaries immediately. 
\begin{corollary}\label{cor-witness-sequence-pinpdiff-large}
Given any $P\in \+P, P'\in \+D'[P]$ where $\pinit\not \subseteq P$ and $\abs{P} >\sizeuperboundpfinal$, assume $\pfinal = P$ and $P'\in \pdiff$.
Then there exists a \emph{w.s.} in $\+W[P']$ occurring in~\Cref{Alg:coupling}.
\end{corollary}

\begin{corollary}\label{cor-witness-sequence-pinpdiff-small}
Given any $P\in \+P, P'\in \+D[P]$ where $\pinit\not \subseteq P$ and $\abs{P} \leq\sizeuperboundpfinal$, assume $\pfinal = P$ and $P'\in \pdisc$.
Then there exists a \emph{w.s.} in $\+W[P,P']$ occurring in~\Cref{Alg:coupling}.
\end{corollary}

With these facts, we are now ready to bound the discrepancy. We deal with the discrepancy analysis for the permutation sets in $\+R$ and $\ol{\+R}$ in Lemmas~\ref{lemma-deviation-witness-small} and~\ref{lemma-deviation-witness-large}, respectively. 
\begin{lemma}\label{lemma-deviation-witness-small}
We have
\begin{equation*}
\begin{aligned}
\sum_{P\in \+R}\E{\Dis(Y'_1,Y'_2,\pfinal)\mid \pfinal = P} \leq  \sum_{P\in \+R}\sum_{P'\in \+D[P]}\sum_{s\in \+W[P,P']}\Pr{{s \mbox{ occurs}}\mid \pfinal =P}\abs{P'}.
\end{aligned}
\end{equation*}
\end{lemma}
\begin{proof}
   Assume $\pfinal=P$ for some $P\in \+R$.
According to \Cref{line-abs-pfinal-lessthan-threshold} in \Cref{Alg:coupling}, we have $\pzeta = \+D'\setminus \+D'[P]\cup \+D[P]$.
Given $\Phi_1$, $\Phi_2$ at \Cref{line-coupling},
let $(\tau_1,\tau_2)$ be a coupling of $\mu_{\Phi_1},\mu_{\Phi_2}$ where $\tau_1(P') = \tau_2(P')$ for each $P'\in \+D[P]\setminus \pdisc$.
According to the construction of $\pdisc$, one can verify that $\pdisc\subseteq \vset$ and the values on the variables in $V\setminus \vset$ can be coupled. Moreover, for any permutation set $P\in \pzeta \setminus \pdisc$, it is assigned by the same domains in $\Phi_1$ and $\Phi_2$. Therefore, the existence of coupling follows immediately.
Therefore, we have
\begin{align}\label{eqn-discrepancy-sum}
    \sum_{P'\in \+D[P]}\abs{\tau_1(P')\setminus\tau_2(P')} \leq \sum_{P'\in \+D[P]\cap  \pdisc}\abs{P'}.
\end{align}
By $\+D$ is a $\theta/\eta$-decomposition of $\+P'$, we have $\+D[P]$ is a a $\theta/\eta$-decomposition of $\+P'[P]$.
Thus, we have 
\[\sum_{P^{\ast}\in \+P'[P]}\abs{\tau_1(P^{\ast})\setminus\tau_2(P^{\ast})} \leq
\sum_{P^{\ast}\in \+P'[P]}\sum_{P'\in \+D[P^{\ast}]}\abs{\tau_1(P')\setminus\tau_2(P')} = \sum_{P'\in \+D[P]}\abs{\tau_1(P')\setminus\tau_2(P')}
\leq \sum_{P'\in \+D[P]\cap  \pdisc}\abs{P'},\]
where the last inequality holds by~\eqref{eqn-discrepancy-sum}.
Thus, under the condition $\pfinal=P$, we have
\begin{align*}
   \E{\sum_{P^{\ast}\in \+P'[P]}\abs{\tau_1(P^{\ast})\setminus\tau_2(P^{\ast})}}\leq \E{ \sum_{P'\in \+D[P]\cap  \pdisc}\abs{P'}}=\sum_{P'\in \+D[P]}\Pr{P'\in \pdisc}\abs{P'},
\end{align*}
where the expectation is taken over the random variables $\Phi_1,\Phi_2,\tau_1,\tau_2$.
Moreover, by {\Cref{lemma-witness-sequence-pinpdiff}}, for each $P'\in \pdisc$, there exists a \emph{w.s.} $(v_0, v_1, \cdots,v_t = P')$ occurring.
Therefore, 
\[\E{\sum_{P^{\ast}\in \+P'[P]}\abs{\tau_1(P^{\ast})\setminus\tau_2(P^{\ast})}}\leq \sum_{P'\in \+D[P]}\sum_{s\in \+W[P,P']}\Pr{{s \mbox{ occurs}}\mid \pfinal =P}\abs{P'}.\]
Given any $\Phi_1$ and $\Phi_2$, $(\sigma_1,\sigma_2)$ is an optimal coupling of $\mu_{\Phi_1},\mu_{\Phi_2}$ where $\E{\sum_{P^{\ast}\in \+P'[P]}\abs{\sigma_1(P^{\ast})\setminus\sigma_2(P^{\ast})}}$ is minimal according to \Cref{line-coupling}.
Thus, we have 
\begin{align*}
\E{\sum_{P^{\ast}\in \+P'[P]}\abs{\sigma_1(P^{\ast})\setminus\sigma_2(P^{\ast})}}&\leq \E{\sum_{P^{\ast}\in \+P'[P]}\abs{\tau_1(P^{\ast})\setminus\tau_2(P^{\ast})}}\\ &\leq \sum_{P'\in \+D[P]}\sum_{s\in \+W[P,P']}\Pr{{s \mbox{ occurs}}\mid \pfinal =P}\abs{P'}.
\end{align*}
Combined with \Cref{line-set-y} and the definition of $\Dis$ in \eqref{eq-def-dis}, we have 
\begin{align*}
\E{\Dis(Y'_1,Y'_2,P)} = \E{\sum_{P^{\ast}\in \+P'[P]}\abs{\sigma_1(P^{\ast})\setminus\sigma_2(P^{\ast})}}\leq \sum_{P'\in \+D'[P]}\sum_{s\in \+W[P,P']}\Pr{{s \mbox{ occurs}}\mid \pfinal =P}\abs{P'}.
\end{align*}
By summing over different $P\in \+R$, the lemma is immediate.
\end{proof}

Before we introduce~\Cref{lemma-deviation-witness-large}, we need more notations.
\begin{definition}\label{def-weight-function}
For any \emph{w.s.} $s = (v_0,\cdots,v_t)$,
define 
\begin{align}
\weightconstant &\triangleq 6000\Delta k^3d^3\sizeuperboundsquare,\label{eq-def-weightconstant}\\
f(s) &= 
\begin{cases}
\abs{v_t}k^2d^2\sizeuperbound & \text{ if $v_t \in \pzeta$,} \label{eq-def-fs}\\
k^2d^2\sizeuperboundsquare& \text{ otherwise.}\\
\end{cases} \\
\weightfunction(s) &= 
\begin{cases}
\abs{v_t}\weightconstant & \text{ if $v_t \in \pzeta$,} \label{eq-def-fs-hat}\\
\sizeuperbound\weightconstant & \text{ otherwise.}
\end{cases} 
\end{align}
\end{definition}
One can verify that $\weightfunction(s)$ is an upper bound of $f(s)$.
Then we have the following lemma, the proof of which is deferred to~\Cref{subsub:boundingdiscrepancy}.
\begin{lemma}\label{lemma-deviation-witness-large}
We have
\begin{align*}
\sum_{P\in \overline{\+R}}\E{\Dis(Y'_1,Y'_2,\pfinal)\mid \pfinal = P} \leq  100kd\Delta^2\sizeuperboundpfinal^{-1} \sum_{s\in \+{W}} f(s)\max_{P^{\ast}\in \+P} \Pr{{s \mbox{ occurs}}\mid \pfinal =P^{\ast} }.
\end{align*}
\end{lemma}

By Lemmas~\ref{lemma-deviation-witness-small} and ~\ref{lemma-deviation-witness-large}, it is crucial to establish the bounds for the probability of occurrence for a specific \emph{w.s.}. As the coupling algorithm exhibits different properties on different components in the \emph{w.s.}, we shall first deal with the components separately.
For any fixed $\pinit$ and any given ${P\in \pcurrent}, P'\in \pzeta[P]$, define
\begin{equation*}
\begin{aligned}
\lambda_{1}(P') = \min\left\{30k\Delta\abs{P}^{2\theta-1},1\right\}, \quad \lambda_{2}(P') = 
\begin{cases}
30\Delta\abs{P}^{2\theta-1}, & \text{ if $\abs{P} \geq \sizeuperboundlp$,} \\
15\Delta p^{\theta}\abs{P}^{2\theta}, & \text{ otherwise.}
\end{cases}
\end{aligned}
\end{equation*}
Furthermore, the parameters $p$, $\theta$ and $\sizeuperboundlp$ will guarantee that if $\abs{P} \leq \sizeuperboundlp$, then
\begin{align}\label{eq-condition-lambda12}
\lambda_2(P') \leq \lambda_1(P'). 
\end{align}
We then define the function $\rho$ for each \emph{w.s.} $s=(v_0,v_1,\cdots,v_t)$ as follows: 
\begin{align}\label{eq-def-rho}
        \rho(s)= p^{(1-\theta)(t-\abs{ V_{P}}-\abs{ V_{C}})}\cdot\Pr{v_0\in \pdiffinit} \prod_{v\in V_{P}}\lambda_1(v)\prod_{v\in V_{C}}\lambda_2(v).
\end{align}

The following lemma shows that the upper bound of the probability for the \emph{w.s.} occurrence can be bounded by $\rho$. For each $v\in V_P$, one can verify that $v\in \pdisc\setminus \pdiffinit$ with probability at most $\lambda_1(v)$ by~\Cref{lemma-trivial-coupling}. For any $v\in V_C$, one can verify that $v\in \pdisci$ if the w.s. occurs, and the probability is bounded by $\lambda_2(v)$ according to~\Cref{lemma-cpstep-failprob}. As for the constraints in $\cdisc$, all permutation sets on it are assigned values, yet it remains unsatisfied, which happens with probability at most $p^{1-\theta}$.
The lemma is immediate by combining these facts.
Its detailed proof is provided in~\Cref{subsub:wit-prop-cal}. 
{
\begin{lemma}\label{lemma-prob-witness}
Assume $8\-e\Delta p^{\theta}\leq 1$. For any \emph{w.s.} $(v_0,\cdots,v_t)$, we have 
\begin{equation}
\begin{aligned}
\Pr{(v_0,\cdots,v_t) \text{ occurs }} <_{q} \rho(v_0,\cdots,v_t).
\end{aligned}
\end{equation}
\end{lemma}
}

Consider the following conditions, where condition \ref{cond-decay5} is from \Cref{condition-couple-first}, condition \ref{cond-decay7} is from \eqref{eq-condition-lambda12}, conditions \ref{cond-decay6} and \ref{cond-decay8} are from the definitions of $\alpha$ and $\theta$ in \Cref{Alg:coupling},
and the other conditions are used in following proofs.
\begin{condition}\label{cond-branching-decay}
    Let $p,\theta\in(0,1)$ and $\Delta\geq 1$. It holds that
    \begin{enumerate}[(1)]
        \item if $\abs{P}\geq \sizeuperboundlp$, then 
        $61\Delta\abs{P}^{3\theta - 1}\weightconstant\left(1 + \abs{P}^{1-\theta}\min\left\{2k\Delta\abs{P}^{2\theta-1},1\right\}\right) \leq 1$;
        \label{cond-decay-small}
        \item if $\abs{P}< \sizeuperboundlp$, then 
        $31\Delta p^{\theta}\abs{P}^{3\theta}\weightconstant\left(1 + \abs{P}^{1-\theta}\min\left\{2k\Delta\abs{P}^{2\theta-1},1\right\}\right)\leq 1$;\label{cond-decay-large}
        \item $2(18k+2)k^2d^2\sizeuperbound\Delta\leq \sizeuperboundpfinal(500kd\Delta^2)^{-1}$;\label{cond-decay4}
        \item $3 p^{1-\theta} \sizeuperbound\weightconstant\leq 1$;\label{cond-decay3}
        \item $8\-e\Delta p^{\theta}\leq 1$;\label{cond-decay5}
        \item $2p^{\theta}\sizeuperbound \leq 1$.\label{cond-decay7}
        \item $\alpha \leq L$;\label{cond-decay6}
        \item $\theta = \eta/c$ for some integer $c$.\label{cond-decay8}
    \end{enumerate}
\end{condition}

The final ingredients are the following two lemmas.
\begin{lemma}\label{lemma-refutation-long-path}
    Given any instance satisfying Conditions 
    \ref{condition-initial} and \ref{cond-branching-decay}, if $\abs{\pfinal} > \sizeuperboundpfinal$, then
    \begin{align*}
    \sum_{s\in\+W}{\rho(s)}\cdot f(s) <_{q}  \sizeuperboundpfinal(500kd\Delta^2\sizeuperbound)^{-1}.
    \end{align*}
\end{lemma}

\begin{lemma}\label{lemma-refutation-long-path-small-pfinal}
    Given any instance satisfying Conditions 
    \ref{condition-initial} and \ref{cond-branching-decay},
    we have 
    \begin{align*}
    \sum_{P\in \+R}\sum_{P'\in \+D[P]}\sum_{s\in \+W[P,P']}\rho(s)\abs{P'} <_{q} 1/5. 
    \end{align*}
\end{lemma}

The proofs of \Cref{lemma-refutation-long-path} and \Cref{lemma-refutation-long-path-small-pfinal} are deferred to~\Cref{subsub:refutation}. 
Now we are ready to prove \Cref{lem-DP-bounds}.
\begin{proof}[Proof of~\Cref{lem-DP-bounds}]
At first, we solve Conditions \ref{condition-initial} and \ref{cond-branching-decay}.
One can verify that the following conditions imply Conditions \ref{condition-initial} and \ref{cond-branching-decay}.
\begin{enumerate}[(1)]
\item $\theta = \eta$ or $2\theta \leq \eta$;\label{final-condition-0}
\item $8\-e\Delta p^{\theta}\leq 1$, $3p^{\theta}\sizeuperbound \leq 1$;\label{final-condition-1}
\item $2\times 10^{4}k^3d^3\Delta^3\sizeuperboundlp^{\theta}\leq \alpha \leq L$;\label{final-condition-2}
\item $2\times 10^{4}k^3d^3\Delta \sizeuperboundlp^{3\theta}p^{1-\theta}\leq 1$;\label{final-condition-3}
\item if $\abs{P} > \sizeuperboundlp$, then $2\times 10^{6}k^4d^3\Delta^3\sizeuperboundlp^{2\theta}\abs{P}^{4\theta - 1}\leq 1$;\label{final-condition-4}
\item if $\abs{P} \leq \sizeuperboundlp$ and $2\Delta\abs{P}^{2\theta - 1}\leq 1$, then $4\-e p\Delta \abs{P}^{2\theta +1} \leq 1$, $10^6k^4d^3\Delta^3\sizeuperboundlp^{2\theta}p^{\theta}\abs{P}^{4\theta}\leq 1$;\label{final-condition-5}
\item if $\abs{P} \leq \sizeuperboundlp$ and $2\Delta\abs{P}^{2\theta - 1}>1$ then $2\-e p \abs{P}^2 \leq 1$, $10^6k^3d^3\Delta^2\sizeuperboundlp^{2\theta}p^{\theta}\abs{P}^{1+2\theta}\leq 1$.\label{final-condition-6}
\end{enumerate}

For general PDC formula, set the parameters as follows:
\begin{align*}
\theta = \frac{1}{8}, \quad \gamma = \Theta(k^{16}d^{12}\Delta^{12}), \quad p = O(k^{-128}d^{-96}\Delta^{-96}),\quad \alpha = \Theta(k^5d^{4.5}\Delta^{4.5}).
\end{align*}
One can verify that all the conditions are satisfied.
Thus, we have 

\begin{align*}
& \sum_{P\in \overline{\+R}}\E{\Dis(Y'_1,Y'_2,\pfinal)\mid \pfinal = P} \\
\leq &100kd\Delta^2\sizeuperboundpfinal^{-1} \sum_{s\in \+{W}} f(s)\max_{P^{\ast}\in \+P} \Pr{{s \mbox{ occurs}}\mid \pfinal =P^{\ast} }\tag{by~\Cref{lemma-deviation-witness-large}}\\
\leq & 100kd\Delta^2\sizeuperboundpfinal^{-1}\sum_{s\in\+W}{\rho(s)}\cdot f(s)\tag{by \ref{cond-decay5} in \Cref{cond-branching-decay} and \Cref{lemma-prob-witness}}\\
\leq & 1/4\tag{by \Cref{lemma-refutation-long-path} and $q_{\min}$ is large}.
\end{align*}
In addition, we also have 
\begin{align*}
& \sum_{P\in \+R}\E{\Dis(Y'_1,Y'_2,\pfinal)\mid \pfinal = P} \\
= &\sum_{P\in \+R}\sum_{P'\in \+D[P]}\sum_{s\in \+W[P,P']}\Pr{{s \mbox{ occurs}}\mid \pfinal =P}\abs{P'}\tag{by~\Cref{lemma-deviation-witness-small}}\\
\leq & \sum_{P\in \+R}\sum_{P'\in \+D[P]}\sum_{s\in \+W[P,P']}\rho(s)\abs{P'}\tag{by \ref{cond-decay5} in \Cref{cond-branching-decay} and \Cref{lemma-prob-witness}}\\
\leq & 1/4\tag{by \Cref{lemma-refutation-long-path} and $q_{\min}$ is large}.
\end{align*}
Combining the two inequalities, \Cref{lem-DP-bounds} is immediate.

For the $(k,q)$-uniform PDC formula, recall the definition of $r$ in \Cref{lemma-rapidmixing}.
We have $p <_{q} r^{-k}$. 
If $k\geq 24$, set
\begin{align*}
\sizeuperboundlp = r+\Theta(1), \quad \theta = 1/2, \quad  p = r^{-k} = O(k^{-12}d^{-8}\Delta^{-8}), \quad \alpha = \Theta(k^3d^{3}\Delta^{3}r^{1/2}).
\end{align*}
We have 
\[
\quad \forall P\in \+P', \quad \abs{P}\leq \gamma, \quad \sizeuperboundlp^{2\theta}p^{\theta}\abs{P}^{4\theta}<_{q} r^{-(k-6)\theta} = p^{\theta(k-6)/k}, \quad \sizeuperboundlp^{3\theta}p^{1-\theta}<_{q} r^{3\theta - k\theta} = p^{\theta(k-3)/k}.
\]
Thus, all the conditions in \ref{final-condition-0}-\ref{final-condition-6} are satisfied.
We also have \Cref{lem-DP-bounds}.
\end{proof}

\section{Analysis of the main algorithm}\label{sec:mainanalysis}

In this section, we analyze the performance of our main sampling algorithm (\Cref{Alg:MCMC}), and complete the proofs of the following theorems.

\begin{theorem} \label{theo:unif-formal}
    The following holds for any sufficiently large $q_{\min}$ and some constant $c$.
    There exists an algorithm such that given as input any $\varepsilon\in (0,1)$ and any $(k,q)$-uniform PDC formula $\Phi=(\+P,\+Q,\+C)$ with $n$ variables satisfying $q\geq q_{\min}$,  $k\geq 24$ and 
    $$ q^k\geq c \zeta k^{24}\Delta^{32}$$ for some $\zeta>1$,
    it outputs a random valid assignment $\sigma$ of $\Phi$ such that $\dtv(\sigma,\mu)\leq \varepsilon$ in running time $\widetilde{O}\left(k\Delta^2 n \left(\frac{n}{\varepsilon}\right)^{\theta}\right)$ where $\theta=c\zeta^{-1/2}k^{-12}\Delta^{-15}$.
\end{theorem}

In addition, we have the following result for general PDC formulas.
\begin{theorem}\label{thm-main-formal}
    The following holds for any sufficiently large $q_{\min}$ and some constant $c$.
    There exists an algorithm such that given as input any $\varepsilon\in (0,1)$ and any $\Phi=(\+P,\+Q,\+C)$ with $n$ variables satisfying $q\geq q_{\min}$ and  
    $${cpk^{518}\Delta^{786}}\leq 1,$$
    it outputs a random valid assignment $\sigma$ of $\Phi$ such that $\dtv(\sigma,\mu)\leq \varepsilon$ in running time $\widetilde{O}({\Delta n^{7}/\varepsilon^2})$.
\end{theorem}

\subsection{Factorization}

Consider any PDC formula $\Phi=(\+P,\+Q,\+C)$ with a decomposition $\+P'$ of $\+P$ satisfying~\Cref{condition-state-compression} and an error parameter $\varepsilon\in (0,1)$. As shown in~\Cref{subsection-sample}, the efficiency and the correctness of $\MCMC(\Phi,\+P',\varepsilon)$ require that the input of the $\samplepermutation(\cdot)$ subroutine satisfies~\Cref{condition-sample-small}. For each $t\in [T+1]$, let $\+Q'_t$ be the updated domain $\+Q'$ at step $t$ in $\MCMC(\Phi,\+P',\varepsilon)$~(\Cref{Alg:MCMC}). 
We also use $\Phi_{1}=(\+P_{1},\+Q_{1},\+C_{1}),\cdots,\Phi_{K}=(\+P_{K},\+Q_{K},\+C_{K})$ to denote the factorization of the random formula $(\+P',\+Q'_t,\+C)$ where $K$ is also a random variable.
Let $\+B_t$ be the event that there exists some $i\in [K]$ such that $\abs{\+C_{i}}>\Delta\log(n/\delta)$ where $\delta = \frac{\varepsilon}{3(T+1)}$. 
{The following lemma bound the probability of the occurrence of the event $\+B_t$.}
\begin{lemma}\label{thm-small-component}
     Given any $\varepsilon\in (0,1)$ and any $\Phi=(\+P,\+Q,\+C)$ with a decomposition $\+P'$ of $\+P$ satisfying ~\Cref{condition-state-compression}, {$8\-e p_\Phi^{\eta}\Delta\leq 1$}, {$8\-e\Delta\leq L$}, and {$2\-ep_{\Phi}^{0.999(1-\eta)}\Delta^2 L^2\leq 1$}, it holds that $\Pr{\+B_t}\leq \delta$ for each $t\in[T+1]$ with $\delta= \frac{\varepsilon}{3(T+1)}$ in~\Cref{Alg:MCMC}.
    
\end{lemma}

To prove~\Cref{thm-small-component}, we introduce the following events for convenience. For each $t\in [T+1]$, we define $\+B_{C}$ for any $C\in \+C$ to denote the event that $C$ is unsatisfied in the formula $(\+P',\+Q'_t, \+C)$, and there exists a factorized formula $\Phi_C=(\+P_C,\+Q_C,\+C_C)$ in $\Phi_{[K]}$ containing the constraint $C$ and $\abs{\+C_C}>\Delta \log(n/\delta)$.
Applying the union bound, it suffices to bound the probability of the event $\+B_{C}$ for any $C\in  \+C$ since
\begin{align}\label{eqn-unionbound-badevents}
    \Pr{\+B_t}\leq \Pr{\bigcup_{C\in \+C} \+B_{C}}\leq \sum_{C\in \+C}\Pr{\+B_{C}}.
\end{align}

We then construct the witness of the bad events $\+B_{C}$ for any $C\in  \+C$.
We introduce the dependency graph $G_{\Phi}^{\!{dep}}[P']$ induced by $\Phi$ and $\+P'$. Let $G_{\Phi}^{\!{dep}}[P'] = (\+{C},E_{\Phi}^{\!{dep}})$ be an undirected graph with the constraints $\+C$ as its vertices. For any $C,C'\in\+{C}$, $\set{{C}, {C'}}\in E_{\Phi}^{\!{dep}}$ if and only if $(\vbl(C)\cap P\neq \emptyset)\land(\vbl(C')\cap P \neq \emptyset)$ for some $P\in \+{P}'$.
That is, $\set{{C}, {C'}}\in E_{\Phi}^{\!{dep}}$ if the constraints $C$ and $C'$ are related to some permutation in $\+P'$. We also use $(G_{\Phi}^{\!{dep}}[P'])^2$ to denote the undirected graph with the same vertices where any two vertices are adjacent when their distance in $G_{\Phi}^{\!{dep}}[P']$ is at most 2. For any $C\in \+C$, we call $C$ a bad constraint if $C$ is not satisfied in the formula $(\+P',\+Q'_t, \+C)$. 

The witness is characterized by the following lemma, and its proof is deferred to~\Cref{subsub:mainanalysis}.
\begin{lemma}\label{lem-sample-witness}
    If the event $\+B_{C}$ occurs, then there exists a collection of bad constraints $\set{C_1,\cdots,C_\ell}\subseteq \+C$ with $\ell= \log(n/\delta)$ such that 1) $C\in \set{C_1,\cdots,C_\ell}$; 2) $\set{C_1,\cdots,C_\ell}$ forms an independent set in the lopsidependency graph $G_{\Phi}^{\!{lop}}$ and a connected component in $(G_{\Phi}^{\!{dep}}[P'])^2$.
\end{lemma}

The number of witnesses can be bounded by the following lemma in ~\cite{BCKLL}.
\begin{lemma}\label{lemma-bounding-connected}
    Let $G=(V,E)$ be a graph with maximum degree $\Delta$ and $v\in V$ be a vertex. Then the number of connected induced subgraphs of size $\ell$ containing $v$ is at most $(\-e\Delta^2)^{\ell-1}/2$.
\end{lemma}

The following lemma bounds the probability of the event that a given collection of constraints is bad, and its proof can be found in~\Cref{subsub:mainanalysis}.

\begin{lemma}\label{lemma-uncut-iteration}
    Given any PDC formula $\Phi=(\+P,\+Q,\+C)$ with a decomposition $\+P'$ of $\+P$ satisfying ~\Cref{condition-state-compression}, {$8\-ep_\Phi^{\eta}\Delta\leq 1$}, and ${8\-e\Delta\leq L}$, it holds that for any collection of lopsided independent constraints $\+C'\subseteq \+C$,
    \begin{align}
        \Pr{\mbox{The constraints $\+C'$ are bad in $(\+P',\+Q'_t,\+C)$} }{<_{q}} \prod_{C\in \+C'} {\^P_{\Phi}[\neg C]^{0.999(1-\eta)}.}
    \end{align}
\end{lemma}

We are now ready to complete the proof of~\Cref{thm-small-component}.
\begin{proof}[Proof of~\Cref{thm-small-component}]
     By~\Cref{lem-sample-witness}, when the event $\+B_C$ happens, there exists a collection of bad constraints $\set{C,C_2,\cdots,C_\ell}\subseteq \+C$ which is an independent set in the lopsidependency graph $G_{\Phi}^{\!{lop}}$, and a connected component in the power graph of the dependency graph $(G_{\Phi}^{\!{dep}}[P'])^2$. Applying Lemmas ~\ref{lemma-uncut-iteration}, ~\ref{lemma-bounding-connected} and the union bound, we have
    \begin{align*}
        \Pr{\+B_C}\leq \frac{(\-e\Delta^{2} L^2)^{\log(n/\delta)-1}}{2}\cdot {p^{0.999(1-\eta)\cdot \log(n/\delta)}} \leq \frac{\delta}{n\Delta},
    \end{align*} where the first inequality holds because the maximum degree of  $G_{\Phi}^{\!{dep}}[P']$ is less than $\Delta L$, and the last inequality holds by {$2 p^{0.999(1-\eta)}\cdot \-e\Delta^{2} L^2\leq 1$.}
    Consequently,
    \begin{align*}
        \Pr{\+B_t}\leq \sum_{C\in\+C} \Pr{\+B_C} \leq \delta,
    \end{align*} which implies the lemma immediately.
\end{proof}

We are now ready to complete the proof of~\Cref{thm-main-formal} and~\Cref{theo:unif-formal}.

\subsection{Analysis for $(k,q)$-uniform PDC formulas}

In this section, we complete the proof of~\Cref{theo:unif-formal}

    Given a sufficiently large $q_{\min}$ and any PDC formula $\Phi=(\+P,\+Q,\+C)$ where {$q\geq q_{\min}$}, we have $p_{\Phi}=_{q} q^{-k}$. Let $L=_{q} q^{\eta}$. For any $\eta\in (0,1/2)$, one can verify that the existence of the decomposition satisfying~\Cref{condition-state-compression}.
    In this case, we further refine the implementation of~\Cref{Alg:MCMC} and its subroutine~\Cref{Alg:samplepermutation}. 
    \begin{itemize}
        \item In~\Cref{Alg:MCMC}, we replace the number of loops at~\Cref{line-main-loops} to $\left\lceil 2(n/q)\log\left({3n}/{\varepsilon}\right) \right\rceil$ where $n/q$ is the number of permutation sets in $\+P$ since the contraction in~\eqref{eqn-contraction} is $1-1/(2\abs{P})$. Moreover, let $\sigma_i\gets \truncatedsampling(\Phi_i,{\lceil (n/\varepsilon)^{\theta/10} \log  (3n(T+1)/\varepsilon)\rceil)}$ at~\Cref{line-main-rej-2} for some $\theta\in (0,1)$ in~\Cref{Alg:MCMC}. We also terminate the algorithm if we detect some formula has more than $\Delta\log(\frac{n}{\varepsilon/(3T+1)})$ constraints at~\Cref{line-main-factorization}.
        \item In~\Cref{Alg:samplepermutation}, let $T\gets \lceil  \log(2/\varepsilon) \rceil$ instead at~\Cref{line-rejectionsampler-time} and $\sigma_i\gets \truncatedsampling(\Phi_i,\lceil (n/\varepsilon)^{\theta/10}\cdot \log  (2nT/\varepsilon)\rceil)$ instead at~\Cref{line-rejection-in-sample-1} for some $\theta\in (0,1)$. 
    \end{itemize}
     We list the following conditions for convenience:
     \begin{enumerate}
         \item\label{unif-condition-main-1} $ {8\-e q^{-\eta k}\Delta\leq 1}, {2\-eq^{-0.999k+\eta\cdot (0.999k+2)}\Delta^2\leq 1}$;
         \item\label{unif-condition-main-2} $
        4q^{-\eta(k-1)}\Delta \leq 1, q^{\eta k}\geq {20} \-e\Delta/{\theta}$ for some $\theta\in (0,1)$;
        \item \label{unif-condition-main-3} $k\geq 24, ck^{12}\Delta^{16}\leq q^{\eta k}$;
        \item \label{unif-condition-main-4} $q^\eta \geq c(dk^2\Delta^2)^{1+5/(3k)}$ for some constant $c$.
     \end{enumerate}
     One can verify that~\eqref{unif-condition-main-1},~\eqref{unif-condition-main-2} and~\eqref{unif-condition-main-3} hold when 
    \begin{align}\label{eqn:anscond1}
        k\geq 24,\quad q^k\geq c \zeta k^{24}\Delta^{32}
    \end{align} for some $\zeta>1$ and constant $c$, by setting $\eta = 1/2$ and $\theta=c\zeta^{-1/2}k^{-12}\Delta^{-15}$. 

\vspace{0.5cm}
     We first claim the correctness of the refined algorithm conditioned on~\eqref{eqn:anscond1}. The condition~\eqref{unif-condition-main-2} ensures that $p_{\max}\abs{P}\Delta\leq 1/4$ always holds at~\Cref{line-if-sample} in~\Cref{Alg:samplepermutation} since $p_{\max}=_{q} q^{-\eta(k-1)-1}$. Therefore, we only implement from~\Cref{line-if-sample} to~\eqref{line-break-sample}. One can verify that given the input satisfying \Cref{condition-sample-small} and $q^{\eta k}\geq {20} \-e\Delta/{\theta}$, $\samplepermutation(\Phi,P,\varepsilon)$ still returns a random valid assignment $\sigma$ of $\Phi$ with total variation distance bounded by $\varepsilon$ by the following facts:
    \begin{itemize}
        \item In \Cref{lemma-simplecase-repeattime}, we have $\Pr[\sigma_{[K]}\sim \mu_{[K]}]{\sigma\in \Omega_{\Phi}} \geq 1/2.$
        \item Given that $q^{\eta k}\geq {20} \-e\Delta/{\theta}$, we set $x(\neg C)=\frac{\theta}{20\Delta}$. One can verify that~\Cref{thm-lopsiLLL} holds and for each $i\in [K-1]$ in ~\Cref{Alg:samplepermutation}, we have
        \begin{align*}
            \^P_{\Phi_i}\left[\bigwedge_{C\in \+C_i} C \right]&\geq \left(1-\frac{\theta}{20\Delta}\right)^{\Delta\log(n/\varepsilon)}\geq  \left(1-\frac{1}{15\Delta/\theta+1}\right)^{\Delta\log(n/\varepsilon)}\geq \exp\left(-\frac{\theta}{15}\log\frac{n}{\varepsilon}\right)\geq \left(\frac{n}{\varepsilon}\right)^{\theta/10}.
        \end{align*}
    \end{itemize}
    Moreover, the rapid mixing property is implied by~\eqref{unif-condition-main-3} or~\eqref{unif-condition-main-4} according to~\Cref{lemma-rapidmixing}.
    
     Let $X_T$ denote the random assignment of $M_{\mathrm{Glauber}}$ runs for $T=\left\lceil 2n/q\log\frac{n}{\varepsilon} \right\rceil$ steps, and $\tau$ be the random assignment sampled from $\mu_{\Phi'}$ where $\Phi'=(\+P',X_T,\+C)$. Recall that $\sigma$ is the random valid assignment returned by~\Cref{Alg:MCMC}.
    We claim that $\dtv(\sigma,\tau)\leq \frac{2\varepsilon}{3}$. At each step of the idealized permutation-wise Glauber dynamics, we couple $M_{\mathrm{Glauber}}$ and $\MCMC(\Phi,\+P',\varepsilon)$ by sharing the randomly chosen permutation $P$. Moreover, we apply the optimal coupling between the $\samplepermutation(\cdot)$ subroutine and the idealized update on the domains in $\+P[P]$. The coupling errors come from the following two circumstances:
    \begin{enumerate}
        \item The input $(\Phi',P,\frac{\varepsilon}{3(T+1)})$ of $\samplepermutation(\cdot)$ subroutine does not satisfy ~\Cref{condition-sample-small};
        \item The intrinsic error comes from $\samplepermutation(\Phi',P,\frac{\varepsilon}{3(T+1)})$ even when $(\Phi',P,\frac{\varepsilon}{3(T+1)})$ satisfying ~\Cref{condition-sample-small}.
    \end{enumerate}  
    According to ~\Cref{thm-small-component}, the probability of the first case mentioned above can be bounded by $\frac{\varepsilon}{3(T+1)}$ at each step. Furthermore, conditioned on $(\Phi',P,\frac{\varepsilon}{3(T+1)})$ satisfying ~\cref{condition-sample-small}, coupling error from $\samplepermutation(\cdot)$ subroutine can be upper bounded by $\frac{\varepsilon}{3(T+1)}$ according to ~\Cref{thm-sample-correctness}. The statement holds immediately.
    Recall that $\nu$ is the stationary distribution of the idealized permutation-wise Glauber dynamics.
    According to~\eqref{unif-condition-main-3} or~\eqref{unif-condition-main-4}, we have $\dtv(X_t,\nu)\leq \frac{\varepsilon}{3}$ by~\Cref{lemma-rapidmixing}, and hence, $\dtv(\sigma,\mu_\Phi)\leq \dtv(X_t,\nu) + \dtv(\sigma,\tau) \leq \varepsilon$.

\vspace{0.5cm}

    As for the efficiency, we emphasize the time complexity of the factorization at~\Cref{line-main-rej-2} and the $\samplepermutation(\cdot)$ subroutine:
    \begin{itemize}
        \item In the factorization process, we take each permutation $P\in \+P'$ as a ``vertex''. At~\Cref{line-main-factorization} in~\Cref{Alg:MCMC}, we transverse all constraints in $\+C(P)$ and implement the DFS algorithm to check whether there exists a connected component containing more than $\Delta\log(\frac{n}{\varepsilon/(3T+1)})$ constraints. Overall, one can verify that it can be done in $O\left(qk\Delta^2\log(\frac{n}{\varepsilon/(3T+1)})\right)$.
        \item According to the execution of~\Cref{Alg:MCMC}, the $\samplepermutation(\cdot)$ subroutine is called only when $\Phi_1,\cdots, \Phi_{K-1}$ have no more than $\Delta\log(n/\varepsilon)$ constraints. Note that $K\leq  q\Delta$. Therefore, we should resample the value on at most $q\Delta^2\log(n/\varepsilon)$ constraints. Instead of sampling the whole formulas $\Phi_1,\cdots, \Phi_{K-1}$, it suffices to reveal the values of the variables on these constraints at~\Cref{{line-rejection-in-sample-1}}. This process requires $O(kq\Delta^2\log(n/\varepsilon))$ time. Additionally, to implement~\eqref{line-break-sample}, we need to further reveal the variables involved in the constraints $\+C(P)$ and some formulas $\Phi_i$ where $i\in [K-1]$, provided they have not already been revealed at~\Cref{{line-rejection-in-sample-1}}. However, these variables can be simply sampled uniformly at random in $O(q\Delta)$ time.
        As mentioned before, the rejection sampling repeats for ${\lceil (n/\varepsilon)^{\theta/10} \log  (3n(T+1)/\varepsilon)\rceil)}$ round. Overall, $\samplepermutation(\cdot)$ subroutine runs in $\widetilde{O}(kq\Delta^2(n/\varepsilon)^{\theta/10})$ time.
    \end{itemize}
    Overall, the time complexity of the refined algorithm is $\widetilde{O}\left(k\Delta^2 n \left(\frac{n}{\varepsilon}\right)^{\theta}\right)$.
   Combinging all these facts, the lemma is immediate.

\subsection{Analysis for general PDC formulas}
In this section, we complete the proof of~\Cref{thm-main-formal}.

Given a sufficiently large $q_{\min}$ and any PDC formula $\Phi=(\+P,\+Q,\+C)$ where {$q\geq q_{\min}$}, {one can verify that the existence of the decomposition satisfying~\Cref{condition-state-compression}.} 
  We list the following conditions for convenience:
     \begin{enumerate}
         \item\label{general-condition-main-1} $ {8\-e p_\Phi^{\eta}\Delta\leq 1},{8\-e\Delta\leq L},{2\-ep_{\Phi}^{0.999(1-\eta)}\Delta^2 L^2\leq 1}$;
         \item\label{general-condition-main-2} $
       p_\Phi^\eta \cdot 6912\e^3\ln2 \cdot k^3L^2\Delta^9\leq 1$;
        \item \label{general-condition-main-3} $ cp_{\Phi}^{\eta}k^{128}\Delta^{192}\leq 1, ck^{128}\Delta^{192}\leq L$.
     \end{enumerate}
     By setting $\eta=1/2$ and $L=_{q} k^{128}\Delta^{192}$, one can verify that ~\eqref{general-condition-main-1},~\eqref{general-condition-main-2}, and ~\eqref{general-condition-main-3} hold when $cpk^{518}\Delta^{786}\leq 1 $ for some constant $c$.

    Let $\left(\Phi',P,\delta\right)$ be the input of the $\samplepermutation(\cdot)$ subroutine in~\Cref{Alg:MCMC} where $\delta= \frac{\varepsilon}{3(T+1)}$. Note that $p_{\Phi'}\leq \min\set{p^{\eta},1/L}$. Therefore, ~\eqref{general-condition-main-1} implies the conditions specified in \Cref{thm-small-component} and the instance $\left(\Phi',P,\delta\right)$ satisfying~\Cref{condition-sample-small} with high probability. 
    Moreover, to ensure the correctness of $\samplepermutation(\cdot)$ subroutine under~\Cref{condition-sample-small}, we require that either $p_{\max}\abs{P}\Delta\leq 1/4$ or $\abs{P}\geq 1728\e^3\ln2 \cdot k^3L^2\Delta^8$ by~\Cref{thm-sample-correctness}. Specifically, we require that $p_{\max}\abs{P}\Delta\leq 1/4$ when $\abs{P}<1728\e^3\ln2 \cdot k^3L^2\Delta^8$. Note that the PDC formula $\Phi'$ is defined on the decomposition $\+P'\circ P$. We claim that for each unsatisfied constraint $C\in \+C(P)$, $\abs{\vbl(C)}\geq 2$. This is because when $\abs{\vbl(C)}=1$ and $C\in \+C(P)$, its violation probability is at most $1/\abs{P}$ which contradicts to $cpk^{518}\Delta^{786}\leq 1$.   
    When $\abs{\vbl(C)}\geq 2$, one can verify that $\^P_{\Phi'}[\neg C]\leq \max \set{p_\Phi^\eta, \frac{1}{ \abs{P}L}}$ for each $C\in \+C(P)$. Therefore $p_{\max}\abs{P}\Delta\leq 1/4$ holds when $\abs{P}<1728\e^3\ln2 \cdot k^3L^2\Delta^8$ by~\eqref{general-condition-main-2}. 
    Moreover, the rapid mixing property holds since~\eqref{general-condition-main-3} implys the specified conditions in~\Cref{lemma-rapidmixing}. Combining all these facts, the correctness is immediate follows from similar argument in the proof of~\Cref{theo:unif-formal}.

     As for the time complexity, the construction of the decomposition can be done in $O(n)$ times. In each iteration of $t$, the algorithm initially locates the factorized formula containing $P$ in $(\+P'\circ P,\+Q'\circ \+Q(P),\+C)$, which costs $\widetilde{O}(\abs{V}+\abs{\+C\cup \+P})=\widetilde{O}(n+n\Delta + n)=\widetilde{O}(n \Delta)$ times. Moreover, the $\samplepermutation(\cdot)$ subroutine cost at most $n^{4}/\delta^2$ by ~\Cref{thm-sample-correctness} where $\delta=\frac{\varepsilon}{3(T+1)}$. As the algorithm iterates for $T$ times, the total time complexity is $\widetilde{O}({\Delta n^{7}/\varepsilon^2})$.

    Combining all these facts, the lemma is immediate.
\section{Applications}\label{sec-application}
This section presents several applications of our PDC sampling algorithm.

\subsection{Resource-allocations}
A teacher assignment problem with parameters $q,k,L$ is defined as follows. 
Consider a grade with $\ell$ classes and $m$ distinct subjects.  Each class must offer $k$ subjects to its students, and—because the classes have different specializations—the required sets of $k$ subjects may differ from one class to another.  
For each subject $i \in [m]$, exactly $q$ teachers are qualified, and at least $q - L$ of them, designated as the subset $S_i$, are senior enough to serve as homeroom teachers.
No teacher is qualified to teach more than one subject.
We wish to assign teachers to classes under the following conditions:
\begin{itemize}
\item Each teacher is assigned to at most one class.
\item Every class receives exactly one teacher for each of its $k$ required subjects.
\item Every class includes at least one senior instructor to serve as the homeroom teacher.
\end{itemize}
To ensure fairness across classes and subjects, an assignment is selected uniformly at random from all allocations that satisfy these constraints.

For each class $i\in[\ell]$, let $\{i_1,\dots,i_k\}\subseteq[m]$ denote the $k$ subjects required by that class.  
We further assume that, across all $\ell$ classes, the total demand for instructors in any fixed subject never exceeds the available supply $q$.  To formalize individual teaching assignments, we attach a unique identifier to every subject request: for every pair $(i,j)$ with $i\in[\ell]$ and $j\in[k]$, let $t^{i}_{j}$ be a distinct label representing class $i$’s request for an instructor in subject $i_{j}$.  Different classes requesting the same subject are assigned distinct labels, ensuring that each teaching demand can be treated as a separate entity in the formal model.


The sampling of a random assignment of teachers can be modeled as the sampling of a PDC formula $\Phi = (\+P,\+Q,\+C)$ as follows.
Here, $\+P = (P_1,P_2,\cdots,P_{m})$ and 
$\+Q = (Q_1,Q_2,\cdots,Q_{m})$ where $Q_i = [q]$ for each $i\in [m]$.
Each class $i\in [\ell]$ defines a set of constraints $\+C_i\subseteq \+C$ as follows:
\begin{itemize}
\item $\left(v^{i_1}_{t^i_1} \neq a_1\right)\lor \cdots\lor \left(v^{i_{k}}_{t^i_{k}} \neq a_k\right)$ for each $a_1\not\in S_{i_1},\cdots,a_k\not\in S_{i_k}$.
\end{itemize}
Thus, the probability $p$ that a constraint is violated is bounded by $q^{-k}$.
Additionally, if $\abs{S_i} \geq q - L$ for each $i\in [m]$,
one can verify that the degree of the lopsidependency graph is upper bounded by $L^k + kqL^{k-1}$,
because of the following reasons:
\begin{itemize}

\item For each $i\in [\ell]$, we have $\abs{\+C_i} \leq L^k$;
\item For each $i\in [\ell]$, no constraint outside $\mathcal{C}_i$ shares a variable with any constraint inside $\mathcal{C}_i$;
\item For every $i \in [m]$ and each $a \in [q] \setminus S_i$, at most $q\,L^{k-1}$ constraints contain a literal of the form $v \neq a$ with $v \in P_i$.
\end{itemize}
By \Cref{theo:unif}, we have the following theorem.

\begin{theorem}\label{theorem-teacher-assignment}
The following holds for some positive constants $q_{\min}$ and $c$. There is an algorithm such that given as input any $\varepsilon\in (0,1)$ and any teacher assignment problem with $q,k,L$ satisfying $q\geq q_{\min}$, $k> 32$ and 
\[ q^{k-32}\geq c k^{56}L^{32k},\]
the algorithm terminates in time $\widetilde{O}({k^3q^{2}L^{2k}n^{1.001} /\varepsilon})$ and outputs a random assignment of teachers that is $\varepsilon$-close in the total variation distance to the uniform distribution of all satisfying assignments. 
\end{theorem}

The theorem demonstrates that for large $q$, $k$, sampling an assignment of teachers is efficient if $L$ is $\tilde{O}(q^{1/32})$, where high-order terms are suppressed.
Before this work, no nontrivial algorithm existed for sampling teacher assignments.

A similar result can also be proved for the reviewer assignment problem.

\subsection{Factors of independent transversals in multipartite hypergraphs.}
A $k$-uniform hypergraph is a set of hyperedges each of which contains exactly $k$ vertices.
An $m$-partite hypergraph is a hypergraph in which the vertices are partitioned into $m$ disjoint sets and each hyperedge contains at most one vertex of each set.
Given an $m$-partite hypergraph $H = (V_1,V_2,\cdots,V_m,E)$, let $V$ denote $V_1\cup \cdots \cup V_m$.
If $\abs{V_i} = q$ for every $i\in [m]$, $H$ is called \emph{$q$-balanced}.
An independent transversal of \( H \) is a set of \( m \) vertices, containing exactly one vertex from each part, such that it does not contain all vertices of any hyperedge.
A factor of independent transversals of $H$ is a collection of disjoint independent transversals that spans all of the vertices in $V$.
We remark that only balanced partite hypergraphs have factors of independent transversals.

The factor of independent transversals is closely related to the perfect matching in multipartite hypergraphs.
A matching in a hypergraph is a set of hyperedges, in which every two hyperedges are disjoint.
A perfect matching is a matching covering all the vertices.
For any balanced $k$-partite $k$-uniform hypergraph $H = (V_1,V_2,\cdots,V_k,E)$, 
each factor of independent transversals in $H$ is exactly a perfect matching of the hypergraph $\overline{H}=(V_1,V_2,\cdots,V_k,\overline{E})$ where $\overline{E} = V_1\times \cdots\times V_k\setminus E$.

Independent transversals of multipartite hypergraphs were initially studied by
Erd\H{o}s, Gy\'{a}rf\'{a}s, and  \L uczak~\cite{erdHos1994independent}.
The problem of establishing sufficient
conditions for the existence of independent transversals and their factors in
$q$-balanced $m$-partite $k$-uniform hypergraphs has been widely explored in the literature, 
though results are only available for certain special cases.
The original work of Erd\H{o}s, Gy\'{a}rf\'{a}s, and  \L uczak focuses on sparse hypergraphs~\cite{erdHos1994independent},
while many subsequent studies concentrate on the simplest case, $k=2$, where the hypergraph degrades into a graph
~\cite{yuster1997independent,yuster2021factors,glock2022average}.
Through the connection between factors of independent transversals and perfect matchings, 
a tight condition can be derived for the existence of a factor of independent transversals in $k$-partite $k$-uniform hypergraphs
~\cite{aharoni2009perfect}.
Moreover, it is worth noting that finding a factor of independent transversals in a multipartite hypergraph is NP-hard, as finding a perfect matching in a given balanced 3-partite 3-uniform hypergraph was one of the first problems proven to be NP-hard~\cite{karp2009reducibility}.

Here, we provide a fast algorithm for sampling factors of independent transversals in multipartite hypergraphs. 
Given a $q$-balanced $m$-partite $k$-uniform hypergraph $H = (V_1,V_2,\cdots,V_m,E)$, 
let $d$ denote its maximum vertex degree.
The sampling of a random factor of independent transversals in $H$ can be modeled as the sampling of a PDC formula $\Phi = (\+P,\+Q,\+C)$ as follows.
Here, $\+P = (P_1,P_2,\cdots,P_{m})$ where $P_i$ contains the vertices in $V_i$ and 
$\+Q = (Q_1,Q_2,\cdots,Q_{m})$ where $Q_i = [q]$ for each $i\in [m]$, and $\+C$ contains all the constraints corresponding to the following conditions:
\begin{itemize}
\item $(v_1 \neq i)\lor \cdots\lor (v_{k} \neq i)$ for each $(v_1,\cdots,v_{k})\in E$ and $i\in [q]$.
\end{itemize}
Thus, the probability $p$ that a constraint is violated is bounded by $q^{-k}$.
Additionally, one can verify that the degree of the lopsidependency graph is $2kdq$,
because each constraint contains $k$ variables, each variable can be contained in $d$ constraints, and can take at most $q$ different values.
By \Cref{theo:unif}, we have the following theorem.

\begin{theorem}\label{theorem-independent-cover}
The following holds for some positive constants $q_{\min}$ and $c$. There is an algorithm such that given as input any $\varepsilon\in (0,1)$ and any $q$-balanced $m$-partite $k$-uniform hypergraph $H$ with maximum vertex degree $d$ satisfying $q\geq q_{\min}$, $k> 32$ and 
\[ q^{k-32}\geq c k^{56}d^{32},\]
the algorithm terminates in time $\widetilde{O}({k^3d^2q^{3.001}m^{1.001} /\varepsilon})$ and outputs a random factor of independent transversals of $H$ that is $\varepsilon$-close in the total variation distance to the uniform distribution of all factors. 
\end{theorem}

The theorem demonstrates that for large $q$, $k$, 
sampling a factor of independent transversals in a $q$-balanced $m$-partite $k$-uniform hypergraph is efficient if a constant power of the maximum vertex degree is bounded by $q^k$.
We emphasize that sampling of factors of independent transversals outside the LLL-like regime is intractable. By introducing \( q-1 \) dummy variables for each vertex, sampling \( q \)-colorings for a \( k \)-uniform hypergraph with \( m \) vertices can be reduced to sampling factors of
independent transversals in a \( q \)-balanced, \( m \)-partite, \( k \)-uniform hypergraph. Moreover, sampling hypergraph colorings is intractable beyond the LLL-like regime~\cite{galanis2021inapproximability}. Therefore, to make sampling factors of independent transversals tractable, it is natural to assume an LLL-like regime. Our theorem establishes such a regime, and the algorithm achieves an almost linear running time in \( n \). To the best of our knowledge, no prior nontrivial sampling algorithm for factors of independent transversals has been proposed.


\section*{Acknowledgement}
We thank Qingyuan Li, Fangjie Peng, Haoran Wang, and Yitong Yin for helpful discussion.

\bibliographystyle{alpha}
\bibliography{refs}

\clearpage


\appendix
\section{Missing proofs in section~\ref{sec-Preliminaries}}
\subsection*{Proof of~\Cref{Coro:LS}} \label{sec:lll}

We first claim that for any event $B\in \+B$ and a subset $S\subseteq \+B \setminus \{B\}$,
\begin{equation} \label{eq:prob_A}
\Pr{B \bigmid \bigwedge_{B'\in S} \overline{B'}} \le x(B).
\end{equation}
We apply induction on the size of $S$. When $\abs{S}=0$, the statement~ \eqref{eq:prob_A} holds immediately by the assumption of the condition \eqref{eq-condition-asym-lop} in \Cref{thm-lopsiLLL}. Suppose the statement~\eqref{eq:prob_A} is true for any $S$ such that $\abs{S}<s$, we show it for any subset $S$ where $\abs{S}=s$ as follows: 

Recall that $\Gamma(B)$ is the set of neighbors of $B$ in the lopsidependency graph $G$ for any event $B\in \+B$.
Let $S_1 = S\cap \Gamma(B)$, and $S_2 = S \setminus \Gamma(B)$. We have
\begin{equation} \label{eq:prob_A2}
\Pr{B \bigmid \bigwedge_{B'\in S} \overline{B'}} =
\frac{\Pr{B \land \left( \bigwedge_{B'\in S_1} \overline{B'} \right) \bigmid \bigwedge_{B'\in S_2} \overline{B'}}} {\Pr{\bigwedge_{B'\in S_1} \overline{B'} \bigmid \bigwedge_{B'\in S_2} \overline{B'}}}.
\end{equation}
The numerator in \eqref{eq:prob_A2} can be bounded by
\begin{equation} \label{eq:numerator}
\begin{aligned}
    \Pr{B \land \left( \bigwedge_{B'\in S_1} \overline{B'} \right) \bigmid \bigwedge_{B'\in S_2} \overline{B'}} &\le \Pr{B  \bigmid \bigwedge_{B'\in S_2} \overline{B'}}\\
    &\le \Pr{B} \quad \text{(due to {non-lopsidependency})}\\
    &\le x(B) \prod_{B'\in \Gamma(B)} \left( 1-x(B') \right).
\end{aligned}
\end{equation}
As for the denominator in \eqref{eq:prob_A2}, if $\abs{S_1}=0$, the denominator is $1$ and \eqref{eq:prob_A} follows from \eqref{eq:numerator}. Otherwise, $\abs{S_1}>0$ and $\abs{S_2}<s$, we can apply the induction hypothesis to $S_2$. Let $S_1 = \left\{ B_1, B_2, \cdots, B_r \right\}$ where $r=\abs{S_2}$. It holds that
\begin{equation} \label{eq:denominator}
\begin{aligned}
    \Pr{\bigwedge_{B'\in S_1} \overline{B'} \bigmid \bigwedge_{B'\in S_2} \overline{B'}} &= \left( 1-\Pr{B_1 \bigmid \bigwedge_{B'\in S_2} \overline{B'}} \right) \left( 1-\Pr{B_2 \bigmid \overline{B}_1 \land \left( \bigwedge_{B'\in S_2} \overline{B'} \right)} \right) \cdots \cdots\\
    &\quad \left( 1-\Pr{B_r \bigmid \overline{B}_1 \land \cdots \land \overline{B}_{r-1} \land \left( \bigwedge_{B'\in S_2} \overline{B'} \right)} \right)\\
    &\ge \left( 1-x(B_1) \right) \left( 1-x(B_2) \right) \cdots \left( 1-x(B_r) \right) \quad \text{(induction hypothesis)}\\
    &= \prod_{B'\in S_1} \left( 1-x(B') \right)\\
    &\ge \prod_{B'\in \Gamma(B)} \left( 1-x(B') \right). \quad \quad (S_1 \subseteq \Gamma(B))
\end{aligned}
\end{equation}
Combining \eqref{eq:prob_A2}, \eqref{eq:numerator} and \eqref{eq:denominator}, we get $\Pr{B \mid \bigwedge_{B'\in S} \overline{B'}} \le x(B)$ and the statement \eqref{eq:prob_A} is proved.

Now \Cref{Coro:LS} can be proved. For any event $A$ defined on the probability space, suppose $T$ is a collection of events in $\+B$ such that $A$ is non-lopsidependent with events in $\+B \setminus T$. We have
\begin{equation} \label{eq:prob_A3}
    \Pr{A \bigmid \bigwedge_{B\in \+B} \overline{B}}
    = \frac{\Pr{A \land \left( \bigwedge_{B\in T} \overline{B} \right) \bigmid \bigwedge_{B\in \+B \setminus T} \overline{B}}} {\Pr{\bigwedge_{B\in T} \overline{B} \bigmid \bigwedge_{B\in \+B \setminus T} \overline{B}}}.
\end{equation}
Again, the numerator in \eqref{eq:prob_A3} can be bounded by $ \Pr{A \mid \bigwedge_{B\in \+B \setminus T} \overline{B}}$ which is at most $\Pr{A}$ by the assumption of non-lopsidependence. Combining this fact with ~\eqref{eq:denominator}, we have
$$
\Pr{A \bigmid \bigwedge_{B\in \+B} \overline{B}}
\le \Pr{A} \prod_{B\in T} \left( 1-x(B) \right)^{-1},
$$
which implies ~\Cref{Coro:LS}. {The second part can be proved by setting $x(B) = \-e \Pr{B}$.}

\section{Missing proofs in section~\ref{sec:compression}}\label{appendix-compression}

\subsection*{Proof of~\Cref{lemma:lop-graph}}

Let $C_1\equiv(v_{1} \ne c_{1}) \lor (v_{2} \ne c_{2}) \lor \cdots \lor (v_{k_1} \ne c_{k_1})$ and $C_2 \equiv (v'_{1} \ne c'_{1}) \lor (v'_{2} \ne c'_{2}) \lor \cdots \lor (v'_{k_2} \ne c'_{k_2}).$
By definition, if $(C_1, C_2) \in G^{\!{lop}}_{\Phi'}$,  there exists some $P'\in \+P'$ such that $ v_{i}, v'_{j} \in P' \text{ and } (v_{i} = v'_{j} \lor c'_{i} = c'_{j})$.
Let $P \in \+P$ be the permutation containing $P'$, then $v_{i}, v'_{j} \in P$ and $(v_{i} = v'_{j} \lor c_{i} = c'_{j})$. Therefore, $(C_1, C_2) \in G^{\!{lop}}_{\Phi}$.
We have proved that $(C_1, C_2) \in G^{\!{lop}}_{\Phi'}$ implies $(C_1, C_2) \in G^{\!{lop}}_{\Phi}$.
Thus, by $(C_1, C_2) \notin G^{\!{lop}}_{\Phi}$ we have 
$(C_1, C_2) \notin G^{\!{lop}}_{\Phi'}$. 
The lemma is immediate.

\subsection*{Proof of~\Cref{lem-violation-decomposed}}

In what follows, we use $a^{\underline{b}}$ to denote the falling factorial $a\cdot (a-1)\cdots(a-b+1)$.

For constraint $C\in \+C$, suppose variables in $\vbl(C)$ are in permutations $P_1, P_2, \cdots, P_m \in \+P$ with $q_i \triangleq \abs{P_i}$, and define $k_i \triangleq \abs{\vbl(C) \cap P_i}$, then the violation probability of $C$ is $\prod_{i} {\left(q_i^{\underline{k_i}}\right)}^{-1}$. Note that
\begin{align}\label{eqn-violationbound1}
    \prod_{i} \left({q_i^{k_i}}\right)^{-1}
\le \prod_{i} \left({q_i^{\underline{k_i}}}\right)^{-1}  
\le p.
\end{align}
Moreover, in the decomposed formula $\Phi'$, suppose $P_i \in \+P$ is decomposed into $P'_{1}, P'_{2}, \cdots, P'_{\ell}$ for each $i\in [m]$. Define $q_{ij} \triangleq \abs{P'_{j}}$ and $k_{ij} \triangleq \abs{\vbl(C) \cap P'_{j}}$ for each $j\in [\ell]$.
One can verify that the violation probability of $C$ in the decomposed formula can be bounded as follows:
\begin{align}\label{eqn-violationbound2}
    \prod_{i,j} \left({q_{ij}^{\underline{k_{ij}}}}\right)^{-1}\le \prod_{i,j} \left( \frac{\-e}{q_{ij}} \right)^{k_{ij}}.
\end{align}
Since $\+P'$ is a $\zeta$-decomposition of $\+P$, we have {$ {q_i^\zeta}=_{q}  q_{ij}$} for each $j\in [\ell]$. Consequently, the violation probability of $C$ in the decomposed formula can then be upper bounded by
$$
\prod_{i,j} \left({q_{ij}^{\underline{k_{ij}}}}\right)^{-1}\leq 
\prod_{i,j} \left( \frac{\-e}{q_{ij}} \right)^{k_{ij}}
=_{q} \prod_{i,j} \left( \frac{\-e}{{q_i^\zeta}} \right)^{k_{ij}}
= \prod_i \left( \frac{\-e}{{q_i^\zeta}} \right)^{k_i}
<_{q} \prod_i  \left({q_i^{\zeta\cdot k_i}}\right)^{-1}
\le p^{\zeta},
$$ where the first inequality holds by~\eqref{eqn-violationbound1} and the last inequality holds by~\eqref{eqn-violationbound2}. Hence, the lemma is immediate.

\subsection*{Proof of~\Cref{lemma-unique-distribution}}

For two states $\+Q_1, \+Q_2 \in \Omega[\Phi, \+P']$, let $M(\+Q_1,\+Q_2)$ be the transition probability of our permutation-wise Glauber dynamics, which can be specified as follows:
\begin{itemize}
    \item If there exists some $P_1,P_2\in \+P$ such that $P_1\neq P_2$ and $\+Q_1(P_1)\neq \+Q_2(P_1),\+Q_1(P_2)\neq \+Q_2(P_2)$, we have $M(\+Q_1,\+Q_2)=0$ since in each step, only on permutation in $\+P$ can be updated.
    \item If there exists a unique $P\in \+P$ such that $\+Q_1(P)\neq \+Q_2(P)$, we should first pick $P$ to sample, which happens with  probability ${1}/{\abs{\+P}}$; then sample $\sigma \sim \Phi'$ where $\Phi'=(\+P'\circ P,\+Q'\circ \+Q(P),\+C)$ and $Q_2(P)=\sigma(P)$. Therefore, $M(\+Q_1,\+Q_2)= {1}/{\abs{\+P}}\cdot \nu_{P}^{\+Q_1} (\+Q_2)$.
    Note that the probability is non-zero since $\-e p_{\Phi'} \Delta_{\Phi'} \leq 1$ for any $\Phi'=\Phi[\+P',\+Q']$ where $\+Q'\in \Omega[\Phi,\+P']$.
    \item If $Q_1=Q_2$, we should ensure that the updated domains are consistent. Therefore, $M(\+Q_1,\+Q_2)=\sum_{P\in \+P} \nu_{P}^{\+Q_1} (\+Q_2)$.
\end{itemize}
As discussed above, one can verify that the idealized permutation-wise Glauber dynamics is irreducible and aperiodic, which implies that $\^P$ has a unique stationary distribution by~\Cref{Theo:Chain}.

Finally, we verify that $\nu$ is the stationary distribution, that is, $\forall \+Q_1,\+Q_2 \in \Omega[\Phi, \+P']$
\begin{align}\label{eqn:detailed-balance}
    \nu(\+Q_1) M(\+Q_1,\+Q_2) = \nu(\+Q_2) M(\+Q_2,\+Q_1).
\end{align}
It is obvious that~\eqref{eqn:detailed-balance} holds when $\+Q_1=\+Q_2$ or $\+Q_1,\+Q_2$ differ in at least two permutation sets among $\+P'$. As for the case where $\+Q_1,\+Q_2$ differ in the unique permutation set $P\in \+P$, 
we have
\begin{align*}
     \nu_{P}^{\+Q_1}(\+Q_2)= \frac{\nu(\+Q_2)}{\Pr[\sigma\sim \mu]{\sigma(\+P\setminus P)=\+Q_1(\+P\setminus P)}}, \quad 
    \nu_{P}^{\+Q_2}(\+Q_1)= \frac{\nu(\+Q_1)}{\Pr[\sigma\sim \mu]{\sigma(\+P\setminus P)=\+Q_1(\+P\setminus P)}}.
\end{align*}
Plugging this into~\eqref{eqn:detailed-balance}, the lemma is immediate.
$$
\nu_{Y \circ P} (X) = \frac{\nu(X)} {\Pr[\sigma \sim \mu]{\sigma[\+P \setminus P] = Y[\+P \setminus P]}}.
$$
Since $X[\+P \setminus P] = Y[\+P \setminus P]$, $\nu(X) \nu_{X \circ P} (Y) = \nu(Y) \nu_{Y \circ P} (X)$, so $\nu(X) \^P(X,Y) = \nu(Y) \^P(Y,X)$. Therefore, the stationary distribution is $\nu$.

\section{Missing proofs in section~\ref{sec:samplealgorithm}}\label{subsec:correctness}


In the remaining section, we complete the proofs of the lemmas in~\Cref{sec:samplealgorithm}.

\subsection*{{{Proof of~\Cref{lemma-simplecase-repeattime}}}}

We first show that, for any PDC formula $\Phi=(\+P,\+Q,\+C)$ satisfying $2\-e p\Delta \leq 1$ with a \pname $P\in \+P$, the violation probability of each constraint in $\+C(P)$ does not deviate significantly in the uniform distribution among all valid assignments of $\Phi$ that satisfy all constraints $\+C\setminus \+C(P)$. 
Specifically, we have the following lemma.

\begin{lemma}\label{lemma-simplecase-constraint-cutdown} 
Given any PDC formula $\Phi=(\+P,\+Q,\+C)$ satisfying $2\-e p\Delta\leq 1$ with a \pname $P\in \+P$,
let $\Phi_1,\dots,\Phi_{K}$ be the factorization of $(\+P,\+Q,\+C\setminus \+C(P))$ where $\Phi_{K}=(\{P\},\{\+Q(P)\},\emptyset)$, and $\sigma_1,\cdots,\sigma_{K}$ be independently drawn from $\mu_{\Phi_1},\cdots,\mu_{\Phi_{K}}$.
For any constraint $C\in \+C(P)$,
\[\Pr[\sigma_{[K]}\sim \mu_{[K]}]{\neg C} \leq 2 \^P[\neg C].
\]
\end{lemma}
\begin{proof}
    Recall that $\^P$ is the uniform distribution of all valid assignments of $\Phi$. 
    We can apply the lopsided LLL to establish an upper bound on the probability of violation for each constraint $C\in \+C(P)$ as follows,
    \begin{align*}
        \Pr[\sigma_{[K]}\sim \mu_{[K]}]{\neg C} &= \^P\left[\neg C \; \bigg \vert \; \bigwedge_{C'\in \+C\setminus\+C(P)} C' \right]\\
        &\leq \^P [\neg C] \cdot \prod_{C'\sim C,C'\in \+C\setminus\+C(P)} (1-\-e\^P[\neg C'])^{-1} \tag{by ~\Cref{prop-PDC-local-uniformity}}\\
        &\leq \^P [\neg C] \cdot (1-\-e p\Delta)^{-1}\leq 2 \^P[\neg C] \tag{by ~\Cref{equality-ab} and $2\-e p\Delta \leq 1$}
    \end{align*}
\end{proof}

Furthermore, the lopsided LLL can also provide a lower bound for the probability of satisfying all constraints in $C\in \+C(P)$. 
\begin{lemma}\label{lem:trivial-ratio-bound}
    Let integer $\Delta\geq 1$. Consider any PDC formula $\Phi=(\+P,\+Q,\+C)$ where $\e p_\Phi \Delta_{\Phi}\leq 1,\Delta_\Phi\leq \Delta$ with a \pname $P\in \+P$ where $\abs{P}\geq 2\e\Delta$,
    let $\Phi_1,\dots,\Phi_{K}$ be the factorization of $(\+P,\+Q,\+C\setminus \+C(P))$ where $\Phi_{K}=(\{P\},\{\+Q(P)\},\emptyset)$, and $\sigma_1,\cdots,\sigma_{K}$ be independently drawn from $\mu_{\Phi_1},\cdots,\mu_{\Phi_{K}}$.
    We have
    \[
        \Pr[\sigma_{[K]}\sim \mu_{[K]}]{\bigwedge_{C\in\+C(P)} C} \geq \exp(-2\-e\Delta).
    \]
\end{lemma}

\begin{proof}
    It suffices to show that for any $\sigma_{[K-1]}$ drawn from $\mu_{[K-1]}$,
    \begin{align*}
        \Pr[\tau \sim \mu_{\Phi'}]{\bigwedge_{C\in\+C(P)} C} \geq \exp(-2\-e\Delta),
    \end{align*} where $\Phi'=\Phi\left[\sigma_{[K-1]}\right]$. Recall that the notation $\Phi\left[\sigma_{[K-1]}\right]$ is defined in \eqref{eq-define-phi-sigma-kminusone}.
    By $\^P_{\Phi'}[\neg C]\leq 1/\abs{P}$ for each constraint $C\in \+C(P)$ and the assumption that $\abs{P}\geq 2\e\Delta$, the lopsided LLL is applicable. Combining with~\Cref{prop-PDC-local-uniformity}, we have
    \begin{align*}
         \Pr[\tau \sim \mu_{\Phi'}]{\bigwedge_{C\in\+C(P)} C}  \geq   \left(1-\frac{\-e}{\abs{P}}\right)^{\abs{P} \Delta} \geq \left(1-\frac{1}{2}\right)^{2\e\Delta} \geq  \exp({-2\-e\Delta}),
    \end{align*}where the third inequality holds by~\Cref{equality-ab}, and the proof is complete.
\end{proof}

We are now ready to prove ~\Cref{lemma-simplecase-repeattime}.
\begin{proof}[Proof of ~\Cref{lemma-simplecase-repeattime}]
    In the case where $p_{\max}\abs{P}\Delta\leq 1/4$, we can directly apply the union bound as follows,
    \begin{align*}
        \Pr[\sigma_{[K]}\sim \mu_{[K]}]{\sigma_{[K]}\in \Omega_{\Phi}}= 1-\Pr[\sigma_{[K]}\sim \mu_{[K]}]{\bigvee_{C\in\+C(P)} \neg C} \geq  1-\sum_{C\in \+C(P)} 2\^P_\Phi[\neg C]\geq 1-2p_{\max}\abs{P}\Delta\geq1/2
    \end{align*} according to ~\Cref{lemma-simplecase-constraint-cutdown}.
    
    Otherwise, we can assume $\abs{P}\geq 2\e\Delta$ by the following argument. If $\abs{P}< 2\e\Delta$, then we have 
    $p_{\max}\abs{P}\Delta \leq 2\-ep\Delta^2\leq \frac{1}{4}$ by the assumption that $8\-ep\Delta^2\leq 1$, which contradicts to $p_{\max}\abs{P}\Delta\geq 1/4$. 
    Therefore, ~\Cref{lem:trivial-ratio-bound} is applicable. {It holds that
    \begin{align*}
        \Pr[\sigma_{[K]}\sim \mu_{[K]}]{\sigma_{[K]}\in \Omega_{\Phi}}=\Pr[\sigma_{[K]}\sim \mu_{[K]}]{\bigwedge_{C\in\+C(P)} C}\geq 2^{-2\-e\Delta}\geq \frac{\varepsilon}{\e^4 n},
    \end{align*} where the second inequality follows from $2\-e\Delta\geq \log(\e^4n/\varepsilon)$. 
    
    Combining all these facts, the proof is immediate.}
\end{proof}

\subsection*{{Proof of~\Cref{lemma-solution-phiprime-phiprimeprime}}}

    By definition, we have 
    \begin{align*}
        \Pr[\sigma\sim \mu_{\Phi'}]{\sigma\in \Omega_{\Phi}}=\^P_{\Phi}\left[\bigwedge_{C\in \+C\setminus \+C^1} C \ \bigg \vert \ \bigwedge_{C'\in \+C^1} C'\right]=1-\^P_{\Phi}\left[\bigvee_{ {C}\in \+C\setminus \+C^1} \neg C \ \bigg \vert \ \bigwedge_{C'\in \+C^1} C'\right].
    \end{align*}  
    It suffices to bound the violation probability of each constraint in $\+C \setminus \+C^1$ conditioned on the satisfaction of all constraints in $\+C^1$.
    Note that $\^P_\Phi[\neg C]={\abs{P}}^{-1}$ for any $C\in \+C^1$, and $\^P_\Phi[\neg C]\leq \frac{1}{\abs{P}\abs{P-1}}$ for any $C\in \+C\setminus \+C^1$. The lopsided LLL is applicable given that $\abs{P}\geq 2\e\Delta$.
    By~\Cref{prop-PDC-local-uniformity}, we have for each $C\in \+C\setminus \+C^1$,
    \begin{align*}
        \^P_{\Phi}\left[\neg C \ \bigg \vert \ \bigwedge_{C'\in \+C^1} C'\right]&\leq  \^P_\Phi[\neg C]\cdot \prod_{C'\in \+C^1: C'\sim C} \left(1- \-e\^P_\Phi[\neg C'] \right)^{-1} \leq \frac{1}{\abs{P}{\abs{P-1}}}\cdot \left(1-\frac{\-e}{\abs{P}}\right)^{-\Delta} \leq \frac{2}{\abs{P}{\abs{P-1}}}
    \end{align*} where the last equality holds by ~\Cref{equality-ab} and $\abs{P}\geq 2\-e\Delta$.
    Combining with union bound and $\abs{ \+C\setminus \+C^1}\leq \abs{P}\Delta$, the proof is immediate.

\subsection*{{{Proof of~\Cref{lemma-concentration-rho}}}}
At last, we prove~\Cref{lemma-concentration-rho}.
According to the ~\Cref{condition-sample-small}, 
$\Phi_1,\cdots,\Phi_{K-1}$ is a collection of small formulas.
Then each solution $\sigma_i$ of $\Phi_i$ has minor impact on the size of $\rho\left(\sigma_{[K-1]}\right) =
\abs{\+C^1\left[\sigma_{[K-1]}\right]}$. This can be captured by the following lemma.

\begin{lemma}\label{prop:Lipschitz-constant}
     Consider any instance satisfying~\Cref{condition-sample-small}.    The function $\rho$ has bounded differences with respect to $\big(\lipschitzconstant,\cdots,\lipschitzconstant\big)$ among all satisfying assignments $\sigma_{[K-1]}$ of $\Phi_{[K-1]}$.
\end{lemma}
\begin{proof}
    Given any pair of satisfying assignments $\sigma_{[K-1]}=\sigma_1,\cdots,\sigma_{K-1}$ and $\sigma'_{[K-1]}=\sigma'_1,\cdots,\sigma'_{K-1}$, we have
    \begin{align}\label{eqn-lipschitz}
        \abs{\rho\left(\sigma_{[K-1]}\right)-\rho\left(\sigma'_{[K-1]}\right)}\leq \sum_{i\in [K-1]} \abs{\rho\left(\sigma_{[K-i]}\circ\sigma'_{[K-i+1,K-1]}\right)-\rho\left(\sigma_{[K-i-1]}\circ\sigma'_{[K-i,K-1]}\right)}
    \end{align} where $\sigma_{[K-i]}\circ\sigma'_{[K-i+1,K-1]}$ is the sequence of $\sigma_1,\cdots,\sigma_{K-i},\sigma'_{K-i+1},\cdots,\sigma'_{K-1}$ for each $i\in [K-1]$. 

    Recall that $\rho(\sigma_{[K-1]})=\abs{\+C^1\left[\sigma_{[K-1]}\right]}$ for any satisifying assignments $\sigma_{[k-1]}$ of $\Phi_{[K-1]}$. The discrepancy between the constraints in $\rho\left(\sigma_{[K-i]}\circ\sigma'_{[K-i+1,K-1]}\right)$ and those in $\rho\left(\sigma_{[K-i-1]}\circ\sigma'_{[K-i,K-1]}\right)$ only comes from the constraints intersecting with the formula $\Phi_i$ in $\+C(P)$, for each $i\in[K-1]$. Consequently, $\abs{\rho\left(\sigma_{[K-i]}\circ\sigma'_{[K-i+1,K-1]}\right)-\rho\left(\sigma_{[K-i-1]}\circ\sigma'_{[K-i,K-1]}\right)}$ can be upper bounded by the number of constraints intersecting with the formula $\Phi_i$ in $\+C(P)$. 
    
    Note that each variable in $\Phi_i$ intersects with at most $\Delta$ constraints. Therefore, the number of constraints intersecting with the formula $\Phi_i$ in $\+C(P)$ does not exceed the multiplication between $\Delta$ and the number of variables in the formulas $\Phi_i$. 
    By \Cref{condition-sample-small}, $\abs{V_{\Phi_i}}$ is at most $kL\Delta \log(n/\varepsilon)$. 
    
    Combining all these facts, we have
    \begin{align}\label{eqn-rho-gap}
       \abs{\rho\left(\sigma_{[K-i]}\circ\sigma'_{[K-i+1,K-1]}\right)-\rho\left(\sigma_{[K-i-1]}\circ\sigma'_{[K-i,K-1]}\right)}\leq \Delta\cdot \left(kL\Delta \log(n/\varepsilon)\right).
    \end{align}
    Plugging ~\eqref{eqn-rho-gap} into ~\eqref{eqn-lipschitz}, the proof is immediate.
\end{proof}

Since $\Phi_1,\cdots,\Phi_{K-1}$ are disjoint, the distribution $\mu_{\Phi_1},\cdots,\mu_{\Phi_{K-1}}$ are mutually independent.
Combining with \Cref{prop:Lipschitz-constant}, we can apply the McDiarmid's inequality to exhibit the concentration of the function $\rho\left(\sigma_{[K-1]}\right)$.
\begin{proof}[Proof of~\Cref{lemma-concentration-rho}]
Note that     
    \begin{align*}
        &\quad \exp\left( -\frac{2\abs{P}^{2}}{144\left(\lipschitzconstant \right)^2\cdot (K-1)} \right)\\ 
        &\leq  \exp\left( -\frac{\abs{P}^{2}}{72\left(\lipschitzconstant \right)^2\cdot k\abs{P}\Delta} \right) \tag{by $K-1\leq k\abs{P}\Delta$}\\
        &=\exp\left( -\frac{\abs{P}}{72 k^3 L^2 \Delta^{5}\log^2 (n/\varepsilon)} \right) \\
        &\leq \exp\left( { \frac{-24\e^3\ln 2\cdot  \Delta^{3}}{\log^2 (n/\varepsilon)} }\right) \tag{by {$\abs{P}\geq 1728\e^3\ln2 \cdot k^3L^2\Delta^8$}} \\
        &\leq {\exp(-6\-e\ln2 \cdot \Delta)} \tag{by $2\-e\Delta\geq \log(\e^4n/\varepsilon)$}.
    \end{align*}
Then the lemma is immediate by \Cref{prop:Lipschitz-constant} and \Cref{thm-McDiarmid-inequlaity}.
\end{proof}

\section{Missing proofs in section~\ref{sec:rapid-mixing}}\label{sec:missingproofs}
\subsection{Missing proofs in \Cref{subsec:rapidmixing}}
\subsection*{{{Proof of~\Cref{lem:distancemetric}}}}

We first claim that for any $\+Q'_1,\+Q'_2\in \+V$, the weighted shortest path distance between $\+Q'_1$ and $\+Q'_2$ is at least $\Dis(\+Q'_1,\+Q'_2)$. Suppose there is a path $(\+Q_1=\+Q'_1,\+Q_2,\cdots,\+Q_{\ell+1}=\+Q'_2)$ between $\+Q'_1$ and $\+Q'_2$ in $\+G$ where $\ell<\Dis(\+Q'_1,\+Q'_2)$. Since $\set{\+Q_i,\+Q_{i+1}}\in \+E$ for each $i\in[\ell]$, we have $\sum_{P'\in \+P'}\abs{\+Q_i(P')\setminus\+Q_{i+1}(P')}=1$, which implies that $\Dis(\+Q_1,\+Q_{i+1})\leq \Dis(\+Q_1,\+Q_{i})+1$. Therefore, $\Dis(\+Q'_1,\+Q'_2)=\Dis(\+Q_1,\+Q_{\ell+1})\leq \ell $, which contradicts our assumption.

To complete the proof, it suffices to construct a path in $\+G$ with length $\Dis(\+Q'_1,\+Q'_2)$. For each $P'\in \+P'$ where $\abs{\+Q'_1(P')\setminus\+Q'_2(P')}\neq \emptyset$, we eliminate the discrepancy between $\+Q'_1$ and $\+Q'_2$ by substituting the elements in $\+Q'_1(P')\setminus\+Q'_2(P')$ to $\+Q'_2(P')\setminus\+Q'_1(P')$ sequentially through the edges in $\+E$. As $\abs{\+Q'_1(P')\setminus\+Q'_2(P')} = \abs{\+Q'_2(P')\setminus\+Q'_1(P')}$, the length of the aforementioned path is $ \sum_{P'\in \+P'}\abs{\+Q'_1(P')\setminus\+Q'_2(P')}=\Dis(\+Q'_1,\+Q'_2)$.

\subsection{Probabilistic properties for \cspformula formulas}

\newcommand{\ppsic}{p'}
\newcommand{\tsize}{1200}
\newcommand{\ratio}{\qgl{r}}
\newcommand{\condcalcu}{{$c_t\in \Lambda\left(C,P\setminus P' \right)$}}
\newcommand{\condcalcuminus}{{$c_t\in \Lambda\left(C,P\setminus P' \right)$}}

In this section, we introduce some useful probabilistic bounds for the analysis of the PDC formulas.

\subsubsection{ \textbf{\emph{Marginal bounds for individual value}}}

The following lemma establishes the upper and lower bounds for the probability that a variable is assigned a specific value in the LLL regime.

\begin{lemma}\label{lem:marginalub} 
Let $p\in (0,1)$ and $\Delta\geq 1$ satisfying $8\-ep\Delta\leq 1$. Consider any \cspformula formula $\Phi = (\+P,\+Q,\+C)$ where $p_\Phi\leq p$, $\Delta_\Phi\leq \Delta$ with a \pname $P\in \+P$. For any $v\in P$ and $c\in \+Q(P)$, let $\+C'\subseteq \+C$ be the set of constraints that the event $v=c$ is non-lopsidependent with the violation of the constraints in $\+C\setminus \+C'$ in the formula $\Phi$. Assume $\^P[\neg C]\leq p'$ for any $C\in \+C'$. It holds that
\begin{align*}
      \frac{1}{\abs{P}} - 4\-ep' \Delta\leq \Pr[\sigma\sim \mu]{\sigma(v)=c} \le  \frac{1}{\abs{P}} \left( 1+4\-e p'\Delta \right). 
\end{align*}
\end{lemma}
\begin{proof}
    By \Cref{prop-PDC-local-uniformity}, we have
    \begin{align}\label{eqn:small-marginal-ub-quantity}
         \Pr[\sigma\sim \mu]{\sigma(v)=c} \leq  \^P[v\gets c]\cdot \left( 1-\-e p' \right)^{-2\Delta} \leq \frac{1}{\abs{P}} \cdot \frac{1}{\left( 1-2\-ep'\Delta \right)}\leq \frac{1}{\abs{P}} \cdot \left( 1+4\-e p'\Delta \right).
    \end{align} 
    where the first inequality follows from that $\abs{\+C'}\leq 2\Delta$, and the second inequality follows from{~\Cref{equality-ab}}.  

     On the other hand, by ~\eqref{eqn:small-marginal-ub-quantity}, we have
    \begin{equation*}
        \begin{aligned}
            \Pr[\sigma \sim \mu]{\tau(v)=c}&= 1-\sum_{c'\in \+Q(P)\setminus \set{c}}  \Pr[\sigma \sim \mu]{\tau(v)=c'} \geq \frac{1}{\abs{P}} - 4\-ep' \Delta.
        \end{aligned}
    \end{equation*}
\end{proof}

{Moreover, if the variable is not assigned a forbidden value, then we can refine the bound as follows.} Recall that the quantity $\Lambda(\cdot)$ was defined in~\eqref{Lambda-def} of~\Cref{subsec:couplingconstruction}.
\begin{lemma}\label{lem:marginallb}
    Given the instance stated in~\Cref{lem:marginalub}, if we further assume that $c\in \+Q(P)\setminus \Lambda(\+C,v)$, we have
    \begin{align*}
        \Pr[\sigma\sim \mu]{\sigma(v)=c} \ge  \frac{1}{\abs{P}}\left(1-\frac{2\Delta}{\abs{P}}\right). 
    \end{align*}
\end{lemma}

\begin{proof}
     
    For any $x\in \supp({\mu_v})$, we use $\Omega_{v=x}$ to denote the set of all satisfying assignment $\sigma\in \Omega$ such that $\sigma(v)=x$, and $\Omega_{v=x}^*$ to denote the set of all valid assignment $\sigma\in \Omega^*$ such that $\sigma(v)=x$. For each $x\in \supp({\mu_v})\setminus \set{c}$, we define the mapping $f_x:\sigma \rightarrow \tau $ from $\Omega_{v=x}$ to $\Omega^*_{v=c}$ satisfying
    \begin{align*}
        \tau(v)=\sigma(u), \tau(u)=\sigma(v), \mbox{and } \tau\left(V\setminus \set{u,v}\right)=\sigma\left(V\setminus \set{u,v}\right)
    \end{align*} where $u$ is the variable such that $\sigma(u)=c$ and $u\in P$. One can verify that the mapping is injective, which implies that
    \begin{align}\label{eqn:sucess-ub}
        \Pr[\sigma\sim \mu]{\sigma\in \Omega_{v=x}, f_x(\sigma)\in \Omega_{v=c}} \leq  \Pr[\sigma\sim \mu]{\sigma\in \Omega_{v=c}}.
    \end{align}
    
    Define the collection of variables $V_x\triangleq \set{u\in V\setminus \set{v} \mid x\in \Lambda(\+C,u)}$, and note that {$\abs{V_x}\leq \Delta$.}
    Given any $\sigma\in \Omega_{v=x}$, if $u\notin V_x$ and $\sigma(u)=c$, we have $f_x(\sigma)\in \Omega_{v=c}$ since $c\in \+Q(P)\setminus \Lambda(\+C,v)$. Let $\+C'\subseteq \+C$ be the set of constraints that the event {$v=x,u=c$} is non-lopsidependent with the violation of the constraints in $\+C\setminus \+C'$.
    Combining all these facts with~\Cref{prop-PDC-local-uniformity}, it holds that for each $x\in \supp({\mu_v})\setminus \set{c}$,
    \begin{equation}\label{eqn:fail-ub}
        \begin{aligned}
             \Pr[\sigma\sim \mu]{\sigma\in \Omega_{v=x}, f_x(\sigma)\notin \Omega_{v=c}}
            &\leq \sum_{u\in V_x} \Pr[\sigma\sim \mu]{\sigma(v)=x, \sigma(u)=c}\\
            &\leq \sum_{u\in V_x} \^P[\sigma(v)=x, \sigma(u)=c]\cdot \left(1-\-ep'\right)^{-4\Delta}\\
            &\leq   \frac{\Delta}{\abs{P}(\abs{P}-1)}\cdot \left(1+8\-e\Delta p'\right),
        \end{aligned}      
    \end{equation}
    where the second inequality holds {since $\abs{\+C'}\leq 4\Delta$, and the last inequality holds by~\Cref{equality-ab}}.
    
    Combining~\eqref{eqn:sucess-ub} and ~\eqref{eqn:fail-ub}, we have
    \begin{align*}
        &\quad\sum_{x\in \supp(\mu_v)\setminus \set{c}} \Pr[\sigma\sim \mu]{\sigma\in \Omega_{v=x}}\\
        &= \sum_{x\in \supp(\mu_v)\setminus \set{c}} \bigg(\Pr[\sigma\sim \mu]{\sigma\in \Omega_{v=x}, f_x(\sigma)\in \Omega_{v=c}} +\Pr[\sigma\sim \mu]{\sigma\in \Omega_{v=x}, f_x(\sigma)\notin \Omega_{v=c}} \bigg)\\
        &\leq \sum_{x\in \supp(\mu_v)\setminus\set{c}} \Pr[\sigma\sim \mu]{\sigma\in \Omega_{v=c}} + \sum_{x\in \supp(\mu_v)\setminus \set{c}} \Delta\cdot \left(1+8\-e\Delta p'\right)/\left(\abs{P}(\abs{P}-1)\right) \\
        &\leq (\abs{P}-1)\cdot  \Pr[\sigma\sim \mu]{\sigma\in \Omega_{v=c}} + {\Delta}\cdot \left(1+8\-e\Delta p'\right)/\abs{P},
    \end{align*} which implies that
    \begin{align*}
          \Pr[\sigma\sim \mu]{\sigma(v)=c}=\Pr[\sigma\sim \mu]{\sigma\in \Omega_{v=c}} \geq \frac{1}{\abs{P}}-\frac{\Delta}{\abs{P}^2} \left(1+8\-e\Delta p'\right) \geq \frac{1}{\abs{P}}\left(1-\frac{2\Delta}{\abs{P}}\right).
    \end{align*} 
\end{proof}

\subsubsection{\textbf{\emph{Marginal bounds for multiple values}}}

In this section, we establish bounds for the probability that a collection of variables is assigned multiple values in the LLL regime.

\begin{lemma}\label{lem:small-marginal-lb} 
    Let $p\in (0,1)$, $\Delta\geq 1$.
    Consider any \cspformula formula $\Phi = (\+P,\+Q,\+C)$ where $p_{\Phi}\leq p$, $\Delta_\Phi\leq \Delta$ with a \pname $P\in \+P$, a decomposition $\+P'$ of $\+P$, a \pname $P'=\set{v_1,\cdots,v_t}\in \+P'[P]$, and a domain $Q'=\set{c_1,\cdots,c_t}\subseteq \+Q(P)$. Assume that for any decomposition $(\+P',\+Q')$ of $(\+P,\+Q)$, we have $p_{\Phi'}\leq p$ where $\Phi'= (\+P',\+Q',\+C)$.
    If {$8\-e {p} \Delta \leq 1 $, $\abs{\abs{P_1} - \abs{P_2}} = O(1)$,
    for any $P_1,P_2\in \+P[P]$,} and $\abs{\+P'[P]}\geq 2$. 
    We have
    \begin{align}\label{eqn:withoutcondition} 
       \Pr[\sigma\sim \mu]{\sigma(P') = Q' }>_{q} \left(\frac{ 1-4\-ep'\Delta  \left({\abs{P}-\abs{P'}+1}\right) }{1+4\-ep'\Delta}\right)^t  \cdot \binom{\abs{P}}{\abs{P'}}^{-1},
    \end{align} 
    and 
    \begin{align}\label{eqn:withoutcondition2}
        \Pr[\sigma\sim \mu]{\sigma(P') = Q' }<_{q} \left(\frac{ 1-4\-ep'\Delta  \left({\abs{P}-\abs{P'}+1}\right) }{1+4\-ep'\Delta}\right)^{-t}  \cdot \binom{\abs{P}}{\abs{P'}}^{-1},
    \end{align}
    where {$p'={\abs{P'}}/{\left(\abs{P}-\abs{P'}+1\right)} \cdot {p}$}. Moreover, we also have
    \begin{align}\label{eqn:withcondition}
       \Pr[\sigma\sim \mu]{\sigma(P') = Q' \mid c_1\in \sigma(P')}>_{q}    \left(\frac{ 1-4\-ep'\Delta  \left({\abs{P}-\abs{P'}+1}\right) }{1+4\-ep'\Delta}\right)^{t-1}\cdot \binom{\abs{P}-1}{\abs{P'}-1}^{-1}.
    \end{align} 
\end{lemma}

To prove \Cref{lem:small-marginal-lb}, we first introduce more notations.
For each $j\in [t]$, let $\Phi_j=(\+P_j,\+Q_j,\+C )$ be the formula where
$${\+P_j=\left(\+P\setminus{\set{P}}\right) \cup \left(\set{P'\setminus \set{v_j}, P\setminus P'\cup \set{v_j}}\right)},$$
$${\+Q_j=\left(\+Q\setminus \set{\+Q(P)}\right) \cup \left(\set{Q'\setminus \set{c_t}, \+Q(P)\setminus Q' \cup \set{c_t} } \right)}.$$
The violation probability for the constraints in these formulas can be upper bounded as follows.
\begin{proposition}\label{lem:violation-prop-bound}
    For each $j\in [t]$ and $C\in \+C$, if $v_j\in \vbl(C)$ or $\left(P\setminus P'\right)\cap \vbl(C)\neq \emptyset$, we have
    \begin{equation*}
        \^P_{\Phi_j}[\neg C] \leq {\abs{P'}}/{\left(\abs{P}-\abs{P'}+1\right)} \cdot {p}.
    \end{equation*}
\end{proposition}

We are now ready to complete the proof of~\Cref{lem:small-marginal-lb}.
\begin{proof}[Proof of \Cref{lem:small-marginal-lb}]
    
    Let $r=\left( 1-4\-ep'\Delta  \left({\abs{P}-\abs{P'}+1}\right) \right) / \left(1+4\-ep'\Delta\right)$ in the following discussion. For each $j\in [t]$, define
    $$B_j \triangleq \set{\sigma(v_j)=c_t\mid  \sigma\left(P'\setminus\set{v_j}\right)=Q'\setminus \set{c_t}}, \ \ A_j=\set{\sigma\left(P'\setminus\set{v_j}\right)=Q'\setminus \set{c_t}},$$
    where $\sigma\sim \mu_\Phi$. 
     By definition, we have
     \begin{equation}\label{eqn:decomposed-small-probability}
        \begin{aligned}
            \Pr[\sigma\sim \mu]{\sigma(P') = Q' } = \sum_{j\in  [t]} \Pr[\sigma\sim \mu]{A_j} \cdot \Pr[\sigma \sim \mu]{B_j}=\sum_{j\in  [t]} \Pr[\sigma\sim \mu]{A_j} \cdot \Pr[\tau\sim \mu_{\Phi_j}]{\tau(v_j)=c_t}.
        \end{aligned}
     \end{equation}
    
    Consider the domain $Q''=\set{c_1,c_2,\cdots,c'_{t}}$ where $c'_t\neq c_t$ and $c'_t\in \+Q(P)$.
    We first claim that
    \begin{align}\label{eqn:small-lb-quantity-distance1}
        r \Pr[\sigma\sim \mu]{\sigma(P') = Q' }\leq  \Pr[\sigma\sim \mu]{\sigma(P') = Q''}.
    \end{align} 
    For each $j\in[t]$, let $\+C'\subseteq \+C$ be the set of constraints that the event $v_j=c_t$ is non-lopsidependent with the violation of the constraints in $\+C\setminus \+C'$ in the formula $\Phi_j$. By definition, for any $C\in \+C'$, either $v_j\in \vbl(C)$ or $\left(P\setminus P'\right)\cap \vbl(C)\neq \emptyset$.
    According to ~\Cref{lem:violation-prop-bound}, we have {$\^P_{\Phi_j}[\neg C] <_{q} p' <_{q} {p}$} for any constraint $C\in \+C'$ as {$\abs{\+P'[P]}\ge 2$}. On the other hand, one can verify that $\^P_{\Phi_j}[\neg C] <_{q} {{p}}$ for any constraint $C\in \+C\setminus \+C'$.     
    Combining with~\Cref{lem:marginalub} and {$8\-e{p}\Delta\leq 1$}, we have, for any $j\in [t]$,
    \begin{align*}
        \Pr[\tau\sim \mu_{\Phi_j}]{\tau(v_j)=c_t}<_{q} \frac{1}{{\abs{P}-\abs{P'}+1}} \cdot \left(1+4\-ep' \Delta\right).
    \end{align*}
    By similar arguments, we have
    \begin{align*}
        \Pr[\tau\sim \mu_{\Phi_j}]{\tau(v_j)=c'_t} >_{q} \frac{1}{{\abs{P}-\abs{P'}+1}} - 4\-e\ppsic \Delta.
    \end{align*}
    Plugging these into \eqref{eqn:decomposed-small-probability}, it holds that
    \begin{align*}
         &\quad r\Pr[\sigma\sim \mu]{\sigma(P') = Q' }
         = r \Pr[\tau\sim \mu_{\Phi_j}]{\tau(v_j)=c_t} \cdot \sum_{j\in  [t]} \Pr[\sigma\sim \mu]{A_j} <_{q} \Pr[\tau\sim \mu_{\Phi_j}]{\tau(v_j)=c'_t} \cdot   \sum_{j\in  [t]} \Pr[\sigma\sim \mu]{A_j},
    \end{align*} 
    which verifies the inequality~\eqref{eqn:small-lb-quantity-distance1}.

    We then consider any domains $Q''$ such that  $\abs{Q'}=\abs{Q''}>1$. We construct a sequence of domains $\left(Q'_1, Q'_2,\cdots, Q'_m\right)$ where $Q'_1=Q',Q'_m=Q'',m=\abs{Q'\setminus Q''}$ and $\abs{Q'_{j}\setminus Q'_{j+1}}=1$ for each $j\in [m-1]$. Combined with the fact that $m\leq t$ and ~\eqref{eqn:small-lb-quantity-distance1}, we have
    \begin{align*}
        r^{t}\Pr[\sigma\sim \mu]{\sigma(P') = Q'}<_{q}  \Pr[\sigma\sim \mu]{\sigma(P') = Q''}<_{q} r^{-t}\Pr[\sigma\sim \mu]{\sigma(P') = Q'},
    \end{align*} 
    which implies the inequalities~\eqref{eqn:withoutcondition} and~\eqref{eqn:withoutcondition2}.
    The correctness of ~\eqref{eqn:withcondition} follows from similar arguments.

\end{proof}

Again, we have an alternative bound for the case where the values are all non-forbidden.

\begin{lemma}\label{lem:marginal-lb} 
    Given the instance stated in~\Cref{lem:small-marginal-lb}, if we further assume that $Q\subseteq \+Q(P)\setminus\Lambda(\+C,P')$ and $\abs{Q}=\abs{P'}$, we have
    \begin{align*} 
       \Pr[\sigma\sim \mu]{\sigma(P') = Q }\geq (\abs{P'}!)\cdot {\prod_{j=0}^{\abs{P'}-1}
       \left(\frac{1}{\abs{P}-j}\left(1-\frac{2\Delta}{\abs{P}-j}\right)\right).}
    \end{align*} 
    Moreover, we also have
    \begin{align*}
          \Pr[\sigma\sim \mu]{\sigma(P') = Q' \mid c_1\in \sigma(P')}\geq  (\abs{P'}-1)!\cdot {\prod_{j=1}^{\abs{P'}-1}
       \left(\frac{1}{\abs{P}-j}\left(1-\frac{2\Delta}{\abs{P}-j}\right)\right).}
    \end{align*}
\end{lemma}

\begin{proof}
    Let $P'=\set{v_1,v_2,\cdots,v_t}$ and $Q=\set{c_1,c_2,\cdots,c_{t}}$. We claim that, for any {permutation $(c'_1,c'_2,\cdots,c'_t)$ of $Q$,} we have
    \begin{align}\label{eqn:mariginal-lb-phi1}
      {\Pr[\sigma\sim \mu]{\sigma(v_j)=c'_j \text{ for all } j\in [t]} }\geq   \prod_{j=0}^{t-1} \frac{\abs{P}-j - 2\Delta}{(\abs{P}-j )^2}.
    \end{align}
    Combining with the fact that the number of the permutations of a given domain of $P'$ equals $\abs{P'}!$,
    the lemma holds immediately from ~\eqref{eqn:mariginal-lb-phi1}.

    In the following, we prove the inequality ~\eqref{eqn:mariginal-lb-phi1}. For each $j\in \set{0,1,\cdots,t-1}$, let $\sigma_j$ be the {partial assignment} where $\sigma_j(v_{j'})=c'_{j'}$ for any $j'\in [j]$, and $\Psi_j=(\+P_j,\+Q_j,\+C_j)=\Phi^{\sigma_j}$ be the formula conditioned on the partial assignment $\sigma_j$. In particular, $\sigma_0$ is {an empty assignment}, $\+P_0=\+P$, $\+Q_0=\+Q$, and $\+C_0=\+C$. It holds that
    \begin{align}\label{eqn-lb-marginal-chainrule}
        {\Pr[\sigma\sim \mu]{\sigma(v_j)=c'_j \text{ for all } j\in [t]} }=\prod_{j=0}^{t-1} \Pr[\tau\sim \mu_{\Psi_j}]{\tau({v_{j+1})}=c'_{j+1}}.
    \end{align}

    For each $j\in \set{0,1,\cdots,t-1}$, let $P_j$ be the permutation containing $v_{j+1}$ in $\Psi_j$. By definition, we have $P_j\subseteq P$ satisfying $\abs{P_j}=\abs{P}-j$, and $c'_{j+1}\in \+Q_j(P_j)\setminus \Lambda(\+C_j,P_j)$ since $Q\subseteq \+Q(P)\setminus\Lambda(\+C,P')$ by assumption. Furthermore, the constraints in $\+C(P\setminus P_j)$ are satisfied. For the constraints $\+C\setminus \+C(P\setminus P_j)$, the violation probability is at most $p$ by the assumption that $\-ep_{\Psi}\Delta\leq 1$ for any $\Psi=(\+P',\+Q',\+C)$ where $(\+P',\+Q')$ is the decomposition of $(\+P,\+Q)$. Therefore, the formula $\Psi_j$ satisfies $8\-ep_{\Psi_j}\Delta\leq 1$. Combined with ~\Cref{lem:marginallb}, we have
    \begin{align}\label{eqn-lb-marginal-variable} 
         \Pr[\tau\sim \mu_{\Psi_j}]{\tau({v_{j+1})}=c_{j+1}} \geq (\abs{P}-j - 2\Delta)/(\abs{P}-j )^2.
    \end{align}
    Plugging \eqref{eqn-lb-marginal-variable} into \eqref{eqn-lb-marginal-chainrule}, the inequality~\eqref{eqn:mariginal-lb-phi1} holds.
\end{proof}

\subsection{Missing proofs in \Cref{subsec:couplingconstruction}}\label{sec:miss-proofs}

\subsection*{{{Proof of~\Cref{lemma-trivial-coupling}}}}

It suffices to show that 
\begin{align}\label{eq-upperbound-e1}
\Pr{Q_i\cap \Lambda(\+C_t,P_i)\neq \emptyset} \leq \frac{2k\Delta\abs{P_i}^2}{\abs{P}}.
\end{align}

For each variable $v\in P_i$,
there are at most $\Delta$ constraints $C\in \+C_t$ such that $v\in \vbl(C)$, which implies that $\abs{\+C_t(P_i)}\leq \Delta\abs{P_i}$.
Combining this with that $C$ is disjunctive and $\abs{\vbl(C)} \leq k$ for each $C\in \+C_t$,
we have $\abs{\Lambda(\+C_t,P_i)} \leq k\Delta\abs{P_i}$ by the definition of $\Lambda(\+C_t,P_i)$.
Moreover, for each $c\in \Lambda(\+C_t,P_i)$, we have $\Pr{c\in Q_i} \leq 2 \abs{P_i}/\abs{P}$
by~\Cref{lem:marginalub} and union bound.
Then \eqref{eq-upperbound-e1} follows immediately by union bound, and the proof is complete.

\subsection*{{{Proof of \Cref{lemma-cpstep-failprob}}}}

Let $\alpha' \triangleq \alpha / \binom{\abs{S}}{\abs{P_i}-\abs{T_1}}$. One can verify that $\alpha'$ is the probability that $Q_i^1$ (or $Q_i^2$) is assigned the value $T\cup T_1$ (or $T\cup T_2$) in \Cref{line-set-both-T} for any $T\subseteq S$ such that $\abs{T} = \abs{P_i}-\abs{T_1}$.
To show that $(Q^1_i,Q^2_i)$ is a coupling of $\nu^1_i$ and $\nu^2_i$, it is sufficient to claim that $\nu_i^1(Q^1_i = T\cup T_1) \ge \alpha'$ holds, and the same holds for $\nu^2_i$ as well. This follows from~\Cref{lem:marginal-lb} when $P_i \notin \psmall$, and~\Cref{lem:small-marginal-lb} when $P_i \in \psmall$ by~\Cref{condition-couple-first}.

We complete the proof by establishing the lower bounds for $\alpha$ since $\Pr{\Succ = \False} = 1-\alpha$.

\smallskip

\paragraph{\textbf{Case I: $P_i \in \psmall$.}} By definition, we have
\begin{align*}
\alpha 
&= \left( (1-4q\-e p'\Delta)/(1 + 4e p'\Delta ) \right)^{\left(\abs{P_i}-\abs{T_1}\right)}\\
&\ge \left( (1-4q\-e p'\Delta )  (1 - 8e p'\Delta ) \right)^{\left(\abs{P_i}-\abs{T_1}\right)}\tag{by {$p'<_{q} p^{\theta}$ as $\abs{\pdecomex[P]}\ge 2$} and $8\-ep^{\theta}\Delta\leq 1$}\\
&\ge (1- 5q\-e p'\Delta)^{\left(\abs{P_i}-\abs{T_1}\right)}>_{q} 1- 5\-e p^{\theta}\Delta \abs{P_i}^2.
\end{align*}
\smallskip

\paragraph{\textbf{Case II: $P_i \notin \psmall$.}}
We first assume that $\ell\geq 3\Delta$.
Let $\alpha_1 = \binom{\abs{S}}{\abs{P_i}-\abs{T_1}}\cdot (\abs{P_i}-\abs{T_1})!\cdot \prod_{j=\abs{T_1}}^{\abs{P_i}-1} 1/(\abs{P}-j )$ and $\alpha_2 = \prod_{j=\abs{T_1}}^{\abs{P_i}-1}(\abs{P}-j - 2\Delta)/(\abs{P}-j )$. Note that $\abs{S}\geq \abs{P_i}$ holds by $\abs{S} \ge \abs{P} -2\Delta \abs{P_i}$ and {$\ell \ge 3\Delta$.} We have
\begin{align*}
\alpha_1
&= \frac{\abs{S} (\abs{S}-1) \cdots (\abs{S}-\abs{P_i}+\abs{T_1}+1)} {(\abs{P}-\abs{T_1}) (\abs{P}-\abs{T_1}-1) \cdots (\abs{P}-\abs{P_i}+\abs{T_1}+1)}\\
&= \frac{(\abs{P}-\abs{P_i}) (\abs{P}-\abs{P_i}-1) \cdots (\abs{S}-\abs{P_i}+\abs{T_1}+1)} {\abs{P} (\abs{P}-1) \cdots (\abs{S}+1)}\\
&\ge \left( 1- \frac{\abs{P_i}}{\abs{S}} \right)^{\abs{P} -\abs{S}}\ge  \left( 1- \frac{\abs{P_i}}{\abs{P} -2\Delta \abs{P_i}} \right)^{2\Delta \abs{P_i}} \tag{by $\abs{S} \ge \abs{P} -2\Delta \abs{P_i}$}\\
&=\left( 1- \frac{1}{\ell -2\Delta} \right)^{2\Delta \abs{P_i}}\ge 1-\frac{2\Delta \abs{P_i}}{\ell -2\Delta}
\ge 1-\frac{3\Delta \abs{P_i}}{\ell}.
\end{align*} 
As for $\alpha_2$, we have
\begin{align*}
\alpha_2
&= \prod_{j=0}^{\abs{P_i}-1} \left( 1- \frac{2\Delta}{\abs{P}-j} \right)
\ge \left( 1- \frac{2\Delta}{\abs{P}-\abs{P_i}} \right)^{\abs{P_i}}\ge 1- \frac{2\Delta \abs{P_i}}{\abs{P}-\abs{P_i}}
= 1- \frac{2\Delta}{\ell-1}.
\end{align*}
Combining all these facts, it holds that
$$
\alpha = \alpha_1 \alpha_2
\ge 1-\frac{3\Delta \abs{P_i}}{\ell} -\frac{2\Delta}{\ell-1}
\ge 1-\frac{4\Delta \abs{P_i}}{\ell}.
$$
For the case when $\ell<3\Delta$, we have $4\Delta\abs{P_i}/\ell>1$, and the bound holds trivially.

Finally, note that $\Succ = \True$ iff $r\leq \alpha$ and $\Succ$ remains $\True$ in \Cref{line:adjust-prob}. If $P_i \notin \psmall$ and $\ell \ge {4\Delta \abs{P_i}}$, we have ${(\ell-4\Delta \abs{P_i})}/{\left(\ell \alpha\right)}\in [0,1]$ by $\ell \ge {4\Delta \abs{P_i}}$.
Consequently,
\begin{align*}
    \Pr{\Succ=\True} &= \Pr{r\leq \alpha} \cdot \Pr{\Succ=\True \mid r\leq \alpha}= \alpha \cdot \frac{\ell-4\Delta \abs{P_i}}{\ell \alpha} = \frac{\ell-4\Delta \abs{P_i}}{\ell},
\end{align*}
which implies $\Pr{\Succ = \False} = 4\Delta\abs{P_i}/\ell$.


\subsection*{{{Proof of~\Cref{lemma-suc-probability-alg-initial-coupling}}}}

Let $\mathsf{Bad}^j_i(Q)$ denote the event $Q \cap \Lambda(\+C_1,P_i)\neq \emptyset$ for each $i\in [\ell]$ and $j\in [2]$.
We will use $\mathsf{Bad}(Q)$ to denote $\mathsf{Bad}^j_i(Q)$ if $i$ and $j$ are clear from the context.

By \Cref{line-init-set-s} in \Cref{alg-initial-coupling},
we have 
\begin{align}\label{eq-decom-exp-s}
\E{\sum_{P'\in S}\abs{P'}} \leq \sum_{r\in [\ell]}\Pr{\mathsf{Bad}(Q^1_r)} \abs{P_r} + \sum_{r\in [\ell]}\Pr{\mathsf{Bad}(Q^2_r)} \abs{P_r}.
\end{align}
In addition, we have
\begin{align*}
\sum_{r\in [\ell]}\Pr{\mathsf{Bad}(Q^1_r)} \abs{P_r}
\leq \Pr{c_1 \in \Lambda(\+C_1,P_i)}\abs{P_i} + \Pr{\mathsf{Bad}^1_i(Q^1_i\setminus c_1)}\abs{P_i} + \sum_{r\in [\ell]\setminus \{i\}}\Pr{\mathsf{Bad}(Q^1_r)}\abs{P_r}.
\end{align*}
We claim that 
\begin{align}
\Pr{c_1 \in \Lambda(\+C_1,P_i)}\abs{P_i} &<_{q} 2\Delta,\label{eq-lemma-initial-c1}\\
\Pr{\mathsf{Bad}^1_i(Q^1_i\setminus c_1)}\abs{P_i} &<_{q}  2k\Delta,\label{eq-lemma-badi-c1}\\
\sum_{r\in [\ell]\setminus \{i\}}\Pr{\mathsf{Bad}(Q^1_r)}\abs{P_r} &<_{q} 16k\Delta\label{eq-lemma-badt}
\end{align}
Thus, we have 
\[\sum_{r\in [\ell]}\Pr{\mathsf{Bad}(Q^1_r)} \abs{P_r} <_{q} (18k + 2)\Delta.\]
Similarly, we also have 
\[\sum_{r\in [\ell]}\Pr{\mathsf{Bad}(Q^2_r)} \abs{P_r} <_{q} (18k + 2)\Delta.\]
Combining the above two inequalities with \eqref{eq-decom-exp-s}, the lemma is immediate.
In the following, we prove \eqref{eq-lemma-initial-c1}, \eqref{eq-lemma-badi-c1}, \eqref{eq-lemma-badt} one by one.

At first, we prove \eqref{eq-lemma-initial-c1}.
We have 
\begin{equation}
\begin{aligned}\label{eq-large-c1inpi}
\Pr{c_1 \in \Lambda(\+C_1,P_i)}\abs{P_i} &\leq \abs{P_i}\Delta\max_{r\in [\ell]}f_1(r)\\
  (\text{by \Cref{lemma-lowerbound-f}})\quad &\leq \abs{P_i}\Delta(1+4\-ep\Delta)\abs{P_r}/\abs{P} 
   \\(\text{by $8\-e\Delta p^{2\theta}\leq 1$ in \Cref{condition-couple-first}})\quad &= 2\Delta\abs{P}^{\theta-1}\notag   \\
   (\text{by $\theta \in (0,1/2]$})\quad &\leq 2\Delta.
\end{aligned}
\end{equation}
\vspace{0.2cm}

In the next, we prove \eqref{eq-lemma-badt} for two different cases.

\noindent{\textbf{Case I: }$t>\sizeuperboundlp$.}
Let $\+A$ denote the event $(c_1\not\in \Lambda(\+C_1,P_i\cup P_j))\land (c_2\not\in \Lambda(\+C_2,P_i\cup P_j))$.
By \eqref{eq-definition-f-lower-bound}, 
$\sum_{r\in [\ell]}\abs{P_r} = \abs{P}$ and $\sum_{r\in [\ell]}\abs{V_r} \leq 2\Delta$, 
we have 
\begin{equation}\label{eq-pro-ieqj-small}
\begin{aligned}
\Pr{i=j}&\geq \sum_{r\in [\ell]}f_{\mathsf{min}}(r) \geq  \sum_{r\in [\ell]}\left(\abs{P_r}-\abs{V_r}\right)\left(\abs{P}-2\Delta\right)\abs{P}^{-2}
= 1- 4\Delta\abs{P}^{-1}.
\end{aligned}
\end{equation}
If $i = j$, by \Cref{lemma-cpstep-failprob} we have 
\begin{align*}
    \Pr{\Succ = \False\mid \+A \land i = j} <_{q} 4\Delta \abs{P}^{2\theta - 1}.
\end{align*}
In addition, by \eqref{eq-large-c1inpi} we have 
\begin{align}
\Pr{\overline{\+A}} \leq 8\Delta\abs{P}^{\theta-1}.
\end{align}
Combining above two inequalities, we have
\begin{align}\label{eq-pro-succ-true}
    \Pr{\Succ = \False\land i = j} 
    \leq \Pr{\overline{\+A}} + \Pr{\Succ = \False\land \+A \land i = j} 
    <_{q} 12\Delta \abs{P}^{2\theta - 1}.
\end{align}
For each $r\in [\ell]\setminus \{i\}$, by \Cref{lemma-trivial-coupling} we have
\begin{align}\label{eq-case-ieqj-succ-false}
\Pr{\mathsf{Bad}(Q^1_r)\mid i= j\land \Succ = \False} <_{q}  \min\left\{\frac{2k\Delta\abs{P_r}^2}{\abs{P}-\abs{P_i}},1\right\}=_{q} \min\left\{2k\Delta\abs{P}^{2\theta - 1},1\right\}.
\end{align}
Combining \eqref{eq-case-ineqj-qt1setminusc1} with \eqref{eq-case-ieqj-succ-false}, we have 
for each $r\in [\ell]\setminus \{i\}$ we have
\begin{align}\label{eq-prqonetcaplambda-not-empty}
\Pr{\mathsf{Bad}(Q^1_r)\mid (i= j\land \Succ = \False)\lor (i\neq j) } <_{q}  \min\left\{2k\Delta\abs{P}^{2\theta - 1},1\right\}.
\end{align}
In addition, by \eqref{eq-pro-ieqj-small} and \eqref{eq-pro-succ-true}, we have
\begin{align}\label{eq-prob-jneqj-succfalse}
\Pr{(i= j\land \Succ = \False)\lor (i\neq j) } <_{q} 16\Delta \abs{P}^{2\theta - 1}.
\end{align}
Combining \eqref{eq-prqonetcaplambda-not-empty} with \eqref{eq-prob-jneqj-succfalse}, we have
\begin{align*}
\sum_{r\in [\ell]\setminus \{i\}}\Pr{\mathsf{Bad}(Q^1_r)}\abs{P_r} &<_{q} 16\Delta \abs{P}^{2\theta - 1}\min\left\{2k\Delta\abs{P}^{2\theta - 1},1\right\}\abs{P}\\
&\leq 16\Delta\abs{P}^{2\theta}\min\left\{2k\Delta\abs{P}^{2\theta - 1},1\right\}\\
&\leq 16k\Delta\tag{by \ref{eq-lem-init-4} in \Cref{condition-initial}}.
\end{align*}

\noindent{\textbf{Case II: }$t\leq r$.}
By \eqref{eq-definition-f-lower-bound}, we have 
\begin{align}\label{eq-pro-ieqj}
\Pr{i=j}\geq \sum_{r\in [\ell]}f_{\mathsf{min}}(r) = \sum_{r\in [\ell]}\left((1+4\-e\Delta p)\abs{P_r}/\abs{P} - 4\-e\Delta p \abs{P_r}\right) = 1 - 4\-e
p\Delta(\abs{P} - 1).
\end{align}
If $i = j$, by \Cref{lemma-cpstep-failprob} we have 
\begin{align}\label{eq-pro-succ-true-small}
    &\quad\Pr{\Succ = \False\mid i = j}<_{q}5\-e\Delta p^{\theta}\abs{P}^{2\theta}.
\end{align}
Combined with \eqref{eq-pro-ieqj}, we have
\begin{align}\label{eq-prob-jneqj-succfalse-small}
\Pr{(i= j\land \Succ = \False)\lor (i\neq j) } <_{q} 4\-e
p\Delta\abs{P} + 5\-e\Delta p^{\theta}\abs{P}^{2\theta}.
\end{align}
Combined \eqref{eq-prqonetcaplambda-not-empty} and \ref{eq-lem-init-3} in \Cref{condition-initial}, we have
\begin{align*}
\sum_{r\in [\ell]\setminus \{i\}}\Pr{\mathsf{Bad}(Q^1_r)}\abs{P_r} &<_{q} 5\-e\Delta (p\abs{P} + p^{\theta}\abs{P}^{2\theta})\min\left\{2k\Delta\abs{P}^{2\theta - 1},1\right\}\abs{P}\notag\leq 10k\Delta.
\end{align*}
\vspace{0.5cm}

At last, we prove \eqref{eq-lemma-badi-c1}.
Note that
\begin{align}\label{eq-sum-badi}
\Pr{\mathsf{Bad}^1_i(Q^1_i\setminus c_1)} \leq  \Pr{\mathsf{Bad}^1_i(Q^1_i\setminus c_1)\mid i \neq j}+ 
\Pr{\mathsf{Bad}^1_i(Q^1_i\setminus c_1)\mid i = j}.
\end{align}
If $t>\sizeuperboundlp$, 
by \Cref{lemma-cpstep-failprob} we have
\begin{align*}
\Pr{\mathsf{Bad}^1_i(Q^1_i\setminus c_1) \mid i = j} <_{q} \min\left\{2k\Delta\abs{P}^{2\theta - 1},1\right\}.
\end{align*}
Combined with \eqref{eq-sum-badi} and \eqref{eq-pro-ieqj-small}, we have
\begin{align*}
\Pr{\mathsf{Bad}^1_i(Q^1_i\setminus c_1)} <_{q}  4\Delta\abs{P}^{-1} + 
\min\left\{2k\Delta\abs{P}^{2\theta - 1},1\right\}.
\end{align*}
Therefore, 
\begin{align*}
\Pr{\mathsf{Bad}^1_i(Q^1_i\setminus c_1)}\abs{P_i} &<_{q} 4\Delta\abs{P}^{\theta-1}+\abs{P}^{\theta}\min\left\{2k\Delta\abs{P}^{2\theta - 1},1\right\}\notag\\
&<_{q} \Delta + k\abs{P}^{2\theta}\min\left\{2\Delta\abs{P}^{2\theta - 1},1\right\}\notag\\
&\leq 2k\Delta\tag{by \ref{eq-lem-init-4} in \Cref{condition-initial}}.
\end{align*}
If $t\leq \sizeuperboundlp$, by \Cref{lemma-cpstep-failprob} we have
\begin{align}\label{eq-case-ineqj-qt1setminusc1}
\Pr{\mathsf{Bad}^1_i(Q^1_i\setminus c_1) \mid i = j} <_{q}  \min\left\{5\-e \Delta p^{\theta} \abs{P}^{2\theta},1\right\}.
\end{align}
Combined with \eqref{eq-sum-badi} and \eqref{eq-pro-ieqj}, we have
\begin{align}\label{eq-case-ieqj-qt1setminusc1}
\Pr{\mathsf{Bad}^1_i(Q^1_i\setminus c_1)} <_{q} 4\-e
p\Delta + \min\left\{5\-e \Delta p^{\theta} \abs{P}^{2\theta},1\right\}.
\end{align}
Therefore,
\begin{align*}
\Pr{\mathsf{Bad}^1_i(Q^1_i\setminus c_1)}\abs{P_i} <_{q} 4\-e
p\Delta\abs{P}^{\theta} + \abs{P}^{\theta}\min\left\{5\-e \Delta p^{\theta} \abs{P}^{2\theta},1\right\}.
\end{align*}
Moreover, by \ref{eq-lem-init-3} in \Cref{condition-initial}, we have 
\[4\-e p\abs{P}^{\theta} <_{q} \min\{\-e p\abs{P}^2, \-e p\abs{P}^{\theta +1}\} \leq \-e p\abs{P}^2 \min\left\{2\Delta\abs{P}^{2\theta - 1},1\right\}\leq 1,\]
\[ \abs{P}^{\theta}\min\left\{5\-e p^{\theta} \abs{P}^{2\theta},1\right\} \leq 5\-e  p^{\theta} \abs{P}^{3\theta}\leq 1.\]
Combining above three inequalities, 
\begin{align*}
\Pr{\mathsf{Bad}^1_i(Q^1_i\setminus c_1)}\abs{P_i} <_{q} 2\Delta.
\end{align*}
In summary, \eqref{eq-lemma-badi-c1} is proved and the lemma is immediate.

\subsection{Missing proofs in \Cref{subsec:discrepancy-analysis}}

\subsubsection{\textbf{\emph{Witness Percolation}}}\label{subsub:witness-percolation}
{In this section, we complete the proof of~\Cref{lemma-witness-sequence-pinpdiff}}.

\begin{lemma}\label{lemma-compisition-pd}
For each $t>0$ and $P\in \pdiff^{t}$, one of the following conditions holds:
\begin{itemize}
\item $P\in \pdisc^t$;
\item $P\in \psmall$, and there exist $P'\in \pdisc^{t}$ and $C\in \+C$ such that $C\in \+C(P)\cap \+C(P')$;
\item $P\in \psmall$, and there exists $C\in \cdisc^t\cap \+C(P)$.
\end{itemize}
\end{lemma}
\begin{proof}
    For each $t>0$, the permutation sets $\pdiff^{t}$ is updated at~\Cref{line-set-vset-pdiff-pdisc},~\Cref{line-set-pdiff} or~\Cref{line-coupling-set-perm}. For the updated permutation sets at~\Cref{line-set-vset-pdiff-pdisc} and~\Cref{line-coupling-set-perm}, we have $P\in  \pdisc^{t}$. 
    
    As for those permutation set $P$ updated at~\Cref{line-set-pdiff}, we have $P\in \psmall$. Moreover, there exists $C$ that is unsatisfied in $\Phi_1$ or $\Phi_2$. If $C\in \cdisc^t$, we have $C\in \cdisc^t\cap \+C(P)$; Otherwise, there exists $P'\in \pdisc$, implying $C\in \+C(P)\cap \+C(P')$.
\end{proof}

\begin{lemma}\label{lemma-compisition-pd1}
$\pdiff^{1}\subseteq \pdiffinit$, $\cdisc^{1} = \emptyset$.
\end{lemma}
\begin{proof}
In the first round of the while loop at \Cref{line-while-loop}, the $\vset$ at \Cref{line-if-c-unsatisfied} is exact the $\vset'$ returned by \initialcouple \ at \Cref{line-initial-coupling}.
Moreover, for the returned variables $(\Phi_1,\Phi_2,\vset',\pdiffinit)$ of \initialcouple \ subroutine and each permutation set $P\subseteq \vset'$,  
by checking \Cref{alg-initial-coupling} one can verify that if any $C\in \cremain(P)$ is unsatisfied in $\Phi_1$ or $\Phi_2$, then 
$P$ is also in $\pdiffinit$.
Therefore, we have the conclusion that $\{P\mid C\in \+C(P)\}\subseteq \pdiffinit$
for each $C$ satisfying the conditions at \Cref{line-if-c-unsatisfied}. 
Combining this conclusion with Lines \ref{line-set-vset-pdiff-pdisc} and \ref{line-set-pdiff}, we have $\pdiff^{1}\subseteq \pdiffinit$ immediately.
Combining this conclusion with Line \ref{line-set-vset-pdiff-pdisc} and the definition of $\cdisc^{1}$, we have
$\cdisc^{1} = \emptyset$.
\end{proof}

\begin{lemma}\label{lem:addp-cases}
For each $t>1$ and $P\in \pdisc^{t} \setminus \pdisc^{t-1}$, one of the following conditions holds:
\begin{enumerate}
\item\label{cond:addp-case1} there exists $P^{\dagger}\in \pdisc^{t-1}$ such that 
$(P^{\dagger},P)\in \+S^{\ast}(\pdisc,\pdisci,2)$;
\item\label{cond:addp-case2} there exists $C\in \cdisc^{t-1}$ such that 
$(C,P)\in \+S(\cdisc,\pdisci,1)$;
\item\label{cond:addp-case3} there exists some $P^{\dagger}\in \pdisc^{t}\setminus\pdisc^{t-1}$ such that
$(P^{\dagger},P)\in \edgepermutation(\pdisci,\pdisc\setminus \pdisci)$.
\end{enumerate}
\end{lemma}
\begin{proof}
Given that $P$ is added to $\pdisc^{t}$ in the $t$-th round, there must exist some constraints in $\cremain^{t-1}$ that intersect with the \pnames in $\pdiff^{t-1}$ at the beginning of the $t$-th round. Let $C$ and $P^{\ast}$ denote the corresponding constraint and \pname specified in \Cref{line-condition-coupling} at the beginning of the $t$-th round, and $P'$ be the \pname specified in \Cref{line-find-p-coupling} of \Cref{Alg:coupling}. 

Note that $P^{\ast}\in \pdiff^{t-1}$, $P\in \pdisc^{t}$ and $C\in \+C(P^{\ast})\cap \+C(P') $ by assumption. We can also infer that $P'\in \pdisc^t\cap \pdisci$, as otherwise we would have $\pdisc' = \emptyset$ in \Cref{line-coupling-set-perm} in the $t$-th round and then $\pdisc^{t-1} = \pdisc^{t}$, which is contradictory with 
$P\in \pdisc^{t} \setminus \pdisc^{t-1}$.
Furthermore, for any $P^{\ast}\in \pdisc^{t}\setminus \pdisc^{t-1}$ where $P^{\ast}\neq P'$, we have $(P',P^{\ast})\in \+E(\pdisci,\pdisc\setminus \pdisci)$ according to the execution of Algorithms \ref{Alg:coupling} and \ref{alg-initial-coupling} , which implies the condition ~\eqref{cond:addp-case3}. 
Consequently, to complete the proof, it suffices to show that either condition \eqref{cond:addp-case1} or \eqref{cond:addp-case2} in this lemma holds as follows.
By $P^{\ast}\in \pdiff^{t-1}$ and \Cref{lemma-compisition-pd}, there are three possibilities:
\smallskip
\paragraph{\textbf{Case I}: $P^{\ast}\in \pdisc^{t-1}$.} In this case, we have $(P^{\ast},P')\in\+S(\pdisc,\pdisci,1)$ by $P^{\ast}\in \pdisc^{t-1}$, $P\in \pdisc^{t}$, and $C\in \+C(P^{\ast})\cap \+C(P')$, which implies the condition~\eqref{cond:addp-case1};


\smallskip
\paragraph{\textbf{Case II}: $P^{\ast}\in \psmall$, and there exists $P''\in \pdisc^{t-1}$ and $C'\in \+C$ such that $C'\in \+C(P^{\ast})\cap \+C(P'')$.} 
We have $(P',C',C,P')$ is a \propagate by $C'\in \+C(P^{\ast})\cap \+C(P'')$, $C\in \+C(P')\cap \+C(P^{\ast})$, and $P^{\ast}\in \psmall$. Therefore, the condition~\eqref{cond:addp-case1} holds by $P''\in \pdisc,P'\in \pdisci$;

\smallskip
\paragraph{\textbf{Case III}: $P^{\ast}\in \psmall$, and there exists $C'\in \cdisc^{t-1} \cap \+C(P^{\ast})$.} By $C'\in \cdisc^{t-1}$, $P'\in \pdisc$ and \Cref{observation-alg-coupling},
we have $C'\in \cdisc^{t-1} \setminus \+C(P')$.
Combining with $C',C\in \+C(P^{\ast})$, $P^{\ast}\in \psmall$, and $C\in \+C(P')$, we have $(C',C,P')$ is a \propagate. Therefore, the condition~\eqref{cond:addp-case2} follows from $P'\in \pdisci$ and $C'\in \cdisc^{t-1} \setminus \+C(P')$.

\end{proof}

\begin{lemma}\label{lem:addconstraint-cases}
For each $t>1$ and $C\in \cdisc^{t} \setminus \cdisc^{t-1}$, one of the following conditions holds:
\begin{enumerate}
\item \label{cond:addconstriant-case1} there exists $P\in \pdisc^{t-1}$ such that 
$(P,C)\in \+S^{\ast}(\pdisc,\cdisc,2)$;
 \item \label{cond:addconstriant-case2} there exists $C'\in \cdisc^{t-1}$ such that 
$(C',C)\in \+S^{\ast}(\cdisc,\cdisc,1)$.
\end{enumerate}
\end{lemma}
\begin{proof}
Given that $C'$ is added to $\cdisc^{t}$ in the $t$-th round, there must exist some constraints in $\cremain^{t-1}$ that intersect with the \pnames in $\pdiff^{t-1}$ at the beginning of the $t$-th round. Let $C$ and $P^{\ast}$ denote the corresponding constraint and \pname specified in \Cref{line-condition-coupling} at the beginning of the $t$-th round, and $P'$ be the \pname specified in \Cref{line-find-p-coupling} of \Cref{Alg:coupling}. 

According to the execution of the coupling algorithm, we have $C\in \+C(P^{\ast})\cap \+C(P')$, $P'\notin \pdisc^t$, and $C'\in \+C(P')$. 
In addition, by $C'\in \cdisc \subseteq \csmall$ and $C'\in \+C(P')$, we have $P'\in \psmall$.
We complete the proof via the discussion for $P^{\ast}$ according to \Cref{lemma-compisition-pd} as follows:
\smallskip
\paragraph{\textbf{Case I}: $P^{\ast}\in \pdisc$.} In this case, we have $(P^{\ast},C,C')$ is a \propagate since $C\in \+C(P^{\ast})\cap \+C(P')$, $C'\in \+C(P')$ and $P'\in \psmall$. Therefore, the condition~\eqref{cond:addconstriant-case1} holds by $P^{\ast}\in \pdisc$ and $C'\in \cdisc$;

\smallskip
\paragraph{\textbf{Case II}: $P^{\ast}\in \psmall$, and there exists $P''\in \pdisc^{t-1}$ and $C''\in \+C$ such that $C''\in \+C(P^{\ast})\cap \+C(P'')$.} Combining with $C\in \+C(P^{\ast})\cap \+C(P')$, $C'\in \+C(P')$, and $P^{\ast},P'\in \psmall$, we have $(P'',C'',C,C')$ is a \propagate.
Therefore, the condition~\eqref{cond:addconstriant-case1} holds by $P''\in \pdisc$ and $C'\in \cdisc$.

\smallskip
\paragraph{\textbf{Case III}: $P^{\ast}\in \psmall$, and there exists $C''\in \cdisc^{t-1} \cap \+C(P^{\ast})$.} Note that $C''\neq C'$ by definition.
Combining with $C\in \+C(P^{\ast})\cap \+C(P')$ and $C'\in \+C(P')$, We have $(C'',C,C')$ is a \propagate. Therefore, the condition~\eqref{cond:addconstriant-case2} holds by $C'',C'\in \cdisc$ and $C''\neq C'$.
\end{proof}

\begin{lemma}\label{lem-unmodifywitness}
For each $P\in \pdiff\setminus \pdiffinit$, there exists a sequence $(v_0, v_1, \cdots,v_t)$ such that
\begin{enumerate}
    \item \label{uniquep0} $v_0\in \pdiffinit$ and $v_i\not\in \pdiffinit$ for each $i\in [t]$;
    \item $\forall i\in [t]$,  $(( v_{i-1},v_{i})\in \edgepermutation(\pdisci,\pdisc)\land i>1)$ or $(v_{i-1},v_{i})\in \+S^{\ast}(\pdisc,\cdisc\cup \pdisci,2)\cup \+S^{\ast}(\cdisc,\cdisc\cup \pdisci, 1)$.
    \item If $P\in\psmall$, either
    $v_t = P$, or $(v_t, P) \in \+S(\pdisc,\psmall,1)$, or $v_t\in \+C(P)$; Otherwise, $v_t=P$.
    \item If $P\in\pdisc$, then $v_t=P$.
\end{enumerate}
\end{lemma}
\begin{proof}
By $P\in \pdiff\setminus \pdiffinit$ and \Cref{lemma-compisition-pd1}, we have $P\not\in \pdiff^{1}$.
Suppose $P\in \pdiff^{r}\setminus \pdiff^{r-1}$ for some $r>1$.
We construct a sequence in a backward way. 
At first, we determine the first element $u_1$ according to~\Cref{lemma-compisition-pd}.
\begin{itemize}
    \item if $P\in \pdisc^{r}$, let $u_1=P$.
    \item Otherwise,
        if $P\in \psmall$ and there exist $P'\in \pdisc^{r}$ and $C\in \+C$ such that $C\in \+C(P)\cap \+C(P')$, let $u_1=P'$ 
    \item Otherwise, we have $P\in \psmall$ and there exists $C\in \cdisc^t\cap \+C(P)$. Let $u_1=C$.
\end{itemize}
According to the construction of $u_1$, we have $u_1\in \pdisc^{r}\cup \cdisc^{r}$. 
In the next, given $u_i\in \pdisc^{r}\cup \cdisc^{r}$ where $i\geq 1$, 
we construct $u_{i+1}$ or stop according to Lemmas~\ref{lem:addp-cases} and~\ref{lem:addconstraint-cases} as follows.
\begin{itemize}
\item if $u_i\in \pdisc^{1}$, stop the construction. Combined with \Cref{lemma-compisition-pd1}, we have $u_i\in \pdisc^{1}\subseteq \pdiffinit$.
\item otherwise, if $u_i\in \pdisc^{r}$, 
we have $u_i\in \pdisc^{k} \setminus \pdisc^{k-1}$ for some $1 < k\leq r$.
We have one of the following occurs according to~\Cref{lem:addp-cases}, 
    \begin{itemize}
        \item if there exists some $P'\in \pdisc^{k-1}$ where $(P',u_i)\in \+S(\pdisc,\pdisci,\leq 2)$, let $u_{i+1}\gets P'$ 
        \item Otherwise, if there exists some $C\in \cdisc^{k-1}$ where $(C,u_i)\in \+S(\cdisc,\pdisci,1)$, let $u_{i+1}\gets C$ 
        \item Otherwise, some $P'\in \pdisc^{k}\setminus\pdisc^{k-1}$ where $(P',u_i)\in \+E(\pdisci,\pdisc\setminus \pdisci)$ exists. Let $u_{i+1}\gets P'$.
        In this case, we have $u_{i+1}\in \pdisc^{k}\setminus\pdisc^{k-1}$.
        Combined with $k>1$, we have $u_{i+1}\not\in \pdisc^{1}$.
    \end{itemize}
Obviously, we have $u_{i+1}\in \pdisc^{k}\cup \cdisc^{k-1}$.
Meanwhile, we also have $(u_i,u_{i+1})\in \+E(\pdisci,\pdisc)\cup \+S^{\ast}(\pdisc,\cdisc\cup \pdisci,2)\cup \+S^{\ast}(\cdisc,\cdisc\cup \pdisci, 1)$.
In addition, if $u_{i+1}\in \pdisc^{k}\setminus \pdisc^{k-1}$,
we have $u_{i+1}\in \pdisci$ and $u_{i}\not\in \pdisci$.

\item Otherwise, $u_i\in \cdisc^{r}$ where $r\geq1$. Combined with $\cdisc^{1} = \emptyset$ due to 
\Cref{lemma-compisition-pd1}, 
we have $u_i\in \cdisc^{k}\setminus \cdisc^{k-1}$ for some $1 < k \leq r$. 
Thus, one of the following occurs according to~\Cref{lem:addconstraint-cases}.
\begin{itemize}
    \item if there exists some $P'\in \pdisc^{k-1}$ such that $(P',u_i)\in \+S(\pdisc,\cdisc,\leq 2)$, let $u_{i+1}\gets P'$ 
    \item Otherwise, some $C'\in \cdisc^{k-1}$ where $(C',u_i)\in \+S(\cdisc,\cdisc,\leq 1)$ exists. Let $u_{i+1}\gets C'$.
\end{itemize}
Obviously, we have $u_{i+1}\in \pdisc^{k-1}\cup \cdisc^{k-1}$. Meanwhile, we also have $(v_{i-1},v_{i})\in \+S^{\ast}(\pdisc,\cdisc\cup \pdisci,2)\cup \+S^{\ast}(\cdisc,\cdisc\cup \pdisci, 1)$.
\end{itemize} 

We remark that this construction eventually stops.
Because if $u_i\in (\pdiff^{k} \setminus \pdiff^{k-1})\cup (\cdisc^{k} \setminus \cdisc^{k-1})$ for some $1 < k\leq r$, then either $u_{i+1} \in \pdisc^{k-1}\cup \cdisc^{k-1}$ or $u_{i+1}\in \pdisci\cap (\pdisc^{k}\setminus \pdisc^{k-1})$.
Thus, we must have $u_{i+2} \in \pdisc^{k-1}\cup \cdisc^{k-1}$ by $u_{i+1}\in \pdisci$.
According to the construction, one can verify that the sequence $\left(u_i,u_{i-1},\cdots,u_1\right)$ satisfy the required conditions.
\end{proof}

\subsection*{Proof of~\Cref{lemma-witness-sequence-pinpdiff}}

By~\Cref{lem-unmodifywitness}, there exists a sequence $s=(u_0,\cdots,u_t)$ satisfying the conditions specified in~\Cref{lem-unmodifywitness}. In the following, we shall construct a w.s. $s'=(v_0,\cdots,v_{t'})$ induced from $s$ such that $v_0\in \pdiffinit$ and $(( v_{i-1},v_{i})\in \edgepermutation(\pdisci,\pdisc)\land i>1)$ or $(v_{i-1},v_{i})\in \+S^{\ast}(\cdisc\cup \pdisci,\cdisc\cup \pdisci,2)$ for each $i\in [t']$, while preserving the conditions specified in~\Cref{lem-unmodifywitness}.

\begin{enumerate}
    \item \label{step2} We first construct a new sequence $s'=(v_0,\cdots,v_{t'})$ satisfying the following condition:
    \begin{align}\label{eqn:inequalcond}
        v_i\neq v_j \  \ \ \forall \ 0\leq i\neq j\leq t'.
    \end{align}
    Let $s'\gets s$ where $s'=(v_0,\cdots,v_{t'})$. According to the condition~\eqref{uniquep0} specified in~\Cref{lem-unmodifywitness}, we have $v_0\neq v_i$ for each $i\in [t']$.
    We repeat the following operation until~\eqref{eqn:inequalcond} holds:
    \begin{enumerate}
        \item Let $i,j\in [t']$ satisfying $v_i=v_j$ and $i<j$;
        \item Update $s'\gets(v_0,v_1,\cdots,v_{i-1},v_{j},\cdots,v_{t'})$, i.e., delete all components from $v_i$ to $v_{j-1}$.
    \end{enumerate}
    Let $s'=(v_0,\cdots,v_{t'})$ be the resulting sequence on which no further operations can be performed. One can verify that $s'$ preserves the conditions specified in~\Cref{lem-unmodifywitness}.
    \item \label{step3} We then construct a new sequence $s'=(v_0,\cdots,v_{t'})$ from the sequence constructed in the previous step, satisfying the following condition:
    \begin{align}\label{eqn:twoatmost}
        (v_i,v_j)\notin \edgepermutation \mbox{ if $j\in [i-2]$ for each $i\in [t']$}.
    \end{align}
    Let $s'=(v_0,\cdots,v_{t'})$ be the sequence constructed in~\eqref{step2}. We repeat the following operation until~\eqref{eqn:twoatmost} holds:
    \begin{enumerate}
        \item Let $i\in [t']$ be the smallest index such that there exists  $j,k\in [t']$ satisfying $i<j<k$ and $v_i,v_j,v_k\subseteq P'$ for some $P'\in \+P'$. 
        \item Update $s'\gets (v_0,v_1,\cdots,v_i, v_k,\cdots,v_{t'})$, i.e., delete all component between $v_i$ and $v_k$.
    \end{enumerate}
    In each operation, we can ensure that $(v_i,v_k)\in \edgepermutation(\pdisci,\pdisc)$ since the chosen $i$ is the smallest index. Let $s'=(v_0,\cdots,v_{t'})$ be the resulting sequence on which no further operations can be performed. One can verify that $s'$ preserves the conditions specified in~\Cref{lem-unmodifywitness}, 
    \item \label{step4} We then construct a new sequence $s'=(v_0,\cdots,v_{t'})$ from the sequence constructed in the previous step, satisfying the following condition:
    \begin{align}\label{eqn:notintersect}
        ((v_i,v_j)\not \in \edgeconstraint)\land(v_i\not\in \+C(v_j) \text{ if } v_j\in \pdisc) \mbox{ for each $i\in [t']$ and $j\in [i]$}.
    \end{align}
    However, the condition $(v_i\not\in \+C(v_j) \text{ if } v_j\in \pdisc) \mbox{ for each $i\in [t']$ and $j\in [i]$}$ holds immediately, since if an edge is in $\cdisc$, then each bucket on this edge is not in $\pdisc$ or $\pdisci$.
    To ensure~\eqref{eqn:notintersect}, we first construct a sequence $s'=(v_0,\cdots,v_{t'})$ satisfying the following condition:
    \begin{align}\label{eqn:minimum}
        (v_i,v_j)\not \in \edgeconstraint \mbox{ for each $i\in [t']$ and $0\leq j\leq i-2$}.
    \end{align}
    Let $s'=(v_0,\cdots,v_{t'})$ be the sequence constructed in~\eqref{step3}. We repeat the following operation until~\eqref{eqn:minimum} holds:
    \begin{enumerate}
        \item Let $i\in [t']$ and $0\leq j\leq i-2$ such that $ (v_i,v_j)\not \in \edgeconstraint)$. 
        \item Update $s'\gets (v_0,v_1,\cdots,v_j, v_i,\cdots,v_{t'})$, i.e., delete all component between $v_i$ and $v_j$.
    \end{enumerate}
    Let $s'=(v_0,\cdots,v_{t'})$ be the resulting sequence on which no further operations can be performed in the previous procedure. One can verify that $s'$ preserves the conditions specified in~\Cref{lem-unmodifywitness} and~\eqref{eqn:minimum} holds.
    Finally, we further repeat the following operation:
    \begin{enumerate}
        \item Let $i\in [t']$ such that $v_{i}\in \cdisc$, $v_{i+1}\in \cdisc$.
        \item Update $s'\gets (v_0,v_1,\cdots,v_i, v_{i+2},\cdots,v_{t'})$, i.e., delete $v_{i+1}$.
    \end{enumerate}
    Note that the above operation could only increase the distance $\+S^{\ast}(\cdisc,\cdisc\cup \pdisci, 1)$ to $\+S^{\ast}(\cdisc,\cdisc\cup \pdisci, 2)$.
    Let $s'=(v_0,\cdots,v_{t'})$ be the resulting sequence on which no further operations can be performed in the previous procedure. Combining with the preserved conditions specified in~\Cref{lem-unmodifywitness}, one can verify that $\forall i\in [t]$,  $(( v_{i-1},v_{i})\in \edgepermutation(\pdisci,\pdisc)\land i>1)$ or $(v_{i-1},v_{i})\in \+S^{\ast}(\pdisc\cup \cdisc,\cdisc\cup \pdisci,2)$. Moreover, $(v_t \rightarrow P) \in \+S(\cdisc,\pzeta,\leq 1)$ if $v_t\in \cdisc$.
    
\end{enumerate}
Combining all these properties, the proof is immediate.

\subsubsection{\textbf{\emph{Bounding the discrepancy}}}\label{subsub:boundingdiscrepancy}

{We complete the proof of ~\Cref{lemma-deviation-witness-large} in this section}.

The proof of ~\Cref{lemma-deviation-witness-large} relies on the following lemma.
\begin{lemma}\label{appendix-lemma-distance-large}
Let $p\in (0,1)$ and integer $\Delta\geq 1$ satisfy $\-ep\Delta\leq 1$.
Consider any formulas $\Phi_1=(\+P,\+Q_1,\+C)$, $\Phi_2 = (\+P,\+Q_2,\+C)$ where $p_{\Phi_1}\leq p, p_{\Phi_2}\leq p, \Delta_{\Phi_1}\leq \Delta, \Delta_{\Phi_2}\leq \Delta$ with a \pname $P\in \+P$ where $\+Q_1(P)=\+Q_2(P)$, $4\-e\Delta\leq \abs{P}$ and a partition $\set{P_1,\cdots,P_\ell}$ of P. 

Let $T = \{P'\in \+P\mid \+Q_1(P')\neq \+Q_2(P')\}$.
Let $S$ be any set of variables satisfying $P \cap S = \emptyset$,
$\bigcup_{P'\in T}P'\subseteq S$,
and that in the factorizations of $(\+P,\+Q_1,\+C\setminus \+C(P))$ and $(\+P,\+Q_2,\+C\setminus \+C(P))$,
each vertex $v\in S$ is not connected to any $u\not\in S$.
Let $(\sigma_1,\sigma_2)$ be an optimal coupling of $\mu_{\Phi_1}$ and $\mu_{\Phi_2}$ such that 
$\E{\sum_{i\in [\ell]}\abs{\sigma_1(P_i)\setminus\sigma_2(P_i)}}$ is minimum.
Then we have 
\[
\E{\sum_{i\in [\ell]}\abs{\sigma_1(P_i)\setminus\sigma_2(P_i)}} \leq 100\Delta^2\abs{P}^{-1}\abs{\+C(S)\cap \+C(P)}.
\]
\end{lemma}

Before we complete the proof of~\Cref{appendix-lemma-distance-large}, we introduce the following lemmas.

\begin{lemma}\label{lem:large-couple-size1}
Given any PDC formulas $\Phi_1 = (\+P,\+Q,\+C_1)$, $\Phi_2 = (\+P,\+Q,\+C_2)$ where $\+P=\{P\}$ and an integer $\Delta\geq 1$,
assume $\Delta_{\Phi_1}\leq \Delta,\Delta_{\Phi_2}\leq \Delta$, {$4\-e\Delta\leq \abs{P}$}, $\+C_1\subseteq \+C_2$, and {$\abs{\vbl(C)}\leq 1$ for any $C\in \+C_2$}.
Let $(\sigma_1,\sigma_2)$ be an optimal coupling between $\mu_{\Phi_1}$ and $\mu_{\Phi_2}$ such that 
$\E{\sum_{v\in P}\id{\sigma_1(v)\neq \sigma_2(v)}}$ is minimum.
Then we have 
\[
\E{\sum_{v\in P}\id{\sigma_1(v)\neq \sigma_2(v)}} \leq \frac{17\Delta\abs{\+C_2\setminus\+C_1}}{\abs{P}}.
\]
\end{lemma}

\begin{proof}
Note that for any constraint $C\in \+C_2$, the violation probability is at most $1/\abs{P}$. Combined with {$\abs{P}\geq k\Delta$,} we have 
\begin{align}\label{eqn:LLL-condition-singlevariable}
    \-ep_{\Phi_1}\Delta_{\Phi_1}\leq 1, \quad \-ep_{\Phi_2}\Delta_{\Phi_2}\leq 1.
\end{align}
Then, we prove the lemma by induction.  

\medskip
\noindent\textbf{Base case:}
we first show that given $\abs{\+C_2\setminus \+C_1}=1$, we have 
\begin{align}\label{eqn:couple-diff-1}
    \E{\sum_{v\in P}\id{\sigma_1(v)\neq \sigma_2(v)}} \leq \frac{17\Delta}{\abs{P}}.
\end{align}

Let $\set{C}=\+C_2\setminus \+C_1$, $\set{u}=\vbl{(C)}$ and $c=\Lambda(C,u)$. For convenience, we use $\mu_1$ and $\mu_2$ to denote $\mu_{\Phi_1}$ and $\mu_{\Phi_2}$, respectively. Moreover, the corresponding supports are denoted by $\Omega_1$ and $\Omega_2$. By definition, we have $\Omega_2\subseteq \Omega_1$, and the density of $\mu_1,\mu_2$ are $1/\abs{\Omega_1},1/\abs{\Omega_2}$, respectively. 
Applying the ~\Cref{prop-PDC-local-uniformity} with the condition ~\eqref{eqn:LLL-condition-singlevariable}, we can compare the size of $\Omega_1$ and $\Omega_2$. Specifically, we have
\begin{align}\label{eqn:bound-of-ratio-omega12}
    \frac{\abs{\Omega_2}}{\abs{\Omega_1}}=(1-\Pr[\sigma\sim \mu_1]{\sigma(u)=c})\geq \left(1-\frac{1}{\abs{P}} \left(1-\frac{\-e}{\abs{P}}\right)^{-\Delta}\right)\geq\left(1- \frac{1}{\abs{P}}\left(1+\frac{2\-e\Delta}{\abs{P}}\right)\right),
\end{align} 
which implies
\begin{align}\label{eqn:lb-density-omega_1}
    \frac{1}{\abs{\Omega_1}}\geq  \frac{1}{\abs{\Omega_2}}\cdot\left(1- \frac{1}{\abs{P}}\left(1+\frac{2\-e\Delta}{\abs{P}}\right)\right).
\end{align}

Then, for any variable $v\in P\setminus \set{u}$ and $\sigma\in \Omega_1\setminus \Omega_2$, we map it to a valid assignment $\tau$ of $\Phi_1$ by the mapping $\tau=f(v,\sigma)$ where 
\begin{align*}
    \tau(v)=\sigma(u), \tau(u)=\sigma(v), \mbox{and } \tau_{P\setminus \set{u,v}}=\sigma_{P\setminus \set{u,v}}.
\end{align*}
{One can verify that the mapping $f(v,\cdot)$ is injective for any $v\in P\setminus \set{u}$, and $f(\cdot,\sigma)$ is injective for any satisfying assignments $\sigma\in \Omega_1\setminus \Omega_2$.}

We then construct a coupling between $\mu_1$ and $\mu_2$ by specifying the density of the joint distribution $\@C:\Omega_1 \times \Omega_2 \rightarrow [0,1]$. Specifically, we would construct coupling tables $\@C_1,\@C_2,\@C_3:\Omega_1 \times \Omega_2 \rightarrow [0,1]$ and let $\@C(\sigma_1,\sigma_2)=\sum_{i\in [3]}\@C_i(\sigma_1,\sigma_2)$ for any $\sigma_1\in \Omega_1, \sigma_2\in \Omega_2$. The coupling tables $\@C_1,\@C_2,\@C_3$ are constructed as follows (any unspecified entries are defined as $0$):
\begin{enumerate}
    \item For any $\sigma\in \Omega_2$, define $\@C_1(\sigma,\sigma)=1/\abs{\Omega_1}$;
    \item For any $\sigma\in \Omega_1\setminus \Omega_2$ and $v\in P\setminus \set{u}$ satisfying $f(v,\sigma)\in \Omega_2$, define \[\@C_2\left(\sigma,f(v,\sigma)\right)=\min\set{\frac{1}{\abs{P}-1}\cdot \frac{1}{\abs{\Omega_1}},\frac{1}{\abs{\Omega_2}}-\frac{1}{\abs{\Omega_1}}};\]
    \item {Complete the coupling table $\@C_3$ such that for any $\sigma\in \Omega_1$, we have $\sum_{i\in[3],\tau\in \Omega_2}\@C_i(\sigma,\tau)=1/{\abs{\Omega_1}}$; for any $\tau\in \Omega_2$, we have $\sum_{i\in[3],\sigma\in \Omega_2}\@C_i(\sigma,\tau)=1/{\abs{\Omega_2}}$}. That is, the coupling table $\@C_3$ is the residual density of the coupling between $\mu_1$ and $\mu_2$.
\end{enumerate}
{We claim that the mapping constructed above defines a valid coupling between $\mu_1$ and $\mu_2$. To verify this, it suffices to show that the marginal density specified by $\@C_1$ and $\@C_2$ is less than their marginal probability respectively:}
\begin{itemize}
    \item For any $\sigma\in \Omega_2$, one can verify that $\sum_{i\in[2],\tau\in \Omega_2}\@C_i(\sigma,\tau)=1/{\abs{\Omega_1}}$. For any $\sigma\in \Omega_1\setminus \Omega_2$, the entries in $\@C_1(\sigma,\cdot)$ are all $0$. As for $\@C_2(\sigma,\cdot)$, there are at most $\abs{P}-1$ non-zero entries and each less than $1/((\abs{P}-1)\cdot \abs{\Omega_1})$ by definition. Therefore, $\sum_{i\in[2],\tau\in \Omega_2}\@C_i(\sigma,\tau)\leq 1/{\abs{\Omega_1}}$.
    \item For any $\tau\in \Omega_2$, we claim that there is at most one non-zero entry in $\@C_2(\cdot,\tau)$. As mentioned before, the mapping $f(v,\cdot)$ is injective for any $v\in P\setminus \set{u}$. Consequently, if there are $\sigma_1$ and $\sigma_2$ such that $\@C_2(\sigma_1,\tau)>0$ and $\@C_2(\sigma_2,\tau)>0$, then $\sigma_1$ and $\sigma_2$ must map to $\tau$ via distinct vertices in $P\setminus \set{u}$. However, this situation cannot occur, as can be verified from the definition of the mapping $f(\cdot,\cdot)$, which assigns value $c$ to different vertices. Therefore, it follows that $\sum_{i\in[2],\sigma\in \Omega_2}\@C_i(\sigma,\tau)\leq 1/{\abs{\Omega_2}}$.
\end{itemize}

In the coupling table constructed above, we have
\begin{equation}\label{eqn:summation-C2}
    \begin{aligned}
        \sum_{(\sigma_1,\sigma_2)\in \Omega_1\times\Omega_2} \@C_2(\sigma_1,\sigma_2)&\leq \sum_{\sigma\in \Omega_1\setminus\Omega_2} \frac{1}{\abs{\Omega_1}}=\left(1-\frac{\abs{\Omega_2}}{\abs{\Omega_1}}\right)\leq \frac{1}{\abs{P}}\left(1+\frac{2\-e\Delta}{\abs{P}}\right),
    \end{aligned} 
\end{equation} where the last inequality follows from ~\eqref{eqn:bound-of-ratio-omega12}. 

We should also analyze the coupling table $\@C_3$. Let $\Omega_2^{\circ}$ be the collection of assignments $\tau \in \Omega_2$ such that $\tau=f(v,\sigma)$ for the variable $v\in P\setminus \set{u}$ and assignment $\sigma\in \Omega_1\setminus \Omega_2$. {By the injection of the mapping}, for any $\sigma_2\in \Omega_2^{\circ}$, we have
    \begin{align*}
        \sum_{\sigma_1 \in \Omega_1} \@C_3(\sigma_1,\sigma_2)&= 
        \left(\frac{1}{\abs{\Omega_2}}-\frac{1}{\abs{\Omega_1}} \right) -\min\set{\frac{1}{\abs{P}-1}\cdot \frac{1}{\abs{\Omega_1}},\frac{1}{\abs{\Omega_2}}-\frac{1}{\abs{\Omega_1}}}\\
        &= \min\set{\frac{1}{\abs{\Omega_2}}\left(1-\frac{\abs{\Omega_2}}{\abs{\Omega_1}}\right)-\frac{1}{\abs{P}-1}\cdot \frac{1}{\abs{\Omega_1}},0}\\
        &\leq  \frac{1}{\abs{\Omega_2}}\left(\frac{1}{\abs{P}}\left(1+\frac{2\-e\Delta}{\abs{P}}\right)\right) -\frac{1}{\abs{\Omega_2}}\left(\frac{1}{\abs{P}-1}\cdot \left(1- \frac{1}{\abs{P}}\left(1+\frac{2\-e\Delta}{\abs{P}}\right)\right) \right)\\
        &\leq  { \frac{1}{\abs{\Omega_2}}\cdot \frac{4\-e\Delta}{\abs{P}^2}},
   \end{align*}
 where the second inequality holds by ~\eqref{eqn:bound-of-ratio-omega12} and ~\eqref{eqn:lb-density-omega_1}.
The above fact implies that
\begin{align}\label{eqn:ub-omega2circ}
    \sum_{\sigma_2\in \Omega_2^\circ} \sum_{\sigma_1 \in \Omega_1} \@C_3(\sigma_1,\sigma_2)\leq \frac{4\-e\Delta}{\abs{P}^2}.
\end{align}
Moreover, one can verify that for any $\sigma_2\in \Omega_2\setminus 
\Omega_2^\circ$,
   \begin{align*}
        \sum_{\sigma_1 \in \Omega_1} \@C_3(\sigma_1,\sigma_2)&= \frac{1}{\abs{\Omega_2}}-\frac{1}{\abs{\Omega_1}}.
   \end{align*}
Consequently, it holds that
\begin{equation}\label{eqn:ub-omega2setminus}
    \begin{aligned}
          \sum_{\sigma_2\in \Omega_2\setminus \Omega_2^\circ} \sum_{\sigma_1 \in \Omega_1} \@C_3(\sigma_1,\sigma_2)&\leq \sum_{\sigma_2\in \Omega_2\setminus \Omega_2^\circ} \left( \frac{1}{\abs{\Omega_2}}-\frac{1}{\abs{\Omega_1}}\right)\\
          &\leq \sum_{\sigma_2\in \Omega_2\setminus \Omega_2^\circ}\frac{1}{\abs{\Omega_2}}\left(\frac{1}{\abs{P}}\left(1+\frac{2\-e\Delta}{\abs{P}}\right)\right)\\
          &= \left(\frac{1}{\abs{P}}\left(1+\frac{2\-e\Delta}{\abs{P}}\right)\right)\cdot 
          \Pr[\sigma\sim \mu_2]{\sigma\in \Omega_2\setminus \Omega_2^\circ}\\
          &\leq \left(\frac{1}{\abs{P}}\left(1+\frac{2\-e\Delta}{\abs{P}}\right)\right)\cdot  \sum_{v\in P\setminus\set{u}} \sum_{c'\in \Lambda(\+C_2,v)\setminus\set{c}} \Pr[\sigma\sim \mu_2]{\sigma(v)=c,\sigma(u)=c'}\\
          &\leq \left(\frac{1}{\abs{P}}\left(1+\frac{2\-e\Delta}{\abs{P}}\right)\right)\cdot  (\abs{P}-1)\cdot \Delta \cdot  \frac{1}{\abs{P}\abs{P-1}}\cdot \left(1+\frac{4\-e\Delta}{\abs{P}}\right)\\
          &\leq \left(\frac{1}{\abs{P}}\left(1+\frac{2\-e\Delta}{\abs{P}}\right)\right)\cdot  \frac{2\Delta}{\abs{P}},
    \end{aligned}
\end{equation} {where the second inequality follows from ~\eqref{eqn:bound-of-ratio-omega12} and the fourth inequality from~\Cref{prop-PDC-local-uniformity}.}

Combining ~\eqref{eqn:summation-C2}.~\eqref{eqn:ub-omega2circ} and ~\eqref{eqn:ub-omega2setminus}, there exists a coupling $\@C$ of $\mu_{\Phi_1}$ and $\mu_{\Phi_2}$ satisfying
\begin{align*}
    &\quad\E[(\sigma_1,\sigma_2)\sim \@C]{\sum_{v\in P}\id{\sigma_1(v)\neq \sigma_2(v)}}\\ &=\sum_{i\in[3]}\sum_{(\sigma_1,\sigma_2)\in \Omega_1\times\Omega_2} \@C_i(\sigma_1,\sigma_2) \cdot \left(\sum_{v\in P}\id{\sigma_1(v)\neq \sigma_2(v)}\right)\\
    &=2 \cdot \sum_{(\sigma_1,\sigma_2)\in \Omega_1\times\Omega_2} \@C_2(\sigma_1,\sigma_2) + \abs{P}\cdot \left(  \sum_{\sigma_2\in \Omega_2^\circ} \sum_{\sigma_1 \in \Omega_1} \@C_3(\sigma_1,\sigma_2)+ \sum_{\sigma_2\in \Omega_2\setminus\Omega_2^\circ} \sum_{\sigma_1 \in \Omega_1} \@C_3(\sigma_1,\sigma_2) \right)\\
    &\leq \frac{2}{\abs{P}}+ \frac{4\-e\Delta}{\abs{P}}+ \frac{4\Delta}{\abs{P}}\leq \frac{17\Delta}{\abs{P}},
\end{align*} which implies the desired bound ~\eqref{eqn:couple-diff-1} in the base case where $\abs{\+C_2\setminus\+C_1}=1$.

\medskip
\noindent\textbf{Induction step:}
Given $\ell>1$, assume the statement holds for any pair of formulas $\Phi_1,\Phi_2$ where $\abs{\+C_2\setminus \+C_1}\leq \ell-1$. We then prove the statement for any formulas $\Phi_1 = (\+P,\+Q,\+C_1)$, $\Phi_2 = (\+P,\+Q,\+C_2)$ satisfying $\abs{\+C_2\setminus \+C_1}=\ell$.
Suppose $\+C_2\setminus \+C_1=\set{C_1,C_2,\cdots,C_\ell}$. We define the formula $\Psi=(\+P,\+Q,\+C_2\setminus \set{C_\ell})$. According to the induction hypothesis, there exists a coupling $\@C_1$ between $\mu_{\Phi_1}$ and $\mu_\Psi$ such that
\begin{align}\label{eqn:induction-hypo}
    \E[(\sigma_1,\tau)\sim \@C_1]{\sum_{v\in P}\id{\sigma_1(v)\neq \tau(v)}}\leq \frac{17\Delta\left(\ell-1\right)}{\abs{P}}.
\end{align}
Applying the coupling construction for the base case, we have a coupling $\@C_2$ between $\mu_{\Psi}$ and $\mu_{\Phi_2}$ satisfying 
\begin{align}\label{eqn:base-hypo}
    \E[(\tau,\sigma_2)\sim \@C_2]{\sum_{v\in P}\id{\tau(v)\neq \sigma_2(v)}}\leq \frac{17\Delta}{\abs{P}}.
\end{align}

We then construct a coupling $\@C$ between $\mu_{\Phi_1}$, $\mu_{\Psi}$ and $\mu_{\Phi_2}$ where
\begin{align*}
    \forall \sigma_1\in \Omega_{\Phi_1},\tau\in \Omega_{\Psi},  \sigma_2\in \Omega_{\Phi_2}: \quad \@C(\sigma_1,\tau,\sigma_2)= \@C_1(\sigma_1,\tau )\cdot \frac{\@C_2(\tau,\sigma_2)}{\mu_\Psi(\tau)}.
\end{align*}
Therefore, the discrepancy in the optimal coupling $\@C_{\!{OPT}}$ between $\mu_{\Phi_1}$ and $\mu_{\Phi_2}$ satisfying 
\begin{align*}
     &\quad\E[(\sigma_1,\sigma_2)\sim \@C_{\!{OPT}}]{\sum_{v\in P}\id{\sigma_1(v)\neq \sigma_2(v)}}\\
     &\leq  \E[(\sigma_1,\tau,\sigma_2)\sim \@C]{\sum_{v\in P}\id{\sigma_1(v)\neq \sigma_2(v)}}\\
     &\leq \E[(\sigma_1,\tau,\sigma_2)\sim \@C]{\sum_{v\in P}\id{\sigma_1(v)\neq \tau(v)} +\sum_{v\in P}\id{\tau(v)\neq \sigma_2(v)}}\\
     & = \E[(\sigma_1,\tau)\sim \@C_1]{\sum_{v\in P}\id{\sigma_1(v)\neq \tau(v)}} +\E[(\tau,\sigma_2)\sim \@C_2]{\sum_{v\in P}\id{\tau(v)\neq \sigma_2(v)}}\\
     &\leq \frac{17\Delta\abs{\+C_2\setminus \+C_1}}{\abs{P}},
\end{align*}where the last inequality holds by ~\eqref{eqn:induction-hypo} and ~\eqref{eqn:base-hypo}.

Putting all these together, the proof is complete.

\end{proof}

We also need the following lemma.

\begin{lemma}\label{lem:large-couple-size2}
Given any PDC formulas $\Phi_1 = (\+P,\+Q,\+C_1)$, $\Phi_2 = (\+P,\+Q,\+C_2)$ where $\+P=\{P\}$ and an integer $\Delta\geq 1$,
assume $\Delta_{\Phi_1}\leq \Delta,\Delta_{\Phi_2}\leq \Delta$,
{$2\-e\Delta\leq \abs{P}$}, $\+C_1\subseteq \+C_2$,
and $\abs{\vbl(C)} \geq 2$ for each $C\in \+C_2\setminus \+C_1$.
Let $(\sigma_1,\sigma_2)$ be an optimal coupling between $\mu_{\Phi_1}$ and $\mu_{\Phi_2}$ such that 
$\E{\sum_{v\in P}\id{\sigma_1(v)\neq \sigma_2(v)}}$ is minimum.
Then we have 
\[
\E{\sum_{v\in P}\id{\sigma_1(v)\neq \sigma_2(v)}} \leq \frac{2\abs{\+C_2\setminus\+C_1}}{\abs{P}-1}.
\]
\end{lemma}
\begin{proof}
Note that for any constraint $C\in \+C_2$, the violation probability is at most $1/\abs{P}$. Combined with $2\-e\Delta\leq \abs{P}$, we have 
\begin{align}\label{eqn:LLL-condition-singlevariable2}
    \-ep_{\Phi_1}\Delta_{\Phi_1}\leq 1, \quad \-ep_{\Phi_2}\Delta_{\Phi_2}\leq 1.
\end{align}
Then, we prove the statement by induction. 

\medskip
\noindent\textbf{Base case:}
We first show that given $\abs{\+C_2\setminus \+C_1}=1$, there exists an coupling $\@C$ between $\mu_{\Phi_1}$ and $\mu_{\Phi_2}$ such that
\begin{align}\label{eqn:couple-diff-2-1}
    \E[(\sigma_1,\sigma_2)\sim \@C]{\sum_{v\in P}\id{\sigma_1(v)\neq \sigma_2(v)}} \leq \frac{2}{\abs{P}-1}.
\end{align}
We construct the coupling by specifying the density of the joint distribution $\@C:\Omega_1\times \Omega_2\rightarrow [0,1]$. Specifically, we construct coupling tables $\@T_{1},\@T_{2}:\Omega_1\times \Omega_2\rightarrow [0,1]$ and let $\@C(\sigma_1,\sigma_2)=\sum_{i\in[2]}\@T_{i}(\tau_1,\tau_2)$ for any $\sigma_1\in \Omega_1$ and $\sigma_2\in \Omega_2$. The coupling tables are defined as follows (any unspecified entries are defined as $0$):
\begin{enumerate}
    \item For any $\sigma \in \Omega_2$, define $\@T_1(\sigma,\sigma)=\mu_{\Phi_1}(\sigma)$;
    \item {Complete the coupling table $\@T_2$ to form a coupling.}
\end{enumerate}
We then analyze the coupling error of $\@C$.

Let $\set{C}=\+C_2\setminus \+C_1$, and use $\mu_1$ and $\mu_2$ to denote $\mu_{\Phi_1}$ and $\mu_{\Phi_2}$, respectively. Moreover, the corresponding supports are denoted by $\Omega_1$ and $\Omega_2$ for convenience. By definition, we have $\Omega_2\subseteq \Omega_1$, and the density of $\mu_1,\mu_2$ are $1/\abs{\Omega_1},1/\abs{\Omega_2}$, respectively. Applying the ~\Cref{prop-PDC-local-uniformity} with the condition ~\eqref{eqn:LLL-condition-singlevariable2}, we have
    \begin{align*}
        \frac{\abs{\Omega_2}}{\abs{\Omega_1}}&=1-\Pr[\sigma\sim \mu_1]{C \mbox{ is unsatisfied under } \sigma}\\
        &\geq 1-\frac{1}{\abs{P}\left(\abs{P}-1\right)} \left(1-\frac{\-e}{\abs{P}}\right)^{-\Delta} \tag{$\^P_{\Phi_1}[C \mbox{ is unsatisfied}]\leq \frac{1}{\abs{P}\left(\abs{P}-1\right)}$}\\
        &\geq 1- \frac{1}{\abs{P}\left(\abs{P}-1\right)} \left(1+\frac{2\-e\Delta}{\abs{P}}\right)\\
        &\geq 1-\frac{2}{\abs{P}\left(\abs{P}-1\right)},
    \end{align*} 
which implies that for any $\sigma \in \Omega_2$,
\begin{align*}
    \mu_{1}(\sigma)\geq \mu_{2}(\sigma)\cdot\left(1-{2}/\left({\abs{P}\left(\abs{P}-1\right)}\right)\right).
\end{align*}

Combining all these facts, it holds that 
\begin{equation}\label{eqn:coupling-error-2}
    \begin{aligned}
        \Pr[\left(\sigma_1,\sigma_2\right)\sim \@C]{\sigma_1\neq \sigma_2}&=1- \Pr[\left(\sigma_1,\sigma_2\right)\sim \@C]{\sigma_1=\sigma_2}\\
        &=1-\sum_{\sigma \in \Omega_2} \mu_1(\sigma)\\
        &\leq 1-\sum_{\sigma \in \Omega_2}\mu_{2}(\sigma)\cdot\left(1-{2}/\left({\abs{P}\left(\abs{P}-1\right)}\right)\right) \\
        &=\frac{2}{\abs{P}\left(\abs{P}-1\right)}.
    \end{aligned}
\end{equation}
Therefore, we have
\begin{align*}
     \E[(\sigma_1,\sigma_2)\sim \@C]{\sum_{v\in P}\id{\sigma_1(v)\neq \sigma_2(v)}}&\leq  \Pr[\left(\sigma_1,\sigma_2\right)\sim \@C]{\sigma_1\neq \sigma_2}\cdot \abs{P}\leq  \frac{2}{\abs{P}-1},
\end{align*} where the last inequality follows from~\eqref{eqn:coupling-error-2}.

\medskip
\noindent\textbf{Induction step:}
Given $\ell>1$, assume the statement holds for any formulas $\Phi_1,\Phi_2$ where $\abs{\+C_2\setminus \+C_1}\leq \ell-1$. 
{The statement holds for any formulas $\Phi_1 = (\+P,\+Q,\+C_1)$, $\Phi_2 = (\+P,\+Q,\+C_2)$ satisfying $\abs{\+C_2\setminus \+C_1}=\ell$ by similar argument stated in the proof of~\Cref{lem:large-couple-size1}.}

\end{proof}

Combining with~\Cref{lem:large-couple-size1} and~\Cref{lem:large-couple-size2}, we have the following lemma according to the triangle inequality for the coupling.
\begin{lemma}\label{lem:large-couple-combined}
Given any PDC formulas $\Phi_1 = (\+P,\+Q,\+C_1)$, $\Phi_2 = (\+P,\+Q,\+C_2)$ where $\+P=\{P\}$ and an integer $\Delta\geq 1$,
assume $\Delta_{\Phi_1}\leq \Delta,\Delta_{\Phi_2}\leq \Delta$, {$4\-e\Delta\leq \abs{P}$} and $\+C_1\subseteq \+C_2$.
Let $(\sigma_1,\sigma_2)$ be an optimal coupling between $\mu_{\Phi_1}$ and $\mu_{\Phi_2}$ such that 
$\E{\sum_{v\in P}\id{\sigma_1(v)\neq \sigma_2(v)}}$ is minimum.
Then we have 
\[
\E{\sum_{v\in P}\id{\sigma_1(v)\neq \sigma_2(v)}} \leq \frac{17\Delta\abs{\+C_2\setminus\+C_1}}{\abs{P}-1}.
\]
\end{lemma}

We are now finishing the proof of~\Cref{appendix-lemma-distance-large}.
\begin{proof}[Proof of \Cref{appendix-lemma-distance-large}]

{If $S= \emptyset$, the two distributions can be coupled completely.} Hence, in the subsequent discussion, it suffices to focus on the case where $S\neq \emptyset$.

For convenience, we use $U$ to denote the variables $V\setminus \left(S\cup P\right)$. Given any pair of {feasible partial assignments $\sigma_{1,S},\sigma_{2,S}$ on $S$ in $\Phi_1$ and $\Phi_2$, respectively}, let {$\Psi_1=\Phi_1^{\sigma_{1,S}}$} and {$\Psi_2=\Phi_2^{\sigma_{2,S}}$} be the formulas conditioned on $\sigma_{1,S}$ and $\sigma_{2,S}$, respectively.
We use $\Omega_{U}$ to denote the support of the distribution $\mu_{\Psi, U}$ where  $\Psi=\left(\+P,\+Q_1,\+C\setminus \+C(P)\right)$. Since $\abs{P}>4\-e\Delta$, for any $\sigma\in \Omega_U$, it is feasible in $\Psi_1$ and $\Psi_2$ by lopsided LLL. Consequently, the support of $\mu_{\Psi_1,U}$ and $\mu_{\Psi_2,U}$ is equal to $\Omega_U$.
Let $\+C_1,\+C_2$ be the sets of unsatisfied constraints in $\+C(P)$ in the formula $\Psi_1$ and $\Psi_2$, respectively. We define $\Psi_3$ as the PDC formula obtained by removing the variables in $S$ while retaining the constraints $\left(\+C\setminus \+C(P)\right)\cup \left(\+C_1\cap \+C_2\right)$. This formula will serve as an intermediate step in the construction of the coupling.

In the following, we specify a coupling between $\mu_{\Psi_1}$ and $\mu_{\Psi_3}$. 
We first construct a coupling $\@C_{U}$ between $\mu_{\Psi_1,U}$ and $\mu_{\Psi_3,U}$ by specifying the density of the joint distribution $\@C_{U}:\Omega_U\times \Omega_U\rightarrow [0,1]$. Specifically, we would construct coupling tables $\@T_{1},\@T_{2}:\Omega_U\times \Omega_U\rightarrow [0,1]$ and let $\@C_{U}(\tau_1,\tau_2)=\sum_{i\in[2]}\@T_{i}(\tau_1,\tau_2)$ for any $\tau_1,\tau_2\in \Omega_U$. The coupling tables are defined as follows (any unspecified entries are defined as $0$):
\begin{enumerate}
    \item For any $\tau\in \Omega_U$, define $\@T_1(\tau,\tau)=\min\set{\mu_{\Psi_1,U}(\tau),\mu_{\Psi_3,U}(\tau)}$;
    \item {Complete the coupling table $\@T_2$ to form a coupling.}
\end{enumerate}
For any $\tau\in \Omega_U$,
let $\Psi_1^{\tau}=\left(\set{P},\set{Q},\+C_1^{\tau}\right)$, and $\Psi_3^{\tau} =\left(\set{P},\set{Q}, \+C_3^{\tau}\right)$. {Applying \Cref{prop-PDC-local-uniformity} with the condition that $2\-e\Delta\leq \abs{P}$}, we have
    \begin{align*}
        \left({\abs{\Omega_{\Psi_3^{\tau}}}-\abs{\Omega_{\Psi_1^{\tau}}}}\right)\big/{\abs{\Omega_{\Psi_3^{\tau}}}}&= \Pr[\sigma\sim \mu_{\Psi_3^{\tau}}]{\+C_1^{\tau}\setminus \+C_3^{\tau} \mbox{ are unsatisfied under } \sigma}\\
        &\leq \sum_{C\in \+C_1^{\tau}\setminus \+C_3^{\tau}} \Pr[\sigma\sim \mu_{\Psi_3^{\tau}}]{C \mbox{ is unsatisfied under } \sigma} \\
        &{\leq \abs{\+C_1\setminus \+C_2}\cdot \frac{1}{\abs{P}}\cdot \left(1+\frac{2\-e\Delta}{\abs{P}}\right)},
    \end{align*}
which implies
\begin{align*}
    \left(1- {\abs{\+C_1\setminus \+C_2}\cdot \frac{1}{\abs{P}}\cdot \left(1+\frac{2\-e\Delta}{\abs{P}}\right)} \right)\cdot \abs{\Omega_{\Psi_3^{\tau}}}\leq \abs{\Omega_{\Psi_1^{\tau}}}\leq \abs{\Omega_{\Psi_3^{\tau}}}.
\end{align*}
To simplify the notation, let $r=1- {\abs{\+C_1\setminus \+C_2}/{\abs{P}}\cdot \left(1+{2\-e\Delta}/{\abs{P}}\right)} $.
Consequently, we have 
\begin{align*}
     \mu_{\Psi_1,U}(\tau)&=\frac{\abs{\Omega_{\Psi_1^{\tau}}}}{\sum_{\tau'\in \Omega_U}\abs{\Omega_{\Psi_1^{\tau'}}}}\geq \frac{r\cdot\abs{\Omega_{\Psi_3^{\tau}}}}{ \sum_{\tau'\in \Omega_U}\abs{\Omega_{\Psi_3^{\tau'}}}}=r\cdot \mu_{\Psi_3,U}(\tau).
\end{align*}
Combining all these facts, it holds that 
\begin{equation}\label{eqn:coupling-error-U}
    \begin{aligned}
        \Pr[\left(\tau_1,\tau_2\right)\sim \@C_{U}]{\tau_1\neq \tau_2}&=1- \Pr[\left(\tau_1,\tau_2\right)\sim \@C_{U}]{\tau_1=\tau_2}\\
        &=1-\sum_{\tau\in \Omega_U} \min\set{\mu_{\Psi_1,U}(\tau),\mu_{\Psi_3,U}(\tau)}\\
        &\leq 1-\sum_{\tau\in \Omega_U} r\cdot\mu_{\Psi_3,U}(\tau) \\
        &=1-r={\abs{\+C_1\setminus \+C_2}\cdot \frac{1}{\abs{P}}\cdot \left(1+\frac{2\-e\Delta}{\abs{P}}\right)}\\
        &\leq {\frac{2\abs{\+C_1\setminus \+C_2}}{\abs{P}}}.
    \end{aligned}
\end{equation}

Given any $\tau_1,\tau_2\in \Omega_U$, we complete our coupling construction by specifying the coupling between $\mu_{\Psi_1^{\tau_1}}$ and $\mu_{\Psi_3^{\tau_2}}$. When $\tau_1=\tau_2=\tau$, let $\Psi_1^{\tau}=\left(\set{P},\set{Q},\+C_1^{\tau}\right)$, and $\Psi_3^{\tau} =\left(\set{P},\set{Q}, \+C_3^{\tau}\right)$. By definition, we have $\+C_3^{\tau}\subseteq \+C_1^{\tau}$ and $\+C_1^{\tau}\setminus \+C_3^{\tau}\subseteq \+C_1\setminus \+C_2$. 
According to~\Cref{lem:large-couple-combined}, there exists a coupling $\@C_P$ between $\mu_{\Psi_1^{\tau}}$ and $\mu_{\Psi_3^{\tau}}$ such that 
\begin{align}\label{eqn:equal-case-bound}
    \E[\left(\tau'_1,\tau'_2\right)\sim\@C_P]{\sum_{v\in P}\id{\tau'_1(v)\neq \tau'_2(v)}}\leq   \frac{17\Delta\abs{\+C_1\setminus \+C_2}}{\abs{P}-1}.
\end{align}
On the other hand, when $\tau_1\neq\tau_2$, let $\Psi_1^{\tau_1}=\left(\set{P},\set{Q},\+C_1^{\tau_1}\right)$, and $\Psi_3^{\tau_2} =\left(\set{P},\set{Q}, \+C_3^{\tau_2}\right)$. By definition, we have $\left(\+C_1^{\tau_1}\setminus \+C_3^{\tau_2}\right) \cup \left(\+C_1^{\tau_1}\setminus \+C_3^{\tau_2}\right) \subseteq \+C(P)$.
According to~\Cref{lem:large-couple-combined}, there exists a coupling $\@C_P'$ between $\mu_{\Psi_1^{\tau_1}}$ and $\mu_{\Psi_3^{\tau_2}}$ such that 
\begin{align}\label{eqn:inequal-case-bound}
    \E[\left(\tau'_1,\tau'_2\right)\sim\@C'_P]{\sum_{v\in P}\id{\tau'_1(v)\neq \tau'_2(v)}}\leq   \frac{17\Delta\abs{\+C(P)}}{\abs{P}-1}\leq \frac{17\abs{P}\Delta^2}{\abs{P}-1} .
\end{align}

We are now ready to specify the coupling $(\sigma'_{1},\sigma'_{2})\sim \@C_1$ between $\mu_{\Psi_1}$ and $\mu_{\Psi_3}$ as follows:
\begin{enumerate}
    \item Sample $\tau_1,\tau_2$ from the distribution $\@C_U$;
    \item If $\tau_1=\tau_2$, sample $\tau'_1,\tau'_2$ from $\@C_P$; Otherwise, sample $\tau'_1,\tau'_2$ from $\@C'_P$;
    \item $\sigma'_{1,U}\gets \tau_1$, $\sigma'_{1,P}\gets \tau'_1$, $\sigma'_{2,U}\gets \tau_2$, and $\sigma'_{2,P}\gets \tau'_2$. 
\end{enumerate}
One can verify that $\left(\sigma'_1,\sigma'_2\right)$ is a coupling. 
Furthermore,
\begin{align*}
    &\quad\E[(\sigma_1',\sigma_2')\sim \@C_1]{\sum_{v\in P}\id{\sigma'_1(v)\neq \sigma'_2(v)}}\\
     &\leq   \Pr{\tau_1=\tau_2}\cdot
    \E{\sum_{v\in P}\id{\tau'_1(v)\neq \tau'_2(v)}} + \Pr{\tau_1\neq \tau_2}\cdot  \E{\sum_{v\in P}\id{\tau'_1(v)\neq \tau'_2(v)}} \\
    &\leq   \Pr{\tau_1=\tau_2}\cdot\frac{17\Delta\abs{\+C_1\setminus \+C_2}}{\abs{P}-1} + \Pr{\tau_1\neq \tau_2}\cdot \frac{17\abs{P}\Delta^2}{\abs{P}-1} \tag{by~\eqref{eqn:equal-case-bound}, and~\eqref{eqn:inequal-case-bound}}\\
    &\leq \frac{17\Delta\abs{\+C_1\setminus \+C_2}}{\abs{P}-1}+{\frac{2\abs{\+C_1\setminus \+C_2}}{\abs{P}}}\cdot \frac{17\abs{P}\Delta^2}{\abs{P}-1}\leq \frac{51\Delta^2\abs{\+C_1\setminus \+C_2}}{\abs{P}-1}.\tag{by~\eqref{eqn:coupling-error-U}}
\end{align*}

By similar arguments, we can show that there exists a coupling $\@C_2$ between $\mu_{\Psi_2}$ and $\mu_{\Psi_3}$ such that 
\begin{align*}
    \E[(\sigma_1',\sigma_2')\sim \@C_2]{\sum_{v\in P}\id{\sigma'_1(v)\neq \sigma'_2(v)}}\leq \frac{51\Delta^2\abs{\+C_2\setminus \+C_1}}{\abs{P}-1}.
\end{align*}

Combining all these facts with {the triangle inequality}, there exists a coupling $\@C$ between $\mu_{\Psi_1}$ and $\mu_{\Psi_2}$ such that
\begin{align*}
    \E[(\sigma_1',\sigma_2')\sim \@C]{\sum_{v\in P}\id{\sigma'_1(v)\neq \sigma'_2(v)}}\leq \frac{51\Delta^2\abs{\+C_1\setminus \+C_2}}{\abs{P}-1}+ \frac{51\Delta^2\abs{\+C_2\setminus \+C_1}}{\abs{P}-1}\leq 100\Delta^2\abs{P}^{-1}\abs{\+C(S)\cap \+C(P)}.
\end{align*}

\end{proof}

\subsection*{Proof of~\Cref{lemma-deviation-witness-large}}\label{subsub:wit-prop}

    Here, we use the notations from ~\Cref{Alg:coupling} for discussion. 

    Let $\Phi_1=(\+P_1,\+Z_1,\+C), \Phi_2=(\+P_2,\+Z_2,\+C)$ be the formulas at~\Cref{line-coupling} in~\Cref{Alg:coupling} where $\+P_1=\+P_2$, and $\pcurrent[P]=(P_1,P_2,\cdots,P_\ell)$.  Note that $\pfinal\in \pcurrent$, and $\qtmpone(\pfinal)=\qtmptwo(\pfinal)$. Let $T=\pdiff$ and $S=\bigcup_{P\in T}P$.
    By definition, we have $\pfinal\cap S = \emptyset$, and that in the factorization of $(\pcurrent,\qtmpone,\cremain\setminus \+C(\pfinal))$ and $(\pcurrent,\qtmptwo,\cremain\setminus \+C(\pfinal))$, each variable $v\in S$ is not connected to any variable $u\notin S$. Combining all these facts with~\Cref{appendix-lemma-distance-large}, it holds that
    \begin{equation}\label{eqn:bound-by-intersection}
        \begin{aligned}
            \sum_{P\in \overline{\+R}}\E{\Dis(Y'_1,Y'_2,\pfinal)\mid \pfinal = P}&=\sum_{P\in \overline{\+R}}\E{\sum_{i\in [\ell]}\abs{\sigma_1(P_i)\setminus \sigma_2(P_i)} \bigg\vert \pfinal = P}\\
            &\leq  100\Delta^2\sum_{P\in \overline{\+R}} \E{ \abs{\+C(S)\cap \+C(\pfinal)}\big\vert \pfinal = P}\abs{P}^{-1}\\
            &\leq  100\Delta^2\alpha^{-1}\sum_{P\in \overline{\+R}} \E{ \abs{\+C(S)\cap \+C(\pfinal)}\big\vert \pfinal = P}.
        \end{aligned}
    \end{equation} 
     Additionally, we emphasize that for each $P\in \overline{\+R}$, we have $\pzeta = \+D'$ under the condition $\pfinal = P$. Thus, we have 
    \begin{align}\label{eqn-intersection-pd}
            \quad & \sum_{P\in \overline{\+R}} \E{\abs{\+C(S)\cap \+C(\pfinal)}\mid \pfinal = P}
            \leq \sum_{P\in \overline{\+R}} \E{\sum_{P'\in \pdiff} \abs{\+C(P')\cap \+C(\pfinal)} \bigg\vert \pfinal=P}  \\  
            = & \sum_{P\in \overline{\+R}} \E{\sum_{P'\in \+D'} \id{P'\in \pdiff}\cdot \abs{\+C(P')\cap \+C(\pfinal)}\bigg\vert \pfinal=P}
            =\sum_{P'\in \+D'} \sum_{P\in \overline{\+R}} \left(\Pr{P'\in \pdiff \mid \pfinal = P} \cdot \abs{\+C(P')\cap \+C(P)}\right)\notag\\
            \leq &\sum_{P'\in \+D'}\left(\max_{P\in \overline{\+R}}\Pr{P'\in \pdiff \mid \pfinal = P} \sum_{P\in \overline{\+R}} \abs{\+C(P')\cap \+C(P)}\right) \leq kd \sum_{P'\in \+D'} \abs{P'}\max_{P\in \overline{\+R}}\Pr{P'\in \pdiff \mid \pfinal = P}\notag,
    \end{align}
    where the last inequality follows from {the fact that each edge in $\+C(P')$ can be reused for $k$ times.}

    Moreover, according to \Cref{cor-witness-sequence-pinpdiff-large} it holds that
    \begin{align}\label{eqn-summation-witnesspath}
        \sum_{P'\in \+D'}\abs{P'}\max_{P\in \overline{\+R}}\Pr{P'\in \pdiff\mid \pfinal = P}\leq \sum_{P'\in \+D'} \abs{P'}\sum_{s\in \witnesswrtp{P'}} \max_{P\in \overline{\+R}} \Pr{{s \mbox{ occurs}}\mid \pfinal =P }.
    \end{align}
    Note that in~\eqref{eqn-summation-witnesspath}, the witness path can be summed more than once. The occurrence can be listed as the following cases:
    
    \smallskip
    \paragraph{\textbf{Case I}:} For each witness path $s=(v_0,v_1,\cdots,v_t)$ where $v_t\in \+D'\cap \psmall$, it could occur in the case where $v_t=P'$ or $(v_t,P')\in \+S^{\ast}(\+D'\cup \+C,\+D',1)$. If it occurs in the case $v_t=P'$, it contributes $$\abs{v_t}\cdot \max_{P\in \overline{\+R}} \Pr{{s \mbox{ occurs}}\mid \pfinal =P }$$ in~\eqref{eqn-summation-witnesspath}. Otherwise, we have $(v_t,P')\in \+S^{\ast}(\+D'\cup \+C,\psmall,1)$. By definition of $\+S^{\ast}(\+D',\psmall,1)$, it contributes  
    $$ \abs{v_t}\cdot d (k-1)\cdot \gamma^{\theta}\cdot  \max_{P\in \overline{\+R}} \Pr{{s \mbox{ occurs}}\mid \pfinal =P } $$
    where $\gamma^{\theta}$ is the upper bound of the size of $P'$.
    
    \smallskip
    \paragraph{\textbf{Case II}:} For each witness path $s=(v_0,v_1,\cdots,v_t)$ where $v_t\in \+D'\setminus \psmall$, it could occur in the case where $v_t=P'$. In this case, it contributes
    \begin{align*}
        \abs{v_t}\cdot  \max_{P\in \overline{\+R}} \Pr{{s \mbox{ occurs}}\mid \pfinal =P }.
    \end{align*}

    \smallskip
    \paragraph{\textbf{Case III}:} For each witness path $s=(v_0,v_1,\cdots,v_t)$ where $v_t\in \csmall$, it occurs in the case where$(v_t,P')\in \+S^{\ast}(\+D'\cup \+C,\+D',1)$.
    By definition, it contributes
    \begin{align*}
        k\cdot \gamma^\theta \cdot d\cdot k\cdot \gamma^\theta\cdot   \max_{P\in \overline{\+R}} \Pr{{s \mbox{ occurs}}\mid \pfinal =P }.
    \end{align*}
    Combining these facts with the definition of $f(s)$,~\eqref{eqn:bound-by-intersection},~\eqref{eqn-intersection-pd} and~\eqref{eqn-summation-witnesspath}, the lemma holds immediately.

\subsubsection{\textbf{\emph{The probability of the occurrence of the witness path}}}\label{subsub:wit-prop-cal}

In this section, we calculate the probability that a given witness sequence occurs and complete the proof of~\Cref{lemma-prob-witness}. 



Given any instance in~\Cref{Alg:coupling}, let $\+A^0$ be the random configuration of the coupling algorithm obtained after implementing the subroutine $\initialcouple(\cdot)$, and $\+A^{t}$ be the random configuration of the coupling algorithm obtained after implementing the $t$-th subroutine $\cpperm(\cdot)$ for each $t\geq 1$ in~\Cref{Alg:coupling}. Given the random configuration $\+A$, let $\ell(\+A,P)$ be the number of the assigned permutation sets among $\pzeta[P]$ until there exists some $P'\in \pdisc$ in the configuration $\+A$ and $\ell(P)=\-e^{-1}\abs{P}^{1-\theta}-(\-e\Delta)^{-1}\abs{P}^{1-2\theta}\ln \abs{P}$, for each $P\in \pcurrent\setminus \+I$ such that $\abs{P}>\gamma$. 

Given any w.s. $s=(v_0,v_1,\cdots,v_\ell)$, we define the following random variables for each integer $t\geq 0$. For any $P\in \pcurrent$ such that $V_C\cap \pzeta[P]\neq \emptyset$ and $\abs{P}>\gamma$,
\begin{align*}
    M_1(\+A^t,P) \triangleq \id{\left(v\notin \pdisci\right) \lor \left(\ell(\+A^t,P)\geq \ell(P)\right)}\cdot \prod_{i=\ell(\+A^t,P)}^{{\ell(P)}-1}\left(1-4\Delta\abs{P}^\theta/(\abs{P}^{1-\theta}-i)\right),
\end{align*}
where $v\in V_C\cap \pzeta[P]$. Intuitively, $M_1(\+A^t,P)$ captures the circumstance where $P$ is consecutively successfully coupled.
Furthermore, let
\begin{align*}
    M_2(\+A^t,P) \triangleq &15\Delta\abs{P}^{2\theta-1}\cdot \id{\left(v\notin \pdisci\right) \land \left(\vset \cap v=\emptyset\right)\land \left(\ell(\+A^t,P)<\ell(P)\right)} \\
    &+\id{\left(v\in \pdisci\right)\land \left(\ell(\+A^t,P)< \ell(P)\right)}.
\end{align*}
any $P\in \pcurrent$ such that $V_C\cap \pzeta[P]\neq \emptyset$ and $\abs{P}\leq \gamma$
\begin{align*}
     M_3(\+A^t) \triangleq \prod_{v\in V_C \cap \pzeta[P]:\abs{P}\leq \gamma}  \left(15\Delta p^{\theta}\abs{P}^{2\theta}\cdot \id{\left(v\notin \pdisci\right) \land \left(\vset \cap v=\emptyset\right)} +\id{v\in \pdisci}\right).
\end{align*}
Finally, we define
\begin{align*}
    M_4(\+A^t) \triangleq \prod_{v\in V_P}  \left(\min\left\{30k\Delta\abs{P}^{2\theta-1},1\right\}\cdot \id{\left(v\notin \pdisc \setminus \pdisci \right) \land \left(\vset \cap v=\emptyset\right)} +\id{v\in  \pdisc \setminus \pdisci}\right) .
\end{align*}
For each $C\in \csmall$, 
define
\begin{align*}
    M_5(\+A^t) \triangleq\prod_{C\in s \cap \csmall} {\^P_{\Phi_1}[\neg C]}\cdot \id{\mbox{the involved permutation sets not in $\pdisc$}}
\end{align*}
{Combining all these random variables, let}
\begin{align*}
    M(\+A^t) = \prod_{\abs{P}>\gamma: V_C \cap \pzeta[P]\neq \emptyset} \left(M_1(\+A^t,P) + M_2(\+A^t,P)\right)\cdot M_3(\+A^t)\cdot M_4(\+A^t)\cdot M_5(\+A^t).
\end{align*}

One can verify that the random variable $\+M$ is a {supermartingale} by the implementation of~\Cref{Alg:coupling}.
\begin{lemma}\label{lem-martingale}
    For each $t\geq 1$, we have $\E{M(\+A^t)}\leq \+A^{t-1}$.
\end{lemma}

The following inequality is used in the analysis.
\begin{lemma}\label{lemma-large-permutation-ell}
Given integers ${q>3b\geq 8}$, $t>0$, and $\Delta>0$, where 
$t\leq  {q}/\left({\-eb}\right) - {q\ln q}/\left({eb^2\Delta}\right)$,
we have 
\begin{equation*}
    \begin{aligned}
    \prod_{i = t}^{\lfloor q/b \rfloor }\left(1 - \frac{4b\Delta}{i}\right) \leq q^{-4}.
    \end{aligned}
\end{equation*}
\end{lemma}

\begin{proof}
Let $r = \lfloor q/b \rfloor $. We have
\begin{equation}\label{eq-bound-harmonic-number}
    \begin{aligned}
    \prod_{i = t}^{r}\left(1 - \frac{4b\Delta}{i}\right) \leq \prod_{i = t}^{r}\left(1 - \frac{1}{i}\right)^{4b\Delta} \leq \prod_{i = t}^{r}\exp\left(-\frac{4b\Delta}{i}\right) = \exp\left(-\sum_{i=t}^{r}\frac{4b\Delta}{i}\right),
    \end{aligned}
\end{equation}
where the second inequality holds by $1-x\leq \exp(-x)$ for any $x\in \^R$.
Moreover, let 
$H_{n}$ denote $\sum_{i=1}^{n}\frac{1}{i}$.
We have 
\[
\sum_{i=t}^{r}\frac{1}{i} = H_{r} - H_{t} \geq \ln \left(r + 1\right) -\ln t - 1,
\]
where the last inequality is by that 
\[
\ln (n+1)\leq H_{n} = \sum_{i=1}^{n}\frac{1}{i}\leq \ln n + 1.
\]
Combining with \eqref{eq-bound-harmonic-number},
we have
\begin{equation}\label{eq-bound-harmonic-number-t}
    \begin{aligned}
    \prod_{i = t}^{r}\left(1 - \frac{4b\Delta}{i}\right) \leq \exp\left(-\sum_{i=t}^{r}\frac{4b\Delta}{i}\right) \leq \left(\frac{\-e t}{r+1}\right)^{4b\Delta}\leq \left(\frac{\-e bt}{q}\right)^{4b\Delta}.
    \end{aligned}
\end{equation}
In addition, by 
\[t\leq \frac{q}{\-eb} - \frac{q\ln q}{eb^2\Delta},\]
we have 
\[
\left(\frac{\-e bt}{q}\right)^{4b\Delta} \leq \left(1 - \frac{\ln q}{b\Delta}\right)^{4b\Delta} \leq \exp\left(-\frac{4b\Delta\ln q }{b\Delta}\right)= q^{-4},
\]
where the second inequality holds by $1-x\leq \exp(-x)$ for any $x\in \^R$.
Combining with \eqref{eq-bound-harmonic-number-t},
the lemma is proved.
\end{proof}

We are now ready to complete the proof of~\Cref{lemma-prob-witness}.
\begin{proof}[Proof of~\Cref{lemma-prob-witness}]
    
According to~\Cref{lem-martingale}, it holds that 
\begin{equation}\label{eqn:originalquant}
\begin{aligned}
    &\quad \E{M(\+A^t)}\leq M(\+A^0)\\
    &\leq  p^{(t-\abs{ V_P}-\abs{ V_C})}\cdot \prod_{v\in V_P}\lambda_1(v)\prod_{v\in V_C}\left(15\Delta\abs{P}^{2\theta-1}+\prod_{i=1}^{{\ell(P)}}\left(\max\set{1-4\Delta\abs{P}^\theta/(\ell+1-i),0}\right)^{-1}\right)\\
    &\leq p^{(t-\abs{ V_P}-\abs{ V_C})}\cdot  \prod_{v\in V_P}\lambda_1(v)\prod_{v\in V_C}(15\Delta\abs{P}^{2\theta-1}+\abs{P}^{-4})\\
    &\leq p^{(t-\abs{ V_P}-\abs{ V_C})} \cdot \prod_{v\in V_P}\lambda_1(v)\prod_{v\in V_C}\lambda_2(v),
\end{aligned}  
\end{equation}
where the second inequality follows from~\Cref{lemma-large-permutation-ell}.

By definition, for any $P\in \pcurrent$ such that $\abs{P}>\gamma$ and $V_C \cap \pzeta[P]\neq \emptyset$, we have $$M_1(\+A^t,P)+M_2(\+A^t,P)\geq 1.$$ 
Thus, for any w.s. $s=(v_0,v_1,\cdots,v_t)$, if $s$ occurs, one can verify that
\begin{align}\label{eqn:failquant}
    M({\+A_\infty})&\geq p^{\theta\cdot (t-\abs{ V_P}-\abs{ V_C})}.
\end{align}

Combining~\eqref{eqn:originalquant}, ~\eqref{eqn:failquant} and Markov's inequality, we have
\begin{align*}
    \Pr{(v_0,\cdots,v_t) \text{ occurs }} &\leq \Pr{v_0\in \pdiffinit}\cdot \frac{\E{M({\+A_\infty})}}{p^{\theta\cdot (t-\abs{ V_P}-\abs{ V_C})}}\\
    &\leq p^{(1-\theta)(t-\abs{ V_P}-\abs{ V_C})}\cdot\Pr{v_0\in \pdiffinit} \prod_{v\in V_P}\lambda_1(v)\prod_{v\in V_C}\lambda_2(v),
\end{align*} which implies the lemma immediately.
\end{proof}

\subsubsection{\textbf{\emph{Refutation of long witness}}}\label{subsub:refutation}

In this section, we complete the proofs of Lemmas ~\ref{lemma-refutation-long-path} and \ref{lemma-refutation-long-path-small-pfinal}.

\subsection*{Witness sequence and suffix}
Define the set of suffices of a \emph{w.s.} as follows.
\begin{definition}[suffix of witness sequence]\label{def-weight-prefix-witness}
Let $\+T$ be $\+W$ or $\bigcup_{P\in \+R,P'\in \+D[P]}\+W[P,P']$.
Given a \emph{w.s.} $(v_0,v_1, \cdots,v_t)$ where $t\geq 0$, 
define
\begin{align*}
\mathsf{suf}(v_0,v_1, \cdots,v_t) &\triangleq\{(v_0, v_1, \cdots, v_t , \cdots, v_{t'})\in \+T\mid t'>t \},\\
\mathsf{suf}_1(v_0,v_1, \cdots,v_t) &\triangleq\{(v_0, v_1, \cdots, v_t,v_{t+1})\in \+T\mid v_{t+1}\in  V_{P}\},\\
\mathsf{suf}_2(v_0,v_1, \cdots,v_t) &\triangleq\{(v_0, v_1, \cdots, v_t,v_{t+1})\in \+T\mid v_{t+1}\not\in  V_{P}\}.
\end{align*}
We remark that $\mathsf{suf}_1(s) = \emptyset$ for some $s\in \+T$. 
Moreover, for each \emph{w.s.} $s$ and $i\in [2]$,
let $\mathsf{suf}_i^{+}(s) \triangleq \mathsf{suf}_i(s)\cup \{s\}$.
Similarly, let $\mathsf{suf}^{+}(s) \triangleq \mathsf{suf}(s)\cup \{s\}$.
Given any $s \in \+T$,
define
\begin{align}\label{eq-def-beta}
\widehat{\beta}(s) \triangleq \sum_{s'\in \mathsf{suf}^+(s)}\rho(s')\cdot \weightfunction(s'), \quad \beta(s) \triangleq \sum_{s'\in \mathsf{suf}^+(s)}\rho(s')\cdot f(s').
\end{align}
\end{definition}

This subsection is devoted to prove the following lemma.
\begin{lemma}\label{lemma-prob-prefix}
Consider the instance with $p,\eta,\Delta$ satisfying Condition \ref{cond-branching-decay}.
For any \emph{w.s.} $s = (v_0,v_1, \cdots,v_t)$ where $t\geq 1$ and $v_t\not\in V_P$, 
we have $\widehat{\beta}(s) \leq \rho(v_0,v_1, \cdots,v_{t-1})$.
\end{lemma}

We only prove \Cref{lemma-prob-prefix} for the case 
\begin{align}\label{eq-def-T}
\+T = \bigcup_{P\in \+R,P'\in \+D[P]}\+W[P,P']
\end{align}
The case $\+T = \+W$ can be proved similarly.
In the following subsection, we always assume \eqref{eq-def-T}.
Recall that $\pzeta$, $\csmall$, $\edgepermutation$, $S(\cdot,\cdot,\cdot)$ and $\+S^{\ast}(\cdot,\cdot,\cdot)$ are dependent on $\pfinal$.
However, given any \emph{w.s.} $(v_0,\cdots,v_t)\in \+T$, one can use these notations without specifying $\pfinal$, because $\pfinal$ is fixed by $v_t\subseteq \pfinal$.

By \Cref{def-weight-prefix-witness}, we have the following lemma.

\begin{lemma}\label{lemma-recursion-betas}
For any $s = (v_0,\cdots,v_t)\in \+T$ where $t\geq 1$, if $v_t\in \csmall$, we have
\begin{equation}\label{beta-vtincs}
  \widehat{\beta}(s) \leq \rho(s)\cdot \weightfunction(s)+ 
 \sum_{s'\in \mathsf{suf}_2(s)}\widehat{\beta}(s').    
\end{equation}
if $v_t\in \pzeta$, we have
\begin{equation}\label{beta-vtinz}
  \widehat{\beta}(s) \leq \sum_{s'\in \mathsf{suf}^{+}_1(s)} \left(\rho(s')\cdot \weightfunction(s')+\sum_{s''\in \mathsf{suf}_2(s')}\widehat{\beta}(s'')\right).    
\end{equation}
\end{lemma}
\begin{proof}
We only prove \eqref{beta-vtinz} here.
The proof of \eqref{beta-vtincs} is similar.
Given any $s=(v_0,v_1, \cdots,v_t)$ where $t\geq 1$, the following are all possibilities for the sequences in $\omega \in \mathsf{suf}(s)$:
\begin{itemize}
\item Either there is a $v_{t+1}$ satisfying $(v_{t},v_{t+1})\in \edgepermutation$ or not.
If there is such a $v_{t+1}$, let $x = v_{t+1}$,
otherwise, let $x = v_{t}$.
By letting $s'=(v_0,v_1,\cdots,x)$, we have 
$s'\in \mathsf{suf}^{+}_1(s)$.
\item Either $x$ has a next vertex in $\omega$ or not.
If there is no such a vertex, we have $\rho(\omega)f(\omega) = \rho(s')f(s')$.
Otherwise, $x$ has a next vertex $\tau$ in $\omega$.
If $x = v_t$, we have $(v_{t},v_{t+1})\not\in \edgepermutation$.
Thus, $(x,\tau)= (v_{t},v_{t+1})\not\in \edgepermutation$.
Otherwise, $x = v_{t+1}$. By $(v_{t},v_{t+1})\in \edgepermutation$ and $(v_{t},v_{t+2})\not\in \edgepermutation$ due to \Cref{def-witness}, we have $(x,\tau) = (v_{t+1},v_{t+2})\in \edgepermutation$.
In summary, we always have 
$(x,\tau)\not\in \edgepermutation$.
Let $s''=(v_0,v_1,\cdots,x,\tau)$.
By $(x,\tau)\not\in \edgepermutation$, we have 
$\tau \not\in V_P$.
Thus, $s''\in \mathsf{suf}_2(s')$.
Therefore, \[\sum_{\omega\in \mathsf{suf}^+(s'')}\rho(\omega)\cdot \weightfunction(\omega) = \widehat{\beta}(s'').\]
\end{itemize}
Combining with \Cref{def-weight-prefix-witness}, 
\eqref{beta-vtinz} is immediate.
\end{proof}

Recall the definitions of $\+S(\cdot,\cdot,\cdot)$ and $\+S^{\ast}(\cdot,\cdot,\cdot)$ in \Cref{def-prop-traj}.
Given any $P'\in \+P$, set $T$ and integer $i>0$, define
\begin{equation*}
\begin{aligned}
\+S_{P'}(P,T,i) = 
\begin{cases}
\+S(\{P\},T,i)  & \text{ if $P\in \pzeta$, $T\subseteq \pzeta\cup \csmall$  under $\pfinal = \+P'$,} \\
\emptyset& \text{ otherwise.}
\end{cases} 
\end{aligned}
\end{equation*}
Define $\+S^{\ast}_{P'}(P,T,i)$, $\+S_{P'}(C,T,i)$ and $\+S^{\ast}_{P'}(C,T,i)$ similarly.
We will omit $P'$ from these notations if $P'$ is clear from the context.
The following lemma is used in the proof of \Cref{lemma-prob-prefix}.
Its proof is immediate by the definition of $\+S(\cdot,\cdot,\cdot)$ and $\+S^{\ast}(\cdot,\cdot,\cdot)$.


\begin{lemma}\label{lemma-upperbound-pluss}
For any $P'\in \+P$, $P$, $C$, and integer $i>0$, we have
\begin{align*}
&\abs{\bigcup_{P'\in \+R}\+S_{P'}(P,\pzeta,i)}\leq \abs{P}(2dk)^{i} \sizeuperboundlp^{(i-1)\theta},\quad \abs{\bigcup_{P'\in \+R}\+S_{P'}(P,\csmall,i)}\leq \abs{P}d^{i+1} (2k)^{i}\sizeuperboundlp^{i\theta}.
\end{align*}
Therefore, 
\begin{align*}
&\abs{\bigcup_{P'\in \+R}\+S^{\ast}_{P'}(P,\pzeta,i)}\leq i\abs{P}(2dk)^{i} \sizeuperboundlp^{(i-1)\theta},\quad \abs{\bigcup_{P'\in \+R}\+S^{\ast}_{P'}(P,\csmall,i)}\leq i\abs{P}d^{i+1} (2k)^{i}\sizeuperboundlp^{i\theta}.
\end{align*}
\end{lemma}
\begin{proof}
We only prove the conclusion about $\+S_{P'}(P,\pzeta,i)$ here. The proofs for other conclusions are similarly.  
Consider the \propagate $(u = P, v_1, \cdots,v_{i+1}=P^{\ast})$ where $P^{\ast}\in \pzeta$.
By \Cref{def-prop-traj}, we have $v_1\in \+C(P)$, $(v_j,v_{j+1})\in \edgeconstraint$ for each $j\in [i-1]$, and $v_{i}\cap \+C(P^{\ast})$.
By $v_1\in \+C(P)$, we have there are at most $d\abs{P}$ choices for $v_1$. 
For any fixed $v_j$ where $j\in [i-1]$, by $(v_j,v_{j+1})\in \edgeconstraint$, we have there are at most $2dk\sizeuperbound$ choice for $v_{j+1}$.
Because the edge $v_j$ has at most $k$ variables.
Each variable can determine one permutation set in $\+D$ and another permutation set in $\+D'$. 
If this permutation set is in $\psmall$, it can connect to at most $d\sizeuperbound$ edges.
For any fixed $v_{i}$, by $v_{i}\in \+C(P^{\ast})$, we have there are at most $2k$ choice for $P^{\ast}$, because $P^{\ast}$ can be chosen from either $\+D$ and $\+D'$.
Thus, the conclusion about $\+S_{P'}(P,\pzeta,i)$ is immediate.
The lemma is proved.
\end{proof}

We have the following corollary by \Cref{lemma-upperbound-pluss}.
One can verify that it also holds for $\+T = \+W$.
\begin{corollary}\label{cor-upbound-presprime}
Given any \emph{w.s.} $s= (v_0,\cdots,v_t)$ where $t\geq 0$, $v_t\in \pzeta$ and 
$s'\in \pref^+_1(s)$  
we have 
\begin{align*}
    \abs{\pref_2(s)}\leq 20\abs{v_{t}}d^3k^2\sizeuperboundsquare,\quad \abs{\pref_2(s')}<_{q} 20\abs{v_{t}}d^3k^2\sizeuperboundsquare.
\end{align*}
In addition, if $v_t\in \psmall$, we have
\begin{align*}
    \abs{\pref_2(s)}\leq 20d^3k^2\sizeuperboundlp^{3\theta}.
\end{align*}
\end{corollary}
\begin{proof}
We only prove the upper bound of $\abs{\pref_2(s')}$ when $v_t\in \pzeta$,
the proofs for other cases are similar.
Given any $s'=(v_0,\cdots,v_{\ell-1},v_{\ell})\in \pref^+_1(s)$ where $\ell = t$ or $t+1$,
consider the \emph{w.s.} $s''=(v_0,\cdots,v_{\ell},v_{\ell+1})\in \pref_2(s')$.
By $s''\in \pref_2(s')$, we have $v_{\ell+1} \not\in V_{P}$.
Combined with \Cref{def-witness}, we have $(v_\ell,v_{\ell+1})\in \+S^{\ast}(v_\ell,\pzeta,2)$ for each $v_{\ell+1}\in \pzeta$ and  $(v_\ell,v_{\ell+1})\in \+S^{\ast}(v_\ell,\csmall,2)$ for each $v_{\ell+1}\in \csmall$.
Combining with \Cref{lemma-upperbound-pluss}, we have
\begin{align*}
    \abs{\pref_2(s')}\leq  \abs{\bigcup_{P'\in \+R}\+S^{\ast}_{P'}(v_{\ell},\pzeta,2)} +  \abs{\bigcup_{P'\in \+R}\+S^{\ast}_{P'}(v_{\ell},\csmall,2)}\leq 2\abs{v_{\ell}}(2dk)^2\sizeuperbound +3\abs{v_{\ell}}d^3(2k)^2\sizeuperboundsquare.
\end{align*}
Moreover, by $s'=(v_0,\cdots,v_{\ell-1},v_{\ell})\in \pref^+_1(s)$ where $\ell = t$ or $t+1$, we have 
either $v_{\ell} = v_t$ or $v_{\ell} = v_{t+1}$ and $(v_t,v_{t+1})\in \edgepermutation$.
Thus, $\abs{v_{\ell}} =_{q} \abs{v_t}$.
We have 
$\abs{\pref_2(s')}<_{q} 20\abs{v_{t}}d^3k^2\sizeuperboundsquare$.
The lemma is immediate.
\end{proof}


Now we can prove \Cref{lemma-prob-prefix}.
\begin{proof}[Proof of \Cref{lemma-prob-prefix}]
    We prove this lemma by structural induction. 
    For any $s = (v_1,v_2, \cdots,v_{t})\in \+T$,  by \Cref{def-witness}, we have $v_i\in \pzeta\cup \csmall$ and $v_i\neq v_j$ for any different $i,j\in [t]$.
    Thus, the possible choices of $v_i$ is bounded and $t$ is also bounded.
    We have $\abs{\+T}$ is also bounded.
    Therefore, for any $s \in \+T$, either $\mathsf{suf}(s) = \emptyset$ or   
    there exists some $s'\in \mathsf{suf}(s)$ where $\mathsf{suf}(s') = \emptyset$. 
    Given any $s = (v_0,v_1, \cdots,v_t)\in \+T$ where $t\geq 1$ and $v_t\not\in V_{P}$, 
     We first show that $\widehat{\beta}(s)\leq \rho(v_1,v_2, \cdots,v_{t-1})$ for any $s$ where $\mathsf{suf}(s) = \emptyset$.  
     By $v_t\not\in V_{P}$,  
    we have either $v_t \in V_C$ or $v_t \in \csmall$.   
    By \eqref{eq-def-beta} and $\mathsf{suf}(s) = \emptyset$,
    we have $\widehat{\beta}(s)= \rho(s)\cdot \weightfunction(s)$.
    In the following, we show $\widehat{\beta}(s)\leq \rho(v_1,v_2, \cdots,v_{t-1})$ for two different cases.
    
\smallskip
\paragraph{\textbf{Case I:} $v_t\in V_C$} We have
\begin{align*}
    \widehat{\beta}(s)&=\rho(s)\cdot \weightfunction(s)\leq \rho(v_0,\cdots,v_{t-1})\cdot \lambda_2(v_t) \cdot \weightfunction(s)\\
    &= \lambda_2(v_t) \abs{v_{t}}\weightconstant\rho(v_0,\cdots,v_{t-1})\leq \rho(v_0,\cdots,v_{t-1}),
\end{align*} where the last inequality holds by~\ref{cond-decay-small} and~\ref{cond-decay-large} in~\Cref{cond-branching-decay}.

\smallskip
\paragraph{\textbf{Case  II:} $v_t\in \csmall$}
We have
\begin{align*}
    \widehat{\beta}(s)&=\rho(s)\cdot \weightfunction(s)\leq \rho(v_0,\cdots,v_{t-1})\cdot p^{1-\theta}\cdot \weightfunction(s)\\
    &\leq \rho(v_0,\cdots,v_{t-1})\cdot p^{1-\theta} \sizeuperbound\weightconstant\leq \rho(v_0,\cdots,v_{t-1}),
\end{align*}where the last inequality holds by~\ref{cond-decay3} in~\Cref{cond-branching-decay}. 

For the induction step, given any \emph{w.s.} $s=(v_0,v_1,\cdots,v_t)$ where $t\geq 1$ and $v_t\not\in V_{P}$, we assume that $\widehat{\beta}(s')\leq \rho(v_0,v_1,\cdots,v_{\ell-1})$ for each $s'=(v_0,v_1,\cdots,v_{\ell})$ where $s'\in \pref(s)$ and $v_{\ell}\not\in V_{P}$.
Similar to the base case, we prove the claim for induction cases by discussing the component $v_t$ respectively.

\smallskip
\paragraph{\textbf{Case I:} $v_t\in V_C$} 
Assume $s = (v_0,\cdots,v_t)$. For any $s'=(v_0,\cdots,v_t,v_{t+1})\in \pref_1(s)$, we have $v_{t+1}\in V_{P}$.
Thus, $(v_t,v_{t+1})\in \edgepermutation$. 
Let $P_1\in \+P$, $P_2\in \+P'$ satisfy $v_t\subseteq P_2\subseteq P_1$.
If $P_1^{\theta} =_{q} \abs{v_t}$, we have $v_t\in \+D$.
Thus $v_{t+1}$ is also in $\+D$.
Therefore, there are at most $P_1^{1-\theta} =_{q} \abs{v_t}^{1/\theta-1}$ choices for $v_{t+1}$ for fixed $v_t$.
Otherwise, $P_2^{\theta} =_{q} \abs{v_t}$, we have $v_t\in \+D'$.
Thus $v_{t+1}$ is also in $\+D'$.
Therefore, there are at most $P_2^{1-\theta} =_{q} \abs{v_t}^{1/\theta-1}$ choices for $v_{t+1}$ for fixed $v_t$.
In summary, we always have
\begin{align*}
\abs{\{v_{t+1}\mid (v_t,v_{t+1})\in \edgepermutation\}}<_{q} {\abs{v_t}^{1/\theta-1}}.
\end{align*}
Thus, we have 
\begin{align*}
    \abs{\pref_1(s)}&\leq \abs{v_t}^{1/\theta-1}.
\end{align*}
Moreover, by \eqref{eq-def-rho} we have 
\[\rho(s') = \rho(s)\cdot \lambda_1(v_{t+1}) = \rho(s)\cdot \lambda_1(v_{t}).\] 
Thus, we have
\begin{align}\label{eq-upperbound-sumsprimepreplus1s}
\sum_{s'\in \mathsf{suf}^{+}_1(s)} \rho(s') \leq \rho(s) + \sum_{s'\in \mathsf{suf}_1(s)}
\rho(s)\cdot \lambda_1(v_t) \leq \rho(s)(1 + \abs{v_t}^{1/\theta-1}\cdot \lambda_1(v_t))
\end{align}
Therefore, we have 
\begin{align*}
\widehat{\beta}(s) & \leq \sum_{s'\in \mathsf{suf}^{+}_1(s)} \left(\rho(s')\cdot \weightfunction(s')+\sum_{s''\in \mathsf{suf}_2(s')}\widehat{\beta}(s'')\right)\tag{by~\eqref{beta-vtinz}}\\
&\leq \sum_{s'\in \mathsf{suf}^{+}_1(s)} \left(\rho(s')\cdot \weightfunction(s')+\sum_{s''\in \mathsf{suf}_2(s')}\rho(s')\right)\tag{by the inductive assumption}\\
&\leq \sum_{s'\in \mathsf{suf}^{+}_1(s)} \left(\rho(s')\cdot \left(\abs{v_{t}}\weightconstant +20\abs{v_{t}}d^3 k^2 \sizeuperboundsquare\right)\right)\tag{by~\eqref{eq-def-weightconstant} and \Cref{cor-upbound-presprime}}\\
&<_{q} \sum_{s'\in \mathsf{suf}^{+}_1(s)} \left(2\abs{v_{t}}\weightconstant\rho(s')\right)\\
&<_{q} 2\abs{v_{t}}\weightconstant(1 + \abs{v_t}^{1/\theta-1}\lambda_1(v_t))\rho(s)\tag{by~\eqref{eq-upperbound-sumsprimepreplus1s}}\\
&<_{q} 2\abs{v_{t}}\weightconstant(1 + \abs{v_t}^{1/\theta-1}\lambda_1(v_t)) \lambda_2(v_t)\rho((v_0,\cdots,v_{t-1}))\\
&\leq \rho(v_0,\cdots,v_{t-1}).\tag{by~\ref{cond-decay-small}, \ref{cond-decay-large} in~\Cref{cond-branching-decay}}
\end{align*}

\smallskip
\paragraph{\textbf{Case II:} $v_t\in \csmall$} We have 
\begin{align*}
\widehat{\beta}(s) & \leq \rho(s)\cdot \weightfunction(s)+ 
 \sum_{s'\in \mathsf{suf}_2(s)}\widehat{\beta}(s')\tag{by~\eqref{beta-vtincs}}\\
&\leq \rho(s)\cdot \weightfunction(s)+ 
 \sum_{s'\in \mathsf{suf}_2(s)}\rho(s)\tag{by the inductive assumption}\\
&\leq \rho(s)\cdot \left(\sizeuperbound\weightconstant +20d^3 k^2 \sizeuperboundlp^{3\theta}\right)\tag{by~\eqref{eq-def-weightconstant} and \Cref{cor-upbound-presprime}}\\
&<_{q} \rho(v_0,\cdots,v_{t-1})\cdot   2p^{1-\theta}\sizeuperbound\weightconstant\tag{by~\eqref{eq-def-rho}} \\
&\leq \rho(v_0,\cdots,v_{t-1}).\tag{by~\ref{cond-decay3} in~\Cref{cond-branching-decay}}
\end{align*}

Combining the proofs for Cases I and II, the induction step is finished. The lemma is immediate.
\end{proof}

\subsection*{Proofs of Lemmas \ref{lemma-refutation-long-path} and \ref{lemma-refutation-long-path-small-pfinal}}
The following two lemmas are used in our proofs.
\begin{lemma}\label{lemma-upperbound-beta-init-small-case}
Assume $\abs{\pinit} \leq \sizeuperboundlp$.
We have
\begin{equation*}
    \begin{aligned}
    \sum_{s\in \+W} \rho(s)\cdot f(s) \cdot \id{\abs{s}> 1}<_{q} (300\Delta)^{-1}kd^2\sizeuperbound\abs{\pinit}^{\theta} \E{\abs{\pdiffinit}}.
    \end{aligned}
\end{equation*}
\end{lemma}
\begin{proof}
By \Cref{def-witness}, we have $v_0\in \pzeta[\pinit]$
and $v_{1} \not\in V_{P}$.
Combined with \Cref{def-weight-prefix-witness},
we have $(v_0,v_1)\in \pref_2((v_0))$.
Meanwhile, by $v_0\in \pzeta[\pinit]$ and $\abs{\pinit} < \sizeuperboundlp$,
we have $v_0\in \psmall \subseteq \pzeta$.
Combined with \Cref{cor-upbound-presprime}, 
we have $\abs{\pref_2((v_0))} <_{q} 20\abs{v_0}d^3 k^2 \sizeuperboundsquare$.
Thus,
\begin{equation}
\begin{aligned}\label{eqn-bound-sum-small-per}
    \abs{\set{v_1\mid (v_0,v_1)\in \+W}} = \abs{\pref_2((v_0))}
    <_{q} 20\abs{v_0}d^3 k^2 \sizeuperboundsquare. 
\end{aligned}
\end{equation}
Thus, we have
\begin{align*}
     \sum_{s\in \+W} \rho(s)\cdot f(s) \cdot \id{\abs{s}> 1}
     &=\sum_{(v_0,v_1)\in \+W}\beta(v_0,v_1)\\
     &\leq \sum_{(v_0,v_1)\in \+W}\widehat{\beta}(v_0,v_1)\cdot k^2d^2\sizeuperbound /\weightconstant\tag{by \Cref{def-weight-function}}\\
     &\leq \sum_{(v_0,v_1)\in \+W}\rho((v_0))\cdot k^2d^2\sizeuperbound /\weightconstant\tag{by $ v_{1} \not\in V_{P}$ and \Cref{lemma-prob-prefix}}\\
     &\leq  \sum_{v_0\in \pzeta[\pinit]} \rho((v_0))\cdot 20\abs{v_0}d^3 k^2 \sizeuperboundsquare \cdot k^2d^2\sizeuperbound /\weightconstant\tag{by~\eqref{eqn-bound-sum-small-per}}\\
     &\leq  (300\Delta)^{-1}kd^2\sizeuperbound\sum_{v_0\in \pzeta[\pinit]} \rho((v_0))\abs{\pinit}^{\theta}\tag{by~\eqref{eq-def-weightconstant}}\\
     &<_{q}   (300\Delta)^{-1}kd^2\sizeuperbound\sum_{v_0\in \pzeta[\pinit]} \Pr{v_0\in \pdiffinit} \cdot\abs{\pinit}^{\theta}\tag{by~\eqref{eq-def-rho}}\\
     &=_{q} (300\Delta)^{-1}kd^2\sizeuperbound\abs{\pinit}^{\theta} \E{\abs{\pdiffinit}}.
\end{align*}
\end{proof}

\begin{lemma}\label{lemma-upperbound-beta-init-large-case}
Assume $\abs{\pinit} > \sizeuperboundlp$.
We have
\begin{equation*}
\begin{aligned}
\sum_{s\in \+W} \rho(s)f(s) <_{q} 2k^2d^2\sizeuperbound\abs{\pinit}^{\theta}\E{\abs{\pdiffinit}}.
\end{aligned}
\end{equation*}
\end{lemma}

The following lemma is also used in the proof of \Cref{lemma-upperbound-beta-init-large-case}.
Its proof is similar to that of \Cref{lemma-recursion-betas}.
\begin{lemma}\label{lemma-recursion-betas-initial}
Assume $\abs{\pinit} \geq \sizeuperboundlp$.
Let $T$ denote the set $\{s\in \+W\mid \abs{s} = 1\}$.
Then 
\begin{equation*}
\begin{aligned}
\sum_{s\in \+W} \rho(s)\weightfunction(s) \leq \sum_{s\in T}\left(\rho(s)\weightfunction(s)+\sum_{\tau\in \mathsf{suf}_2(s)}\widehat{\beta}(\tau)\right).
\end{aligned}
\end{equation*}
\end{lemma}

\begin{proof}[Proof of \Cref{lemma-upperbound-beta-init-large-case}]
Let $T$ denote the set $\{s\in \+W\mid \abs{s} = 1\}$.
By Lemmas \ref{lemma-recursion-betas-initial} and ~\ref{lemma-prob-prefix}, we have
\begin{equation*}
\begin{aligned}
\sum_{s\in \+W} \rho(s)\weightfunction(s) &\leq \sum_{s\in T}\left(\rho(s)\weightfunction(s)+\sum_{\tau\in \mathsf{suf}_2(s)}\widehat{\beta}(\tau)\right)
\leq \sum_{s\in T}\left(\rho(s)\weightfunction(s)+\sum_{\tau\in \mathsf{suf}_2(s)}\rho(s)\right)\\
&= \sum_{(v_0)\in T}\left(\rho((v_0))\weightfunction((v_0))+\sum_{\tau\in \mathsf{suf}_2((v_0))}\rho((v_0))\right).
\end{aligned}
\end{equation*}
In addition, by ~\eqref{eq-def-fs-hat} we have 
\begin{equation*}
\begin{aligned}
\rho((v_0))\weightfunction((v_0))\leq \abs{v_0}\weightconstant\rho((v_0)).
\end{aligned}
\end{equation*}
Meanwhile, by~\Cref{cor-upbound-presprime}, we have
\begin{equation*}
\begin{aligned}
\sum_{\tau\in \mathsf{suf}_2((v_0))}\rho((v_0)) \leq 20\abs{v_{0}}d^3k^2\sizeuperboundsquare\rho((v_0)).
\end{aligned}
\end{equation*}
Combining the above three inequalities with \eqref{eq-def-weightconstant}, 
we have
\begin{equation*}
\begin{aligned}
&\sum_{s\in \+W} \rho(s)\weightfunction(s)
\leq &\sum_{(v_0)\in T}2\abs{v_0}\weightconstant\rho((v_0)).
\end{aligned}
\end{equation*}
Combined with \Cref{def-weight-function}, we have
\begin{equation*}
\begin{aligned}
&\sum_{s\in \+W} \rho(s)f(s)
\leq &\sum_{(v_0)\in T}2\abs{v_0}k^2d^2\sizeuperbound\rho((v_0)).
\end{aligned}
\end{equation*}
Combined with \eqref{eq-def-rho}, we have
\[
\sum_{s\in \+W} \rho(s)f(s) \leq 2k^2d^2\sizeuperbound\sum_{(v_0)\in T}\abs{v_0}\rho((v_0)) \leq 2k^2d^2\sizeuperbound\sum_{v_0\in \pzeta[\pinit]} \abs{v_0}\Pr{v_0\in \pdiffinit} <_{q}2k^2d^2\sizeuperbound\abs{\pinit}^{\theta}\E{\abs{\pdiffinit}}.
\]
Thus, the lemma is immediate by above two inequalities.
\end{proof}

We are now ready to complete the proof of~\Cref{lemma-refutation-long-path}.

\begin{proof}[Proof of~\Cref{lemma-refutation-long-path}]
If $\abs{\pinit} < \sizeuperboundlp$,
by ~\eqref{eq-def-fs} and \eqref{eq-def-rho}, we have
\begin{equation}\label{eq-sum-rho-fs-sequalone}
\begin{aligned}
     \sum_{s\in \+W} \rho(s)\cdot f(s)\cdot \id{\abs{s}= 1}&=  \sum_{v_0\in \pzeta[\pinit]} \rho((v_0))\cdot f(v_0)= \sum_{v_0\in \pzeta[\pinit]} \rho((v_0))\cdot  \abs{v_0}k^2d^2\sizeuperbound\\
     &\leq \sum_{v_0\in \pzeta[\pinit]} \rho((v_0))\cdot \abs{\pinit}^{\theta}k^2d^2\sizeuperbound 
     \leq \abs{\pinit}^{\theta}k^2d^2\sizeuperbound \cdot \E{\abs{\pdiffinit}}.
\end{aligned}
\end{equation}
Thus, we have
\begin{align*}
     \sum_{s\in \+W} \rho(s)\cdot f(s)&=\sum_{s\in \+W} \rho(s)\cdot f(s)\cdot \id{\abs{s}= 1}+ \sum_{s\in \+W} \rho(s)\cdot f(s)\cdot \id{\abs{s}\geq 1}\\
     &<_{q} 2k^2d^2\sizeuperbound\abs{\pinit}^{\theta}\E{\abs{\pdiffinit}} \tag{by \eqref{eq-sum-rho-fs-sequalone} and \Cref{lemma-upperbound-beta-init-small-case}}.
\end{align*}
If $\abs{\pinit} \geq \sizeuperboundlp$, by \Cref{lemma-upperbound-beta-init-large-case} we also have
\begin{align*}
\sum_{s\in \+W} \rho(s)f(s) &<_{q} 2k^2d^2\sizeuperbound\abs{\pinit}^{\theta}\E{\abs{\pdiffinit}}.
\end{align*}
Therefore, for any $\pinit$ we have
\begin{align*}
     \sum_{s\in \+W} \rho(s)\cdot f(s)
     &<_{q} 2k^2d^2\sizeuperbound\abs{\pinit}^{\theta}\E{\abs{\pdiffinit}} \notag\\
     &<_{q} 2(18k+2)k^2d^2\sizeuperbound\Delta \tag{by \Cref{condition-initial} and \Cref{lemma-suc-probability-alg-initial-coupling}}\\
     &<_{q} \sizeuperboundpfinal(500kd\Delta^2)^{-1}. \tag{by~\ref{cond-decay4} in~\Cref{cond-branching-decay}}
\end{align*}
The lemma is proved.
\end{proof}

In the next, we complete the proof of~\Cref{lemma-refutation-long-path-small-pfinal}.
Similar to \Cref{lemma-upperbound-beta-init-small-case}, one can also prove the following lemma.
\begin{lemma}\label{lemma-upperbound-beta-init-small}
We have
\begin{equation*}
    \begin{aligned}
    \sum_{P\in \+R}\sum_{P'\in \+P'[P]}\sum_{s\in \+{W}[P,P']}\rho(s) f(s)\id{\abs{s}> 1}  <_{q} (300\Delta)^{-1}kd^2\sizeuperbound\abs{\pinit}^{\theta}\E{\abs{\pdiffinit}}.
    \end{aligned}
\end{equation*}
\end{lemma}

\begin{proof}[Proof of~\Cref{lemma-refutation-long-path-small-pfinal}]
Given any $P\in \+R$, $P'\in \+P'[P]$  and $s\in \+W[P,P']$, 
we have $\abs{s} \geq 1$.
Because the first element $v_0$ in $s$ satisfies $v_0 \subseteq \pinit $.
Moreover, by $\+R = \{P\in \+P\mid \abs{P} \leq \sizeuperboundpfinal, \pinit\not \subseteq P\}$, $P\in \+R$ and $P'\in \+P'[P]$, we have $v_0\neq P'$. 
Combining with $s\in \+W[P,P']$,
we have $\abs{s} \geq 1$.
Therefore, we have
\begin{align*}
    \sum_{P\in \+R}\sum_{P'\in \+P'[P]}\sum_{s\in \+{W}[P,P']}\rho(s)\abs{P'} &=\quad \sum_{P\in \+R}\sum_{P'\in \+P'[P]}\sum_{s\in \+{W}[P,P']}\rho(s)\abs{P'}\id{\abs{s}> 1}\\
    &= \sum_{P\in \+R}\sum_{P'\in \+P'[P]}\sum_{s\in \+{W}[P,P']}\rho(s) f(s) \id{\abs{s}> 1}  \abs{P'}/f(s) \\
    &\leq \sum_{P\in \+R}\sum_{P'\in \+P'[P]}\sum_{s\in \+{W}[P,P']}\rho(s) f(s)  \id{\abs{s}> 1} /(k^2d^2\sizeuperbound)\tag{by \eqref{eq-def-fs}}\\
    &\leq (300k\Delta)^{-1}\abs{\pinit}^{\theta}\E{\abs{\pdiffinit}}\tag{by \Cref{lemma-upperbound-beta-init-small}}\\
    &<_{q} (300k\Delta)^{-1}(36k+4)\Delta \tag{by \Cref{condition-initial} and \Cref{lemma-suc-probability-alg-initial-coupling}}\\
    &<_{q} \frac{1}{5}\tag{by $q_{\min}$ is large}.
\end{align*}
The lemma is immediate.
\end{proof}

\section{Missing proofs in section~\ref{sec:mainanalysis}}\label{subsub:mainanalysis}

\subsection*{Proof of~\Cref{lem-sample-witness}}

    Given that the event $\+B_C$ occurs, let $G$ be the dependency graph, and $G'$ be the lopsidependency graph of formula $\Phi_C=(\+P_C,\+Q_C,\+C_C)$, respectively.
    Note that both $G$ and $G'$ are the subgraph of $G_{\Phi}^{\!{lop}}$ and $G_{\Phi}^{\!{dep}}[P']$ respectively, and $G'$ is the subgraph of $G$.
    For any collection of constraints $S\subseteq \+C$, let $\Gamma^+_{\!{dep}}(S)$ be the extended neighbors of $S$, i.e.,
    \begin{align*}
        \Gamma^+_{\!{dep}}(S)\triangleq \set{C\in V(G_{\Phi}^{\!{dep}}[P'])\mid C\in S \mbox{ or there exists $C'\in S$ such that $\set{C,C'}\in E(G_{\Phi}^{\!{dep}}[P'])$} }.
    \end{align*}
    Similarly, we can define the notation $\Gamma^+_{\!{lop}}(S)$. 
    
    Given that $\abs{\+C_C}\geq \Delta\log(n/\delta)$, we construct a collection of constraints $\set{C_1,\dots,C_{\ell'}}$ in a greedy way. 
    We first initialize $T_1=\set{C}$ and then repeat the following operations: at the $i$-th step, select an arbitrary constraint $C_i$ in $\+C_{C}\setminus \Gamma^+_{\!{dep}}(T_{i-1})$ such that $\!{dist}_{G_{\Phi}^{\!{dep}}[P']}(C_i,T_{i-1})=2$, and update $T_i\gets T_{i-1}\cup \set{C_i}$. Let $\ell'$ be the step that $\Gamma^+_{\!{dep}}(T_{\ell'})=\+C_C$.
    We can see that $T_{\ell'}=\set{C,C_1,\dots,C_{\ell'}}$ form a connected component in the power dependency graph $(G_{\Phi}^{\!{dep}}[P'])^2$. Moreover, $T_{\ell'}=\set{C,C_1,\dots,C_{\ell'}}$ is also an independent set in the $G_{\Phi}^{\!{dep}}[P']$, which also an independent set in lopsidependency graph $G_{\Phi}^{\!{lop}}$ immediately. 
    
    We again continue to construct a maximal independent set in the lopsidependency graph $G_{\Phi}^{\!{lop}}$ containing the constraints in $T_{\ell'}$ in a greedy way. Let $S_1=T_{\ell'}$ and repeat the following: at the $i$-th step, select an arbitrary constraint $C_i$ in $\+C_C\setminus \Gamma^+_{\!{lop}}(S_{i-1})$, and update $S_i\gets S_{i-1}\cup \set{C_i}$. Let $\set{C_1,\cdots,C_{\ell''}}$ be the final set constructed, and $\ell'\geq \log(n/\delta)$ due to its maximality. The proof is complete by selecting the subset containing $C$ in $\set{C_1,\cdots,C_{\ell''}}$ with size $\ell$.

\subsection*{Proof of~\Cref{lemma-uncut-iteration}}

    If all constraints in $\+C'$ are bad in $(\+P',\+Q'_t,\+C)$, then for each constraint $C=(v_1\neq c_1) \lor \cdots \lor (v_{\ell}\neq c_{\ell})\in \+C'$, we have $c_i\in \+Q'_t(P')$ for each $i\in [\ell]$ where $v_i\in P'\in \+P'$.

    Each time a permutation $P\in \+P$ is resampled, let $\set{P_1, P_2, \cdots, P_m}\in \+P'[P]$ be the permutations contains some variables in $\vbl(C)$ for $C\in \+C'$ and $k_i$ be the number of those variables in $P_i$ for each $i\in [m]$. W.L.O.G., we can always reorder these permutations such that $k_1 \ge k_2 \ge \cdots \ge k_m \ge 1$. We first calculate the probability of the event that the forbidden values of the constraints $\+C'$ lie in the corresponding permutations, in the underlying distribution $\^P$. Sampling $P$ is equivalent to consecutively sampling $P_1, P_2, \cdots, P_m$. Let $B_i$ be the event that when $P_i$ is sampled, {the forbidden values of all constraints $\+C'$ lie in $P_i$.} It holds that 
    $$
    \^P[B_i] \le
    \frac{ \abs{P_i}^{\underline{k_i}}}{( \abs{P} - \abs{P_1} -\cdots -\abs{P_{i-1}} )^{ \underline{k_i}}}
    \le \left( \frac{\abs{P_i}} {\abs{P} - \abs{P_1} - \cdots - \abs{P_{i-1}}} \right)^{k_i}.
    $$
    {Examing our algorithm carefully, the probability of the event $B_i$ in~\Cref{Alg:MCMC} can be regarded as the deviation between the underlying distribution $\^P$.} According to~\Cref{lem-violation-decomposed}, we have $P_{\Psi}\leq \min\set{p_\Phi^{\eta},1/L}$ for any $\Psi =(\+P',\+Q'_t,\+C)$. By {$8\-ep_\Phi^{\eta}\Delta\leq 1$}, ${8\-e\Delta\leq L}$, and~\Cref{lem:small-marginal-lb} , we have
    $$
    \Pr{B_i} <_{q} \^P[B_i] \leq \left( \frac{\abs{P_i}} {\abs{P} - \abs{P_1} - \cdots - \abs{P_{i-1}}} \right)^{k_i},
    $$ where $\Pr{\cdot}$ is the randomness in~\Cref{Alg:MCMC}.
    
    Let $k_P \triangleq \sum_{i=1}^m k_i$ be the number of forbidden values of $\+C'$ on $P$ and $m'=\floor{0.999m}$. In our analysis, we only exploit the probability of the first $m'=\ceil{0.999m}$ terms because the last $\floor{0.001m}$ terms are too large. Specifically, in the first $m'$ terms, the number of literals we exploit is $k'_P \triangleq \sum_{i=1}^{m'} k_i \ge 0.999 k_P$ because $k_1 \ge k_2 \ge \cdots \ge k_m$, and the number of unused sub-permutations is at least $\ell - \ceil{0.001m} \ge \floor{0.001\ell}$ by the decomposition specified in~\Cref{condition-state-compression}. Since $P_1,P_2, \cdots P_m$ is evenly decomposed, each denominator $\abs{P} - \abs{P_1} - \cdots - \abs{P_{i-1}}>_{q} \abs{P}$ for each $i\in [m']$. Moreover, we have $\abs{P_i}<_{q} \abs{P}^{\eta}$.
    Therefore,
    \begin{align*}
        \Pr{\bigwedge_{i=1}^m B_i}
        &<_{q}  \left( \frac{\abs{P_1}} {\abs{P}} \right)^{k_1} \left( \frac{\abs{P_2}} {\abs{P} - \abs{P_1}} \right)^{k_2}  \cdots \cdots \left( \frac{\abs{P_{m'}}} {\abs{P} - \abs{P_1} - \cdots - \abs{P_{m'-1}}} \right)^{k_{m'}}\\
        &<_{q} \left( \frac{\abs{P}^{\eta}} {\abs{P}} \right)^{k_1} \left( \frac{\abs{P}^{\eta}} {\abs{P}} \right)^{k_2}  \cdots \cdots \left( \frac{\abs{P}^{\eta}} {\abs{P}} \right)^{k_{m'}}\leq  \left( \frac{\abs{P}^{\eta}} {\abs{P}} \right)^{0.999 k_P} \leq \^P_\Phi[\neg C]^{0.999(1-\eta)}. 
    \end{align*}
    The lemma follows immediately by applying the above bound to each constraint in $\+C'$.

\end{document}